\documentclass[11pt]{now}
\usepackage{amsmath}
\usepackage{amssymb}
\usepackage{graphicx}
\usepackage{Steven}

\newtheorem{theorem}{Theorem}[chapter]
\newtheorem{proposition}{Proposition}[chapter]
\newtheorem{definition}{Definition}[chapter]
\newtheorem{assumption}{Assumption}[chapter]
\newtheorem{fact}{Fact}[chapter]
\newtheorem{remark}{Remark}[chapter]
\newtheorem{corollary}{Corollary}[chapter]
\newtheorem{lemma}{Lemma}[chapter]

\newcommand{\adhoc}{{\em ad hoc} }
\newcommand{\ie}{{\em i.e.}}
\newcommand{\eg}{{\em e.g.}}
\newcommand{\sinr}{\mathsf{sinr}}
\newcommand{\sir}{\mathsf{sir}}
\newcommand{\snr}{\mathrm{snr}}
\newcommand{\ebno}{\frac{\varepsilon_b}{\eta}}
\newcommand{\nr}{{n_{\mathrm{\scriptscriptstyle R}}}}
\newcommand{\nt}{{n_{\mathrm{\scriptscriptstyle T}}}}
\newcommand{\Sisf}{\mathsf{\Sigma}}
\newcommand{\Lisf}{\mathsf{\Lambda}}
\newcommand{\arc}{{\rm arc}}

\begin{document}

\isbn{xxxxxxxxxxx}

\DOI{xxxxxx}


\abstract{Transmission capacity (TC) is a performance metric for wireless networks that measures the spatial intensity of successful transmissions per unit area, subject to a constraint on the permissible outage probability (where outage occurs when the SINR at a receiver is below a threshold).  This volume gives a unified treatment of the TC framework that has been developed by the authors and their collaborators over the past decade.  The mathematical framework underlying the analysis (reviewed in Ch.\ 2) is stochastic geometry: Poisson point processes model the locations of interferers, and (stable) shot noise processes represent the aggregate interference seen at a receiver.  Ch.\ 3 presents TC results (exact, asymptotic, and bounds) on a simple model in order to illustrate a key strength of the framework: analytical tractability yields explicit performance dependence upon key model parameters.  Ch.\ 4 presents enhancements to this basic model --- channel fading, variable link distances, and multi-hop.  Ch.\ 5 presents four network design case studies well-suited to TC: $i)$ spectrum management, $ii)$ interference cancellation, $iii)$ signal threshold transmission scheduling, and $iv)$ power control.  Ch.\ 6 studies the TC when nodes have multiple antennas, which provides a contrast vs.\ classical results that ignore interference.}

\articletitle{Transmission capacity of wireless networks}

\authorname1{Steven Weber}
\affiliation1{}
\author1address2ndline{Drexel University}
\author1email{sweber@coe.drexel.edu}

\authorname2{Jeffrey G. Andrews}
\affiliation2{}
\author2address2ndline{The University of Texas at Austin}
\author2email{jandrews@ece.utexas.edu}

\relax \input sample.jtl 
\volume{xx}
\issue{xx}
\copyrightowner{xxxxxxxxx}
\pubyear{xxxx}

\maketitle

\cleardoublepage \pagenumbering{roman}

\tableofcontents

\clearpage

\setcounter{page}{0}
\pagenumbering{arabic}

%
%
\chapter{Introduction and preliminaries}
\label{cha:int}

Wireless networks are becoming ever more pervasive, and the correspondingly denser deployments make interference management and spatial reuse of spectrum defining aspects of wireless network design.  Understanding the fundamentals of the performance and behavior of such networks is an important theoretical endeavor, but one with only limited success to date.  Information theoretic approaches, well-summarized by \cite{ElGKim2012}, have been most successfull when applied to small isolated networks, where background interference and spatial reuse are not considered. Large network approaches, typified by transport capacity scaling laws \cite{XueKum2006}, have given considerable insight into scaling laws, but are generally unable to quantify the relative merits of candidate design choices or provide a tractable approach to analysis for spatial reuse or the SINR statistics.  Our hope for the transmission capacity framework has been to develop a tractable approach to large network throughput analysis, that while falling short of information theory's ideals of inviolate upper bounds, nevertheless provides a rigorous and flexible approach to the same sort of questions, and ultimately provides the types of broad design insights that information theory has been able to achieve for small networks.

\section{Motivation and assumptions}
\label{sec:motass}

This monograph presents a framework for computing the outage probability (OP) and transmission capacity (TC) \cite{WebYan2005,WebAnd2010} in a wireless network.  The OP is defined as the probability that a ``typical'' transmission attempt fails (is in outage) at the intended receiver, where outage occurs when the signal to interference plus noise ratio (SINR) at the receiver is below a threshold.  Basing outage on the SINR, it is assumed that interference is treated as noise.  The TC is defined as the maximum average number of concurrent successful transmissions per unit area taking place in the network, subject to a constraint on OP.  The OP constraint may be thought of as a reliability and/or quality of service (QoS) parameter --- strict requirements on the fraction of failed transmissions result in low spatial reuse, low area spectral efficiency (ASE, measured in $\mathrm{bps/Hz}$ per unit area), and thus lower TC, while relaxing the outage requirement improves, up to a point, the ASE and thus TC.  Viewing the OP as a (strictly increasing) function of the intensity of attempted transmissions, the TC is computed by inverting this function for the transmission intensity at the target OP.

Note we use the word {\em capacity} in a distinctly different manner from its information--theoretic sense, \ie, Shannon capacity: the TC framework typically treats interference as noise\footnote{see \S\ref{sec:intcan} and the results of Chapter 6 as an exception: even here though the background (uncancelled) interference is then treated as noise.} while Shannon theory does not, and TC measures capacity in a spatial sense, while Shannon theory does not.   The {\em capacity} in TC is also distinct from the transport capacity of \cite{GupKum2000}, defined as the maximum weighted sum rate of communication over all pairs of nodes, where each pair's communication rate is weighted by the distance separating them.  The transport capacity optimizes over all scheduling and routing algorithms and the focus is on the asymptotic rate of growth of the sum rate in the number of nodes $n$, either keeping the network area fixed or letting the network area grow linearly with $n$.  TC, on the other hand, is a medium access control (MAC) layer metric that neither precludes nor addresses routing\footnote{see \S\ref{sec:multihop} for an exception, where a simple multihop model is added.}.  Although transport capacity is more general in that it optimizes scheduling and routing, the cost of this generality is that typically the transport capacity results are less specific than those obtainable under the TC framework.  The results are less specific in the sense that results on the asymptotic rate of growth of the transport capacity as a function of $n$ often do not specify the pre-constant.

The advantages of using TC as a metric for wireless network performance are: $i)$ it can be exactly derived in some important cases, and tightly bounded in many others, $ii)$ performance dependencies upon fundamental network parameters are thereby illuminated, and $iii)$ design insights are obtainable from these performance expressions.  More fundamentally, the TC captures in a natural way essential performance indicators like network efficiency (ASE), reliability (OP), and throughput (TP).  In fact, TC is precisely maximization of TP under an OP constraint, as discussed in \S\ref{sec:tpandtc}, and is proportional to ASE, as discussed in \S\ref{sec:specman}.

One limitation of the TC framework, at least as in this monograph, is the implicit assumption that the network employs the simplistic and sub-optimal slotted Aloha protocol at the MAC layer. The TC can also be extended to model other contention based MAC protocols at the cost of some tractability \cite{GanHae2009a,GanAndHae2011}, but we elect to stick to the simple slotted Aloha protocol, where each transmitter (Tx) independently elects whether or not to transmit to its receiver (Rx) in each time slot by flipping an independent biased coin \cite{Abr1977}.  If the point process describing the locations of contending transmitters at a snapshot in time form a Poisson point process (PPP), which we assume is the case, then under the Aloha protocol the locations of the active transmitters at some point in time also form a PPP, obtained by independent sampling of the node location PPP.  The PPP model is necessary for preserving the highest level of analytical tractability of the TC framework, but of course it means that the computed TC is sub-optimal.  The difficulty in relaxing the Aloha assumption lies in the fact that any realistic and useful randomized MAC protocol involves coordination among competing transmitters, which necessarily spoils the crucial independence property of the PPP.  One's valuation of the TC framework typically rests on weighing the advantage of having an explicit expression for an insightful network performance indicator with the disadvantage of that performance corresponding to a suboptimal control law.  Throughout this monograph, we generally adopt the following assumptions in order for the OP and TC to be computable.  See Fig.\ \ref{fig:overview}.

\begin{assumption}
\label{ass:keyass}
The following assumptions are made:
\begin{enumerate}
\item The network is viewed at a single snapshot in time for the purpose of characterizing its spatial statistics.
\item Every potential Tx is matched with a prearranged intended Rx at a fixed distance $u$ (meters) away\footnote{The extension to random distances is straightforward and given in \S\ref{sec:vardist}.}: these Tx--Rx pairs are in one to one correspondence.

\item When mapping our results to a specific bit rate, we assume each Rx treats (uncancelled) interference as noise, and the rate from a particular Tx to its Rx at location $o$ is given by the Shannon capacity $\csf(o) \equiv \frac{1}{2} \log_2(1+\sinr(o))$.

\item The potential transmitters form a homogeneous PPP on the network arena, taken to be $\Rbb^d$, for $d \in \{1,2,3\}$.  This implies $i)$ the number of nodes in the network is countably infinite, and $ii)$ the number of potential transmitters in two disjoint bounded sets of the plane are independent Poisson random variables (RV).  See Fig.\ \ref{fig:overview}.

\item Every potential Tx decides independently whether or not to transmit with a common probability $p_{\rm tx}$.  It follows that the set of actual transmitters is also a (thinned) PPP.
\end{enumerate}
\end{assumption}

A few remarks are in order:
\begin{enumerate}
\item The TC computes the maximum spatial reuse which is computable by looking at the network at a single snapshot in time.  This perspective neither addresses nor precludes multi-hop or routing considerations.

\item Ass.\ (3) can be easily softened to account for any modulation and coding type that is characterized by a SINR ``gap'' from capacity.  Typically, we directly utilize SINR for computing outage probability and TC and do not include the per-link rate in the results.

\item Ass.\ (4) makes clear our focus is on networks whose arena is the entire plane $\Rbb^2$, and which have a countably infinite number of nodes.  This along with Ass.\ (5) removes any concern about boundary effects and makes each node ``typical'' in a sense described below.
\end{enumerate}

\begin{figure}[!htbp]
\centering
\includegraphics[width=0.6\textwidth]{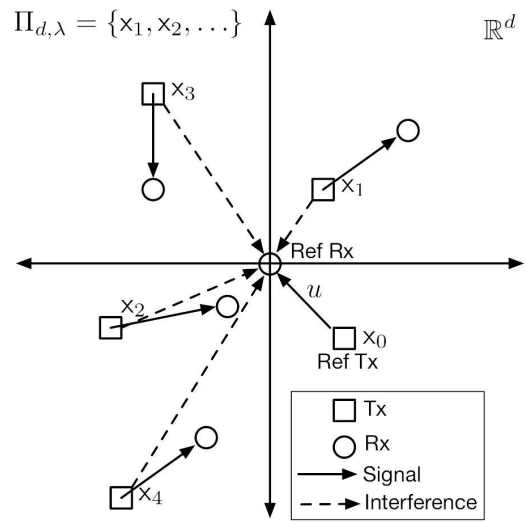}
\caption{A reference (typical) Rx located at the origin $o \in \Rbb^d$ is paired with a reference Tx at distance $u$, and is subject to interference (dashed lines) from a PPP $\Pi_{d,\lambda}$ of interfering Tx's (each of which has a unique associated Rx).}
\label{fig:overview}
\end{figure}

\section{Key definitions: PPP, OP, and TC}
\label{sec:def}

Ass.\ (1)--(3) allow us to formally define the OP.
\begin{definition}
\label{def:op}
{\bf Outage probability (OP).}
Define the constant $R$ to be the spectral efficiency (in bits per channel use per Hz) of the channel code employed by each Tx--Rx pair in the network.  Define the SINR threshold $\tau \equiv 2^{2R}-1$ so that $R = \frac{1}{2}\log_2(1+\tau)$.  For an arbitrary Tx--Rx pair with the Rx positioned at the origin $o \in \Rbb^d$, let $\csf(o) \equiv \frac{1}{2} \log_2(1+\sinr(o))$ be the {\em random} Shannon spectral efficiency of the channel connecting them when interference is treated as noise.  The OP is the probability the random spectral efficiency of the channel falls below the spectral efficiency of the code, or equivalently, the probability that the random SINR at the Rx is below the threshold $\tau$:
\begin{eqnarray}
q(o)
&\equiv& \Pbb(\csf(o) < R) \nonumber \\
&=& \Pbb \left( \frac{1}{2}\log_2(1+\sinr(o)) < \frac{1}{2}\log_2(1+\tau) \right) \nonumber \\
&=& \Pbb(\sinr(o) < \tau). \label{eq:outage}
\end{eqnarray}
\end{definition}
It is worth emphasizing that the RV in $\Pbb(\csf(o) < R)$ is the capacity $\csf(o)$ of the channel connecting the Tx--Rx pair, computed at the snapshot in time at which we observe the network, and not the rate $R$, which is assumed fixed.  In particular, $\csf(o)$ is a function of the RV $\sinr(o)$, which is quite sensitive to the distances between the Rx at $o$ and the random set of interfering transmitters at the observation instant.  The OP is the cumulative distribution function (CDF) of the RVs $\csf(o)$ and $\sinr(o)$.

By the assumption that the transmitters (and receivers) form a PPP, it follows that all Tx--Rx pairs are typical, hence $q(o) = q$.  More formally, we can condition on the presence of a test Tx--Rx pair where, without loss of generality, we assume the test Rx to be located at the origin $o$.  The distribution of the PPP of potential transmitters is unaffected by the addition of this test pair:\footnote{This result is due to Slivnyak \cite{Sli1964}.  See, \eg, \cite{BacBla2009a} Thm.\  1.13 (p.30), \cite{HaeGan2008} Thm.\  A.5 (p.113), \cite{StoKen1995} p.41 and Example 4.3 (p.121).}
\begin{equation}
\label{eq:slivnyak}
q(o) \equiv \Pbb(\csf(o) < R|\mbox{ Rx at $o$}) = \Pbb(\csf(o) < R).
\end{equation}
Ass.\ (4)--(6) and the definition of OP allow us to formally define the TC.  We first define a homogeneous PPP on $\Rbb^d$.  Fig.\ \ref{fig:PPP} shows a portion of a sample PPP on $\Rbb^2$ and illustrates the fact that the number of points in each compact set is a Poisson RV.
\begin{definition}
\label{def:ppp}
{\bf Homogeneous Poisson point process (PPP).}
A PPP with intensity $\lambda > 0$ in $d$-dimensions is a random countable collection of points $\Pi_{d,\lambda} = \{\xsf_1,\xsf_2,\ldots\} \subset \Rbb^d$ such that
\begin{itemize}
\item For two disjoint subsets $A,B \subset \Rbb^d$ the number of points from $\Pi$ in these sets are independent RVs: $\Pi(A) \perp \Pi(B)$, where $\Pi(A) \equiv |\Pi \cap A|$ is the number of points in $\Pi$ in $A$.
\item The number of points in any compact set $A \subset \Rbb^d$ is a Poisson RV with parameter $\lambda |A|$, for $|A|$ the volume of $A$.  That is:
\begin{equation}
\Pi(A) \sim \mathrm{Po}(\lambda |A|),
\end{equation}
or equivalently,
\begin{equation}
\Pbb(\Pi(A) = k) = \frac{1}{k!} \erm^{-\lambda |A|} (\lambda |A|)^k , ~ k \in \Zbb_+.
\end{equation}
\end{itemize}
\end{definition}

\begin{figure}[!htbp]
\centering
\includegraphics[width=0.87\textwidth]{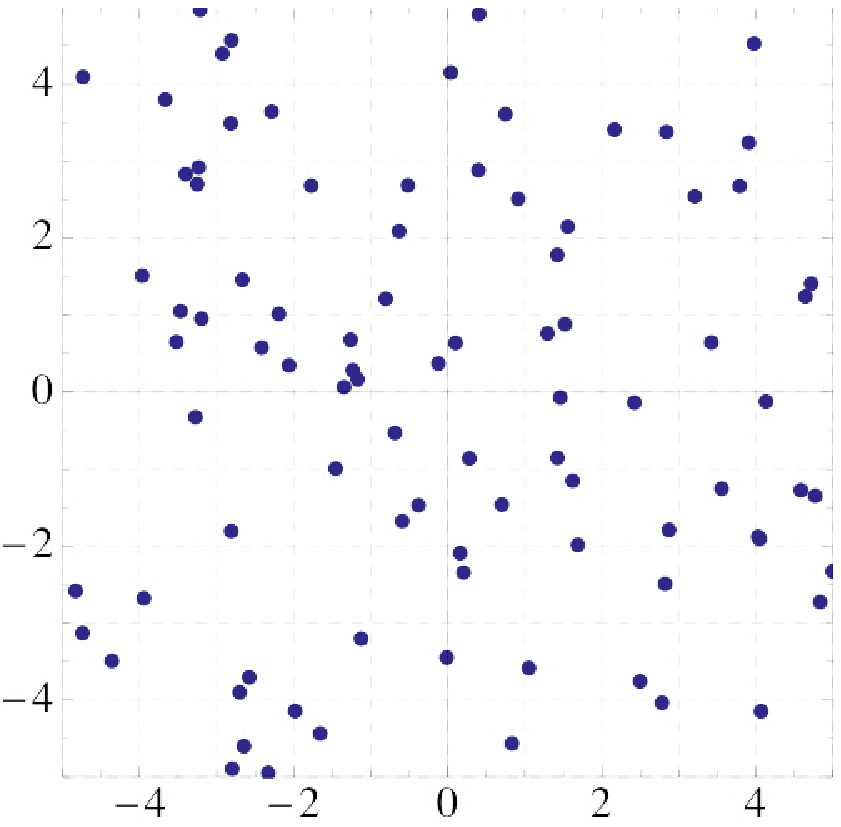}
\includegraphics[width=0.87\textwidth]{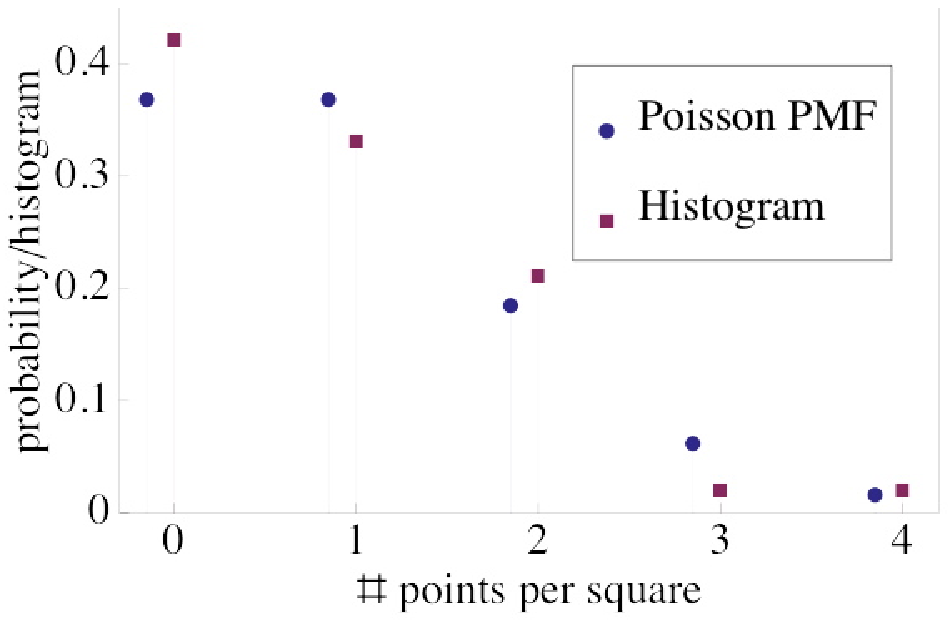}
\caption{{\bf Top:} a realization of a homogeneous Poisson point process (PPP) of intensity $\lambda = 1$ on $\Rbb^2$.  {\bf Bottom:} the histogram of the number of points in each of the $10 \times 10$ unit area squares and the corresponding Poisson PMF for $\lambda = 1$.}
\label{fig:PPP}
\end{figure}

Suppose $\Pi_{\rm pot} \equiv \Pi_{d,\lambda_{\rm pot}}$ is the PPP of potential transmitters with spatial intensity $\lambda_{\rm pot}$ discussed in Ass.\ \ref{ass:keyass} (4) and (5).  Let $p_{\rm tx} \in (0,1)$ be the common transmission probability employed by each node in Ass.\ \ref{ass:keyass} (5).  It follows that the PPP of actual transmitters at the observation instant (denoted $\Pi_{d,\lambda}$) is a thinned version of $\Pi_{\rm pot}$ with corresponding thinned intensity $\lambda \equiv \lambda_{\rm pot} p_{\rm tx}$.

It is intuitive that the OP is increasing in $\lambda$: a higher spatial intensity of transmission attempts yields larger interference at each Rx, which decreases the SINR.  We emphasize this dependence by writing the OP as $q(\lambda)$ and thereby view $q : \Rbb_+ \to [0,1]$ as a map from the spatial intensity of transmission attempts to the corresponding OP.\footnote{This redefines the OP $q(\lambda)$ as a function of the intensity $\lambda \in \Rbb_+$ of the PPP --- in \eqref{eq:outage} and \eqref{eq:slivnyak} $q(x)$ denoted the OP at location $x \in \Rbb^d$.}
\begin{fact}
\label{fac:opinv}
The OP $q(\lambda)$ is continuous, strictly increasing, and onto $[q(0),1)$, where $q(0)$ is the OP in the absence of interference.
\end{fact}
Because of this fact, the inverse $q^{-1} : [q(0),1) \to \Rbb_+$ is well-defined.  For an outage constraint $q^* \in [q(0),1)$, the inverse OP $q^{-1}(q^*)$ is the (unique) intensity of transmission attempts associated with an outage probability of $q^*$.  Each such transmission succeeds with probability $1-q^*$, and as such $q^{-1}(q^*)(1-q^*)$ is the spatial intensity of successful transmissions.  This is what we call TC.
\begin{definition}
\label{def:tc}
{\bf Transmission capacity (TC).}
Fix a maximum permissible OP $q^* \in [q(0),1)$.  The TC is the maximum spatial intensity of successful transmissions subject to an OP of $q^*$:
\begin{equation}
\lambda(q^*) \equiv q^{-1}(q^*) (1-q^*).
\end{equation}
\end{definition}
The intensity of failed transmissions is $q^{-1}(q^*) q^*$, and the summed intensity of successful and failed transmissions is naturally $q^{-1}(q^*)$.  The TC (and all spatial intensities) are measured in units of ($\mathrm{meters}^{-d}$), \ie, an ``average'' number of nodes per unit area.
\begin{remark}
\label{rem:ptx}
{\bf TC and slotted Aloha.}
The TC $\lambda(q^*)$ has operational sigificance for a wireless network of potential transmitters positioned according to a PPP of intensity $\lambda_{\rm pot}$ and employing the slotted Aloha MAC protocol with transmission probability $p_{\rm tx}$.  Namely, if $(q^*,\lambda_{\rm pot})$ are such that $\lambda(q^*) <\lambda_{\rm pot} \times (1-q^*)$ then select
\begin{equation}
p_{\rm tx} = \frac{\lambda(q^*)}{\lambda_{\rm pot} \times (1-q^*)}.
\end{equation}
The resulting intensity of attempted transmissions $p_{\rm tx} \lambda_{\rm pot} = q^{-1}(q^*)$ will be such that the OP is $1-q^*$. If $\lambda(q^*) \geq \lambda_{\rm pot} \times (1-q^*)$ then the network does not need an Aloha MAC throttling transmission attempts to achieve an OP of $q^*$: setting $p_{\rm tx} = 1$ will result in an OP $q(\lambda_{\rm pot}) < q^*$.
\end{remark}

\section{Overview of the results}
\label{sec:sum}

The results presented in this volume are listed in Tables \ref{tab:cha1} (Ch.\ \ref{cha:int}) through \ref{tab:cha6} (Ch.\ \ref{cha:MIMO}).  We briefly discuss each chapter.

{\bf Ch.\  \ref{cha:int} (Table \ref{tab:cha1}).}  The key concepts are in \S\ref{sec:def}, specifically, Def.\ \ref{def:op} of the outage probability (OP), Def.\ \ref{def:ppp} of the (homogeneous) Poisson point process (PPP), and Def.\ \ref{def:tc} of the transmission capacity (TC).

{\bf Ch.\  \ref{cha:matpre} (Table \ref{tab:cha2}).}  We first define the ball and annulus in $\Rbb^d$. (Def.\ \ref{def:ball}) and gives their volumes (Prop.\ \ref{pro:ballvol}).  All results are given for arbitrary dimension $d$, where $\{1,2,3\}$ are the three relevant values.

Throughout the volume we denote RVs in sans-serif font (Rem.\ \ref{rem:notation} in \S\ref{sec:mccine}), \eg, $\xsf$.  Note the acronyms and notation for standard probabilistic concepts in Def.\ \ref{def:standardprobdefs}.  \S\ref{sec:pppvd} gives a short but essential coverage of the void probability (Prop.\ \ref{pro:void}), the mapping theorem (Thm.\ \ref{thm:map}) and a derivative result on mapping distances (Prop.\ \ref{pro:distmap}).  The void probability underlies most performance bounds derived in this volume, and the mapping result allow translation from a PPP on $\Rbb^d$ of intensity $\lambda \in \Rbb_+$ ($\Pi_{d,\lambda}$) to an ``equivalent'' unit intensity PPP on $\Rbb^1$ ($\Pi_{1,1}$).

We cover (spatial) shot noise (SN) processes in \S\ref{sec:snp} (Def.\ \ref{def:snp}), which are used to model the aggregate interference experienced by a reference Rx at the origin.  We focus on power law SN (Def.\ \ref{def:intsn}) by assuming the impulse response function in the SN definition is taken to be the standard pathloss attenuation $|\xsf|^{-\alpha}$ where $\alpha$ is the pathloss exponent.  We also introduce here the characteristic exponent $\delta = d/\alpha$, where to avoid trivialities we assume $\delta \in (0,1)$ (\ie, $\alpha > d$) throughout.  The sum SN process ($\Sisf$) adds the interference contributions while the max SN ($\Msf$) takes the largest contribution.  The simple inequality $\Msf < \Sisf$ forms the basis of most of the bounds in this volume in that the distribution of $\Msf$ (Cor.\ \ref{cor:maxsnrvdis}) is Frech\'{e}t (Def.\ \ref{def:frechet}), and also can be derived directly from the void probability in Prop.\ \ref{pro:void}.  The Campbell-Mecke result (Thm.\ \ref{thm:cam}) allows computation of moments (Prop.\ \ref{pro:camz}) of SN RVs.  More important for us will be the series expansions of the SN distribution (Prop.\ \ref{pro:snserexp}) as these directly yield the asymptotic (tail) distributions (Cor.\ \ref{cor:snasypdf}), which yield all the asymptotic performance results in this volume.

A critical observation is that the SN is a stable RV, this is the focus of \S\ref{sec:stadis}.  We define this class (Def.\ \ref{def:stadis} and \ref{def:staparam}), and introduce the L\'{e}vy distribution (Def.\ \ref{def:levy}) which is the only stable distribution of relevance to us with a closed form CDF, and corresponding to $\delta = \frac{1}{2}$.  This allows the exact performance results in Ch.\  \ref{cha:bm}.  We introduce the probability generating functional (PGFL) (Def.\ \ref{def:pgfl}), and identify its connection with the Laplace transform, the moment generating function, and the characteristic function of the SN RV.

The results in this chapter are tied together in \S\ref{sec:maxsum} where we demostrate the key property that the the simple bound $\Msf < \Sisf$ is tight in the sense that the ratio of the CCDFs for these RVs approaches unity in the limit (Prop.\ \ref{pro:intsummaxtight}).  We derive a similar result using subexponential distributions (Def.\ \ref{def:subexp}) for a binomial point process (BPP).

{\bf Ch.\  \ref{cha:bm} (Table \ref{tab:cha3}).}  This chapter presents the main results on OP and TC in their barest, simplest form, so as to achieve maximum clarity.  Exact OP and TC results are in \S\ref{sec:exactTC}.  SINR is defined (Def.\ \ref{def:sinrnf}) and it is observed that the OP is the CCDF of the SN evaluated at a certain value.  An explicit expression for the OP and TC (for $\delta = \frac{1}{2}$) is given (Cor.\ \ref{cor:oplev} and \ref{cor:tclev}).

Asymptotic OP and TC results are in \S\ref{sec:asympTC}.  The asymptotic CCDF of the SN (Cor.\ \ref{cor:snasypdf}) yields the asymptotic OP (as $\lambda \to 0$) and TC (as $q^* \to 0$) in Prop.\ \ref{pro:asymoptc}.  The asymptotic TC is interpreted as sphere packing in $\Rbb^d$, where the sphere radius depends upon the key model parameters $\delta,u,\tau,d$ (Rem.\ \ref{rem:spherepack}).

The $\Msf < \Sisf$ SN inequality forms the basis for the OP lower bound (LB) and TC upper bound (UB) in \S\ref{sec:ubTC}.  We adopt the language of dominant interferers (Def.\ \ref{def:domint}) to describe interferers capable by themselves of reducing the SINR seen at the origin below its threshold $\tau$, but observe this concept is equivalent to taking the maximum interferer (Rem.\ \ref{rem:dommaxequiv}).  The main result is the bound on OP and TC in Prop.\ \ref{pro:oplb}.

In \S\ref{sec:tpandtc} we turn our attention to a third performance metric, the MAC layer throughput (TP), $\Lambda(\lambda)$, defined (Def.\ \ref{def:tp}) as the spatial intensity of succesful transmissions.  A TP UB is obtained from the OP LB (Prop.\ \ref{pro:tp}).  We make the key observation that ``blind'' maximization of TP leads to an associated OP of $67\%$.  The natural design objective of maximizing TP subject to an OP constraint is shown to be precisely the TC, giving a more natural justification for this quantity as a meaningful performance measure (Prop.\ \ref{pro:tcopt}).  In fact the TP and TC have the same unconstrained maximum and we relate their maximizers (Prop.\ \ref{pro:optTPandTC}).

Finally, \S\ref{sec:lbTC} gives an UB on OP and a LB on TC.  A useful expression for the OP in terms of its LB is derived (Prop.\ \ref{pro:domnonop}), which the OP LB and the three basic inequalities in \S\ref{sec:mccine} (Markov, Chebychev, and Chernoff) are combined to give three OP UBs.  These are observed to vary both in terms of their tightness and their simplicity.

{\bf Ch.\ \ref{cha:modenh} (Table \ref{tab:cha4})} extends the basic model in three ways: fading (\S\ref{sec:fading}), variable link distances (\S\ref{sec:vardist}), and multi-hop (\S\ref{sec:multihop}).

The bulk of this chapter is on fading (\S\ref{sec:fading}); the SINR under fading is defined (Def.\ \ref{def:fad}).  \S\ref{sec:fading} is split into three subsections: exact results (\S\ref{ssec:fadexact}), asymptotic results (\S\ref{ssec:fadasymp}), and bounds (\S\ref{ssec:fadlb}).   The main result in \S\ref{ssec:fadexact} is Prop.\ \ref{pro:optcrayfadsig} which gives the exact OP and TC under the assumption that the signal fading is Rayleigh (exponential).  Note this exact result holds for all $\delta$, while the only exact result available under the basic model in Ch.\ \ref{cha:bm} is for $\delta = \frac{1}{2}$ (Cor.\ \ref{cor:oplev} and \ref{cor:tclev}).  For the asymptotic results in \S\ref{ssec:fadasymp} we introduce the formalism of the marked PPP (MPPP) and exploit the important marking theorem (Thm.\ \ref{thm:mark}) which allows us to extend the distance and interference mapping results for PPPs from \S\ref{sec:pppvd} to the MPPP case.  The series expansions of the interference under fading (Prop.\ \ref{pro:fadintserrep}) is used to derive the asymptotic OP and TC (Prop.\ \ref{pro:asymoptcfad}).  An important observation is that fading in general degrades performance relative to the non-fading case (Cor.\ \ref{cor:optcfadnoncomp}).  In \S\ref{ssec:fadlb} the concept of dominant interferers used in Def.\ \ref{def:domint} is extended to incorporate fading (Def.\ \ref{def:domintfad}), but under fading the strongest interferer need not be the nearest interferer to the origin.  The main result is the OP LB (Prop.\ \ref{pro:fadoplb}), where we observe the LB is in fact the MGF of a certain function of the signal fading RV.

\S\ref{sec:vardist} addresses variable link distances, \ie, the Tx--Rx distance is a RV.  The SINR and OP for this model are defined in Def.\ \ref{def:vardistsinr} and \ref{def:opvld}, respectively, and we present asymptotic results (Prop.\ \ref{prop:vardistoptc}) and exact results (Cor.\ \ref{cor:vardistexactopnn}).

\S\ref{sec:multihop} extends the TC framework to a multihop scenario where sources send packets to destinations $M$ hops away over a total distance $R$.  Multihop TC is defined in Def.\ \ref{def:MHTC}. Although some fairly strong assumptions must be made to preserve tractability, plausible insights can be drawn about the optimum hop count (given in Prop.\ \ref{pro:mstar}) and end-to-end TC in terms of all the network parameters.

{\bf Ch.\ \ref{cha:destech} (Table \ref{tab:cha5}).}   The chapter on design techniques studies four natural approaches to improve the performance of a wireless network: \S\ref{sec:specman} studies the performance when the spectrum is split into a number of channels, \S\ref{sec:intcan} considers performance when receivers are equipped with interference cancellation capabilities, \S\ref{sec:sched} evaluates the performance when nodes only transmit when their signal fade is above a specified threshold, and \S\ref{sec:power} considers power control.

In \S\ref{sec:specman} the design objective is to optimize the number of bands to form from the available spectrum, where each Tx selects a band uniformly at random.  The intuitive tradeoff is that more bands gives fewer interferers but this also means the bandwidth per band is smaller, and thus a higher SINR threshold is required to achieve a given data rate.  We define the model in Def.\ \ref{def:specdefs} and \ref{def:outageequiv}.  The spectral efficiency optimization problem is formalized in Prop.\ \ref{pro:tcspec}, and we characterize the solution in Prop.\ \ref{pro:specoptnu}, and then specialize the result to both the high (Cor.\ \ref{cor:specasymoptnu}) and low (Prop.\ \ref{pro:specasymnulowsnr}) SNR regimes.

In \S\ref{sec:intcan} the usual limitations of interference cancellation (IC) are captured through the $(\kappa,K,P_{\rm min})$ Rx model (Def.\ \ref{def:icparam}) where $\kappa$ is cancellation effectiveness, $K$ is the maximum number of cancellable nodes, and $P_{\rm min}$ is the minimum received power.  The SINR is defined in Def.\ \ref{def:sicsinr}, and the main result is the OP LB (Prop.\ \ref{pro:sicmain}).

In \S\ref{sec:sched} the fading coefficient threshold (Def.\ \ref{def:sched}) used to throttle transmission attempts naturally trades off between the quality and quantity of transmission attempts, and the TP metric $\Lambda$ (Def.\ \ref{def:schedoptptc}) illustrates this tradeoff.  Asymptotic results are given in Prop.\ \ref{pro:asymtpsched} and a LB on OP is given in Prop.\ \ref{pro:schedlb}.

In \S\ref{sec:power} the notion of fractional power control (FPC) is introduced, where the power control exponent sweeps between fixed power and channel inversion (Def.\ \ref{def:power}).  The asymptotic results (Prop.\ \ref{pro:fpcasymp}) yield the optimal exponent is $1/2$ (Prop.\ \ref{pro:fpcasympoptf}).  The notion of dominant interferers is used once again (Def.\ \ref{def:domintfpc}) to compute the OP LB (Prop.\ \ref{pro:fpcoplb}).

{\bf Ch.\ \ref{cha:MIMO} (Table \ref{tab:cha6}).} The final chapter introduces multiple antennas at both the Tx and Rx, resulting in some of the first analytical work on MIMO that properly accounts for background interference. The results are broken into two main categories, which are defined along with basics of the models in \S\ref{sec:model}. \S\ref{sec:singlestream} considers the case where despite the multiple antennas, only a single data stream is sent, with the balance of the antennas being used for diversity and/or interference cancellation.  \S\ref{sec:multistream} considers the more general multistream case, where transmitters send more than one simultaneous stream to either a single receiver (spatial multiplexing) or to multiple users (space division multiple access).  Finally, the practical implications and limitations of the results are discussed in \S\ref{sec:takeaways}.

In \S\ref{sec:singlestream}, the results are further categorized into diversity (\S\ref{ssec:diversity}) and interference cancellation (\S\ref{ssec:MIMO-IC}).  For receive diversity, the OP of MRC is given in Prop.\ \ref{pro:OP-MRC} and the corresponding TC in \ref{pro:TC-MRC}.  The result is equivalent for MRT (transmit MRC) and the generalization to $\nt \times \nr$ diversity beamforming is discussed in Rem.\ \ref{rem:eigen}.  In \S\ref{ssec:MIMO-IC}, a TC lower bound is given on a suboptimal technique called partial ZF in Prop.\ \ref{pro:tcpzflb} and a TC upper bound for MMSE in Prop.\ \ref{pro:MMSE-UB}. These respectively bound the TC of MMSE and we see linear scaling can be achieved with the number of antennas.

In \S\ref{sec:multistream}, we first consider a class of results for spatial multiplexing in \S\ref{ssec:SM}, where multiple streams are transmitted from a single Tx to a single Rx.  Prop.\ \ref{pro:SM-MRC} and \ref{pro:SM-ZF} give the optimal number of streams $K^*$ and TC scaling in terms of $\nt \leq \nr$ for MRC and ZF receivers, respectively. This is extended to a BLAST receiver in Prop. \ref{pro:SM-BLAST}.  Then in \S\ref{ssec:SDMA}, we turn our attention to streams being sent to multiple Rx's at the same time. The main result for MRC receivers is given in Prop.\ \ref{pro:DPC-TC}, with the appropriate scaling results given in Prop.\ \ref{pro:DPC-scaling}.

%
%
\chapter{Mathematical preliminaries}
\label{cha:matpre}

In this chapter we present some necesssary mathematical preliminaries, mostly related to probabilistic analysis of functionals of PPPs.  The most important example of such a functional is the aggregate interference experienced by a typical Rx in a wireless network when the locations of interfering nodes form a PPP and the channel is distance dependent.  Many of the results in this chapter are also found in the excellent monographs by Haenggi and Ganti \cite{HaeGan2008} and Baccelli and B{\l}aszczyszyn \cite{BacBla2009a,BacBla2009b}.  Our treatment of this large field is quite selective: we present only those results directly relevant to computing the OP and TC.  We recommend both these monographs for a more in depth treatment of application of the mathematical field of stochastic geometry to the performance analysis of wireless networks.  To the extent possible we have used notation consistent with that used in \cite{HaeGan2008,BacBla2009a,BacBla2009b}.  Moreover, whenever possible we give references in \cite{HaeGan2008,BacBla2009a,BacBla2009b} to corresponding results presented in this chapter.

Denote the reals by $\Rbb$, the natural numbers by $\Nbb = \{1,2,3,\ldots\}$, the integers by $\Zbb$, and the complex numbers by $\Cbb$ (and $\sqrt{-1}$ by $\irm$).  We use $\equiv$ for equality that holds by definition.  We work in $\Rbb^d$ where $d$ is the spatial dimension of the wireless network.  Our analysis holds for general $d \in \Nbb$, but $d \in \{1,2,3\}$ are the relevant cases.  A point in $\Rbb^d$ is denoted by $x = (x_1,\ldots,x_d)$.  The Euclidean ($L_2$) norm $\|x\|_2 \equiv \sqrt{\sum_{i=1}^d x_i^2}$ is denoted by $|x|$, and the origin $(0,\ldots,0)$ by $o$.  We denote the natural log as $\log x$.  For a natural number $N \in \Nbb$, we write $[N]$ for the set $\{1,\ldots,N\}$.  We use the shorthand $a \land b \equiv \min\{a,b\}$ and $a \lor b \equiv \max\{a,b\}$.  We use standard asymptotic order notation $\Omc(\cdot),\Omega(\cdot),\Theta(\cdot)$.

We begin with the $d$-dim.\ ball and annulus and their volumes.
\begin{definition}
\label{def:ball}
{\bf Ball and annulus.}
The $d$-dim.\ ball ($d \in \Nbb$) centered at $c \in \Rbb^d$ with radius $r \in \Rbb_+$ is:
\begin{equation}
\brm_d(c,r) \equiv \{ x \in \Rbb^d : |x-c| \leq r\}.
\end{equation}
The $d$-dim.\ annulus centered at $c \in \Rbb^d$ with radii $0 < r_1 < r_2$ is:
\begin{equation}
\arm_d(c,r_1,r_2) \equiv \{ x \in \Rbb^d : r_1 \leq |x-c| \leq r_2\}.
\end{equation}
\end{definition}
Volume of a set $S \subset \Rbb^d$ is denoted $|S|$.  The ball and annulus volumes are given below (\cite{Hae2005} (3)).
\begin{proposition}
\label{pro:ballvol}
{\bf Ball and annulus volume.}
The $d$-dim.\ ball $\brm_d(c,r)$ and annulus $\arm_d(c,r_1,r_2)$ have volume
\begin{equation}
|\brm_d(c,r)| = c_d r^d, ~~~ |\arm_d(c,r_1,r_2)| = c_d (r_2^d - r_1^d),
\end{equation}
where
\begin{equation}
c_d \equiv \left\{ \begin{array}{ll}
\frac{\pi^{\frac{d}{2}}}{(d/2)!}, \; & d \mbox{ even} \\
\frac{1}{d!} \pi^{\frac{d-1}{2}}2^d \left(\frac{d-1}{2}\right)!, \; & d \mbox{ odd},
\end{array} \right.
\end{equation}
\end{proposition}
The relevant values of $c_d$ are:
\begin{equation}
c_1 = 2, ~ c_2 = \pi, ~ c_3 = \frac{4}{3} \pi.
\end{equation}
We will also have frequent use for the gamma function.
\begin{definition}
\label{def:gamfun}
{\bf Gamma function.}
The gamma and incomplete gamma function are, respectively,
\begin{equation}
\Gamma(z) \equiv \int_0^{\infty} t^{z-1} \erm^{-t} \drm t, ~~~~
\Gamma(z,t_l,t_h) \equiv \int_{t_l}^{t_h} t^{z-1} \erm^{-t} \drm t
\end{equation}
for $z \in \Cbb$ and $0 \leq t_l \leq t_h \leq \infty$.
\end{definition}
Note $\Gamma(z,0,\infty) = \Gamma(z)$ and that $\Gamma(z) = (z-1)!$ for $z \in \Nbb$.  We will have use for the identity:
\begin{equation}
\label{eq:gamident}
\Gamma(1-\delta)\Gamma(1+\delta) = \frac{\pi \delta}{\sin (\pi \delta)}.
\end{equation}
and the fact that $\Gamma(-1/2) = -2 \sqrt{\pi}$.  See Fig.\ \ref{fig:fadnonasymp} for \eqref{eq:gamident}.

\section{Probability: notations, definitions, key inequalities}
\label{sec:mccine}

The material in this section is quite standard and is available in most textbooks on probability.  The following is a key notational convention.
\begin{remark}
\label{rem:notation}
{\bf RV notation.}
We indicate random variables (RVs) with a sans-serif non-italic font, \eg, $\xsf, \hsf, \usf, \csf$, and their realizations (as well as other non-random quantities) with an italicized serif font, \eg, $x,h,u,c$.  A notable exception is the use of  $\Pi$ (Def.\ \ref{def:ppp}) to indicate a random point process.
\end{remark}
Standard probabilistic quantities are denoted as follows.
\begin{definition}
\label{def:standardprobdefs}
{\bf Standard probability definitions.}
Let $\xsf$ denote a continuous real-valued RV, and let $t \in \Rbb$, $\theta \in \Rbb_+$, and $s \in \Cbb$.
\begin{enumerate}
\item The cumulative distribution function (CDF) is $F_{\xsf}(t) \equiv \Pbb(\xsf \leq t)$ for $t \in \Rbb$.  Denote the CDF for random $\xsf$ by $\xsf \sim F_{\xsf}$.
\item The complementary CDF (CCDF) is $\bar{F}_{\xsf}(t) \equiv 1 - F_{\xsf}(t) = \Pbb(X > t)$.
\item The inverse CDF and inverse CCDF are $F_{\xsf}^{-1}(p)$ and $\bar{F}_{\xsf}^{-1}(p)$ for $p \in [0,1]$.
\item The probability density function (PDF) is $f_{\xsf}(t) \equiv \frac{\drm}{\drm t} F_{\xsf}(t)$.
\item Expectation is denoted by $\Ebb[\xsf]$, variance is denoted $\mathrm{Var}(\xsf)$.
\item The Laplace transform (LT) is $\Lmc[\xsf](s) \equiv \Ebb[\erm^{-s \xsf}]$, note $s \in \Cbb$.
\item The characteristic function (CF) is $\phi[\xsf](t) \equiv \Ebb[\erm^{\irm t \xsf}]$, note $t \in \Rbb$.
\item The moment generating function (MGF) is $\Mmc[\xsf](\theta) \equiv \Ebb[\erm^{\theta \xsf}]$, note $\theta \in \Rbb_+$.
\item The hazard rate function (HRF) is $\Hmc[\xsf](x) \equiv \frac{\drm }{\drm x} -\log \bar{F}_{\xsf}(x) = \frac{f_{\xsf}(x)}{\bar{F}_{\xsf}(x)}$.
\item A normal RV with $\Ebb[\xsf]=\mu$ and $\mathrm{Var}(\xsf) = \sigma^2$ is denoted $\xsf \sim N(\mu,\sigma)$.  A standard normal is denoted $\zsf \sim N(0,1)$ with CDF $F_{\zsf}(t)$ and CCDF $\bar{F}_{\zsf}(t)$.
\item Equality in distribution between RVs $\xsf,\ysf$ is denoted $\xsf \stackrel{\drm}{=} \ysf$.
\end{enumerate}
\end{definition}
The LT with argument $s \in \Cbb$ is more general than both the CF and the MGF, but the LT and MGF need not exist, while the CF is guaranteed to exist.  When all three exist, the CF and MGF are obtainable from the LT:
\begin{equation}
\label{eq:ltcfmgf}
\phi[\xsf](t) = \Lmc[\xsf](-\irm t), ~ \Mmc[\xsf](\theta) = \Lmc[\xsf](-\theta).
\end{equation}
We will have use for Jensen's inequality (\eg, Cor.\ \ref{cor:optcfadnoncomp}).
\begin{proposition}
\label{pro:jensen}
{\bf Jensen's inequality.}
For a RV $\xsf$, if $f$ is a convex function then $\Ebb[f(\xsf)] \geq f(\Ebb[\xsf])$, with equality holding for $f$ affine.
\end{proposition}

The three inequalities of Markov, Chebychev, and Chernoff are each UBs on tail probabilities.  The three inequalities build upon one another.  In general it is fair to say that Markov is simpler to apply than Chebychev, and in turn Chebychev is simpler to apply than Chernoff.  This is on account of the fact that Markov relies only upon the mean $\Ebb[\xsf]$, while Chebychev depends upon the variance $\mathrm{Var}(\xsf)$, and Chernoff is a function of the moment generating function $\Mmc[\xsf](s) = \Ebb[\erm^{s \xsf}]$.  In general (but not always) it is further the case that the Chernoff bound is tighter than the Chebychev bound, and the Chebychev bound is tighter than the Markov bound.  The tightness of the Chernoff bound also comes about through the flexibility to tune the free parameter $s$.  These inequalities will be applied in \S\ref{sec:lbTC} to derive  an UB on the OP.  We begin with Markov's inequality.
\begin{proposition}
\label{pro:mar}
{\bf Markov's inequality.} For a nonnegative RV $\xsf$ and $t \in \Rbb_+$:
\begin{equation}
\Pbb(\xsf > t) \leq \frac{\Ebb[\xsf]}{t}.
\end{equation}
\end{proposition}
\begin{proof}
Define the Bernoulli indicator RV $\mathbf{1}_{\xsf > t}$ and observe $\xsf \geq t \mathbf{1}_{\xsf > t}$ for all $t \in \Rbb_+$.   Taking expectations yields
\begin{equation}
\Ebb[\xsf] \geq t \Ebb[\mathbf{1}_{\xsf > t}] = t \Pbb(\xsf > t).
\end{equation}
\end{proof}
Chebychev's inequality is obtained by applying Markov's inequality to the nonnegative RV $|\xsf-\Ebb[\xsf]|$.
\begin{proposition}
\label{pro:cheb}
{\bf Chebychev's inequality.}
For a RV $\xsf$ and $t \in \Rbb_+$:
\begin{equation}
\Pbb(|\xsf-\Ebb[\xsf]|>t) \leq \frac{\mathrm{Var}(\xsf)}{t^2}.
\end{equation}
\end{proposition}
\begin{proof}
Apply Markov's inequality with $|\xsf-\Ebb[\xsf]|$:
\begin{equation}
\Pbb(|\xsf-\Ebb[\xsf]|>t) = \Pbb((\xsf-\Ebb[\xsf])^2 > t^2) \leq \frac{\Ebb[(\xsf-\Ebb[\xsf])^2]}{t^2}.
\end{equation}
\end{proof}
Finally, Chernoff's inequality is obtained by applying Markov's inequality to the nonnegative RV $\erm^{\theta \xsf}$.
\begin{proposition}
\label{pro:cher}
{\bf Chernoff's inequality.}
For a nonnegative RV $\xsf$ and $t \in \Rbb_+$:
\begin{equation}
\Pbb(\xsf > t) \leq \inf_{\theta > 0} \Ebb[\erm^{\theta \xsf}] \erm^{-\theta t}
\end{equation}
\end{proposition}
\begin{proof}
Observe the equality of the events $\{\xsf > t\}$ and $\{\erm^{\theta \xsf} > \erm^{\theta t}\}$ for all $\theta > 0$ and apply Markov's inequality to the nonnegative RV $\erm^{\theta \xsf}$
\begin{equation}
\Pbb(\xsf > t) = \Pbb\left( \erm^{\theta \xsf} > \erm^{\theta t} \right) \leq \Ebb[\erm^{\theta \xsf}] \erm^{-\theta t}.
\end{equation}
The above inequality holds for all $\theta >0$ and hence in particular for that $\theta$ that minimizes the UB.
\end{proof}

\section{PPP void probabilities and distance mappings}
\label{sec:pppvd}

Recall $\Pi_{d,\lambda} = \{\xsf_i\}$ denotes a PPP with points $\{\xsf_i\} \subset \Rbb^d$ of intensity $\lambda$.  The two most important examples for us are $\Pi_{2,\lambda}$ and $\Pi_{1,1}$.  We often will write $\Pi_{1,1} = \{\tsf_i\}$ in accordance with the usual interpretation of the points in a one dimensional point process as times.
\begin{assumption} {\bf Labeling convention for PPP.}
\label{ass:pppdistorder}
All point processes are assumed to number  points in order of increasing distance from $o$: $\Pi_{d,\lambda} = \{\xsf_i\}$ with $|\xsf_1| < |\xsf_2| < \cdots$.
\end{assumption}
We present two key facts about distances for PPPs in this section.  First, the {\em void probability} $\Pbb(|\xsf_1|>r)$ is the probability that there are no points from $\Pi_{d,\lambda}$ in the ball $\brm_d(o,r)$, \ie, that the nearest neighbor to $o$ in $\Pi_{d,\lambda}$ is at least at distance $r$.
\begin{proposition}
\label{pro:void}
{\bf Void probability.}
The RV $|\xsf_1|$ has distribution
\begin{equation}
\Pbb(\Pi_{d,\lambda}(\brm_d(o,r)) = 0) = \Pbb(|\xsf_1| > r) = \erm^{-\lambda c_d r^d}, ~ r \in \Rbb_+
\end{equation}
where $c_d$ is defined in Prop.\ \ref{pro:ballvol}.
\end{proposition}
\begin{proof}
In words, $|\xsf_1| > r$ is the event that the nearest point in PPP $\Pi_{d,\lambda}$ to $o$ is at least a distance $r$ away.  This is the same as there being no points in the PPP lying in the ball $\brm_d(o,r)$.  Thus, $|\xsf_1| > r \Leftrightarrow \Pi_{d,\lambda}(\brm_d(o,r)) = 0$.  Recall from Def.\ \ref{def:ppp} that the RV $\Pi_{d,\lambda}(\brm_d(o,r))$ is Poisson with intensity $\lambda |\brm_d(o,r)| = \lambda c_d r^d$, and thus:
\begin{equation}
\Pbb(|\xsf_1|>r) = \Pbb(\Pi_{d,\lambda}(\brm_d(o,r)) = 0) = \erm^{-\lambda c_d r^d}.
\end{equation}
\end{proof}
For $d=1$ we recover the elementary fact that $|\xsf_1| \sim \mathrm{Exp}(2\lambda)$, \ie, $\Pbb(|\xsf_1|>r) = \erm^{-\lambda 2 r}$ and for $d=2$ we have $\Pbb(|\xsf_1|>r) = \erm^{-\lambda \pi r^2}$.  Prop.\ \ref{pro:void} will $i)$ form the basis for the max SN distribution (Cor.\ \ref{cor:maxsnrvdis}) which in turn will yield the LB on the OP in \S\ref{sec:ubTC}, and will be generalized to $ii)$ a non-homogeneous marked PPP (MPPP) in Prop.\ \ref{pro:voidnonhomo}, and $iii)$ distances to the $k^{th}$ nearest neighbor in Thm.\ \ref{thm:eucdistnn}.

The second result in this section is a special case of a more general mapping theorem given below (\cite{HaeGan2008}, Thm.\  A.1, p.\ 107 and \cite{Kin1993} \S2.3):
\begin{theorem}
\label{thm:map}
{\bf Mapping theorem} (\cite{HaeGan2008} Thm.\ A.1).
Let $\Phi$ be an inhomogeneous PPP on $\Rbb^d$ with intensity function $\Lambda$, and let $f:\Rbb^d \to \Rbb^s$ be measurable and $\Lambda(f^{-1}\{y\}) = 0$ for all $y \in \Rbb^s$.  Assume further that $\mu(B) = \Lambda(f^{-1}(B))$ satisfies $\mu(B) < \infty$ for all bounded $B$.  Then $f(\Phi)$ is a non-homogeneous PPP on $\Rbb^s$ with intensity measure $\mu$.
\end{theorem}
We refer the interested reader to \cite{HaeGan2008} for the formal definition of inhomogeneous PPP, intensity function, and measurability.  The following proposition is a special case of Thm.\  \ref{thm:map}.
\begin{proposition}
\label{pro:distmap}
{\bf Distance mapping.}
Let $\Pi_{d,\lambda} = \{\xsf_i\}$ be a PPP in $\Rbb^d$ of intensity $\lambda$, and $\Pi_{1,1} = \{\tsf_i\}$ a PPP in $\Rbb$ of intensity $1$.  Then:
\begin{equation}
\lambda c_d |\xsf_i|^d \stackrel{\drm}{=} 2|\tsf_i|, ~ i \in \Nbb.
\end{equation}
\end{proposition}
\begin{proof}
Consider Thm.\  \ref{thm:map} for $s=1$, $\Phi = \Pi_{d,\lambda}$ homogeneous with intensity $\Lambda(A) = \lambda |A|$ for all compact $A \subseteq \Rbb^d$ for some $\lambda \in \Rbb_+$, and $f(x) = \frac{\lambda c_d}{2} |x|^d \mathrm{sign}(x^{(1)})$, where $x^{(1)}$ is the first component in vector $x = (x^{(1)},\ldots,x^{(n)})$ and $\mathrm{sign}(a) = \mathbf{1}_{a \geq 0} - \mathbf{1}_{a \leq 0}$ is the sign of $a \in \Rbb$.  Consider bounded symmetric intervals $B$ of the form $[-t,t]$ for $t \in \Rbb_+$.
\begin{eqnarray}
\Lambda(f^{-1}([-t,t])) &=& \lambda |f^{-1}([-t,t])| = \lambda \left| \left\{ x \in \Rbb^d : f(x) \in [-t,t] \right\} \right| \nonumber \\
&=& \lambda \left| \left\{ x \in \Rbb^d  : \frac{\lambda c_d}{2} |x|^d \mathrm{sign}(t) \in [-t,t] \right\} \right| \nonumber \\
&=& \lambda \left| \left\{ x \in \Rbb^d : |x| \leq \left( \frac{2t}{\lambda c_d} \right)^{\frac{1}{d}} \right\} \right| \nonumber \\
&=& \lambda \left| \brm_d\left(o,\left( \frac{2t}{\lambda c_d} \right)^{\frac{1}{d}} \right) \right| = 2t
\end{eqnarray}
By the mapping theorem $\mu([-t,t]) = 2t$, which is to say that $f(\Pi_{d,\lambda})$ is a homogeneous PPP of unit intensity, \ie, $\Pi_{1,1}$.
\end{proof}
In particular, for $d=2$ this result states $\pi \lambda |\xsf_i|^2 \sim 2|\tsf_i|$.  Prop.\ \ref{pro:void} and \ref{pro:distmap} are easily seen to be consistent for $i=1$ in that they both give:
\begin{equation}
\Pbb \left( \frac{\lambda c_d}{2} |x_1|^d > r \right)
= \Pbb \left( |\xsf_1| > \left(\frac{2r}{\lambda c_d}\right)^{\frac{1}{d}} \right)
= \erm^{-2r} = \Pbb(|\tsf_1| > r).
\end{equation}
Prop.\ \ref{pro:distmap} is somewhat analogous to the standardization of normal $N(\mu,\sigma)$ RVs to $N(0,1)$, \ie, for $\xsf \sim N(\mu,\sigma)$ and $\zsf \sim N(0,1)$ the standardization of $\xsf$ is $(\xsf-\mu)/\sigma$, which is equal in distribution to $\zsf$.   Prop.\ \ref{pro:distmap} is used below in Prop.\ \ref{pro:imap} for mapping probabilities associated with functionals of distances in $\Pi_{d,\lambda}$ to probabilities associated with functionals of distances in $\Pi_{1,1}$.

\section{Shot noise (SN) processes}
\label{sec:snp}

\begin{figure}[!htbp]
\centering
\includegraphics[width=\textwidth]{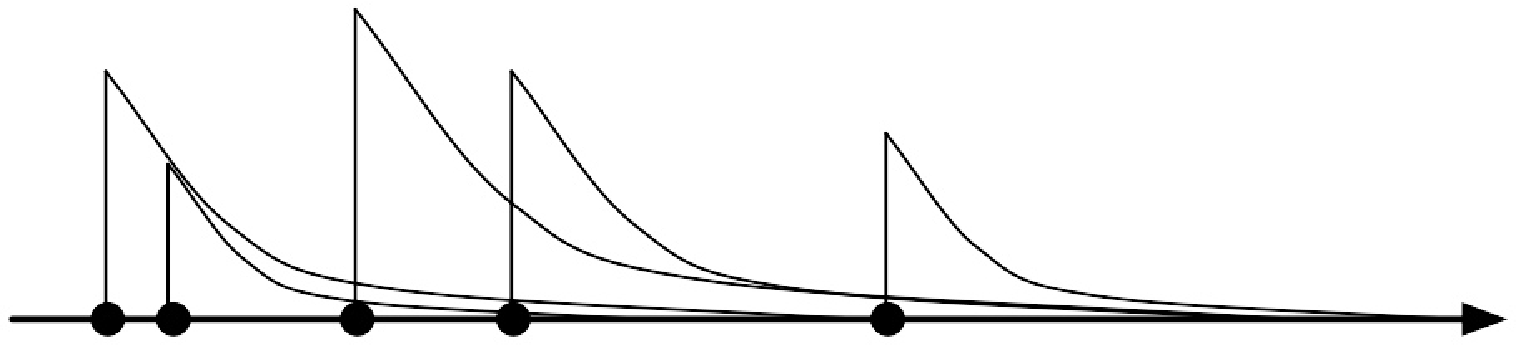}
\vspace{0.5in} \\
\includegraphics[width=\textwidth]{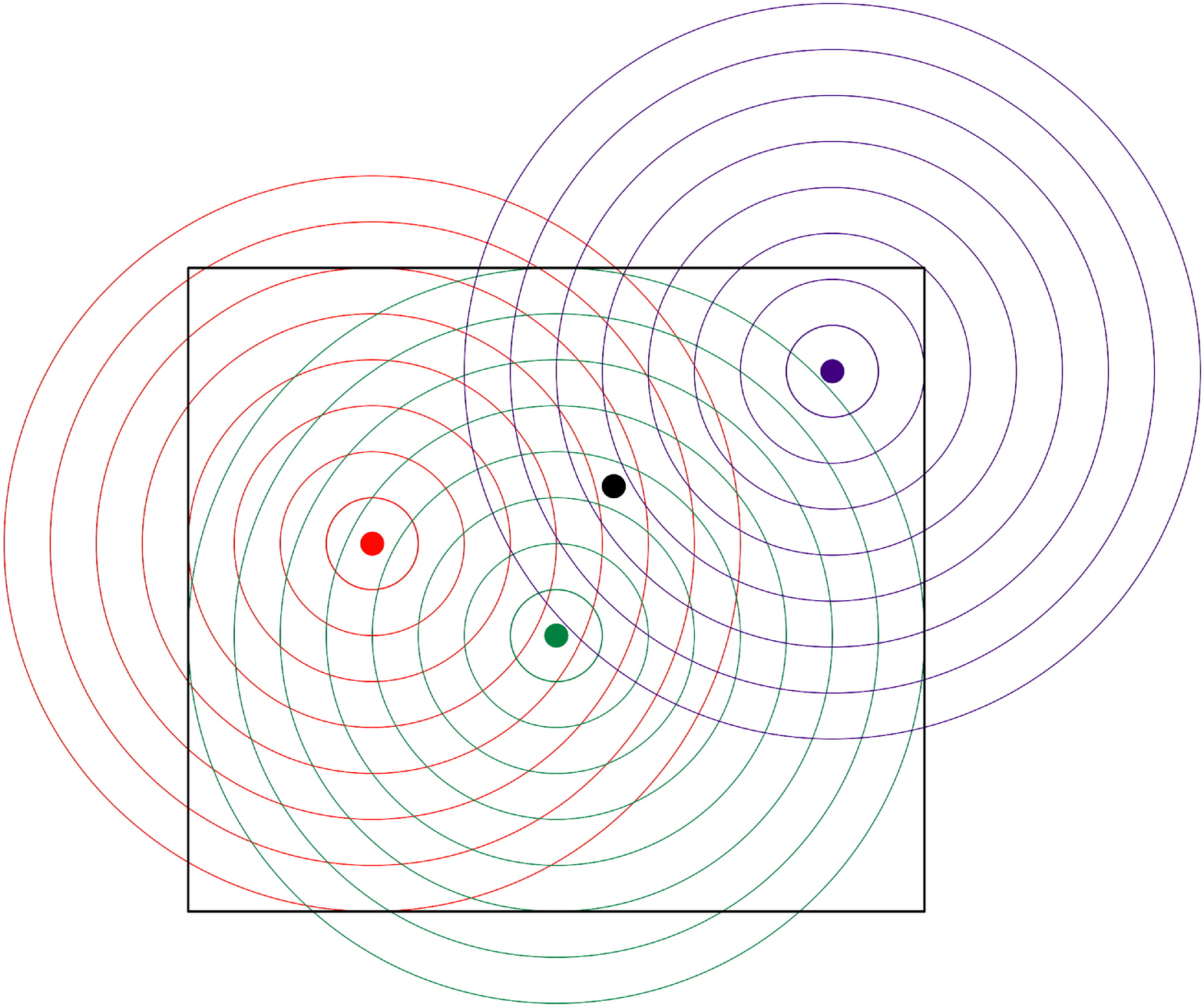}
\caption{{\bf Top:} a temporal SN process for $d=1$ is the superposition of appropriately attenuated discrete time injections of energy / noise into a system.  {\bf Bottom:} a spatial SN process for $d=2$ in the context of a wireless network is the superposition of interferences, where each interferer is appropriate attenuated through the corresponding channel.}
\label{fig:SN}
\end{figure}

Consider a system given injections of energy or noise at a sequence of random times (shock times), where each energy injection attenuates in time according to an impulse response function, so that the random cumulative energy seen at any given time $t$ is the superposition of attenuated shocks from all shock instances prior to $t$.  Such a process is termed a (temporal) shot noise (SN) process, and was first used by Schottky in 1918 \cite{Sch1918} to explain how transfers of charge at random discrete units in time in vaccuum tubes give rise to current fluctuations.  Fig.\ \ref{fig:SN} shows a sample SN process.  See \eg, \cite{BacBla2009a,BacBla2009b,Gub1996,Ric1944},\cite{Kin1993} (Ch.\  3), \cite{DalVer2003} (Vol.\ 1, Ex.\ 6.1(d)) for more information.  The following is a general definition of a SN process for arbitrary dimension $d$.
\begin{definition}
\label{def:snp}
{\bf SN process.}
A (sum) SN process is a real-valued random process $\{\Sisf(x)\}$, indexed by the continuous parameter $x \in \Rbb^d$, that is a functional of an underlying (stationary) point process $\Pi = \{\xsf_i\} \subset \Rbb^d$, where
\begin{equation}
\label{eq:isng}
\Sisf^l_{\Pi}(x) \equiv \sum_{i \in \Pi} \hsf_i l(|\xsf_i-x|), ~ x \in \Rbb^d.
\end{equation}
Here $l: \Rbb_+ \to \Rbb_+$ is a linear time-invariant impulse response function and $\{\hsf_i\}$ is a collection of i.i.d. nonnegative RVs.  A {\em max} SN process $\{\Msf(x)\}$ formed from $\Pi,l$ is
\begin{equation}
\label{eq:msn}
\Msf^l_{\Pi}(x) \equiv \max_{i \in \Pi} \hsf_i l(|\xsf_i-x|), ~ x \in \Rbb^d.
\end{equation}
Note $\Msf^l_{\Pi}(x) < \Sisf^l_{\Pi}(x)$ a.s. for each $x \in \Rbb^d$.
\end{definition}
\begin{remark}
\label{rem:indexconvention}
{\bf SN index convention.} 
Throughout this volume we write functionals of PPPs by summing over their indices rather than their points, \eg, $\sum_{i \in \Pi} f(\xsf_i)$ instead of $\sum_{\xsf_i \in \Pi} f(\xsf_i)$.  Although the latter is maybe clearer in this case, it becomes awkward for {\em marked} point processes, say $\Phi = \{(\xsf_i,\msf_i), i \in \Nbb\}$, with points $\{\xsf_i\}$ and marks $\{\msf_i\}$.  In this case, writing $\sum_{i \in \Phi} f(\xsf_i,\msf_i)$ is more clear and compact than $\sum_{(\xsf_i,\msf_i) \in \Pi} f(\xsf_i,\msf_i)$.  
\end{remark}
The case $d=1$ is most common in the stochastic process literature, but the case $d=2$ is most relevant for spatial models of wireless networks.  The interpretation of $l(|\xsf_i-x|)$ for $d=1$ is the energy injected at time $\xsf_i$ attenuated over the time interval $[\xsf_i,x]$, and thus $\Sisf^l_{\Pi}(x)$ is the superposition of all time-attenuated energy injections seen at time $x$.  The interpretation of $l(|\xsf_i-x|)$ in the context of a $d$-dimensional wireless network is the interference generated by the node at position $\xsf_i$ attenuates in space over the distance $|\xsf_i-x|$ at position $x$, and thus $\Sisf^l_{\Pi}(x)$ is the superposition of all distance-attenuated interferences seen at position $x$.  The simple LB $\Msf(x) < \Sisf(x)$ will be shown to be asymptotically tight in Prop.\ \ref{pro:intsummaxtight}, and will form the basis for the various LBs on OP and the UBs on TC.
\begin{remark}
\label{rem:radsym}
{\bf Radial symmetry.}
We have restricted Def.\ \ref{def:snp} to radially symmetric functions $l(r)$ for $r=|x|$, based on the assumption that the impact on $x$ of a shock at $\xsf_i$ depends only on the distance $|x-\xsf_i|$.  This assumption allows all integrals of $l$ over $\Rbb^d$ to be replaced by integrals over $\Rbb_+$ (see Thm.\  \ref{thm:baker}).
\end{remark}
For our purposes it suffices to consider a rather specific case.
\begin{assumption}
\label{ass:snp}
{\bf Power law impulse response.}
Assume the following for the SN process $\{\Sisf^l_{\Pi}(x)\}$ in Def.\ \ref{def:snp}:
\begin{enumerate}
\item The impulse response function is a power-law truncated around $o$
\begin{equation}
\label{eq:tpl}
l_{\alpha,\epsilon}(r) \equiv r^{-\alpha}\mathbf{1}_{r \geq \epsilon}, ~ r \in \Rbb_+,
\end{equation}
for $\alpha > 0,\epsilon \geq 0$;
\item The stationary point process $\Pi$ is a PPP $\Pi_{d,\lambda}$;
\item The amplitude RVs $\{\hsf_i\}$ are all unity;
\item We restrict our attention to the origin $\Sisf^l_{\Pi}(o)$.
\end{enumerate}
\end{assumption}
The assumption on the amplitudes $\{\hsf_i\}$ will be relaxed in \S\ref{sec:fading}.  We will employ a special notation for $\Sisf^l_{\Pi}(x)$ under Ass.\ \ref{ass:snp}.
\begin{definition}
\label{def:intsn}
{\bf Power law SN and characteristic exponent.}
The SN RVs at $o$ under Ass.\ \ref{ass:snp} are denoted
\begin{eqnarray}
\Sisf^l_{\Pi}(o) &=& \Sisf_{\Pi_{d,\lambda}}^{l_{\alpha,\epsilon}}(o) \equiv \Sisf_{d,\lambda}^{\alpha,\epsilon}(o) \nonumber \\
\Msf^l_{\Pi}(o) &=& \Msf_{\Pi_{d,\lambda}}^{l_{\alpha,\epsilon}}(o) \equiv \Msf_{d,\lambda}^{\alpha,\epsilon}(o) \label{eq:ig}
\end{eqnarray}
The characteristic exponent of $\Sisf_{d,\lambda}^{\alpha,\epsilon}(o)$ is defined as
\begin{equation}
\label{eq:cexp}
\delta \equiv \frac{d}{\alpha}.
\end{equation}
\end{definition}
\begin{remark}
\label{rem:path}
{\bf Pathloss attenuation and the singularity at the origin.}
The impulse response function $l_{\alpha,0}(r) = r^{-\alpha}$ is a common choice to model the attenuation due to pathloss in wireless communication but suffers the drawback of modeling amplification $|x|^{-\alpha} > 1$ rather than attenuation of received energy at distances $|x| < 1$, and in fact this amplification grows without bound as $|x| \to 0$; this is further discussed in Prop.\ \ref{pro:camz} below, in \cite{HaeGan2008} (p.\ 24), and in \cite{InaWic2009}.  Generalizing $l_{\alpha,0}$ by truncating around the origin $l_{\alpha,\epsilon}$ removes this singularity.  Note $l_{\alpha,\epsilon}$ is used in \cite{HaeGan2008} (\S3.7.1) to model carrier sense multiple access (CSMA).
\end{remark}
The max SN RV $\Msf^{\alpha,0}_{d,\lambda}(o)$ will be shown to obey the Frech\'{e}t distribution, defined below.
\begin{definition}
\label{def:frechet}
The {\bf Frech\'{e}t distribution} with parameters $\gamma > 0$, $\sigma > 0$, and $\mu \in \Rbb$ has CDF
\begin{equation}
F_{\xsf}(x) \equiv \exp \left\{ - \left( \frac{x-\mu}{\sigma} \right)^{-\gamma} \right\}, ~ x \geq \mu,
\end{equation}
and for $\mu = 0$ has moments up to order $\gamma$:
\begin{equation}
\Ebb[\xsf^p] = \left\{ \begin{array}{ll}
\sigma^p \Gamma(1-p/\gamma), \; & p < \gamma \\
\infty, \; & \mbox{else} \end{array} \right., ~ \mu = 0.
\end{equation}
The Frech\'{e}t is one of three extreme value distributions \cite{KotNad2001}.
\end{definition}
The max SN RV $\Msf^{\alpha,\epsilon}_{d,\lambda}(o)$ CDF is immediate from Prop.\ \ref{pro:void}.
\begin{corollary}
\label{cor:maxsnrvdis}
The {\bf max SN RV CDF} is
\begin{equation}
\label{eq:maxsnrvdis}
\Pbb\left( \Msf^{\alpha,\epsilon}_{d,\lambda}(o) \leq y \right) = \left\{ \begin{array}{ll}
\exp \left\{ - \lambda c_d \left( y^{-\delta} - \epsilon^d \right) \right\}, \; & 0 \leq y \leq \epsilon^{-\alpha} \\
1, \; & \mbox{else} \end{array} \right.
\end{equation}
For $\epsilon = 0$ the max SN RV has the Frech\'{e}t distribution in Def.\ \ref{def:frechet} with $\gamma = \delta$, $\sigma = (\lambda c_d)^{\frac{1}{\gamma}}$ and $\mu = 0$.
\end{corollary}
Having characterized $\Msf^{\alpha,\epsilon}_{d,\lambda}(o)$, we now focus on characterizing the RV $\Sisf_{d,\lambda}^{\alpha,\epsilon}(o)$, as it represents the aggregate interference seen at a typical location when interferers are positioned according to $\Pi_{d,\lambda}$ and the  pathloss attenuation function $l_{\alpha,\epsilon}$ is assumed.  Cases of particular interest are
\begin{equation}
\label{eq:is}
\Sisf_{2,\lambda}^{\alpha,0}(o), \Sisf_{1,1}^{\alpha,\epsilon}(o), \Sisf_{1,1}^{\alpha,0}(o), \Sisf_{1,1}^{2,0}(o).
\end{equation}
Many results will hold provided the characteristic exponent $\delta < 1$.  For the important case $d=2$ this translates to $\alpha > 2$.  Prop.\ \ref{pro:distmap} is used below to show that it suffices to consider $\Sisf_{1,1}^{\alpha,\epsilon}(o)$.
\begin{proposition}
\label{pro:imap}
{\bf Interference mapping.}
The following RVs are equal in distribution
\begin{equation}
\Sisf_{d,\lambda}^{\alpha,\epsilon}(o) \stackrel{\drm}{=} \left(\frac{\lambda c_d}{2} \right)^{\frac{\alpha}{d}} \Sisf_{1,1}^{\frac{\alpha}{d},\lambda c_d \epsilon^d/2}(o),
~~ \mbox{and} ~~
\Sisf_{d,\lambda}^{\alpha,0}(o) \stackrel{\drm}{=} \left(\frac{\lambda c_d}{2} \right)^{\frac{\alpha}{d}} \Sisf_{1,1}^{\frac{\alpha}{d},0}(o).
\end{equation}
\end{proposition}
\begin{proof}
Using Prop.\ \ref{pro:distmap} gives:
\begin{equation}
\Sisf_{d,\lambda}^{\alpha,0}(o) = \left(\frac{\lambda c_d}{2} \right)^{\frac{\alpha}{d}}  \sum_{i \in \Pi_{d,\lambda}}
\left( \frac{\lambda c_d}{2} |\xsf_i|^d \right)^{-\frac{\alpha}{d}}
= \left(\frac{\lambda c_d}{2} \right)^{\frac{\alpha}{d}} \sum_{i \in \Pi_{1,1}} |\tsf_i|^{-\frac{\alpha}{d}}.
\end{equation}
For $\Sisf_{d,\lambda}^{\alpha,\epsilon}(o)$ simply observe $\{|\xsf_i| > \epsilon\} = \{|\tsf_i| > \lambda c_d \epsilon^d/2\}$.
\end{proof}
Prop.\ \ref{pro:imap} is important because it expresses the SN RV formed from $\Pi_{d,\lambda}$ with exponent $\alpha$ as a scaling of a SN RV formed from $\Pi_{1,1}$ with exponent $\alpha/d$.  In this sense $d$ and $\lambda$ are inessential parameters.

The next result is called the Campbell-Mecke Theorem; the version below is a special case of a much more general theory on moments of functionals of PPPs (see \eg, Thm.\ A.2 and Lem.\  A.3 in \cite{HaeGan2008}).  Our specialization is to homogeneous PPPs with measure $\Lambda(\drm x) = \lambda \drm x$, and to radially symmetric functions $l(|x|)$.  As mentioned in Rem.\ \ref{rem:radsym} and \ref{rem:path}, this assumption is natural for wireless networks, and has the advantage of allowing the $d$-dimensional integrals to be replaced with single dimensional integrals using the following theorem (from \cite{Bak1999}).
\begin{theorem}
\label{thm:baker}
{\bf Integration of radially symmetric functions} (\cite{Bak1999}).
Let $l:\Rbb_+ \to \Rbb$ be Riemann integrable on $\Rbb_+$, and let $\tilde{l}(x) = l(|x|)$ for $x \in \Rbb^d$.  Then $\tilde{l}$ is Riemann integrable on $\Rbb^d$ and
\begin{equation}
\label{eq:baker}
\int_{\Rbb^d} \tilde{l}(x)\drm x = d c_d \int_0^{\infty} l(r) r^{d-1} \drm r.
\end{equation}
\end{theorem}
The above theorem is used to simplify the integrals in the Campbell-Mecke Theorem below and in several other places throughout this monograph.
\begin{theorem}
\label{thm:cam}
{\bf Campbell-Mecke} (\cite{HaeGan2008} Thm.\ A.2, \cite{BacBla2009a} Thm. 1.11, \cite{Kin1993} \S3.2).
The mean and variance of the RV $\Sisf^l_{\Pi}(y)$ in \eqref{eq:isng} for PPP $\Pi_{d,\lambda}$ and measurable $l : \Rbb_+ \to \Rbb_+$ are:
\begin{eqnarray}
\Ebb[\Sisf^l_{\Pi_{d,\lambda}}(o)] &=& \lambda \Ebb[\hsf] \int_{\Rbb^d} l(|x|) \drm x = \lambda d c_d \Ebb[\hsf] \int_0^{\infty} l(r) r^{d-1} \drm r \nonumber \\
\mathrm{Var}(\Sisf^l_{\Pi_{d,\lambda}}(o)) &=& \lambda \Ebb[\hsf^2] \int_{\Rbb^d} l(|x|)^2 \drm x = \lambda d c_d \Ebb[\hsf^2] \int_0^{\infty} l(r)^2 r^{d-1} \drm r \nonumber \\
\end{eqnarray}
\end{theorem}
The proof of Thm.\ \ref{thm:cam} is essentially an exchange of the order of integration and summation and is omitted.  In particular, for $l_{\alpha,0}(x) = |x|^{-\alpha}$ we change variables from $|x|$ for $x \in \Rbb^d$ to $r \in \Rbb_+$ to exploit the radial symmetry of the function $l$:
\begin{equation}
\label{eq:intinf}
\Ebb[\Sisf_{d,\lambda}^{\alpha,0}(o)] = \lambda \int_{\Rbb^d} |x|^{-\alpha} \drm x = \lambda d c_d \int_0^{\infty} r^{-\alpha} r^{d-1} \drm r = \left. \frac{\lambda d c_d}{d-\alpha} r^{d-\alpha} \right|_0^{\infty}.
\end{equation}
As discussed in \cite{HaeGan2008} ((3.4) p.24) this integral diverges for $\alpha < d$ due to the upper limit of integration, and for $\alpha > d$ it diverges due to the lower limit of integration, which in turn is attributable to the singularity at the origin of the function $r^{-\alpha}$.  As stated in Ass.\ \ref{ass:snp}, we use $l_{\alpha,\epsilon}(r) = r^{-\alpha}\mathbf{1}_{r \geq \epsilon}$ for $r \in \Rbb_+$, which can be interpreted as assuming a Rx has perfect interference cancellation within a ball of radius $\epsilon$.  The following proposition summarizes this discussion.
\begin{proposition}
\label{pro:camz}
{\bf SN mean and variance.}
The means and variances of the SN RVs \eqref{eq:ig} are:
\begin{eqnarray}
\Ebb[\Sisf_{d,\lambda}^{\alpha,\epsilon}(o)] &=& \left\{ \begin{array}{ll}
\infty, \; & \alpha < d \mbox{ or } \epsilon = 0\\
\frac{\lambda d c_d }{\alpha-d} \epsilon^{d-\alpha}, \; & \mbox{else} \end{array} \right. \nonumber \\
\mathrm{Var}(\Sisf_{d,\lambda}^{\alpha,\epsilon}(o)) &=& \left\{ \begin{array}{ll}
\infty, \; & \alpha < d/2 \mbox{ or } \epsilon = 0\\
\frac{\lambda d c_d }{2\alpha-d} \epsilon^{d-2\alpha}, \; & \mbox{else} \end{array} \right.
\end{eqnarray}
\end{proposition}
The proof is a straightforward modification of \eqref{eq:intinf}.  Prop.\ \ref{pro:camz} will be used along with the Markov and Chebychev inequalities in \S\ref{sec:mccine} to form the LBs on TC (UBs on OP) in \S\ref{sec:lbTC}.

The last SN specific result we will have use for is the series expansion of the PDF and CDF of a SN RV for $d=1$.  The following result is adapted\footnote{Eq. \eqref{eq:intserexp} have a factor of $2$ in front of $\lambda$ not present in \cite{LowTei1990} (29) due to the fact that their impulse response function (4) does not count contributions from $t < 0$.} from \cite{LowTei1990} (Eq. (29)).
\begin{proposition}
\label{pro:snserexp}
{\bf SN series expansion} (\cite{LowTei1990}).
The series expansions of the PDF and CCDF of the RV $\Sisf^{\alpha,0}_{1,\lambda}$ for $\delta = \frac{1}{\alpha} < 1$ are:
\begin{eqnarray}
f_{\Sisf^{\alpha,0}_{1,\lambda}}(y) &=& \frac{1}{\pi y} \sum_{n=1}^{\infty} \frac{(-1)^{n+1}}{n!} \Gamma(1+n \delta) \sin (\pi n \delta)(2\lambda \Gamma(1-\delta) y^{-\delta})^n \nonumber \\
\bar{F}_{\Sisf^{\alpha,0}_{1,\lambda}}(y) &=& \frac{1}{\pi \delta} \sum_{n=1}^{\infty} \frac{(-1)^{n+1}}{n n!} \Gamma(1+n \delta) \sin(\pi n \delta)(2\lambda \Gamma(1-\delta) y^{-\delta})^n \nonumber \\
\label{eq:intserexp}
\end{eqnarray}
\end{proposition}
The asymptotic PDF and CCDF as $y \to \infty$ is immediate upon taking the dominant $n=1$ term from the above expansions.
\begin{corollary}
\label{cor:snasypdf}
The {\bf Asymptotic PDF and CCDF of the SN RV} $\Sisf^{\alpha,0}_{1,\lambda}$ as $y \to \infty$ for $\delta = \frac{1}{\alpha} < 1$ are:
\begin{eqnarray}
f_{\Sisf^{\alpha,0}_{1,\lambda}}(y) &=& 2\lambda \delta y^{-1-\delta} + \Omc(y^{-1-2\delta}), ~ y \to \infty \nonumber \\
\bar{F}_{\Sisf^{\alpha,0}_{1,\lambda}}(y) &=& 2\lambda y^{-\delta} + \Omc(y^{-2 \delta}), ~ y \to \infty
\end{eqnarray}
\end{corollary}
Employ \eqref{eq:gamident} to simplify the $n=1$ term.  Cor.\ \ref{cor:snasypdf} gives the asymptotic approximations of OP (as $\lambda \to 0$) and TC (as $q^* \to 0$) in Prop.\ \ref{pro:asymoptc}.

\section{Stable distributions, Laplace transforms, and PGFL}
\label{sec:stadis}

A RV is said to be stable if iid sums of that RV are equal to an affine function of the original RV (closure under summation).  Equivalently, a distribution is stable if convolutions of the distribution yield a translation and/or scaling of the distribution  (closure under convolution).
\begin{definition}
\label{def:stadis}
{\bf Stable RV and distribution.}
Let $\xsf \sim F$ and $(\xsf_1,\ldots,\xsf_n)$ be iid from $F$.  Say $\xsf$ is a {\em stable} RV ($F$ is a stable distribution) if for each $n \in \Nbb$ there exists numbers $(a_n,b_n)$ such that
\begin{equation}
a_n \xsf + b_n \stackrel{\drm}{=} \xsf_1 + \cdots + \xsf_n.
\end{equation}
Moreover, if it exists, $a_n = n^{1/\delta}$ for a characteristic exponent $\delta \in [0,2]$.  If $b_n = 0$ for all $n$ then $\xsf$ ($F$) is a {\em strictly stable} RV (distribution).
\end{definition}
See \eg, \cite{Nol2012} Def.\ 1.5.  Perhaps the most familiar example of a stable distribution is the normal: let $(\xsf,\xsf_1,\ldots,\xsf_n) \sim N(\mu,\sigma)$ be independent normal RVs with mean $\mu$ and standard deviation $\sigma$.  Then $\xsf_1 + \cdots + \xsf_n \sim N(n \mu, \sqrt{n} \sigma)$ and $a \xsf + b \sim N(a \mu + b, a \sigma)$.  Choose $a_n = \sqrt{n}$ and $b_n = (n-\sqrt{n})\mu$ to satisfy the requirement in Def.\ \ref{def:stadis}.

Aside from a few special cases (the normal, Cauchy, and L\'{e}vy distributions), a stable RV $\xsf$ does not admit a closed form CDF.  It does, however, have a special form for its CF $\phi[\xsf](t) = \Ebb[\erm^{\irm t \xsf}]$ for $t \in \Rbb$.  We are interested in a specific sub-class of stable RVs appropriate for modeling SN RVs, and consequently the definitions in the remainder of this section are specialized to that class.
\begin{definition}
\label{def:staparam}
{\bf Stable CF.}
The RV $\xsf$ is stable with characteristic exponent $\delta \in (0,1)$, dispersion coefficient $\gamma > 0$ if it has CF
\begin{equation}
\label{eq:staparam}
\phi[\xsf](t) \equiv
\exp\left\{ -\gamma^{\delta} |t|^{\delta}(1-\irm \tan (\pi \delta/2) \mathrm{sign}(t)) \right\}, ~ t \in \Rbb.
\end{equation}
\end{definition}
The above definition is adopted from Def.\  1.8 from \cite{Nol2012}.  It is specialized in that we have fixed the location parameter at $0$ and the skewness parameter at its maximum value of $1$ (``totally skewed to the right'' in \cite{Nol2012} p.12).  The support of this RV is $\Rbb_+$ (\cite{Nol2012} Lem.\  1.10).  As discussed in \cite{Nol2012} \S1.3, there is a wide variety of parameterizations for stable distributions found in the technical literature, and their differences may easily lead to confusion.

The case $\delta = \frac{1}{2}$ corresponds to the L\'{e}vy distribution (\cite{Nol2012} \S1.1).
\begin{definition}
\label{def:levy}
The {\bf L\'{e}vy distribution} with parameter $\gamma \geq 0$ and support $\Rbb_+$ has PDF, CDF, and CF
\begin{eqnarray}
f_{\xsf}(x) &\equiv& \sqrt{\frac{\gamma}{2 \pi}} \frac{\erm^{-\frac{\gamma}{2x}}}{x^{\frac{3}{2}}}, ~ x \in \Rbb_+ \nonumber \\
F_{\xsf}(x) &\equiv& 2 \bar{F}_Z \left( \sqrt{\frac{\gamma}{x}} \right), ~ x \in \Rbb_+ \nonumber \\
\phi[\xsf](t) &\equiv& \erm^{-\sqrt{-2 \irm \gamma t}}, ~ t \in \Rbb \label{eq:cflev}
\end{eqnarray}
where $\bar{F}_{\zsf}$ is the normal $N(0,1)$ CCDF.
\end{definition}
Observe from \eqref{eq:cflev} and Def.\ \ref{def:staparam} that a L\'{e}vy RV is stable with characteristic exponent $\delta = \frac{1}{2}$ and dispersion coefficient $\gamma$.  To see the equivalence between \eqref{eq:staparam} with $\delta = 1/2$ and \eqref{eq:cflev} note:
\begin{equation}
\tan (\pi/4) = 1, ~
\phi[\xsf](t)= \left\{ \begin{array}{ll}
\erm^{-\sqrt{-\gamma t}(1+\irm)}, \; & t < 0 \\
\erm^{-\sqrt{\gamma t}(1-\irm)}, \; & t > 0
\end{array} \right., ~ 1\pm \irm = \sqrt{\pm 2 \irm}.
\end{equation}
Fig.\ \ref{fig:levpdf} shows the L\'{e}vy PDF and CCDF.  Note the heavy tail in the right plot compared with the light tailed normal distribution.  The characteristic exponent $\delta$ fixes which moments of a stable RV are finite.
\begin{proposition}
\label{pro:stamom}
{\bf Stable moments.}
For $\xsf$ stable with characteristic exponent $\delta$ and $p \in \Rbb_+$:
\begin{equation}
\Ebb[|\xsf|^p] \left\{ \begin{array}{ll}
< \infty, \; & p < \delta < 2 \\
=\infty, \; & \mbox{else} \end{array} \right. .
\end{equation}
\end{proposition}
For example, a L\'{e}vy RV $\xsf$ has $\Ebb[\xsf^p]=\infty$ for $p \geq \delta = \frac{1}{2}$.
\begin{figure}[!htbp]
\centering
\includegraphics[width=0.49\textwidth]{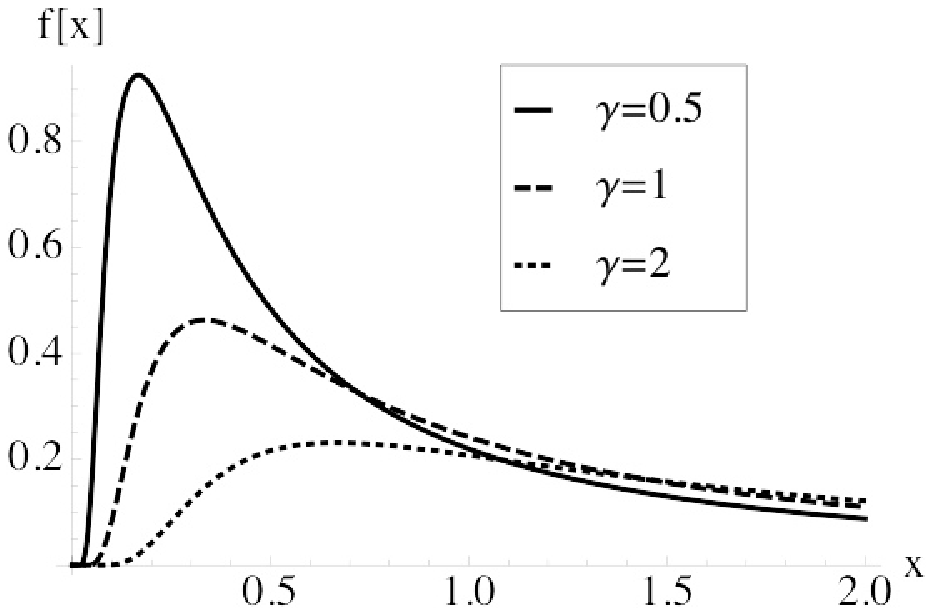}
\includegraphics[width=0.49\textwidth]{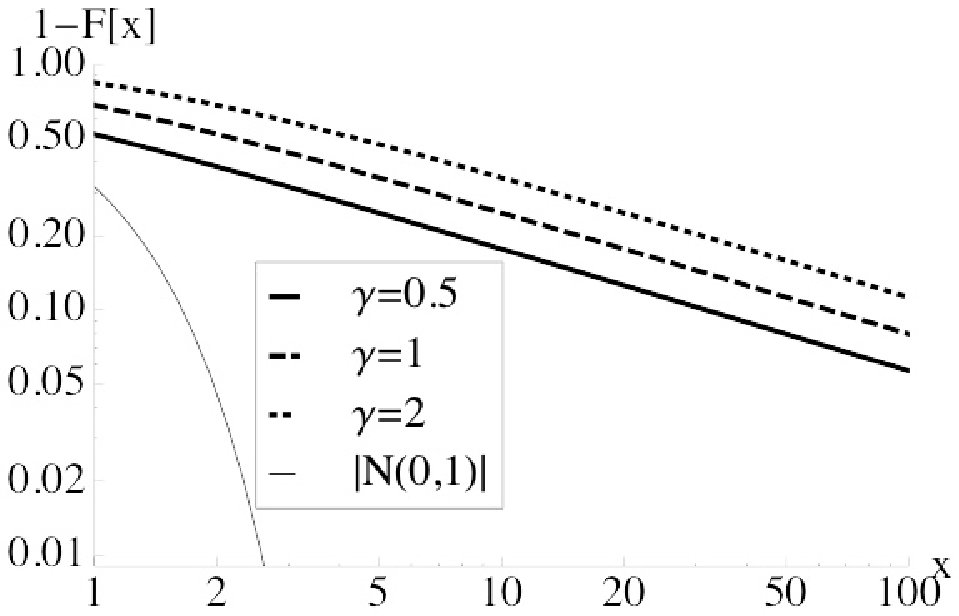}
\caption{The PDF (left) and CCDF (right) for the L\'{e}vy distribution with parameter $\gamma \in \{0.5,1,2\}$.  The CCDF is on log-log scale; the normal $|N(0,1)|$ CCDF is shown as well. }
\label{fig:levpdf}
\end{figure}

The notation for the characteristic exponent $\delta$ in Def.\ \ref{def:staparam} has been deliberately chosen to coincide with the notation for the characteristic exponent $\delta$ of $\Sisf_{d,\lambda}^{\alpha,\epsilon}(o)$ in Def.\ \ref{def:intsn}.  Our objective in this section is to characterize conditions under which $\Sisf^{\alpha,\epsilon}_{d,\lambda}(o)$ in \eqref{eq:ig} is stable.  Cor.\ \ref{cor:lapint} will show $\Sisf_{d,\lambda}^{\alpha,0}(o)$ is stable with characteristic exponent $\delta = d/\alpha < 1$.  To get to this result, we must find the CF of $\Sisf^{\alpha,\epsilon}_{d,\lambda}(o)$ and compare with Def.\ \ref{def:staparam}.  We will show that the LT of a SN RV $\Sisf^l_{\Pi}(o)$ in \eqref{eq:isng} is expressible in terms of the probability generating functional (PGFL) of the underlying point process $\Pi$, and that this PGFL admits a closed form for PPP $\Pi_{d,\lambda}$ and the truncated power law function $l_{\alpha,\epsilon}$ \eqref{eq:tpl}.  We start by defining the PGFL (\cite{HaeGan2008} Def.\ A.5).  Our methodology in the development that follows is to present results in a somewhat general form and then specialize as needed.  Thus we distinguish between $i)$ a stationary point process $\Pi$ and a PPP $\Pi_{d,\lambda}$, $ii)$ generic measurable functions $l,\nu$ and the specific function $l_{\alpha,\epsilon}$ \eqref{eq:tpl}, and SN RVs $\Sisf^l_{\Pi}(o)$ \eqref{eq:isng}, $\Sisf^l_{\Pi_{d,\lambda}}(o)$, and $\Sisf^{\alpha,\epsilon}_{d,\lambda}(o)$ \eqref{eq:ig}.  We aim to clarify the impact of the assumptions of a PPP $\Pi_{d,\lambda}$ and a particular pathloss function $l_{\alpha,\epsilon}$.
\begin{definition}
\label{def:pgfl}
{\bf Point process PGFL.}
The probability generating functional (PGFL) of a point process $\Pi$ and a measurable function $\nu : \Rbb^d \to \Rbb_+$ is defined as
\begin{equation}
\Gmc[\Pi,\nu] \equiv \Ebb \left[ \prod_{i \in \Pi} \nu(x_i) \right].
\end{equation}
\end{definition}
If the point process $\Pi$ is a PPP $\Pi_{d,\lambda}$ the PGFL simplifies (\cite{HaeGan2008} (A.3)).
\begin{proposition}
\label{pro:pgfl}
{\bf PPP PGFL.}
For a PPP $\Pi_{d,\lambda}$ the PGFL is
\begin{equation}
\Gmc[\Pi_{d,\lambda},\nu] = \exp \left\{ - \lambda \int_{\Rbb^d} (1-\nu(x)) \drm x \right\}.
\end{equation}
\end{proposition}
The PGFL yields the LT of a SN RV.
\begin{corollary}
\label{cor:pgflsn}
{\bf Point process SN LT.}
The SN RV $\Sisf^l_{\Pi}(o)$ in \eqref{eq:isng} for a stationary point process $\Pi$ with each $\hsf_i = 1$ has a LT expressible in terms of its PGFL:
\begin{equation}
\Lmc[\Sisf^l_{\Pi}(o)](s) = \Gmc[\Pi,\erm^{-s l(|\cdot|)}], ~ s \in \Cbb.
\end{equation}
\end{corollary}
\begin{proof}
\begin{equation}
\Ebb \left[ \erm^{-s \Sisf^l_{\Pi}(o)} \right]
= \Ebb \left[ \exp \left\{-s \sum_{i \in \Pi} l(|\xsf_i|) \right\} \right]
= \Ebb \left[ \prod_{i \in \Pi} \erm^{-s l(|\xsf_i|)} \right].
\end{equation}
\end{proof}
When the point process $\Pi$ is a PPP $\Pi_{d,\lambda}$ we can combine Prop.\ \ref{pro:pgfl} and Cor.\  \ref{cor:pgflsn} to get the LT of $I^l_{\Pi_{d,\lambda}}$.
\begin{corollary}
\label{cor:pgflsnppp}
{\bf PPP SN LT.}
The SN RV $\Sisf^l_{\Pi_{d,\lambda}}(o)$ in \eqref{eq:isng} for a PPP $\Pi_{d,\lambda}$ with each $\hsf_i = 1$ has a LT
\begin{equation}
\label{eq:pgflsnppp}
\Lmc[\Sisf^l_{\Pi_{d,\lambda}}(o)](s) = \exp \left\{ - \lambda d c_d \int_0^{\infty} \left(1-\erm^{-s l(r)} \right) r^{d-1} \drm r \right\},
\end{equation}
for all $s \in \Cbb$ for which the integral exists.
\end{corollary}
\begin{proof}
Prop.\ \ref{pro:pgfl} and Cor.\  \ref{cor:pgflsn} yield:
\begin{equation}
\Lmc[\Sisf^l_{\Pi_{d,\lambda}}(o)](s) = \exp \left\{ - \lambda \int_{\Rbb^d} \left(1-\erm^{-s l(|x|)} \right) \drm x \right\}.
\end{equation}
Now change variables from $|x|$ to $r$ using Thm.\  \ref{thm:baker}.
\end{proof}
We next fix $l$ to be $l_{\alpha,\epsilon}$ \eqref{eq:tpl} and obtain the MGF of $\Sisf^{\alpha,\epsilon}_{d,\lambda}(o)$.  Recall that we can obtain the MGF from the LP \eqref{eq:ltcfmgf}.
\begin{corollary}
\label{cor:mgfint}
{\bf Pathloss SN MGF.}
For $\delta < 1$ and $\epsilon > 0$ the MGF of $\Sisf^{\alpha,\epsilon}_{d,\lambda}(o)$ in Def.\ \ref{def:intsn} is (for $\theta \in \Rbb_+$):
\begin{equation}
\label{eq:mgfint}
\Mmc[\Sisf^{\alpha,\epsilon}_{d,\lambda}(o)](\theta)
= \exp \left\{ \frac{\lambda d c_d}{\alpha} \int_0^{\epsilon^{-\alpha}} \left( \erm^{\theta y}-1 \right) y^{-\delta-1} \drm y \right\}.
\end{equation}
\end{corollary}
\begin{proof}
Substituting $l_{\alpha,\epsilon}$ \eqref{eq:tpl} into \eqref{eq:pgflsnppp} gives
\begin{equation}
\Lmc[\Sisf^{\alpha,\epsilon}_{d,\lambda}(o)](s)
= \exp \left\{ - \lambda d c_d \int_{\epsilon}^{\infty} \left(1-\erm^{-s r^{-\alpha}} \right) r^{d-1} \drm r \right\}.
\end{equation}
The change in variable:
\begin{equation}
y = r^{-\alpha}, r = y^{-\frac{1}{\alpha}}, r^{d-1} = y^{-\frac{d-1}{\alpha}}, \drm r = -\frac{1}{\alpha} y^{-\frac{1}{\alpha}-1} \drm y,
\end{equation}
yields
\begin{equation}
\label{eq:lapinteps}
\Lmc[\Sisf^{\alpha,\epsilon}_{d,\lambda}(o)](s)
= \exp \left\{ - \frac{\lambda d c_d}{\alpha} \int_0^{\epsilon^{-\alpha}} \left( 1-\erm^{-s y} \right) y^{-\delta-1} \drm y \right\}.
\end{equation}
Specializing to the assumed $s = - \theta$ for $\theta \in \Rbb_+$ yields the proposition.  For $\delta > 1$ or $\epsilon = 0$ this quantity diverges.
\end{proof}
Finally, we fix $l$ to be $l_{\alpha,0}$ \eqref{eq:tpl} and obtain the CF of $\Sisf^{\alpha,0}_{d,\lambda}(o)$ \eqref{eq:ig}.  Recall that we can obtain the CF from the LP \eqref{eq:ltcfmgf}.
\begin{corollary}
\label{cor:lapint}
{\bf Pathloss SN CF.}
For $\delta < 1$ the CF of $\Sisf^{\alpha,0}_{d,\lambda}(o)$ in Def.\ \ref{def:intsn} is (for $t \in \Rbb$):
\begin{equation}
\phi[\Sisf^{\alpha,0}_{d,\lambda}(o)](t) = \exp \left\{ \lambda \delta c_d \Gamma(-\delta) \cos(\pi \delta/2) |t|^{\delta} \left(1 - \irm \tan(\pi \delta/2) \mathrm{sign}(t) \right) \right\}. \label{eq:cfint}
\end{equation}
In particular, $\Sisf^{\alpha,0}_{d,\lambda}(o)$ is stable as in Def.\ \ref{def:staparam} with stability coefficient $\delta < 1$ and dispersion coefficient
\begin{equation}
\label{eq:intstadis}
\gamma^{\delta} = -\lambda \delta c_d \Gamma(-\delta) \cos(\pi \delta/2).
\end{equation}
\end{corollary}
\begin{proof}
Fix $\epsilon = 0$ and $s = - \irm t$ for $t \in \Rbb$ in \eqref{eq:lapinteps}, and assume $\delta < 1$:
\begin{equation}
\phi[\Sisf^{\alpha,0}_{d,\lambda}(o)](t) = \exp \left\{ - \lambda \delta c_d |t|^{\delta-1} \Gamma(-\delta) \left( - |t| \cos(\pi \delta/2) + \irm t \sin(\pi \delta/2) \right) \right\}.
\end{equation}
Straightforward manipulations give \eqref{eq:cfint}.
\end{proof}
For the special case of $\delta = \frac{1}{2}$ we can apply Def.\ \ref{def:levy}.
\begin{corollary}
\label{cor:levint}
{\bf Pathloss SN for $\delta=\frac{1}{2}$.}
For $\delta = \frac{1}{2}$, \eqref{eq:intstadis} gives $\gamma = \frac{\pi}{2} (\lambda c_d)^2$ and the RV $\Sisf^{\alpha,0}_{d,\lambda}(o)$ is L\'{e}vy with parameter $\gamma$.  In particular, $\gamma = 2 \pi$ for $\Sisf^{2,0}_{1,1}(o)$.
\end{corollary}

\section{Maximums and sums of RVs}
\label{sec:maxsum}

To finish this chapter, we combine several previous results to illustrate the asymptotic tightness of the LB $\Msf^{\alpha,0}_{d,\lambda}(o) < \Sisf^{\alpha,0}_{d,\lambda}(o)$.
\begin{proposition}
\label{pro:intsummaxtight}
{\bf Sum and max SN CCDF ratio.}
For $\epsilon = 0$ and $\delta < 1$ the CCDFs of the RVs $\Msf^{\alpha,0}_{d,\lambda}(o)$ and $\Sisf^{\alpha,0}_{d,\lambda}(o)$ have a ratio that converges to unity:
\begin{equation}
\lim_{y \to \infty} \frac{\Pbb(\Sisf^{\alpha,0}_{d,\lambda}(o) > y)}{\Pbb(\Msf^{\alpha,0}_{d,\lambda}(o) > y)} = 1.
\end{equation}
\end{proposition}
\begin{proof}
First apply Prop.\ \ref{pro:imap} and then apply Cor.\ \ref{cor:snasypdf} to $\Sisf^{\alpha,0}_{d,\lambda}(o)$:
\begin{eqnarray}
\Pbb\left( \Sisf^{\alpha,0}_{d,\lambda}(o) > y \right)
&=& \Pbb\left( \left( \frac{\lambda c_d}{2} \right)^{\frac{1}{\delta}} \Sisf^{\alpha,0}_{1,1}(o) > y \right) \nonumber \\
&=& \Pbb\left( \Sisf^{\alpha,0}_{1,1}(o) > \left( \frac{\lambda c_d}{2} \right)^{-\frac{1}{\delta}} y \right) \nonumber \\
&=& \lambda c_d y^{-\delta} + \Omc(y^{-2\delta}) \label{eq:intsumccdfasym}
\end{eqnarray}
Using Cor.\ \ref{cor:maxsnrvdis}, the first order series expansion of the CDF of $\Msf^{\alpha,0}_{d,\lambda}(o)$ is:
\begin{equation}
\label{eq:intmaxccdfasym}
\Pbb \left( \Msf^{\alpha,0}_{d,\lambda}(o) > y \right) = 1 - \erm^{-\lambda c_d y^{-\delta}} = \lambda c_d y^{-\delta} + \Omc(y^{-2\delta}).
\end{equation}
The limit of the ratio of the CCDFs \eqref{eq:intsumccdfasym} and \eqref{eq:intmaxccdfasym} as $y \to \infty$ is one.
\end{proof}
Fig.\ \ref{fig:intsummaxtight} shows the (exact) CCDF for $\Sisf$ (for $\delta = 1/2$) from Cor.\ \ref{cor:lapint}, the asymptotic CCDF for $\Sisf^{\alpha,0}_{d,\lambda}(o)$ from Cor.\ \ref{cor:snasypdf}, and the CCDF for $\Msf^{\alpha,0}_{d,\lambda}(o)$ from Cor.\ \ref{cor:maxsnrvdis}.  Observe the ratio of the CCDFs appears to converge to one as $y \to \infty$.  This convergence will be used to establish the asymptotic tightness of the OP LBs and TC UBs in what follows.
\begin{figure}[!htbp]
\centering
\includegraphics[width=0.75\textwidth]{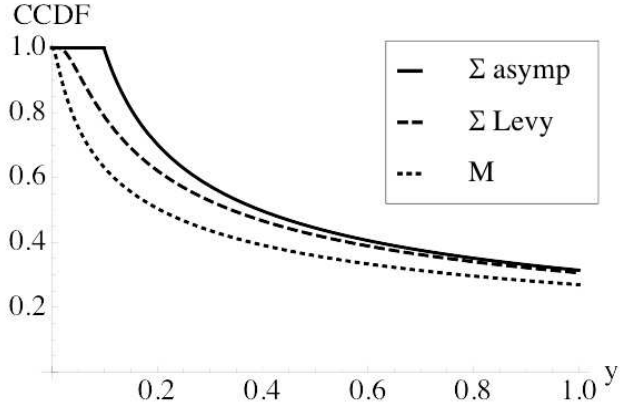}
\caption{The exact and asymptotic CCDFs for the RV $\Sisf^{\alpha,0}_{d,\lambda}(o)$ and the CCDF for the RV $\Msf^{\alpha,0}_{d,\lambda}(o)$.  The parameter values are $\epsilon = 0$, $\lambda = 1/10$, $d=2$, and $\alpha = 4$ ($\delta = 1/2$).  The $\Msf^{\alpha,0}_{d,\lambda}(o)$ CCDF is asymptotically equal to the $\Sisf^{\alpha,0}_{d,\lambda}(o)$ CCDF as $y \to \infty$.}
\label{fig:intsummaxtight}
\end{figure}

The relationship between the max and sum of a sequence of RVs has a long history in the literature on probability.  L\'{e}vy (1935) \cite{Lev1935}, Darling (1952) \cite{Dar1952}, Chistyakov (1964) \cite{Chi1964}, and Chow and Teugels (1978) \cite{ChoTeu1978} are important early works.  Somewhat more recently, Goldie and Kl\"{u}ppelberg (1997) \cite{GolKlu1997} characterize the class of ``subexponential distributions'' (first introduced in \cite{Chi1964}).
\begin{definition}
\label{def:subexp}
{\bf Subexponential distribution} (\cite{GolKlu1997}, Def.\ 1.1).  Let $\{\xsf_i\}$ for $i \in \Nbb$ be iid positive RVs with CDF $F$ such that $F(x) < 1$ for all $x > 0$.  $F$ is a subexponential CDF if one of the following equivalent conditions holds:
\begin{eqnarray}
\lim_{x \to \infty} \frac{\Pbb(\xsf_1+\cdots+\xsf_N > x)}{\bar{F}(x)} & = & n, ~ \forall n \geq 2 \\
\lim_{x \to \infty} \frac{\Pbb(\xsf_1 +\cdots+\xsf_N > x)}{\Pbb(\max\{\xsf_1,\ldots,\xsf_N\}>x)} & = & 1, ~ \forall n \geq 2. \label{eq:subexpdef}
\end{eqnarray}
\end{definition}
Note in particular \eqref{eq:subexpdef} states subexponential distributions have the ratio of the CCDF of the sum over the CCDF of the max approaching one asymptotically, which is precisely result in Prop.\ \ref{pro:intsummaxtight}.  For our purposes we require only one of their results, a sufficient condition for a distribution to be subexponential, which we condense and adapt below.  Recall the hazard rate function $\Hmc(x)$ in Def.\ \ref{def:standardprobdefs}.
\begin{proposition}
\label{pro:subexp}
{\bf Sufficient subexponential condition} (\cite{GolKlu1997}, Prop.\ 3.8).
\begin{equation}
\limsup_{y \to \infty} y \Hmc(y) < \infty \Rightarrow F \mbox{ is subexponential }.
\end{equation}
\end{proposition}
A simple and natural way to apply Prop.\ \ref{pro:subexp} to our case is to condition on the number of nodes $N$ from $\Pi_{d,\lambda}$ within a bounded domain, say $\Bmc \subset \Rbb^d$ .  More formally, suppose $\Pi_{d,\lambda}(\Bmc) = N$.  In this case the $N$ points are independent and uniformly distributed on $\Bmc$, and form a so-called binomial point process (BPP) \cite{HaeGan2008} (\S A.1.1).
\begin{definition}
\label{def:BPP}
Fix $N \in \Nbb$ and bounded $\Bmc \subset \Rbb^d$.  The {\bf binomial point process (BPP)} $\Pi_{\Bmc,N} = \{\xsf_1,\ldots,\xsf_N\}$ consists of $N$ points independently and distributed uniformly at random in $\Bmc$.
\end{definition}
We fix $\Bmc = \brm_d(o,R)$ for $R \in \Rbb_+$, and derive the CDF and HRF for the interference contributions from each of the nodes in a BPP seen at  $o$.
\begin{lemma}
\label{lem:binint}
{\bf BPP distances and interference.}
Let $\Pi_{R,N} = \{\xsf_1,\ldots,\xsf_N\}$ be the BPP on $\brm_d(o,R)$.  The CDF for the RV $|\xsf_i|$ is
\begin{equation}
\Pbb(|\xsf_i| \leq r) = \left\{ \begin{array}{ll}
\frac{|\brm_d(o,r)|}{|\brm_d(o,R)|} = \left( \frac{r}{R} \right)^d, \; & 0 \leq r \leq R \\
1, \; & r > R \end{array} \right.  .
\end{equation}
Assuming $\epsilon = 0$, the interference contribution RV $\Isf_i = P |\xsf_i|^{-\alpha}$ from each node $i$ has support $\Isf_i \in [y_{\rm min},\infty)$ for $y_{\rm min} = P R^{-\alpha}$ and CCDF:
\begin{equation}
\label{eq:binintccdf}
\Pbb(\Isf_i > y) = \Pbb \left( |\xsf_i| \leq \left( \frac{P}{y} \right)^{\frac{1}{\alpha}} \right) = \left\{ \begin{array}{ll} \frac{P^{\delta}}{R^d} y^{-\delta}, \; & y \geq y_{\rm min} \\
1, \; & y < y_{\rm min} \end{array} \right.
\end{equation}
The hazard rate function of $\Isf$ is
\begin{equation}
\label{eq:bininthrf}
\Hmc[\Isf](y) = \frac{f(y)}{\bar{F}(y)} = \left\{ \begin{array}{ll}
\frac{\delta}{y}, \; & y > y_{\rm min} \\
0, \; & \mbox{else} \end{array} \right.  .
\end{equation}
The fractional order moments are
\begin{equation}
\Ebb[\Isf^p] = \left\{ \begin{array}{ll}
\frac{(y_{\rm min})^p \delta}{\delta - p}, \; & p < \delta \\
\infty, \; & \mbox{else} \end{array} \right.
\end{equation}
\end{lemma}
Applying Prop.\ \ref{pro:subexp} to the HRF \eqref{eq:bininthrf} gives that $\Isf$ is subexponential.
\begin{corollary}
\label{cor:subexp}
{\bf Subexponential BPP interferences.}
The RV $\Isf$ is subexponential as the HRF \eqref{eq:bininthrf} for $\Isf$ in Lem.\ \ref{lem:binint} satisfies Prop.\ \ref{pro:subexp}:
\begin{equation}
\lim_{y \to \infty} y \Hmc[\Isf](y) = \delta < \infty.
\end{equation}
\end{corollary}
Prop.\ \ref{pro:intsummaxtight} (for the PPP) and Cor.\ \ref{cor:subexp} (for the BPP) both demonstrate that the sum and max RVs of the interference contributions under the pathloss model $l(|x|) = |x|^{-\alpha}$ for $\delta \in (0,1)$ have CCDFs that are asymptotically equal.  More succinctly, the probability of the sum being large is roughly the same as the probability of the max being large.  Large sums occur due to a small number of large individual contributions; they do not occur due to a large number of small individual contributions.  This intuition helps explain why the LB on the OP and UB on the TC, which are ultimately derived from the simple bound on the interference RVs $\Msf < \Sisf$ (Def.\ \ref{def:intsn}), are asymptotically tight.

%
%
\chapter{Basic model}
\label{cha:bm}

In this chapter we consider the most basic model for computing the OP and TC.  The wireless channel between any two nodes consists of pathloss attenuation with no fading.   As indicated in Def.\ \ref{def:op} and \ref{def:tc}, the key quantity is the SINR, defined below.
\begin{definition}
\label{def:sinrnf}
{\bf Basic model SINR.}
Under the basic model, the SINR seen by a reference Rx located at $o$ when all nodes use constant power $P$, the interferers form a PPP $\Pi_{d,\lambda}$, the noise power is $N$, the channel model is $l_{\alpha,\epsilon}(r)$ as in Ass.\ \ref{ass:snp}, and each Tx is positioned at a fixed distance $u > \epsilon$ from its Rx is
\begin{equation}
\label{eq:sinr}
\sinr(o) \equiv \frac{S}{\Sisf(o) + N},
\end{equation}
where the received signal and interference powers are
\begin{equation}
\label{eq:sn}
S \equiv P l_{\alpha,\epsilon}(u) = P u^{-\alpha} \mbox{ and } \Sisf(o) \equiv \sum_{i \in \Pi_{d,\lambda}} P l_{\alpha,\epsilon}(|\xsf_i|).
\end{equation}
\end{definition}
We emphasize the only random quantity in $\sinr(o)$ in \eqref{eq:sinr} is $\Sisf(o)$, and the only random quantity in $\Sisf(o)$ in \eqref{eq:sn} is the PPP $\Pi_{d,\lambda}$.  The following quantity will be used frequently throughout this chapter.
\begin{definition}
\label{def:xisnr}
{\bf Rx SNR.} Define $\xi \equiv \left(\frac{u^{-\alpha}}{\tau} - \frac{N}{P}\right)^{-\frac{1}{\alpha}}$.  Define the Rx SNR $\snr \equiv \frac{P u^{-\alpha}}{N}$.
\end{definition}
Considering the $N=0$ case ($\xi = u \tau^{\frac{1}{\alpha}}$) makes plain that $\xi$ has units of meters.  To simplify the analysis that follows it is convenient to make the following assumptions.
\begin{assumption}
\label{ass:epsnoi}
{\bf SNR LB.}
The Rx SNR obeys $\snr > \tau$.  Moreover, assume $\epsilon < \xi$.
\end{assumption}
The first assumption states that the received SNR exceeds the SINR threshold, \ie, in the absence of interference a transmission attempt is successful.  As will be clear below, the second assumption states that dominant interferers are possible.  We compute the exact OP and TC in \S\ref{sec:exactTC}, asymptotic exact OP and TC for $\lambda \to 0$ and $q^* \to 0$ in \S\ref{sec:asympTC}, UBs on TC (LBs on OP) in \S\ref{sec:ubTC}, and LBs on TC (UBs on OP) in \S\ref{sec:lbTC}.  Several extensions on this basic model are presented in Ch.\  \ref{cha:modenh}.

\section{Exact OP and TC}
\label{sec:exactTC}

The next result gives the OP in terms of the CCDF of the SN RV representing the aggregate interference seen at $o$ under the PPP $\Pi_{1,1}$.
\begin{proposition}
\label{pro:opnf}
{\bf OP is SN CCDF.}
The OP for the SINR in Def.\ \ref{def:sinrnf} is expressible as the tail probability of a SN RV $\Sisf_{1,1}^{1/\delta,\epsilon'}(o)$ on $\Pi_{1,1}$ evaluated at $y = (\lambda c_d/2)^{-\frac{\alpha}{d}} \xi^{-\alpha}$:
\begin{equation}
q(\lambda) = \Pbb\left( \Sisf_{1,1}^{1/\delta,\epsilon'}(o) > y \right), ~ \delta^{-1} = \frac{\alpha}{d}, ~ \epsilon' = \lambda c_d \epsilon^d/2.
\end{equation}
\end{proposition}
\begin{proof}
By manipulation of the outage event $\{\sinr(o) < \tau\}$ and employing Def.\ \ref{def:intsn}, and Prop.\ \ref{pro:imap}:
\begin{eqnarray}
\left\{ \sinr(o) < \tau \right\}
&=& \left\{ \frac{S}{\Sisf(o) + N} < \tau \right\} \nonumber \\
&=& \left\{ \Sisf(o) > \frac{S}{\tau} - N \right\} \nonumber \\
&=& \left\{ \sum_{i \in \Pi_{d,\lambda}} l_{\alpha,\epsilon}(|\xsf_i|) > \frac{u^{-\alpha}}{\tau} - \frac{N}{P} \right\} \nonumber \\
&=& \left\{ \Sisf^{\alpha,\epsilon}_{d,\lambda}(o) > \xi^{-\alpha} \right\} \nonumber \\
&=& \left\{ (\lambda c_d/2)^{1/\delta} \Sisf_{1,1}^{1/\delta,\lambda c_d \epsilon^d/2}(o) > \xi^{-\alpha} \right\}.
\label{eq:abc}
\end{eqnarray}
\end{proof}
For no receiver guard zone ($\epsilon = 0$) and a path loss exponent of $\alpha =4$ ($\delta = \frac{1}{2}$) we use Cor.\ \ref{cor:levint} to express the OP in a more simple and explicit form.
\begin{corollary}
\label{cor:oplev}
{\bf Explicit OP for $\delta=\frac{1}{2}$.}
For $\epsilon = 0$ and $\delta = \frac{1}{2}$ the OP in Prop.\ \ref{pro:opnf} is
\begin{equation}
\label{eq:oplev}
q(\lambda) = 2 F_{\zsf} \left( \sqrt{ \frac{\pi/2}{\frac{u^{-2 d}}{\tau} - \frac{N}{P}}} c_d \lambda \right)-1.
\end{equation}
where $F_{\zsf}$ is the standard normal $N(0,1)$ CDF.
\end{corollary}
\begin{proof}
Write $\Sisf = \Sisf^{2,0}_{1,1}(o)$ for this proof.  For the assumed $\delta = 1/2$ we have from Cor.\ \ref{cor:levint} that $\Sisf$ is a Levy RV (Def.\ \ref{def:levy}) with $\gamma = 2 \pi$ and CCDF $\bar{F}_{\Sisf}(x) = 1 - 2 \bar{F}_{\zsf}\left(\sqrt{\frac{2 \pi}{x}}\right)= 2 F_{\zsf}\left(\sqrt{\frac{2 \pi}{x}}\right) - 1$.  Using Prop.\ \ref{pro:opnf} gives
\begin{equation}
q(\lambda) = \bar{F}_{\Sisf}\left( (\lambda c_d/2)^{-2} \xi^{-2d} \right) = 2 F_{\zsf}\left(\sqrt{\frac{2 \pi}{(\lambda c_d/2)^{-2} \xi^{-2d}}}\right)-1.
\end{equation}
\end{proof}
For $\epsilon = 0$ we use Prop.\ \ref{pro:opnf} and Def. \ref{def:tc} to express the TC in terms of the inverse of the CCDF of the RV $\Sisf_{1,1}^{\frac{\alpha}{d},0}(o)$.
\begin{proposition}
\label{pro:tcexzereps}
{\bf TC ($\epsilon = 0$).}
For $\epsilon = 0$ the TC equals
\begin{equation}
\lambda(q^*) =
\frac{2\left(\frac{u^{-\alpha}}{\tau} - \frac{N}{P}\right)^{\delta} (1-q^*)}{c_d (\bar{F}_{\Sisf}^{-1}(q^*))^{\delta}},
\end{equation}
where $\bar{F}_{\Sisf}^{-1}$ is the inverse CCDF of the RV ${\Sisf}_{1,1}^{1/\delta,0}(o)$.
\end{proposition}
\begin{proof}
Write $\Sisf = \Sisf_{1,1}^{1/\delta,0}(o)$ for this proof.  Let $\bar{F}_{\Sisf}(y)$ be its CCDF and $\bar{F}_{\Sisf}^{-1}(q)$ the inverse CCDF.  Equate the OP $q(\lambda)$ with the target OP $q^*$:
\begin{equation}
q^* = \bar{F}_{\Sisf}\left( (\lambda c_d/2)^{-1/\delta} \xi^{-\alpha}\right) \Leftrightarrow
\bar{F}_{\Sisf}^{-1}(q^*) = (\lambda c_d/2)^{-1/\delta} \xi^{-\alpha}.
\end{equation}
Now solve for $\lambda$:
\begin{equation}
\lambda = \frac{2}{c_d \xi^d (\bar{F}_{\Sisf}^{-1}(q^*))^{\delta}}.
\end{equation}
Finally, multiply by $1-q^*$ as in Def.\ \ref{def:tc}.
\end{proof}
The condition $\epsilon = 0$ is necessary in order to decouple the RV $\Sisf_{1,1}^{1/\delta,0}(o)$ from the parameter $\lambda$, thereby enabling solving for $\lambda$.  An analogous result for the OP may be derived for $\epsilon > 0$ in terms of the inverse CCDF of the RV $\Sisf^{\alpha,\epsilon}_{d,\lambda}(o)$, but this expression cannot be solved for $\lambda$, and therefore there is no explicit expression for the TC.  For $\epsilon = 0$ and $\delta = \frac{1}{2}$ we can use Cor.\ \ref{cor:oplev} to get a more explicit expression for the TC.
\begin{corollary}
\label{cor:tclev}
{\bf TC ($\epsilon = 0$ and $\delta = \frac{1}{2}$).}
For $\epsilon = 0$ and $\delta = \frac{1}{2}$ the TC in Prop. \ref{pro:tcexzereps} is
\begin{equation}
\label{eq:tclev}
\lambda(q^*) =
\frac{1}{c_d} \sqrt{\frac{\frac{u^{-2 d}}{\tau} - \frac{N}{P}}{\pi/2}} F_{\zsf}^{-1}\left( \frac{1+q^*}{2}\right) (1-q^*),
\end{equation}
where $F_{\zsf}^{-1}$ is the inverse of the normal $N(0,1)$ CDF $F_{\zsf}$.
\end{corollary}
The corollary is immediate from Cor.\ \ref{cor:oplev} and Prop.\ \ref{pro:tcexzereps}.  The expressions for exact TC and OP for $\delta = \frac{1}{2}$ ($d = 2$ and $\alpha = 4$) in Cor.\ \ref{cor:oplev} and \ref{cor:tclev} are shown in Fig.\ \ref{fig:optclev} for varying SINR thresholds $\tau$, with $u=1,N=0,P=1$.  Additional exact OP and TC expressions will be given for the case of Rayleigh fading in \S\ref{sec:fading}.

\begin{figure}[!htbp]
\centering
\includegraphics[width=0.49\textwidth]{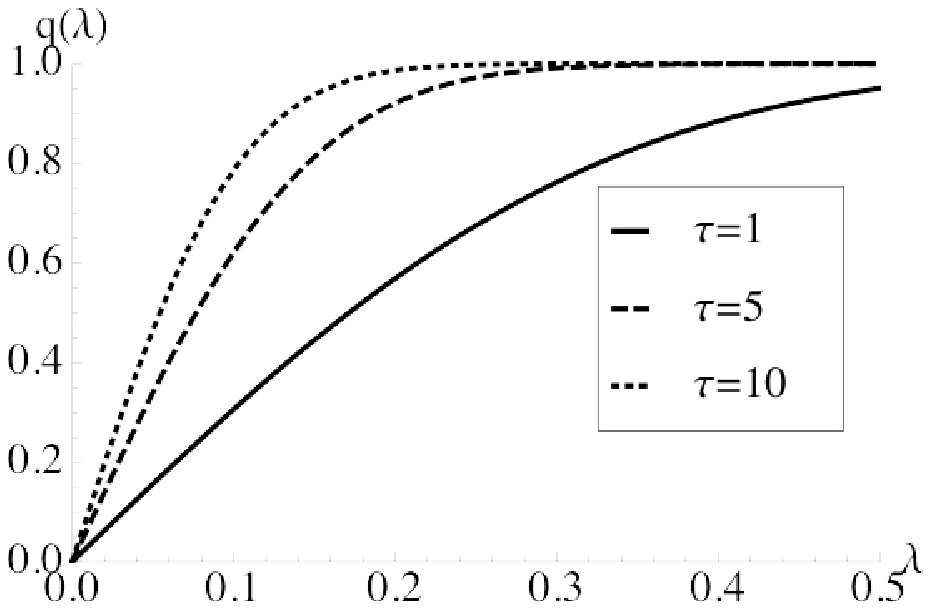}
\includegraphics[width=0.49\textwidth]{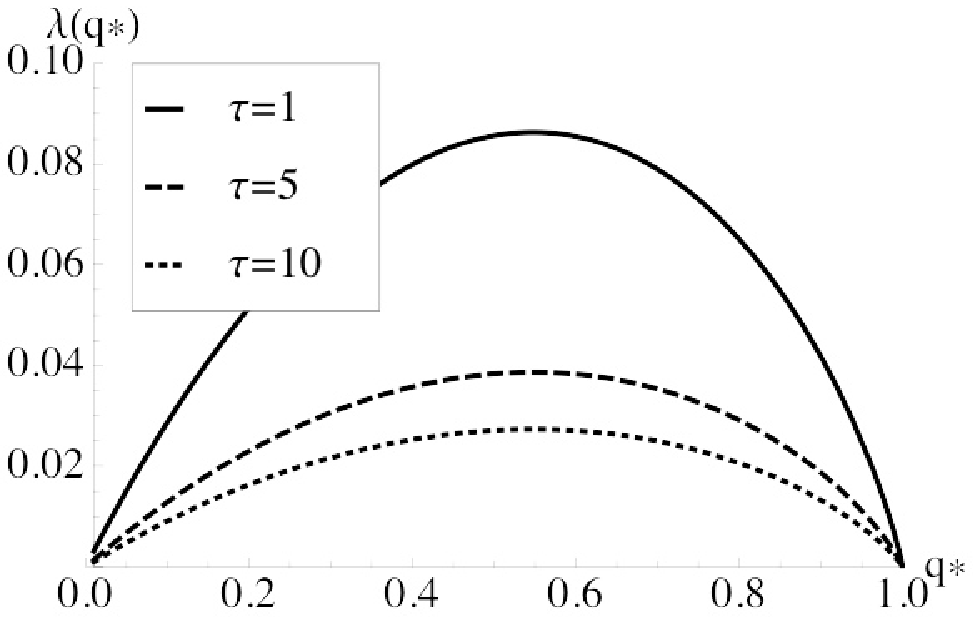}
\caption{The OP $q(\lambda)$ vs.\ $\lambda$ (left) and the TC $\lambda(q^*)$ vs.\ $q^*$ (right) for $\delta = \frac{1}{2}$ ($d=2$ and $\alpha = 4$) and SINR threshold $\tau \in \{1,5,10\}$.}
\label{fig:optclev}
\end{figure}

\section{Asymptotic OP and TC}
\label{sec:asympTC}

In this section we obtain the asymptotic OP in the limit as $\lambda \to 0$ and the asymptotic TC in the limit as $q^* \to 0$ by applying Cor.\  \ref{cor:snasypdf} to Prop.\ \ref{pro:opnf} and Prop.\ \ref{pro:tcexzereps}, both valid for the special case of $\epsilon = 0$.
\begin{proposition}
\label{pro:asymoptc}
{\bf Asymptotic OP and TC.}
For $\epsilon = 0$ the asymptotic OP as $\lambda \to 0$ is:
\begin{equation}
q(\lambda) = \frac{c_d \lambda}{\left( \frac{u^{-\alpha}}{\tau} - \frac{N}{P} \right)^{\delta}} + \Omc(\lambda^2), ~ \lambda \to 0.
\end{equation}
For $\epsilon = 0$ the asymptotic TC as $q^* \to 0$ is:
\begin{equation}
\lambda(q^*) = \frac{1}{c_d} \left( \frac{u^{-\alpha}}{\tau} - \frac{N}{P} \right)^{\delta} q^* + \Omc(q^*)^2, ~ q^* \to 0.	
\end{equation}
\end{proposition}
Thus the OP is linear in $\lambda$ for small $\lambda$ and the TC is linear in $q^*$ for small $q^*$.  These asymptotic approximations will be used in \S\ref{sec:vardist} on variable link distances.  The following remark gives the asymptotic approximations for OP and TC for the special case $\delta = 1/2$ (and $N=0$) in  Cor.\  \ref{cor:tclev}.
\begin{remark}
\label{rem:spherepack}
{\bf TC as sphere packing.}
The first order Taylor series expansion of $(1-q^*)F_{\zsf}^{-1}((1+q^*)/2)$ in \eqref{eq:tclev} around $q^* = 0$ is
\begin{equation}
(1-q^*) F_{\zsf}^{-1}\left( \frac{1+q^*}{2}\right) = \sqrt{\frac{\pi}{2}} q^* + \Omc(q^*)^2.
\end{equation}
Using this in Cor.\  \ref{cor:tclev} for no thermal noise ($N=0$) and rearranging gives a low OP approximation for the TC for $\delta = 1/2$:
\begin{equation}
\label{eq:spherepack}
\lambda(q^*) = \frac{1}{c_d \tilde{u}(q^*)^d}  + \Omc(q^*)^2, ~ q^* \to 0, ~~ \tilde{u}(q^*) = u \left( \frac{\tau^{\delta}}{q^*} \right)^{\frac{1}{d}}.
\end{equation}
In particular, \eqref{eq:spherepack} may be interpreted as the number of $d$-dim.\ spheres per unit area, each with radius $\tilde{u}(q^*)$.  Observe that $\tilde{u}(q^*)/u$ is the guard zone factor by which each Tx-Rx distance $u$ must be expanded to account for the required SINR threshold $\tau$, the required outage probability $q^*$, and the stability exponent $\delta$.  Also note the asymptotically tight UB on TC in Prop.\ \ref{pro:oplb} has the same expansion (for $\epsilon = 0$ and $N = 0$) as Prop.\ \ref{pro:asymoptc}.  In this case the series expansion being $(1-q^*) \log(1-q^*) = q^* + O(q^*)^2$ as $q^* \to 0$.  These expansions are seen to be accurate over a reasonable range of $q^*$ in Fig.\ \ref{fig:spherepack}.
\end{remark}

\begin{figure}[!htbp]
\centering
\includegraphics[width=0.49\textwidth]{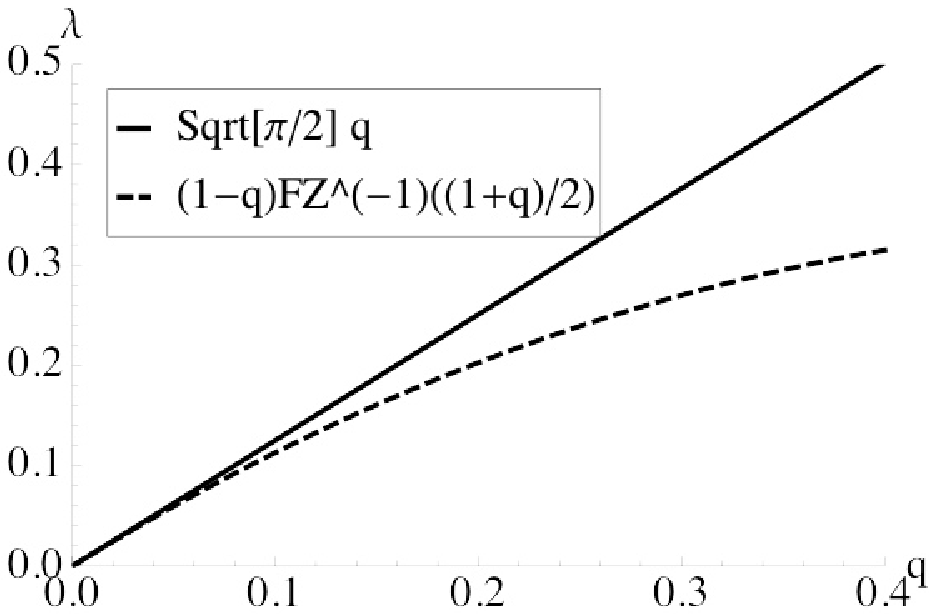}
\includegraphics[width=0.49\textwidth]{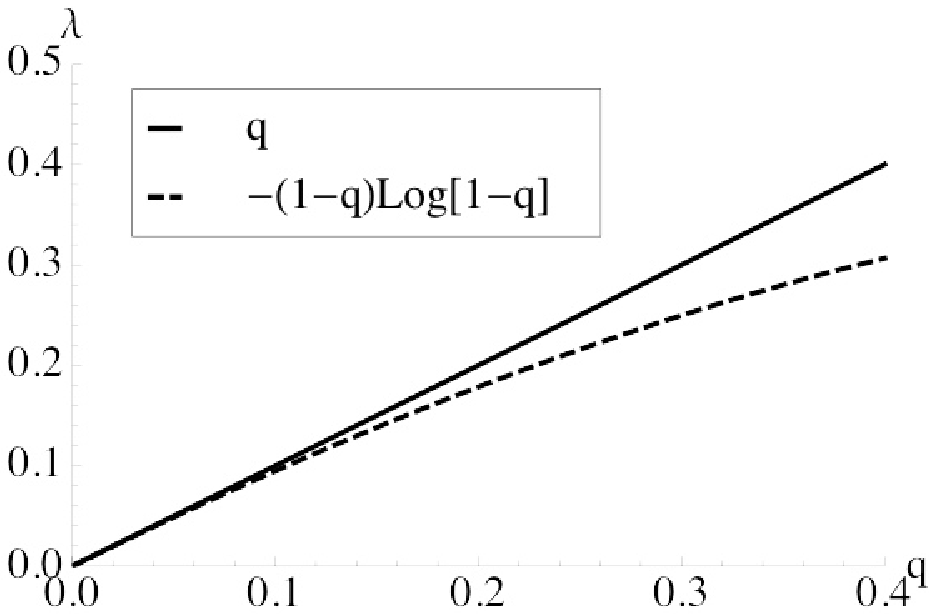}
\caption{The first order Taylor series expansion ($\sqrt{\pi/2}q^*$) of $(1-q^*)F_Z^{-1}((1+q^*)/2)$ for $F_Z$ the standard normal CDF used in the exact TC for $\delta = 1/2$ in Cor.\  \ref{cor:tclev} (left), and the expansion ($q^*$) of $-(1-q^*)\log(1-q^*)$ used in the TC UB (for $\delta \in (0,1)$) in Prop.\ \ref{pro:oplb} (right). }
\label{fig:spherepack}
\end{figure}

\section{Upper bound on TC and lower bound on OP}
\label{sec:ubTC}

In this section we obtain an UB on TC (LB on OP).  The bound is based on considering only ``dominant'' interferers and interference.
\begin{definition}
\label{def:domint}
{\bf Dominant interferers and interference.}
An interferer $i \in \Pi_{d,\lambda}$ is dominant at $o$ under threshold $\tau$ if its interference contribution is sufficiently strong to cause an outage for the reference Rx at $o$:
\begin{equation}
l_{\alpha,\epsilon}(|\xsf_i|) > \xi^{-\alpha} \Leftrightarrow \epsilon \leq |\xsf_i| \leq \xi.
\end{equation}
Else $i$ is non-dominant.  The set of dominant and non-dominant interferers at $o$ under $\tau$ is
\begin{equation}
\hat{\Pi}_{d,\lambda}(o) \equiv \left\{ i \in \Pi_{d,\lambda} : \epsilon \leq |\xsf_i| \leq \xi \right\}, ~
\tilde{\Pi}_{d,\lambda}(o) \equiv \Pi_{d,\lambda} \setminus \hat{\Pi}_{d,\lambda}(o).
\end{equation}
The dominant and non-dominant interference at $o$ under $\tau$
\begin{equation}
\hat{\Sisf}^{\alpha,\epsilon}_{d,\lambda}(o) \equiv \sum_{i \in \hat{\Pi}_{d,\lambda}(o)} l_{\alpha,\epsilon}(|\xsf_i|), ~
\tilde{\Sisf}^{\alpha,\epsilon}_{d,\lambda}(o) \equiv \sum_{i \in \tilde{\Pi}_{d,\lambda}(o)} l_{\alpha,\epsilon}(|\xsf_i|)
\end{equation}
are the interference generated by the dominant and non-dominant nodes.  Note $\Sisf^{\alpha,\epsilon}_{d,\lambda}(o) = \hat{\Sisf}^{\alpha,\epsilon}_{d,\lambda}(o) + \tilde{\Sisf}^{\alpha,\epsilon}_{d,\lambda}(o)$.
\end{definition}
The LB on OP is obtained by observing the aggregate interference exceeds the dominant interference, and thus the probability of the aggregate interference exceeding some value exceeds the probability of the dominant interference exceeding that value.
\begin{proposition}
\label{pro:oplb}
{\bf OP LB and TC UB.}
The OP has a LB
\begin{equation}
\label{eq:oplb}
q^{\rm lb}(\lambda) = 1 - \erm^{- \lambda c_d (\xi^d  - \epsilon^d)}.
\end{equation}
The TC has an UB
\begin{equation}
\label{eq:tcub}
\lambda^{\rm ub}(q^*) = \frac{-(1-q^*)\log(1-q^*)}{c_d (\xi^d  - \epsilon^d)}.
\end{equation}
When $\epsilon = 0$, the bounds are tight for $\lambda \to 0$ and $q^* \to 0$, respectively.
\end{proposition}
\begin{proof}
The key observation is the equivalence of the events $\{\hat{\Sisf}^{\alpha,\epsilon}_{d,\lambda}(o) > \xi^{\alpha}\}$ and $\{\hat{\Pi}_{d,\lambda}(o) \neq \emptyset\}$, where we observe nodes in the annulus $\arm_d(o,\epsilon,\xi)$ are dominant interferers.  From here we compute the corresponding void probability for $\Pi_{d,\lambda}$ using Prop.\ \ref{pro:void}.
\begin{eqnarray}
q^{\rm lb}(\lambda) & = & \Pbb \left( \hat{\Sisf}^{\alpha,\epsilon}_{d,\lambda}(o) > \xi^{-\alpha} \right) \nonumber \\
&=& \Pbb \left( \hat{\Pi}_{d,\lambda}(o) \neq \emptyset \right) = 1 - \Pbb \left( \hat{\Pi}_{d,\lambda}(o) = \emptyset \right) \nonumber \\
&=& 1 - \Pbb \left( \Pi_{d,\lambda}\left(\arm_d(o,\epsilon,\xi)\right) = 0 \right) \nonumber \\
&=& 1 - \erm^{- \lambda c_d (\xi^d  - \epsilon^d)}.
\end{eqnarray}
Set this last equation equal to $q^*$, solve for $\lambda$, and multiply by $1-q^*$ as in Def.\ \ref{def:tc} to get the TC UB.

Finally, the tightness of the bounds can be verified by comparison with Prop. \ref{pro:asymoptc} with Taylor series expansions of $\erm^{- \lambda c_d (\xi^d  - \epsilon^d)}$ and $\log(1-q^*)$ around $\lambda$ and $q^*$.
\end{proof}
\begin{remark}
\label{rem:dommaxequiv}
{\bf Dominant and maximum interferers.}
The LB on OP obtained via dominant interferers is exactly equivalent to the LB on OP obtained by retaining only the largest interferer:
\begin{eqnarray}
\hat{\Pi}_{d,\lambda}(o) = \emptyset & \Leftrightarrow & \Pi_{d,\lambda}(\arm_d(o,\epsilon,\xi)) = 0 \nonumber \\
& \Leftrightarrow & \min_{i \in \Pi_{d,\lambda} : |\xsf_i| > \epsilon} |\xsf_i| > \xi^{-\frac{1}{\alpha}} \nonumber \\
& \Leftrightarrow & \max_{i \in \Pi_{d,\lambda}} |\xsf_i|^{-\alpha} \mathbf{1}_{|\xsf_i| > \epsilon} < \xi \nonumber \\
& \Leftrightarrow & \Msf^{\alpha,\epsilon}_{d,\lambda}(o) < \xi,
\end{eqnarray}
for $\Msf^{\alpha,\epsilon}_{d,\lambda}(o)$ in \eqref{eq:msn} and Def.\ \ref{def:intsn}.  In our extensions of the OP LB ({\em c.f.}\ Def.\ \ref{def:domintfad}, Prop.\ \ref{pro:fadoplb}, Lem.\  \ref{lem:sicdomcan}, Prop.\ \ref{pro:schedlb}, Def.\ \ref{def:domintfpc}, and Prop.\ \ref{pro:fpcoplb}) we won't make this correspondence with the largest interferer explicit, although adapting the above derivation to establish this relationship in those cases is straightforward.
\end{remark}
Specializing Prop.\ \ref{pro:oplb} to the case $\epsilon = 0, N = 0$ and $\delta = \frac{1}{2}$ and comparing  with Cor.\ \ref{cor:oplev} and \ref{cor:tclev}  (with $N = 0$) gives the following corollary.
\begin{corollary}
\label{cor:bndexcomp}
{\bf OP and TC bounds ($\epsilon = 0$, $N=0$, $\delta = \frac{1}{2}$).}
The LB on the OP and the exact OP are:
\begin{equation}
\label{eq:oplbexcomp}
q^{\rm lb}(\lambda) = 1 - \erm^{-\lambda c_d u^d \sqrt{\tau}} \leq q(\lambda) = 2 F_{\zsf} \left( \sqrt{\pi/2}  \lambda c_d u^d \sqrt{\tau} \right) - 1.
\end{equation}
The UB on the TC and the exact TC are:
\begin{equation}
\label{eq:tcubexcomp}
\lambda^{\rm ub}(q^*) = \frac{-(1-q^*)\log(1-q^*)}{c_d u^d \sqrt{\tau}} \geq \lambda(q^*) = \frac{(1-q^*)F_{\zsf}^{-1}\left(\frac{1+q^*}{2}\right)}{\sqrt{\pi/2} c_d u^d \sqrt{\tau}}.
\end{equation}
\end{corollary}
These expressions are plotted in Fig.\ \ref{fig:exactboundcomp} for $d=2$ and $\alpha = 4$ (with $u = 1$ and $\tau = 5$).  Note the OP LB appears to be asymptotically exact as $\lambda \to 0$, and the TC UB appears to be asymptotically exact as $q^* \to 0$.  The expressions in \eqref{eq:oplbexcomp} may be simplified to
\begin{equation}
\label{eq:fzlb}
F_{\zsf}(z) \geq F_{\zsf}^{\rm lb}(z) = 1 - \frac{1}{2} \erm^{-\sqrt{\frac{2}{\pi}} z}, ~ z \geq 0,
\end{equation}
by defining $z = \lambda c_d u^d \sqrt{\tau}$.  This inequality is shown in Fig.\ \ref{fig:FZLB}; again note the bound appears to be asymptotically exact as $z \to 0$.

\begin{figure}[!htbp]
\centering
\includegraphics[width=0.49\textwidth]{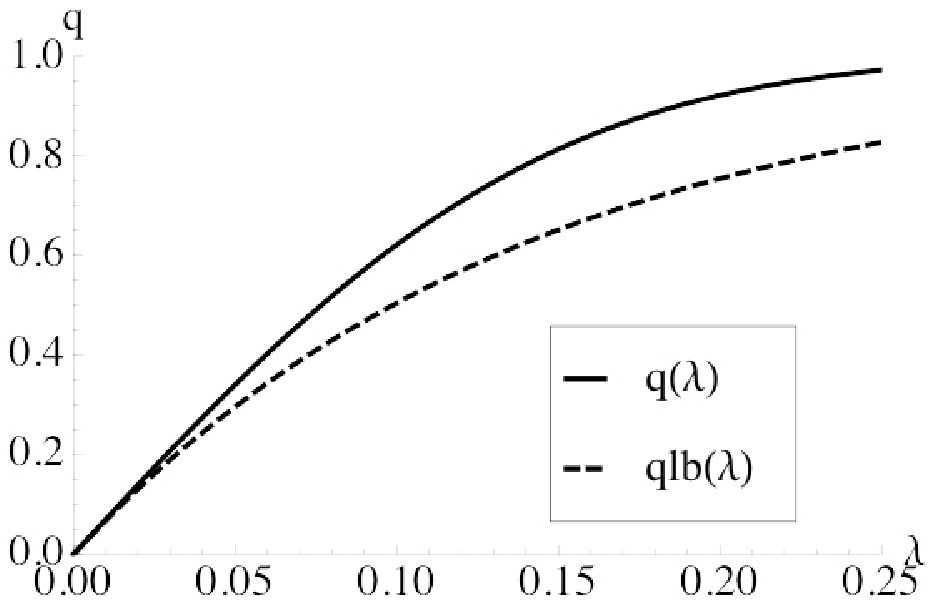}
\includegraphics[width=0.49\textwidth]{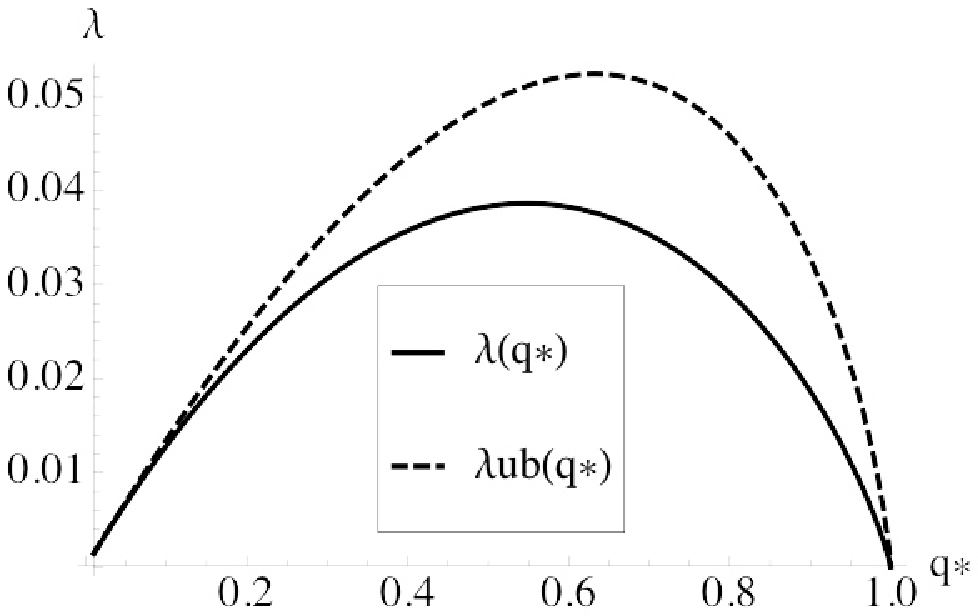}
\caption{The OP $q(\lambda)$ and its LB $q^{\rm lb}(\lambda)$ vs.\ $\lambda$ (left) and the TC $\lambda(q^*)$ and its UB $\lambda^{\rm ub}(q^*)$ vs.\ $q^*$ (right) for $\delta = \frac{1}{2}$ ($d=2$ and $\alpha = 4$).}
\label{fig:exactboundcomp}
\end{figure}

\begin{figure}[!htbp]
\centering
\includegraphics[width=0.49\textwidth]{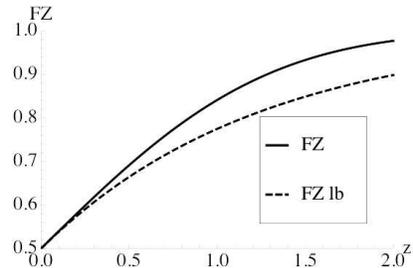}
\caption{The LB on the standard normal CDF $F_{\zsf}(z)$ in \eqref{eq:fzlb} that underlies the OP LB.}
\label{fig:FZLB}
\end{figure}

\section{Throughput (TP) and TC }
\label{sec:tpandtc}

One justification for claiming that TC is a natural performance metric is obtained by comparison with MAC layer TP, as discussed below.
\begin{definition}
\label{def:tp}
The {\bf MAC layer TP} of a wireless network employing the slotted Aloha MAC protocol, where the active transmitters form a PPP $\Pi_{d,\lambda}$, is
\begin{equation}
\label{eq:tp}
\Lambda(\lambda) \equiv \lambda(1-q(\lambda)),
\end{equation}
where $q(\lambda)$ is the OP in Def.\ \ref{def:op}.
\end{definition}
The TP has units of successful transmissions per unit area, and \eqref{eq:tp} may be read as saying the TP is the intensity of attempted transmissions per unit area ($\lambda$) thinned by the success probability ($1-q(\lambda)$) of each transmission.  In light of Rem.\  \ref{rem:ptx}, the design question for the TP is: given $\lambda_{\rm pot}$ how to select $p_{\rm tx}$ so as to maximize $\Lambda(p_{\rm tx} \lambda_{\rm pot})$.  Before considering this question, we first recall the following basic facts about (saturated) slotted Aloha in a wireless uplink setting under the collision channel model with a single collision domain (no spatial reuse).
\begin{proposition}
\label{pro:slotaloha}
{\bf Slotted Aloha TP and OP.}
For slotted Aloha in a single collision domain under the collision channel model with $N$ saturated (backlogged) users transmitting to a common base station, employing a common transmission probability $p$, the TP and OP are
\begin{equation}
\label{eq:alotpop}
\Lambda(N,p) \equiv N p (1-p)^{N-1}, ~~ q(N,p) \equiv 1-(1-p)^{N-1}.
\end{equation}
The TP optimal $p$ is $p^*(N) = 1/N$ with associated TP and OP
\begin{equation}
\label{eq:alo}
\Lambda(N,p^*(N)) = \left(1-\frac{1}{N} \right)^{N-1}, ~~ q(N,p^*(N)) = 1-\left(1-\frac{1}{N}\right)^{N-1}.
\end{equation}
For large $N$ and $p = \frac{\lambda}{N}$ the asymptotic TP and OP are
\begin{eqnarray}
\Lambda(\lambda) &=& \lim_{N \to \infty} \Lambda(N,\lambda/N) = \lambda \erm^{-\lambda} \nonumber \\
q(\lambda) &=& \lim_{N \to \infty} q(N,\lambda/N) = 1-\erm^{-\lambda}.
\end{eqnarray}
The asymptotic TP optimal choice for $\lambda$ is $\lambda^* =1$ with associated asymptotic TP and OP
\begin{equation}
\Lambda(\lambda^*) = 1 \erm^{-1} \approx 0.367879, ~~ q(\lambda^*) = 1 - \erm^{-1} \approx 0.632121.
\end{equation}
\end{proposition}
The statements in Prop.\ \ref{pro:slotaloha} are simple to prove.  Note that achieving the maximum asyptotic TP of $37\%$ requires an incurred OP of $63\%$.  These relationships are illustrated in Fig.\ \ref{fig:slotaloha}.
\begin{figure}[!htbp]
\centering
\includegraphics[width=0.49\textwidth]{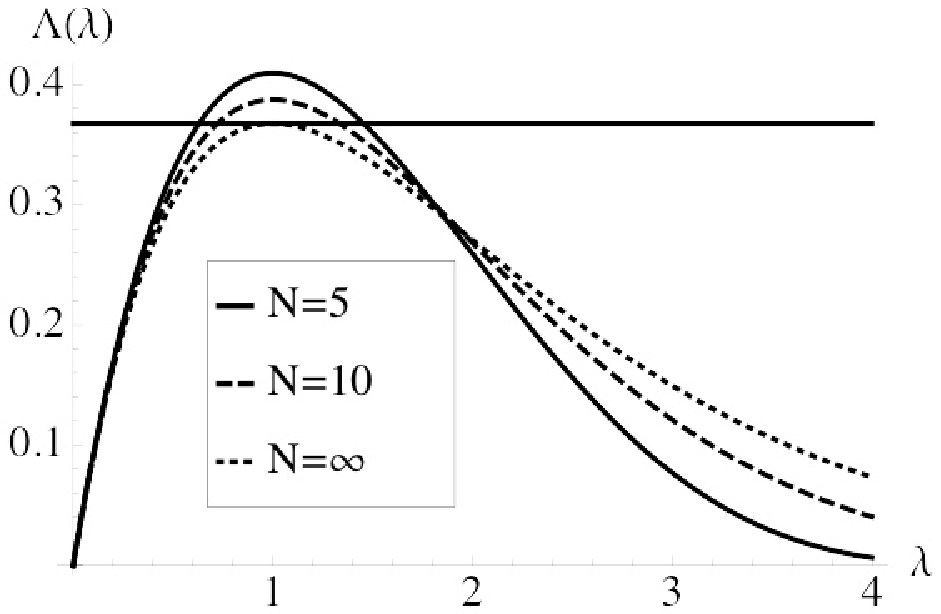}
\includegraphics[width=0.49\textwidth]{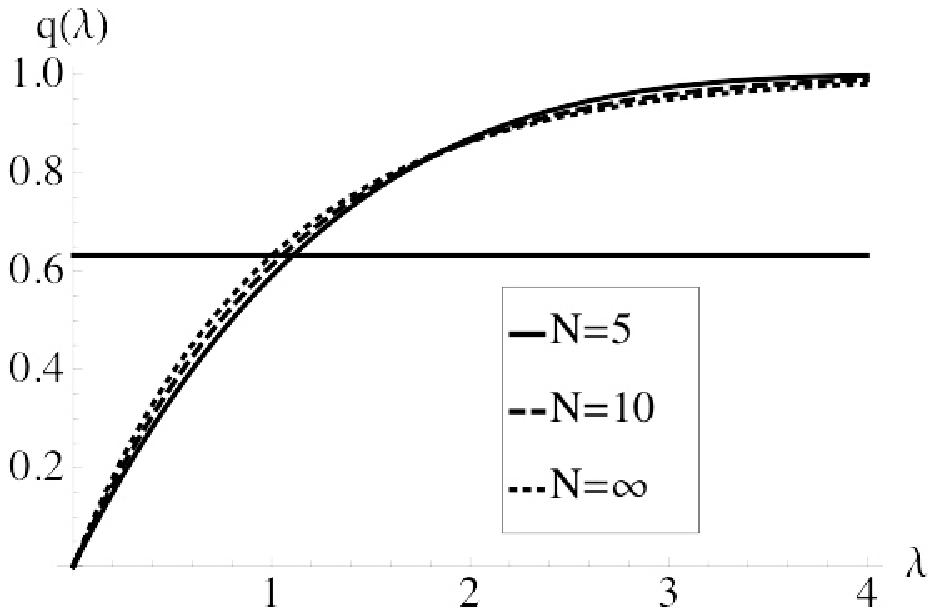}
\caption{The TP (left) and OP (right) of slotted Aloha in a single collision domain (no spatial reuse) for $N$ backlogged users employing transmission probability $p=\lambda/N$, for $N \in \{5,10\}$  Also shown is the TP and OP as $N \to \infty$.  Achieving a high TP requires incurring a high OP.}
\label{fig:slotaloha}
\end{figure}
We now return to the context of wireless networks with spatial reuse.  The LB on the OP in Prop. \ref{pro:oplb} leads directly to an UB on the TP in Def.\ \ref{def:tp}.  This UB is equivalent to the TP in the non-spatial context of Prop.\ \ref{pro:slotaloha}.
\begin{proposition}
\label{pro:tp}
{\bf MAC layer TP UB.}
The TP in Def.\ \ref{def:tp} has UB
\begin{equation}
\label{eq:tpub}
\Lambda(\lambda) \leq \Lambda^{\rm ub}(\lambda) = \lambda(1-q^{\rm lb}(\lambda)) = \lambda \erm^{-\lambda c_d (\xi^d - \epsilon^d)}.
\end{equation}
The TP bound optimal $\lambda$ is
\begin{equation}
\lambda^* = \frac{1}{c_d (\xi^d - \epsilon^d)},
\end{equation}
and the associated TP UB and OP LB are
\begin{equation}
\Lambda^{\rm ub}(\lambda^*) = \frac{1}{\erm c_d (\xi^d - \epsilon^d)}, ~~ q^{\rm lb}(\lambda^*) = 1 - \erm^{-1} \approx 0.632121.
\end{equation}
\end{proposition}
The main point of Prop.\ \ref{pro:tp} is that achieving the optimal TP requires incurring an OP of $63\%$.  Given that wireless devices are energy constrained and that failed attempted transmissions are wasted energy, it is natural to question if unconstrained TP maximization is the right design objective.  TC is in fact the TP of a wireless network under Aloha subject to an OP constraint, as shown below.
\begin{proposition}
\label{pro:tcopt}
{\bf TC is constrained TP maximization.}
The optimization problem of maximizing TP subject to a constraint on the outage probability $q^* \in (0,1)$
\begin{equation}
\label{eq:tcopt}
\max_{\lambda} \left\{ \Lambda(\lambda) ~:~ q(\lambda) \leq q^* \right\}
\end{equation}
has solution $\lambda^* = q^{-1}(q^*)$ and TP equal to the TC in Def.\ \ref{def:tc}:
\begin{equation}
\Lambda(\lambda^*) = \lambda^*(1-q(\lambda^*)) = q^{-1}(q^*)(1-q^*) = \lambda^*(q^*).
\end{equation}
\end{proposition}
To clarify, $\lambda^* = q^{-1}(q^*)$ is the solution of \eqref{eq:tcopt}, and $\lambda^*(q^*)$ is the TC as defined in Def.\ \ref{def:tc}.  In summary, TC is TP under an OP constraint.  Note that maximization of TP $\Lambda(\lambda)$ over $\lambda \in \Rbb$ is equivalent to maximization of the TC $\lambda^*(q^*)$ over the target OP $q^* \in [0,1]$, as made precise below.
\begin{proposition}
\label{pro:optTPandTC}
{\bf Maximum TP equals maximum TC.}
The TP optimization problem
\begin{equation}
\max_{\lambda \in \Rbb_+} \Lambda(\lambda)
\end{equation}
has a unique maximizer $\lambda_{\rm opt}$ and an associated maximum value $\Lambda_{\rm opt} = \Lambda(\lambda_{\rm opt})$.  The TC optimization problem
\begin{equation}
\max_{q^* \in [0,1]} \lambda^*(q^*)
\end{equation}
has a unique maximizer $q^*_{\rm opt}$ and an associated maximum value $\lambda^*_{\rm opt} = \lambda^*(q^*_{\rm opt})$.  Furthermore, the maximum values are equal and the maximizers are related through the OP function $q(\lambda)$:
\begin{equation}
\Lambda_{\rm opt} = \lambda^*_{\rm opt} \mbox{ and } q(\lambda_{\rm opt}) = q^*_{\rm opt}.
\end{equation}
\end{proposition}
\begin{proof}
We first prove $q(\lambda_{\rm opt}) = q^*_{\rm opt}$.  The monotonicity of $q(\lambda)$ guarantees that $\Lambda(\lambda)$ has a unique stationary point on $(0,\infty)$ and that this is the unique maximizer $\lambda_{\rm opt} > 0$.  Taking the derivative
\begin{equation}
\Lambda'(\lambda) = 1 - q(\lambda) - \lambda q'(\lambda)
\end{equation}
and equating with zero at the stationary point gives
\begin{equation}
\Lambda'(\lambda_{\rm opt}) = 1 - q(\lambda_{\rm opt}) - \lambda_{\rm opt} q'(\lambda_{\rm opt}) = 0,
\end{equation}
or equivalently, after rearranging,
\begin{equation}
\label{eq:tpopteq}
q'(\lambda_{\rm opt}) = \frac{1-q(\lambda_{\rm opt})}{\lambda_{\rm opt}}.
\end{equation}
It is likewise straightforward to show $\lambda^*(q^*)$ is concave on $[0,1]$ and therefore the unique optimizer is the unique stationary point on $(0,1)$.  Taking the derivative and applying the inverse function theorem gives:
\begin{equation}
\lambda^{*'}(q^*) = (1-q^*) q^{-1'}(q^*) - q^{-1}(q^*) = \frac{1-q^*}{q'(q^{-1}(q^*))} - q^{-1}(q^*).
\end{equation}
Equating with zero at the stationary point gives
\begin{equation}
\lambda^{*'}(q^*_{\rm opt}) = \frac{1-q^*_{\rm opt}}{q'(q^{-1}(q^*_{\rm opt}))} - q^{-1}(q^*_{\rm opt}) = 0,
\end{equation}
or equivalently, after rearranging,
\begin{equation}
\label{eq:tcopteq}
q'(q^{-1}(q^*)) = \frac{1-q^*_{\rm opt}}{q^{-1}(q^*_{\rm opt})}.
\end{equation}
There is a unique solution $\lambda_{\rm opt}$ for \eqref{eq:tpopteq}, and likewise there is a unique solution $q^*_{\rm opt}$ for \eqref{eq:tcopteq}.  The two optimizers are related by $q(\lambda_{\rm opt}) = q^*_{\rm opt}$ (equivalently, $\lambda_{\rm opt} = q^{-1}(q^*_{\rm opt})$) since under this relationship these two equations are the same (in that $\lambda_{\rm opt}$ solves \eqref{eq:tpopteq} iff $q^*_{\rm opt}$ solves \eqref{eq:tcopteq}).   We next prove $\Lambda_{\rm max} = \lambda^*_{\rm max}$.   Rearranging \eqref{eq:tpopteq} and \eqref{eq:tcopteq} gives
\begin{equation}
\Lambda_{\rm max} = \lambda_{\rm opt}^2 q'(\lambda_{\rm opt}), ~~ \lambda^*_{\rm max} = \frac{(1-q^*_{\rm opt})^2}{q'(q^{-1}(q^*_{\rm opt}))}.
\end{equation}
The square root of their ratio is
\begin{equation}
\sqrt{\frac{\Lambda_{\rm max}}{\lambda^*_{\rm max}}} = \frac{\lambda_{\rm opt}}{1-q(\lambda_{\rm opt})} q'(\lambda_{\rm opt}) = 1,
\end{equation}
and thus $\Lambda_{\rm max} = \lambda^*_{\rm max}$.
\end{proof}
The corollary below shows that Prop.\ \ref{pro:optTPandTC} holds for the TP UB in Prop.\ \ref{pro:tp} and the TC UB in Prop.\ \ref{pro:oplb}.
\begin{corollary}
\label{cor:optTPandTC}
{\bf Maximizing TP and TC UBs.}
Denote $c_d(\xi^d-\epsilon^d)$ found in Prop.\ \ref{pro:oplb} by $a$.  The TP UB optimization problem
\begin{equation}
\max_{\lambda \in \Rbb_+} \Lambda^{\rm ub}(\lambda)
\end{equation}
for $\Lambda^{\rm ub}(\lambda)$ in \eqref{eq:tpub} has a unique maximizer $\lambda_{\rm opt} = 1/a$ and an associated maximum value $\Lambda^{\rm ub}_{\rm max} = \frac{1}{\erm a}$.  The TC UB optimization problem
\begin{equation}
\max_{q^* \in [0,1]} \lambda^{\rm ub}(q^*)
\end{equation}
for $\lambda^{\rm ub}(q^*)$ in \eqref{eq:tcub} has a unique maximizer $q^*_{\rm opt} = 1\!-\!\frac{1}{\erm}$ and an associated maximum value $\lambda^{\rm ub}_{\rm max} = \frac{1}{\erm a}$.  Like Prop.\ \ref{pro:optTPandTC}, the maximum values are equal and the optimizers obey $q^*_{\rm opt} = q^{\rm lb}(\lambda_{\rm opt})$ for $q^{\rm lb}$ in \eqref{eq:oplb}.
\end{corollary}
The corollary is obtained by simple calculus on the functions $\Lambda^{\rm ub}(\lambda)$ and $\lambda^{\rm ub}(q^*)$.  Fig.\ \ref{fig:tp12} illustrates the quantities in Cor.\  \ref{cor:optTPandTC}.
\begin{figure}[!htbp]
\centering
\includegraphics[width=0.49\textwidth]{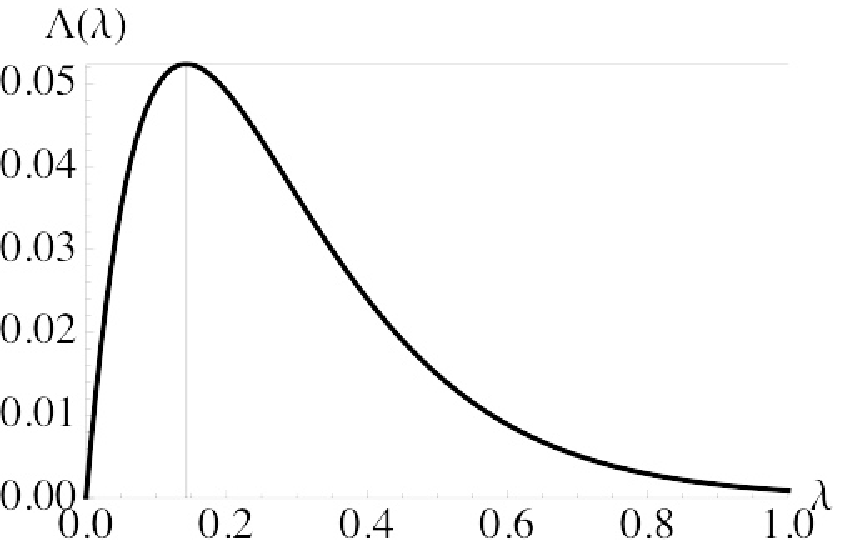}
\includegraphics[width=0.49\textwidth]{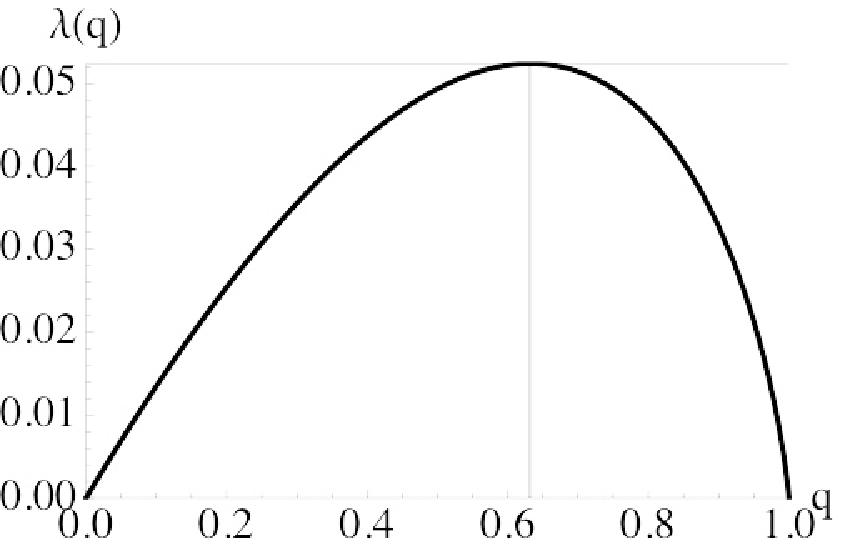}
\caption{The spatial TP UB $\Lambda^{\rm ub}(\lambda)$ \eqref{eq:tpub} vs.\ the spatial intensity of attempted transmissions $\lambda$ (left), and the TC UB $\lambda^{\rm ub}(q^*)$ vs.\ the target OP $q^*$ (right).  For $\alpha=4,d=2,u=1,\tau=5,N=0$ we have $a = c_d \xi^d = \sqrt{5} \pi \approx 7.02$ and optimal $\lambda^* = 1/a \approx 0.14$ and maximum TP $\Lambda_{\rm max} = 1/(\erm a)  \approx 0.052$, $q^*_{\rm opt} = 1-1/\erm \approx 0.63$, and $q^*_{\rm opt} = q^{\rm lb}(\lambda_{\rm opt})$.}
\label{fig:tp12}
\end{figure}

\section{Lower bounds on TC and upper bounds on OP}
\label{sec:lbTC}

In this section we obtain three LBs on TC (UBs on OP).  The bounds are obtained by using the bound from \S\ref{sec:ubTC} along with three UBs on the tail probability for the non-dominant interference.  The tail UBs are given in \S\ref{sec:mccine} and these in turn employ the moments from the Campbell-Mecke theorem given in Thm.\ \ref{thm:cam}.  We express the OP in terms of the distributions of the dominant and non-dominant interference.
\begin{proposition}
\label{pro:domnonop}
{\bf Exact OP in terms of OP LB.}
Define the CCDFs of the independent RVs $\hat{\Sisf}^{\alpha,\epsilon}_{d,\lambda}(o),\tilde{\Sisf}^{\alpha,\epsilon}_{d,\lambda}(o)$ from Def.\ \ref{def:domint} as:
\begin{equation}
\bar{F}_{\hat{\Sisf}}(y) \equiv \Pbb(\hat{\Sisf}^{\alpha,\epsilon}_{d,\lambda}(o) > y), ~~ \bar{F}_{\tilde{\Sisf}}(y) \equiv \Pbb(\tilde{\Sisf}^{\alpha,\epsilon}_{d,\lambda}(o) > y).
\end{equation}
The OP equals:
\begin{equation}
\label{eq:domnonop1}
q(\lambda) = \bar{F}_{\hat{\Sisf}}(\xi^{-\alpha}) + F_{\hat{\Sisf}}(\xi^{-\alpha}) \bar{F}_{\tilde{\Sisf}}(\xi^{-\alpha}).
\end{equation}
Moreover, the OP equals:
\begin{equation}
\label{eq:domnonop2}
q(\lambda) = q^{\rm lb}(\lambda) + (1-q^{\rm lb}(\lambda)) \bar{F}_{\tilde{\Sisf}}(\xi^{-\alpha})
\end{equation}
for $q^{\rm lb}(\lambda)$ given in Prop.\ \ref{pro:oplb}.
\end{proposition}
\begin{proof}
Write $\Sisf = \Sisf^{\alpha,\epsilon}_{d,\lambda}(o), \hat{\Sisf} = \hat{\Sisf}^{\alpha,\epsilon}_{d,\lambda}(o), \tilde{\Sisf} = \tilde{\Sisf}^{\alpha,\epsilon}_{d,\lambda}(o)$ for this proof.  Recall $\Sisf = \hat{\Sisf} + \tilde{\Sisf}$. Express the outage event for $\Sisf$ in terms of $(\hat{\Sisf},\tilde{\Sisf})$ and decompose it into three (overlapping) regions:
\begin{eqnarray}
\{ \Sisf > \xi^{-\alpha} \} &=& \{ (\hat{\Sisf},\tilde{\Sisf}) : \hat{\Sisf} + \tilde{\Sisf} > \xi^{-\alpha} \} \\
&=& \{ \hat{\Sisf} > \xi^{-\alpha}\} \cup \{ \tilde{\Sisf} > \xi^{-\alpha}\} \cup \{ \hat{\Sisf} \leq \xi^{-\alpha}, \tilde{\Sisf} \leq \xi^{-\alpha}, \Sisf > \xi^{-\alpha} \}.  \nonumber
\end{eqnarray}
Recall the equivalence of the events $\{\hat{\Sisf} < \xi^{-\alpha}\}$ and $\{\hat{\Sisf} = 0\}$, and observe this implies the third event becomes null: $\{\tilde{\Sisf} \leq \xi^{\alpha} \cap \tilde{\Sisf} > \xi^{-\alpha}\} = \emptyset$.  The OP is therefore
\begin{equation}
q(\lambda) = \Pbb(\hat{\Sisf} > \xi^{-\alpha}) + \Pbb(\tilde{\Sisf} > \xi^{-\alpha}) - \Pbb(\hat{\Sisf} > \xi^{-\alpha} \cap \tilde{\Sisf} > \xi^{-\alpha}).
\end{equation}
Apply the independence of $(\hat{\Sisf},\tilde{\Sisf})$ to the third term and group terms to get \eqref{eq:domnonop1}.  Finally, recognize $q^{\rm lb}(\lambda) = \bar{F}_{\hat{\Sisf}}(\xi^{-\alpha})$ to get \eqref{eq:domnonop2}.
\end{proof}
In \eqref{eq:domnonop2} it is clear that we can UB $q(\lambda)$ by an UB on the CCDF of the non-dominant interference $\bar{F}_{\tilde{\Sisf}}(\xi^{-\alpha})$.  We now apply the Markov, Chebychev, and Chernoff bounds to the RV $\tilde{\Sisf}^{\alpha,\epsilon}_{d,\lambda}(o)$.
\begin{proposition}
\label{pro:oplbmar}
{\bf Markov inequality OP UB} for $\alpha > d$ is
\begin{equation}
q^{\rm ub,Mar}(\lambda) = q^{\rm lb}(\lambda) + (1-q^{\rm lb}(\lambda)) \frac{ \lambda d c_d}{\alpha-d} \xi^d.
\end{equation}
\end{proposition}
\begin{proof}
Markov inequality (Prop.\ \ref{pro:mar}) for $\tilde{\Sisf}^{\alpha,\epsilon}_{d,\lambda}(o)$ in Def.\ \ref{def:domint} gives
\begin{equation}
\Pbb( \tilde{\Sisf}^{\alpha,\epsilon}_{d,\lambda}(o) > \xi^{-\alpha} ) \leq \xi^{\alpha} \Ebb[ \tilde{\Sisf}^{\alpha,\epsilon}_{d,\lambda}(o)].
\end{equation}
Now apply Campbell's Theorem (Thm. \ref{thm:cam}) to compute $\Ebb[ \tilde{\Sisf}^{\alpha,\epsilon}_{d,\lambda}(o)]$:
\begin{equation}
\Ebb \left[ \sum_{i \in \Pi_{d,\lambda}} |\xsf_i|^{-\alpha} \mathbf{1}_{|\xsf_i| > \xi} \right] = \lambda d c_d \int_{\xi}^{\infty} r^{-\alpha} r^{d-1} \drm r = \frac{\lambda d c_d}{\alpha-d} \xi^{-\alpha + d}.
\end{equation}
The integral is finite for $\alpha > d$.  Substitution yields the proposition.
\end{proof}
Second, apply the Chebychev inequality.
\begin{proposition}
\label{pro:oplbcheb}
{\bf Chebychev inequality OP UB} for $\alpha > d$ is
\begin{equation}
q^{\rm ub,Cheb}(\lambda) = q^{\rm lb}(\lambda) + (1-q^{\rm lb}(\lambda)) \frac{\lambda d c_d}{2 \alpha - d} \frac{\xi^d}{\left(1 - \frac{\lambda d c_d}{\alpha-d} \xi^d \right)^2}.
\end{equation}
The bound is trivial for
\begin{equation}
\label{eq:chebtrivthresh}
\lambda > \frac{\alpha-d}{d c_d \xi^d} .
\end{equation}
\end{proposition}
\begin{proof}
Denote $\tilde{\Sisf} = \tilde{\Sisf}^{\alpha,\epsilon}_{d,\lambda}(o)$ for this proof.  The Chebychev inequality (Prop.\ \ref{pro:cheb}) applied to the RV $\tilde{\Sisf}$ in Def.\ \ref{def:domint} gives
\begin{eqnarray}
\Pbb( \tilde{\Sisf} > \xi^{-\alpha} ) & = & \Pbb( \tilde{\Sisf} - \Ebb[\tilde{\Sisf}] > \xi^{-\alpha} - \Ebb[\tilde{\Sisf}] ) \nonumber \\
& \leq & \Pbb( |\tilde{\Sisf} - \Ebb[\tilde{\Sisf}] | > \xi^{-\alpha} - \Ebb[\tilde{\Sisf}] ) \nonumber \\
& \leq & \frac{\mathrm{Var}(\tilde{\Sisf})}{(\xi^{-\alpha} - \Ebb[\tilde{\Sisf}])^2}.
\end{eqnarray}
Now apply Campbell's Theorem (Thm. \ref{thm:cam}) to compute $\mathrm{Var}(\tilde{\Sisf})$:
\begin{equation}
\mathrm{Var} \left( \sum_{i \in \Pi_{d,\lambda}} |\xsf_i|^{-\alpha} \mathbf{1}_{|\xsf_i| > \xi} \right) = \lambda d c_d \int_{\xi}^{\infty} r^{-2\alpha} r^{d-1} \drm r = \frac{\lambda d c_d}{2 \alpha - d} \xi^{-2 \alpha + d}.
\end{equation}
Evaluating the integral for $\alpha > d/2$, substituting $\Ebb[\tilde{\Sisf}]=\frac{\lambda d c_d}{\alpha-d} \xi^{-\alpha+d}$, and cancelling the common term $\xi^{-2\alpha}$ yields the proposition.  The first inequality used above yields a trivial bound for $\xi^{-\alpha} < \Ebb[\tilde{\Sisf}]$, which may be expressed as bound on $\lambda$:
\begin{equation}
\xi^{-\alpha} < \Ebb[\tilde{\Sisf}] \Leftrightarrow \xi^{-\alpha} < \frac{\lambda d c_d}{\alpha-d} \xi^{-\alpha + d}.
\end{equation}
\end{proof}
The threshold \eqref{eq:chebtrivthresh} is sufficient but not necessary for the bound to be trivial.  Third, apply the Chernoff inequality.
\begin{proposition}
\label{pro:oplbcher}
{\bf Chernoff inequality OP UB } for $\alpha > d$ is
\begin{equation}
q^{\rm ub,Cher}(\lambda) = q^{\rm lb}(\lambda) + (1-q^{\rm lb}(\lambda)) \erm^{- c(\lambda)},
\end{equation}
where
\begin{equation}
\label{eq:cherbd}
c(\lambda) = \sup_{\theta \geq 0} \left( \theta \xi^{-\alpha} - \frac{\lambda d c_d}{\alpha} \int_0^{\xi^{-\alpha}} \left( \erm^{\theta y} -1\right) y^{-\delta-1} \drm y \right).
\end{equation}
\end{proposition}
\begin{proof}
Denote $\tilde{\Sisf} = \tilde{\Sisf}^{\alpha,\epsilon}_{d,\lambda}(o)$ for this proof.  The Chernoff inequality (Prop.\ \ref{pro:cher}) applied to the RV $\tilde{\Sisf}$ in Def.\ \ref{def:domint} gives
\begin{equation}
\Pbb( \tilde{\Sisf} > \xi^{-\alpha} ) \leq \inf_{\theta > 0} \Ebb[\erm^{\theta \tilde{\Sisf}}] \erm^{-\theta \xi^{-\alpha}}.
\end{equation}
The MGF of $\tilde{\Sisf}$ is obtained by selecting $\epsilon = \xi$ in Cor.\  \ref{cor:mgfint} yielding:
\begin{equation}
\label{eq:mgfint2}
\Mmc[\tilde{\Sisf}](\theta)
= \exp \left\{ \frac{\lambda d c_d}{\alpha} \int_0^{\xi^{-\alpha}} \left( \erm^{\theta y} -1\right) y^{-\delta-1} \drm y \right\}.
\end{equation}
Simple manipulations yield the proposition.
\end{proof}
Finding the optimal $\theta^*$ in \eqref{eq:cherbd} must in general be done numerically, although certain simplifications hold for $\delta = 1/2$.

The Markov, Chebychev, and Chernoff bounds on the OP and TC are shown in Fig.\ \ref{fig:optcmarkchebcher}.  The top two plots are OP vs.\ $\lambda$ and the bottom two plots are TC vs.\ $q^*$.  The right plots are an inset of the left plots.  Consider the OP plots.  The OP LB is seen to be tighter than each of the three OP UBs.  For $\lambda$ small (corresponding to small $q(\lambda$)) the Chernoff and Chebychev bounds are nearly equivalent and are better than the Markov bound.  For moderate to large $\lambda$ (corresponding to moderate to large $q(\lambda)$) the Markov and Chernoff are nearly equivalent and better than the (trivial) Chebychev bound.  All three bounds thus have their value: $i)$ Markov is bad for small $\lambda$ but good for larger $\lambda$ and is simple, $ii)$ Chebychev is good for small $\lambda $ but bad for larger $\lambda$ and is intermediate in simplicity between Markov and Chernoff, and $iii)$ Chernoff is as good as Markov and Chebychev for all $\lambda$, but is much more complicated than the other two.  Similar trends naturally follow for the three TC plots as for the OP plots.

\begin{figure}[!htbp]
\centering
\includegraphics[width=0.49\textwidth]{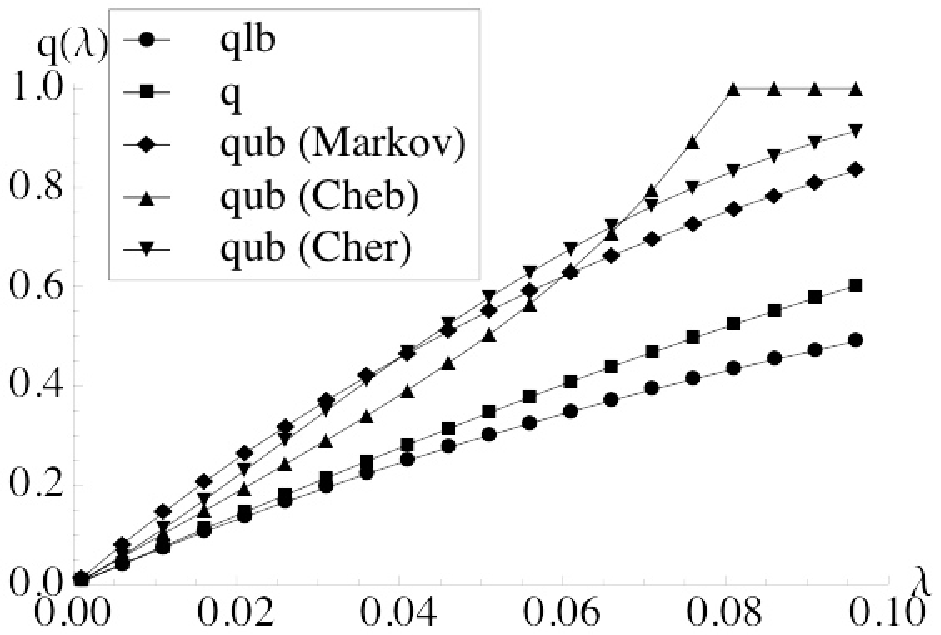}
\includegraphics[width=0.49\textwidth]{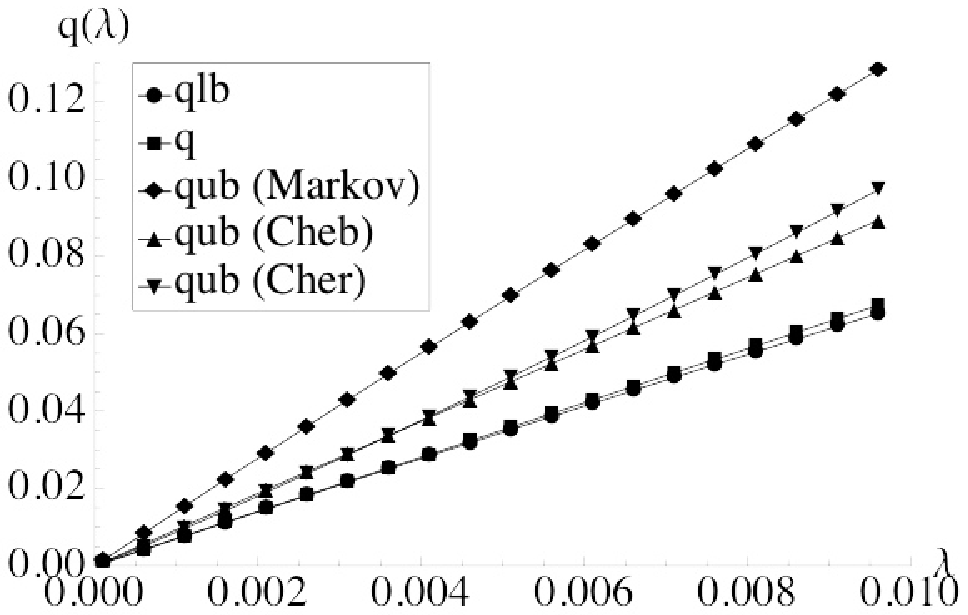}
\includegraphics[width=0.49\textwidth]{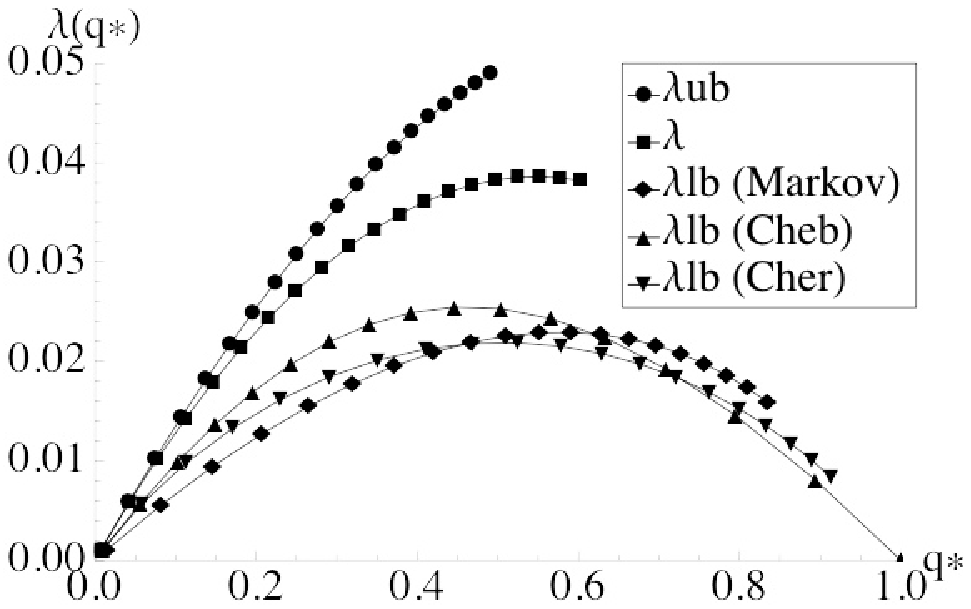}
\includegraphics[width=0.49\textwidth]{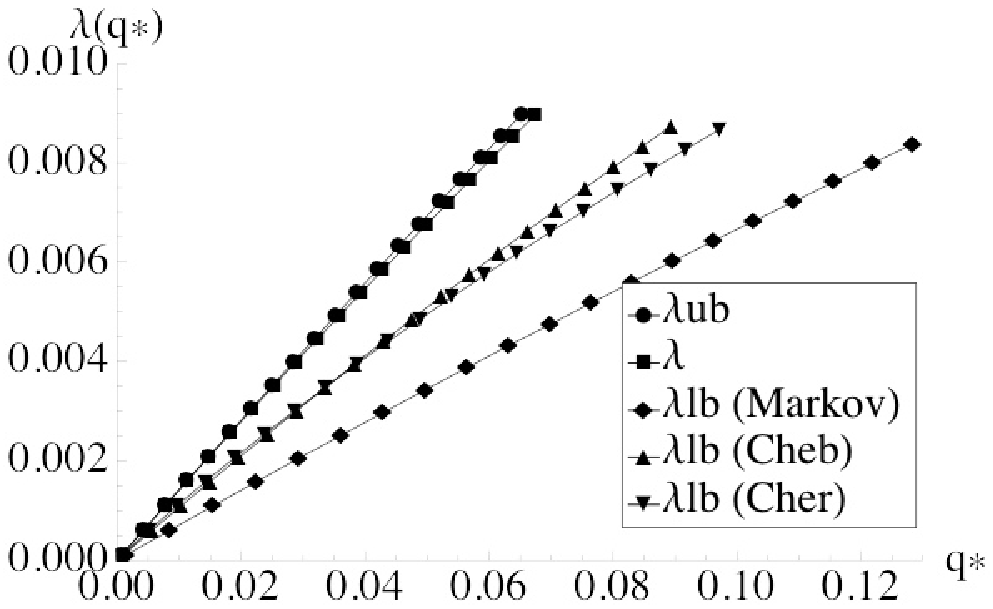}
\caption{{\bf Top:} the OP $q(\lambda)$, the LB, and the three UBs (Markov, Chebychev, Chernoff)  vs.\  $\lambda \in [0,0.1]$ (left) and $\lambda \in [0,0.01]$ (right).  {\bf Bottom:} the TC $\lambda(q^*)$, the UB, and the three LBs (Markov, Chebychev, Chernoff) vs.\ $q^* \in [0,1]$ (left) and $q^* \in [0,0.1]$ (right).}
\label{fig:optcmarkchebcher}
\end{figure}

%
%
\chapter{Extensions to the basic model}
\label{cha:modenh}

In this chapter we study three extensions of the basic model in Ch.\ \ref{cha:bm}:
\begin{enumerate}
\item {\em Channel fading:} allow for iid channel fading on top of the pathloss attenuation, with Rayleigh fading a special focus.
\item {\em Variable link distances:} allow each Tx--Rx pair to be separated by a random distance $\usf$, iid across pairs.
\item {\em Multi-hop networks:} measure performance for a multi-hop extension of the model using suitably modified OP and TC metrics.
\end{enumerate}

These three extensions are chosen because they are among the most obvious steps towards a more realistic decentralized network model.  We shall see that fading can be added to the model without any major difficulties, and in fact when all fading is Rayleigh, exact results are easier to compute than without fading.  Variable link distances are straightforward to include, and result only in a multiplicative constant, which justifies the use of the less realistic but simpler fixed distance model adopted in the rest of the monograph (and most of the literature).  Multihop is a nontrivial extension and requires end-to-end definitions of OP and TC, but under a simple model we are able to preserve tractability and determine quantities like the optimum number of hops.

\section{Channel fading}
\label{sec:fading}

In this section we modify the basic model of Ch.\  \ref{cha:bm} to include channel fading.  Recall in Ch.\  \ref{cha:matpre} and \ref{cha:bm} we operated under Ass.\ \ref{ass:snp} where in particular the amplitudes $\{\hsf_i\}$ in Def.\ \ref{def:snp}  were assumed to be unity.  We now relax that assumption.
\begin{definition}
\label{def:fad}
{\bf SINR under fading.}
Fix $\epsilon = 0$ in this section.  Let $\hsf_0$ be the fading coefficient on the signal channel between the reference Tx and the reference Rx at $o$, and let $\{\hsf_i\} = \hsf_1,\hsf_2,\ldots$ be iid RVs representing the fading coefficients on the channels between each interferer and the reference Rx at $o$ (as in Def.\ \ref{def:snp}).  Let $\Ssf(o),\Sisf(o)$ denote the random signal and interference powers seen at $o$ each normalized by the transmission power $P$:
\begin{equation}
\Ssf(o) \equiv \hsf_0 u^{-\alpha}, \mbox{ and } \Sisf(o) \equiv \Sisf^{\alpha,\hsf}_{d,\lambda}(o) \equiv \sum_{i \in \Pi_{d,\lambda}} \hsf_i |\xsf_i|^{-\alpha}.
\end{equation}
The SINR at $o$ is as in the basic model in Def.\ \ref{def:sinrnf} with $\Ssf(o)$ and $\Sisf(o)$ updated as above:
\begin{equation}
\sinr(o) \equiv \frac{\Ssf(o)}{\Sisf(o)+N/P}.
\end{equation}
Note we have normalized the noise by $P$ since we defined $\Ssf,\Sisf$ to be normalized by $P$ as well.
\end{definition}
\begin{remark}
\label{rem:fadsigint}
{\bf Signal and interference fading coefficients.}
Under Ass.\ \ref{def:fad} both the received signal $\Ssf(o)$ and interference power $\Sisf(o)$ are random, where randomness in $\Ssf(o)$ is due to $\hsf_0$, and randomness in $\Sisf(o)$ is due to both the random positions in $\Pi_{d,\lambda}$ and the fading coefficients $\{\hsf_i\}$.  Note the convention that $\{\hsf_i\} = \hsf_1,\hsf_2,\ldots$, and in particular the signal channel fade $\hsf_0$ is not in the collection of interference channel fades $\{\hsf_i\}$.  To be clear, $\hsf$ denotes a generic interferer fading coefficient, $\hsf_i$ denotes the fading coefficient for interferer $i$, and $\hsf_0$ denotes the signal fading coefficient.  Throughout this section it is important to observe the distinct impacts of $\hsf$ vs.\ $\hsf_0$ on performance, and the distinct requirements for analytic tractability.
\end{remark}
This section is divided into three subsection: exact OP and TC (\S\ref{ssec:fadexact}), asymptotic OP (as $\lambda \to 0$) and TC (as $q^* \to 0$) (\S\ref{ssec:fadasymp}), and a lower (upper) bound on OP (TC) (\S\ref{ssec:fadlb}).

\subsection{Exact OP and TC with fading}
\label{ssec:fadexact}

The Laplace transfom of the interference $\Sisf^{\alpha,\hsf}_{d,\lambda}(o)$ is given in the following proposition (\cite{HaeGan2008} (3.20)).
\begin{proposition}
\label{pro:lapintfad}
{\bf LT of the interference.}
The Laplace transform of the interference $\Sisf^{\alpha,\hsf}_{d,\lambda}(o)$ under Def.\ \ref{def:fad} and for $\delta < 1$ is
\begin{equation}
\label{eq:lapintfad}
\Lmc[\Sisf^{\alpha,\hsf}_{d,\lambda}(o)](s) = \exp \left\{ - \lambda c_d \Ebb[\hsf^{\delta}] \Gamma(1-\delta) s^{\delta} \right\}, ~ s \in \Cbb.
\end{equation}
The RV $\Sisf^{\alpha,\hsf}_{d,\lambda}(o)$ is stable: the characteristic function $\phi[\Sisf^{\alpha,\hsf}_{d,\lambda}(o)](t)$ is given by Def.\ \ref{def:staparam} where the characteristic exponent is $\delta < 1$ and the dispersion coefficient is
\begin{equation}
\label{eq:dispintfad}
\gamma = \left(\lambda c_d \Ebb[\hsf^{\delta}] \Gamma(1-\delta) \cos(\pi \delta/2) \right)^{\frac{1}{\delta}}.
\end{equation}
For the special case $\delta = 1/2$ the RV $\Sisf^{\alpha,\hsf}_{d,\lambda}(o)$ is L\'{e}vy as defined in Def.\ \ref{def:levy} with parameter
\begin{equation}
\gamma = \frac{\pi}{2} \left( \lambda c_d \Ebb[\sqrt{\hsf}] \right)^2.
\end{equation}
\end{proposition}
\begin{proof}
The proof of \eqref{eq:lapintfad} is given in \cite{HaeGan2008} (3.20).  Write $\Sisf = \Sisf^{\alpha,\hsf}_{d,\lambda}(o)$ for this proof.  We find the CF for $\Sisf$ from \eqref{eq:lapintfad}.  Using the expression for $\gamma$ in \eqref{eq:dispintfad} we write the Laplace transform as
\begin{equation}
\Lmc[\Sisf](s) = \exp \left\{ - \frac{\gamma^{\delta}}{\cos(\pi \delta/2)} s^{\delta} \right\}.
\end{equation}
In general we can obtain the CF from the LT via \eqref{eq:ltcfmgf}.
\begin{equation}
\phi[\Sisf](t) = \exp \left\{ - \frac{\gamma^{\delta}}{\cos(\pi \delta/2)} (-\irm t)^{\delta} \right\}.
\end{equation}
To put this in the form of Def.\ \ref{def:staparam} we must establish:
\begin{equation}
- \frac{\gamma^{\delta}}{\cos(\pi \delta/2)} (-\irm t)^{\delta} = - \gamma^{\delta} |t|^{\delta}(1 - \irm \tan(\pi \delta/2) \mathrm{sign}(t)).
\end{equation}
We consider the case $t > 0$, the case $t < 0$ is similar:
\begin{eqnarray}
-\irm &=& \erm^{-\irm \pi/2} \nonumber \\
(-\irm)^{\delta} &=& \erm^{-\irm \pi \delta/2} \nonumber \\
(-\irm)^{\delta} &=& \cos(\pi \delta/2) - \irm \sin(\pi \delta/2) \nonumber \\
\frac{(-\irm)^{\delta}}{\cos(\pi \delta/2)} &=& 1 - \irm \tan(\pi \delta/2).
\end{eqnarray}
\end{proof}
We obtain an explicit expression for the OP and TC when the signal fade $\hsf_0$ is assumed to be exponentially distributed (usually interpreted as modeling Rayleigh fading).  The following result is found in \cite{BacBla2006} and is discussed in \cite{HaeGan2008} \S3.3.
\begin{proposition}
\label{pro:optcrayfadsig}
{\bf OP and TC under Rayleigh signal fading, general interference fading} (\cite{BacBla2006}).
The OP and TC under Def.\ \ref{def:fad} for signal fading coefficient $\hsf_0 \sim \mathrm{Exp}(1)$ (Rayleigh signal fading) are:
\begin{eqnarray}
q(\lambda) &=& 1 - \exp \left\{ - \lambda c_d \Ebb[\hsf^{\delta}] \Gamma(1-\delta) \tau^{\delta} u^d -\frac{\tau}{\snr} \right\} \nonumber \\
\lambda(q^*) &=& \frac{\left(-\log(1-q^*)-\frac{\tau}{\snr}\right)(1-q^*)}{c_d \Ebb[\hsf^{\delta}] \Gamma(1-\delta) \tau^{\delta} u^d}
\end{eqnarray}
\end{proposition}
\begin{proof}
Write $\Sisf = \Sisf^{\alpha,\hsf}_{d,\lambda}(o)$ for this proof.  Solve the outage event for $\hsf_0$, condition on $\Sisf$, and apply the exponential CCDF:
\begin{eqnarray}
1-q(\lambda) &=& \Pbb(\sinr(o) > \tau) \nonumber \\
&=& \Pbb\left( \hsf_0 > \tau u^{\alpha} (\Sisf+N/P) \right) \nonumber \\
&=& \Ebb \left[ \Pbb\left( \left. \hsf_0 > \tau u^{\alpha} (\Sisf+N/P) \right| \Sisf \right) \right] \nonumber \\
&=& \Ebb \left[ \exp \left\{ - \tau u^{\alpha} (\Sisf+N/P) \right\} \right] \nonumber \\
&=& \erm^{-\tau u^{\alpha}N/P} \Ebb \left[ \erm^{-\tau u^{\alpha} \Sisf} \right].
\end{eqnarray}
Now recognize the Laplace transform $\Lmc[\Sisf](s)$ at $s=\tau u^{\alpha}$, use the definition of Rx SNR in Ass.\ \ref{ass:epsnoi}, and apply Prop.\ \ref{pro:lapintfad}:
\begin{eqnarray}
1-q(\lambda) &=& \left. \erm^{-\frac{\tau}{\snr}} \Lmc[\Sisf](s) \right|_{s = \tau u^{\alpha}} \nonumber \\
&=& \left. \erm^{-\frac{\tau}{\snr}} \exp \left\{ - \lambda c_d \Ebb[\hsf^{\delta}] \Gamma(1-\delta) s^{\delta} \right\} \right|_{s = \tau u^{\alpha}}.
\end{eqnarray}
Simplification gives the proposition.
\end{proof}
The following corollary gives the OP and TC for the special case when the interference coefficients are also unit exponentials (\ie, Rayleigh fading for both signal and interference channels).  The key step to obtain the corollary is the following expression for the fractional order moment of the exponential distribution.
\begin{lemma}
\label{lem:expmom}
{\bf Moments of exponential RV.}
For a unit rate exponential RV $\hsf \sim \mathrm{Exp}(1)$ and $\delta \in (0,1)$:
\begin{equation}
\Ebb[\hsf^{\delta}] = \Gamma(1+\delta), ~~ \Ebb[\hsf^{-\delta}] = \Gamma(1-\delta).
\end{equation}
\end{lemma}
Lem.\  \ref{lem:expmom} and the identity \eqref{eq:gamident} yield the following corollary, which is one of the most widely used results appearing in this monograph, and due to \cite{BacBla2006}.
\begin{corollary}
\label{cor:optcrayfadall}
The {\bf OP and TC under Rayleigh fading} ($\hsf_0,\hsf_1,\hsf_2,\ldots \sim \mathrm{Exp}(1)$) are:
\begin{eqnarray}
q(\lambda) &=& 1 - \exp \left\{ - \lambda \frac{\pi \delta c_d}{\sin (\pi \delta)} \tau^{\delta} u^d -\frac{\tau}{\snr} \right\} \nonumber \\
\lambda(q^*) &=& \frac{\sin(\pi \delta) \left(-\log(1-q^*)-\frac{\tau}{\snr}\right)(1-q^*)}{\pi \delta c_d \tau^{\delta} u^d}
\end{eqnarray}
\end{corollary}
Notice that this is an exact closed-form result for OP and TC that does not require bounds, asymptotics, or approximations.  We can specialize it even further by fixing $\delta = 1/2$ and $N = 0$ ($\snr = \infty$) in the above corollary.
\begin{corollary}
\label{cor:optcrayfadlev}
The {\bf OP and TC under Rayleigh fading ($\delta=\frac{1}{2}$, $N=0$)} are
\begin{eqnarray}
q(\lambda) &=& 1 - \exp \left\{ - \lambda \frac{\pi}{2} c_d  \sqrt{\tau} u^d \right\} \nonumber \\
\lambda(q^*) &=& \frac{-2\log(1-q^*)(1-q^*)}{\pi c_d \sqrt{\tau} u^d}
\end{eqnarray}
\end{corollary}
We compare the impact of fading on the OP and TC in Fig.\ \ref{fig:fadnoncomp} by plotting the OP and TC with Rayleigh fading for $\delta = 1/2$ (from Cor.\  \ref{cor:optcrayfadlev}) with the non-fading case for $\delta = 1/2$ (from Cor.\ \ref{cor:oplev} and \ref{cor:tclev}).  The plots illustrate that fading degrades performance in the low outage probability regime.  This may be roughly understood by arguing that fading is a source of variability (and hence diversity) but our slotted Aloha MAC protocol does not exploit that diversity, and therefore performance suffers under the variability.  This degradation is quantified in the asymptotic ($\lambda \to 0$ and $q^* \to 0$) regime in Cor.\ \ref{cor:optcfadnoncomp} and Fig.\ \ref{fig:fadnonasymp}.  We will see how this diversity may be exploited via scheduling (\S\ref{sec:sched}) and power control (\S\ref{sec:power}) leading to OP and TC that outperform the non-fading counterpart.
\begin{figure}[!htbp]
\centering
\includegraphics[width=0.49\textwidth]{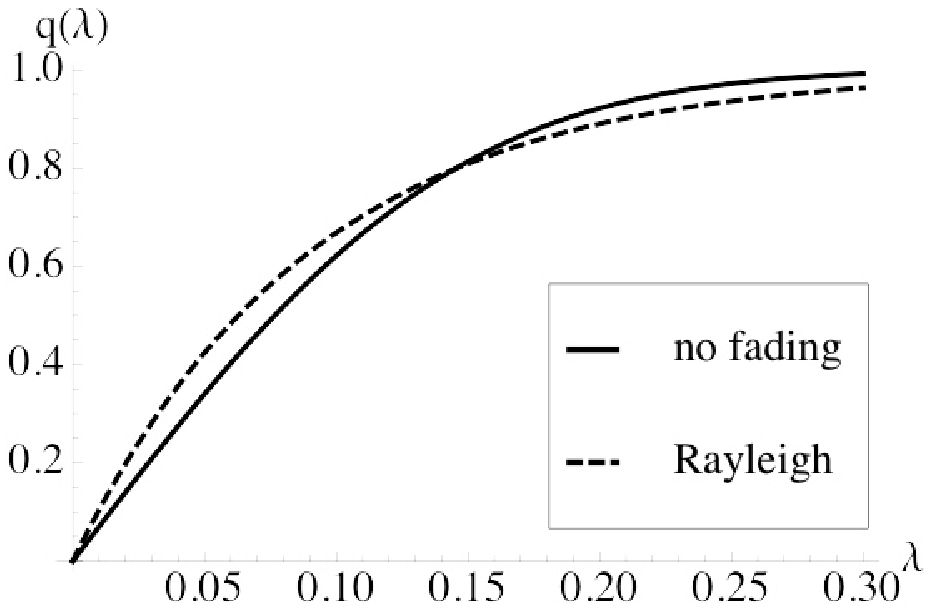}
\includegraphics[width=0.49\textwidth]{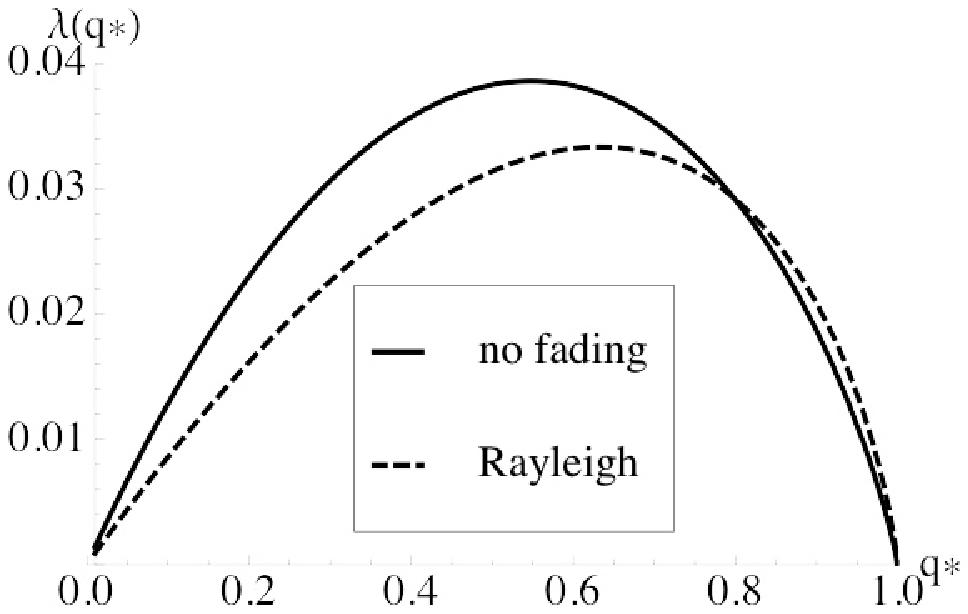}
\caption{The OP (left) and TC (right) for both no fading and Rayleigh fading at $\delta = 1/2$ and $N=0$.  The other parameters are $u =1$, $\tau = 5$, and $d=2$ ($\alpha = 4$).}
\label{fig:fadnoncomp}
\end{figure}

\subsection{Asymptotic OP and TC with fading}
\label{ssec:fadasymp}

The asymptotic OP and TC in the basic model of Ch.\  \ref{cha:bm} were given in \S\ref{sec:asympTC} (Prop. \ref{pro:asymoptc}).  These results were obtained by $i)$ mapping $\Sisf^{\alpha,0}_{d,\lambda}(o)$ to $\Sisf^{1/\delta,0}_{1,1}(o)$ using Prop.\ \ref{pro:imap}, $ii)$ applying the series expansions of the CCDF of $\Sisf^{\alpha,0}_{1,\lambda}$ in Prop.\ \ref{pro:snserexp}, and $iii)$ using this in Prop.\ \ref{pro:opnf} which gives the OP in terms of the CCDF of $\Sisf^{1/\delta,0}_{1,1}(o)$ evaluated at a certain $y$.  We now extend this same procedure to the fading model in Def.\ \ref{def:fad}.
\begin{remark}
\label{rem:MPPP}
{\bf Marked PPP (MPPP).}
The appropriate formalism to incorporate channel fading under Def.\ \ref{def:fad} is that of the marked Poisson point process (MPPP), denoted by $\Phi_{d,\lambda} \equiv \{(\xsf_i,\hsf_i)\}$ where the RV $\hsf_i \in \Rbb_+$ is the ``mark'' associated with point $\xsf_i \in \Rbb^d$.  In general the distribution of the mark is allowed to depend upon the point, but marks must otherwise be independent of other marks and other points.  Let $\hsf_i \in \Rbb_+$ be the fading coefficient for the channel between the interferer at $\xsf_i$ and the reference Rx at $o$, and assume that $\hsf_i$ is independent of $\xsf_i$ and the other points and marks.  The key result is that the MPPP $\Phi_{d,\lambda}$ with pairs $(\xsf_i,\hsf_i)$ taking values in $\Rbb^d \times \Rbb_+$ is in fact a (in general non-homogeneous) PPP with ``points'' $(\xsf_i,\hsf_i)$ taking values in $\Rbb^d \times \Rbb_+$.  Thus the addition of marks does not spoil the tractability of the unmarked PPP framework.  See \cite{Kin1993} Ch.\  5, and in particular the marking theorem in \S5.2, given below.
\end{remark}
\begin{theorem}
\label{thm:mark}
{\bf PPP marking theorem} (\cite{Kin1993}).
The MPPP $\Phi_{d,\lambda} \equiv \{(\xsf_i,\hsf_i)\}$, where $\{\xsf_i\}$ is a (in general, non-homogeneous) PPP on $\Rbb^d$ of intensity $\lambda(x)$ and the marks admit a conditional PDF $f_{\hsf|\xsf}(h|x)$ for each $x \in \Rbb^d$, is a (non-homogeneous) PPP on $\Rbb^d \times \Rbb_+$ with intensity measure
\begin{equation}
\mu(C) =  \int_{(x,h) \in C} \lambda(x)  f_{\hsf|\xsf}(h|x) \drm x \drm h,
\end{equation}
for all Borel sets $C \subseteq \Rbb^d \times \Rbb_+$.
\end{theorem}
Prop.\ \ref{pro:void} gives the void probability for a homogeneous PPP; this is generalized for non-homogeneous MPPPs below.
\begin{proposition}
\label{pro:voidnonhomo}
The {\bf void probability of the non-homogeneous MPPP} $\Phi_{d,\lambda}$ with intensity measure $\mu(C)$ in Thm. \ref{thm:mark} is
\begin{equation}
\Pbb(\Phi_{d,\lambda}(C) = 0) = \erm^{-\mu(C)},
\end{equation}
for all Borel sets $C \subseteq \Rbb^d \times \Rbb_+$.
\end{proposition}
Thm.\ \ref{thm:mark} allows for marked versions of Prop.\ \ref{pro:distmap} and \ref{pro:imap}.
\begin{proposition}
\label{pro:markdistmap}
{\bf MPPP distance and interference mapping.}
Let $\Phi_{d,\lambda} = \{(\xsf_i,\hsf_i)\}$ be a MPPP where $\{\xsf_i\}$ is a homogeneous PPP on $\Rbb^d$ of intensity $\lambda$ and the marks $\{\hsf_i\}$ are independent of the points, and let $\Phi_{1,1} = \{(\tsf_i,\hsf_i)\}$ be a MPPP where $\{\tsf_i\}$ is a homogeneous PPP in $\Rbb$ of intensity $1$.  Then:
\begin{equation}
(\lambda c_d |\xsf_i|^d,\hsf_i) \stackrel{\drm}{=} (2|\tsf_i|,\hsf_i), ~ i \in \Nbb.
\end{equation}
Further, the following RVs are equal in distribution:
\begin{equation}
\Sisf^{\alpha,\hsf}_{d,\lambda}(o) \stackrel{\drm}{=} \left( \frac{\lambda c_d}{2} \right)^{\frac{\alpha}{d}} \Sisf^{\frac{\alpha}{d},\hsf}_{1,1}(o).
\end{equation}
\end{proposition}
Thus, as with the non-fading case, the interference may be treated as a scaled version of that generated by a unit rate PPP on $\Rbb$.  Mapping from $\Rbb^d$ to $\Rbb$ is important because it allows us to apply a series representation of the CCDF (as in Prop.\ \ref{pro:snserexp}) and take the dominant term to get the asymptotic CCDF as $y \to \infty$ (as in Corr.\ \ref{cor:snasypdf}).  Prop.\ \ref{pro:snserexp} (taken from \cite{LowTei1990} (29)) ignored the marks since in that section we assumed $\hsf_i = 1$ (Ass.\ \ref{ass:snp}), but in fact \cite{LowTei1990} gives the series expansion for $\Sisf^{\alpha,\hsf}_{1,\lambda}$, given below.  Recall the footnote under Prop. \ref{pro:snserexp}.
\begin{proposition}
\label{pro:fadintserrep}
{\bf SN series expansion} (\cite{LowTei1990}).
The series expansions of the PDF and CCDF of the SN RV $\Sisf^{\alpha,\hsf}_{1,\lambda}$ for $\delta = \frac{1}{\alpha} < 1$ are:
\begin{eqnarray}
f_{\Sisf^{\alpha,\hsf}_{1,\lambda}}(y) &=& \frac{1}{\pi y} \sum_{n=1}^{\infty} \frac{(-1)^{n+1}}{n!} \Gamma(1+n \delta) \sin (\pi n \delta)(2\lambda \Gamma(1-\delta) \Ebb[\hsf^{\delta}] y^{-\delta})^n \nonumber \\
\bar{F}_{\Sisf^{\alpha,\hsf}_{1,\lambda}}(y) &=& \frac{1}{\pi \delta} \sum_{n=1}^{\infty} \frac{(-1)^{n+1}}{n n!} \Gamma(1+n \delta) \sin(\pi n \delta)(2\lambda \Gamma(1-\delta) \Ebb[\hsf^{\delta}] y^{-\delta})^n \nonumber \\
\end{eqnarray}
The asymptotic PDF and CCDF as $y \to \infty$ are:
\begin{eqnarray}
f_{\Sisf^{\alpha,\hsf}_{1,\lambda}}(y) &=& 2\lambda \delta \Ebb[\hsf^{\delta}] y^{-1-\delta} + \Omc(y^{-1-2\delta}), ~ y \to \infty \nonumber \\
\bar{F}_{\Sisf^{\alpha,\hsf}_{1,\lambda}}(y) &=& 2\lambda \Ebb[\hsf^{\delta}] y^{-\delta} + \Omc(y^{-2 \delta}), ~ y \to \infty
\end{eqnarray}
\end{proposition}
The following theorem establishes the stability (in the sense of Def.\ \ref{def:staparam}) of $\Sisf^{1/\delta,\hsf}_{1,1}(o)$ (adapted from \cite{SamTaq1994} Thm.\ 1.4.5 and \cite{IloHat1998} Thm.\ 3).
\begin{theorem}
\label{thm:fadstab}
{\bf Interference with fading is stable} (\cite{SamTaq1994}).
For $\delta < 1$, $\Sisf^{1/\delta,\hsf}_{1,1}(o)$ is stable (with CF in Def.\ \ref{def:staparam}), with characteristic exponent $\delta$ and dispersion coefficient
\begin{equation}
\gamma = \left( \frac{1}{1-\delta} \Gamma(2-\delta) \cos(\pi \delta/2) \Ebb[\hsf^{\delta}] \right)^{\frac{1}{\delta}}.
\end{equation}
\end{theorem}
\begin{remark}
\label{rem:fadint}
{\bf Interference and fading moments.}
In both Prop.\ \ref{pro:fadintserrep} and Thm.\ \ref{thm:fadstab} the impact of the random marks (fading coefficients) is restricted to the fractional order moment $\Ebb[\hsf^{\delta}]$.  See \cite{HaeGan2008} (p.33) for an extended discussion of this fact.
\end{remark}
The asymptotic CCDF of $\Sisf^{\alpha,\hsf}_{1,\lambda}$ in Prop.\ \ref{pro:fadintserrep} leads directly to the asymptotic OP (as $\lambda \to 0$) and TC (as $q^* \to 0$) under fading, just as Cor. \ref{cor:snasypdf} led directly to Prop.\ \ref{pro:asymoptc}.   A key difference is that random signal fading invalidates Ass.\ \ref{ass:epsnoi}.
\begin{remark}
\label{rem:fadlowout}
{\bf Fading and outage with no interference.}
With random signal fading $\hsf_0$ there is the possibility of a bad fade causing outage even in the absence of interference.  This outage event
\begin{equation}
\Emc_0 \equiv \left\{ \frac{P u^{-\alpha} \hsf_0}{N} < \tau \right\} = \left\{ \hsf_0 < \frac{\tau}{\snr} \right\} = \left\{ \frac{\hsf_0}{\tau u^{\alpha}} - \frac{N}{P} < 0 \right\},
\end{equation}
for $\snr$ defined in Ass.\ \ref{ass:epsnoi} has probability
\begin{equation}
q(0) \equiv \Pbb\left( \Emc_0 \right) = F_{\hsf_0}\left(\frac{\tau}{\snr}\right).
\end{equation}
Note $q(0)$ is the OP evaluated at $\lambda = 0$.  Denote the complement of $\Emc_0$ by $\bar{\Emc}_0$, and its probability by
\begin{equation}
\bar{q}(0) \equiv 1 -q(0) = \Pbb(\bar{\Emc}_0) = \bar{F}_{\hsf_0}\left(\frac{\tau}{\snr}\right).
\end{equation}
All analysis of OP (and hence TC) must therefore condition on $\hsf_0$ being above or below $\tau/\snr$ to distinguish between the case of outage being possible vs.\ outage being guaranteed.  The range of the OP $q(\lambda)$ is $[q(0),1]$, and the domain of $q^*$ in the TC $\lambda(q^*)$ is $[q(0),1]$.
\end{remark}
\begin{proposition}
\label{pro:asymoptcfad}
The {\bf asymptotic OP under fading} as $\lambda \to 0$ is:
\begin{equation}
\label{eq:fadasyop}
q(\lambda) = 1 - \left(1 - c_d \Ebb[\hsf^{\delta}] \Ebb \left[ \left. \left( \frac{\hsf_0}{\tau u^{\alpha}} - \frac{N}{P}\right)^{-\delta} \right| \bar{\Emc}_0 \right] \lambda \right) \bar{q}(0) + \Omc(\lambda^2) \end{equation}
The asymptotic TC under fading as $q^* \to q(0)$ is
\begin{equation}
\label{eq:fadasytc}
\lambda(q^*) = \frac{\left(\frac{q^*-q(0)}{1-q(0)} \right)}{c_d  \Ebb[\hsf^{\delta}] \Ebb \left[ \left. \left( \frac{\hsf_0}{\tau u^{\alpha}} - \frac{N}{P}\right)^{-\delta} \right| \bar{\Emc}_0 \right]}  + \Omc(q^*-q(0))^2.
\end{equation}
In the no noise case ($N=0$, $q(0) = 0$) these expressions become
\begin{eqnarray}
q(\lambda) &=& c_d \tau^{\delta} u^d \Ebb[\hsf^{\delta}] \Ebb[\hsf_0^{-\delta}]  \lambda + \Omc(\lambda^2) \nonumber \\
\lambda(q^*) &=& \frac{1}{c_d \tau^{\delta} u^d\Ebb[\hsf^{\delta}] \Ebb[\hsf_0^{-\delta}]} q^* + \Omc(q^*)^2 \label{eq:fadoptc2}
\end{eqnarray}
In the no noise and Rayleigh fading case ($\hsf_0$ and $\{\hsf_i\}$ exponential RVs):
\begin{eqnarray}
q(\lambda) &=& \frac{\pi c_d \delta \tau^{\delta} u^d }{\sin(\pi \delta)}  \lambda + \Omc(\lambda^2) \nonumber \\
\lambda(q^*) &=& \frac{\sin(\pi \delta)}{\pi c_d \delta \tau^{\delta} u^d} q^* + \Omc(q^*)^2 \label{eq:fadoptc3}
\end{eqnarray}
\end{proposition}
\begin{proof}
The proof is analogous to that of Prop.\ \ref{pro:opnf}, with a key difference being the requirement to condition on possible vs.\ guaranteed outage, {\em c.f.}\ Rem.\  \ref{rem:fadlowout}.  Applying Def.\ \ref{def:fad} to Def.\ \ref{def:op}, and conditioning on $\hsf_0$ leaves $\Sisf^{\alpha,\hsf}_{d,\lambda}$ as the sole source of randomness in $\sinr(o)$:
\begin{eqnarray}
q(\lambda) &=& \Pbb\left( \sinr(o) < \tau \right) \nonumber \\
&=&  \Pbb\left( \sinr(o) < \tau| \bar{\Emc}_0 \right) \bar{q}(0) + 1 (1-\bar{q}(0)) \nonumber \\
&=& 1 - \left(1 - \Pbb\left( \sinr(o) < \tau| \bar{\Emc}_0 \right) \right) \bar{q}(0) \nonumber \\
&=& 1 - \left(1 - \Ebb \left[ \left. \Pbb\left( \left. \sinr(o) < \tau \right|\hsf_0 \right) \right| \bar{\Emc}_0 \right] \right) \bar{q}(0) \nonumber \\
&=& 1 - \left(1 - \underbrace{\Ebb \left[ \left. \Pbb\left( \left. \Sisf^{\alpha,\hsf}_{d,\lambda}(o) > \frac{\hsf_0}{\tau u^{\alpha}} - \frac{N}{P}  \right|\hsf_0 \right) \right| \bar{\Emc}_0 \right]}_{\Ebb[\cdot]} \right) \bar{q}(0) \label{eq:fadasympf1}
\end{eqnarray}
Note the probability is a RV that is a function of the RV $\hsf_0$, and the expectation is with respect to the conditional PDF for $\hsf_0$:
\begin{equation}
f_{\hsf_0}\left( \left. h \right| \bar{\Emc}_0 \right) = \left\{ \begin{array}{ll}
\frac{f_{\hsf_0}(h)}{\bar{F}_{\hsf_0}\left( \frac{\tau}{\snr} \right)}, \; & h > \frac{\tau}{\snr} \\
0, \; & \mbox{else}
\end{array} \right.
\end{equation}
Now apply Prop.\ \ref{pro:markdistmap} and \ref{pro:fadintserrep} to the interference probability noting that $\hsf_0$ and $\Sisf$ are independent.
\begin{eqnarray}
\Ebb[\cdot]&=& \Ebb \left[ \left. \Pbb \left( \left. \Sisf^{1/\delta,\hsf}_{1,1}(o) > \underbrace{\left( \frac{\lambda c_d}{2} \right)^{-\frac{1}{\delta}} \left(\frac{\hsf_0}{\tau u^{\alpha}} - \frac{N}{P} \right)}_{\ysf} \right| \hsf_0\right) \right| \bar{\Emc}_0 \right] \nonumber \\
&=& \Ebb \left[ \left. \Pbb \left( \left. \Sisf^{1/\delta,\hsf}_{1,1}(o) > \ysf  \right| \ysf \right) \right| \ysf > 0 \right]\nonumber \\
&=& \Ebb \left[ \left. \Ebb[\hsf^{\delta}] \ysf^{-\delta} + \Omc(\ysf^{-2 \delta}) \right| \ysf > 0 \right]
\end{eqnarray}
Note $\bar{\Emc}_0 = \left\{ \hsf_0 > \frac{\tau}{\snr} \right\}$ is equivalent to $\{\ysf > 0\}$.  Simplification yields \eqref{eq:fadasyop}.  Solving $q(\lambda) = q^*$ for $\lambda$ yields \eqref{eq:fadasytc}.  Expressions \eqref{eq:fadoptc2} are immediate upon substituting $N=0$, and expressions\eqref{eq:fadoptc3} are immediate upon applying Lem.\  \ref{lem:expmom} and \eqref{eq:gamident}.
\end{proof}
\begin{remark}
\label{rem:fadoptc}
{\bf OP and fading moments.}
In Rem.\  \ref{rem:fadint} we noted the distribution of the interference depended upon the fading coefficients only through $\Ebb[\hsf^{\delta}]$.  In the no noise case of Prop.\ \ref{pro:asymoptcfad} we see the asymptotic   OP and TC depend upon the signal fading coefficient $\hsf_0$ and the interference fading coefficients $\{\hsf_i\}$ only through the product $\Ebb[\hsf^{\delta}]\Ebb[\hsf_0^{-\delta}]$.
\end{remark}
Comparing the no-noise asymptotic OP and TC with fading in \eqref{eq:fadoptc2} in Prop.\ \ref{pro:asymoptcfad} with the analogous results without fading in Prop.\ \ref{pro:asymoptc} yields the following corollary.
\begin{corollary}
\label{cor:optcfadnoncomp}
{\bf Fading degrades performance.}
Let $q(\lambda),\lambda(q^*)$ denote the OP and TC without fading, and $\tilde{q}(\lambda),\tilde{\lambda}(q^*)$ denote the OP and TC with fading.  Let the signal and interference fading distributions be equal ($\hsf \stackrel{\drm}{=} \hsf_0$).  Under no noise ($N=0$) the asymptotic OP (as $\lambda \to 0$) and TC (as $q^* \to 0$) with and without fading have ratio
\begin{equation}
\frac{\tilde{q}(\lambda)}{q(\lambda)} = \frac{\lambda(q^*)}{\tilde{\lambda}(q^*)} = \Ebb[\hsf^{\delta}]\Ebb[\hsf_0^{-\delta}] > 1.
\end{equation}
For Rayleigh fading this ratio is given by Lem.\  \ref{lem:expmom} and \eqref{eq:gamident}:
\begin{equation}
\label{eq:fadnonasympcomp}
\Ebb[\hsf^{\delta}]\Ebb[\hsf_0^{-\delta}] = \Gamma(1+\delta)\Gamma(1-\delta) = \frac{\pi \delta}{\sin(\pi \delta)} > 1.
\end{equation}
Fading degrades asymptotic performance relative to non-fading.
\end{corollary}
\begin{proof}
Jensen's inequality (Prop.\ \ref{pro:jensen}) asserts $\Ebb[f(\xsf)] \geq f(\Ebb[\xsf])$ for convex $f(\cdot)$ and RV $\xsf$.  Use convex function $f(x) = 1/x$ and RV $\hsf^{\delta}$.
\end{proof}
Fig.\ \ref{fig:fadnonasymp} shows \eqref{eq:fadnonasympcomp} vs.\ $\delta$.  For $\delta = 1/2$ (\eg, $d=2$ and $\alpha = 4$) this quantity is $\pi/2 \approx 1.57$, \ie, the asymptotic OP / TC is $57\%$ worse in the presence of Rayleigh fading than without fading.
\begin{figure}[!htbp]
\centering
\includegraphics[width=0.75\textwidth]{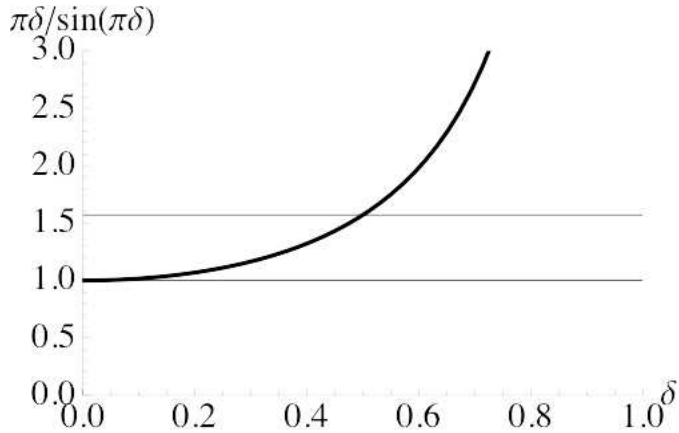}
\caption{The quantity $\pi \delta / \sin(\pi \delta)$ in \eqref{eq:fadnonasympcomp} and \eqref{eq:gamident} vs.\ $\delta$.  Also shown is $1$ and $\pi/2 \approx 1.57$.}
\label{fig:fadnonasymp}
\end{figure}

\subsection{LB on OP (UB on TC) with fading}
\label{ssec:fadlb}

The OP LB (and TC UB) was established for the basic model through the use of dominant interferers, defined in Def.\ \ref{def:domint} as those nodes whose interference strength (under pathloss without fading) was sufficient to individually cause outage at the reference Rx.  The following definition extends the concept to allow for fading.
\begin{definition}
\label{def:domintfad}
{\bf Dominant interferers and interference.}
An interferer $i$ in the MPPP $\Phi_{d,\lambda}$ defined in Prop.\ \ref{pro:markdistmap} is dominant at $o$ under threshold $\tau$ and signal fade $\hsf_0$ if its interference contribution is sufficiently strong to cause an outage for the reference Rx at $o$:
\begin{equation}
\frac{\hsf_0 u^{-\alpha}}{\hsf_i |\xsf_i|^{-\alpha} + N/P} < \tau \Leftrightarrow \hsf_i |\xsf_i|^{-\alpha} > \frac{\hsf_0}{\tau u^{\alpha}} - \frac{N}{P}.
\end{equation}
Else $i$ is non-dominant.  The set of dominant and non-dominant interferers at $o$ under $(\tau,\hsf_0)$ is
\begin{equation}
\hat{\Phi}_{d,\lambda}(o) \equiv \left\{ i \in \Phi_{d,\lambda} : \hsf_i |\xsf_i|^{-\alpha} > \frac{\hsf_0}{\tau u^{\alpha}} - \frac{N}{P} \right\}, ~
\tilde{\Phi}_{d,\lambda}(o) \equiv \Phi_{d,\lambda} \setminus \hat{\Phi}_{d,\lambda}(o).
\end{equation}
The dominant and non-dominant interference at $o$ under $(\tau,\hsf_0)$
\begin{equation}
\hat{\Sisf}^{\alpha,\hsf}_{d,\lambda}(o) \equiv \sum_{i \in \hat{\Phi}_{d,\lambda}(o)} \hsf_i |\xsf_i|^{-\alpha}, ~
\tilde{\Sisf}^{\alpha,\hsf}_{d,\lambda}(o) \equiv \sum_{i \in \tilde{\Phi}_{d,\lambda}(o)} \hsf_i |\xsf_i|^{-\alpha}
\end{equation}
are the interference generated by the dominant and non-dominant nodes.  Note $\Sisf^{\alpha,\hsf}_{d,\lambda}(o) = \hat{\Sisf}^{\alpha,\hsf}_{d,\lambda}(o) + \tilde{\Sisf}^{\alpha,\hsf}_{d,\lambda}(o)$.
\end{definition}
This definition leads directly to a LB on the OP and an (in general, numerically computed) UB on the TC.
\begin{proposition}
\label{pro:fadoplb}
{\bf OP LB.}
Under Def.\ \ref{def:fad} the OP has a LB
\begin{equation}
\label{eq:fadoplb}
q^{\rm lb}(\lambda) = 1 - \Ebb \left[ \left. \exp \left\{ - \lambda c_d \Ebb[\hsf^{\delta}]  \left(\frac{\hsf_0}{\tau u^{\alpha}} - \frac{N}{P} \right)^{-\delta} \right\} \right| \bar{\Emc}_0 \right] \bar{q}(0)
\end{equation}
where the outer expectation is w.r.t. the random signal fade $\hsf_0$.  In the case of no noise ($N=0$, $q(0)=0$) the LB is:
\begin{equation}
\label{eq:fadoplbnn}
q^{\rm lb}(\lambda) = 1 - \Ebb \left[ \exp \left\{ - \lambda c_d \tau^{\delta} u^d \Ebb[\hsf^{\delta}] \hsf_0^{-\delta}  \right\}  \right],
\end{equation}
In particular, with no noise the LB is expressible in terms of the MGF of the RV $-\hsf_0^{-\delta}$ at a certain $\theta$:
\begin{equation}
q^{\rm lb}(\lambda) = 1 - \left. \Mmc[-\hsf_0^{-\delta}](\theta)\right|_{\theta = \lambda c_d \tau^{\delta} u^d \Ebb[\hsf^{\delta}]}.
\end{equation}
\end{proposition}
\begin{proof}
Repeat the proof of Prop.\ \ref{pro:asymoptcfad} up to \eqref{eq:fadasympf1}, LB in terms of $\hat{\Sisf}^{\alpha,\hsf}_{d,\lambda}(o)$, and then express the LB outage event in terms of $\hat{\Phi}_{d,\lambda}(o)$:
\begin{eqnarray}
q(\lambda) &=& \Pbb(\sinr(o) < \tau) \nonumber \\
&=& 1 - \left(1 - \Ebb \left[ \left. \Pbb\left( \left. \Sisf^{\alpha,\hsf}_{d,\lambda}(o) > \frac{\hsf_0}{\tau u^{\alpha}} - \frac{N}{P}  \right|\hsf_0 \right) \right| \bar{\Emc}_0 \right] \right) \bar{q}(0) \nonumber \\
& > & 1 - \left(1 - \Ebb \left[ \left. \Pbb\left( \left. \hat{\Sisf}^{\alpha,\hsf}_{d,\lambda}(o) > \frac{\hsf_0}{\tau u^{\alpha}} - \frac{N}{P}  \right|\hsf_0 \right) \right| \bar{\Emc}_0 \right] \right) \bar{q}(0) \nonumber \\
&=& 1 - \left(1 - \Ebb \left[ \left. \Pbb\left( \left. \hat{\Phi}_{d,\lambda}(o) \neq \emptyset \right|\hsf_0 \right) \right| \bar{\Emc}_0 \right] \right) \bar{q}(0) \nonumber \\
&=& 1  - \Ebb \left[ \left.  \Pbb\left( \left. \hat{\Phi}_{d,\lambda}(o) = \emptyset \right| \hsf_0 \right) \right| \bar{\Emc}_0 \right] \bar{q}(0) \label{eq:mida}
\end{eqnarray}
The PPP $\Phi_{d,\lambda}$ is a homogeneous PPP with intensity measure given by Thm.\ \ref{thm:mark} with $\lambda(x) = \lambda$ and $f_{\hsf|\xsf}(h|x) = f_{\hsf}(h)$.  The key observation is that the probability that $\hat{\Phi}_{d,\lambda}(o)=\emptyset$ equals the void probability for $\Phi_{d,\lambda}$ on the set $C_0 = \{ (x,h) : h |x|^{-\alpha} > w_0\}$ for $w_0 = \frac{h_0}{\tau u^{\alpha}} - \frac{N}{P}$.  Note $\hsf_0$ is random and hence so is $\Csf_0$ and $\wsf_0$.  Using Prop.\ \ref{pro:voidnonhomo} and simplifying gives:
\begin{eqnarray}
\Pbb\left( \left. \hat{\Phi}_{d,\lambda}(o) = \emptyset \right| \hsf_0 \right) &=&
\Pbb\left( \left. \Phi_{d,\lambda}(\Csf_0) = 0 \right| \hsf_0 \right) \nonumber \\
&=& \exp \left\{ - \lambda \int_{(x,h) \in \Csf_0} f_{\hsf}(h) \drm x \drm h \right\} \nonumber \\
&=& \exp \left\{ - \lambda \int_{\Rbb^d} \int_{\Rbb_+} \mathbf{1}_{h |x|^{-\alpha} > \wsf_0} f_{\hsf}(h) \drm x \drm h \right\} \nonumber \\
&=& \exp \left\{ - \lambda \int_{\Rbb^d} \Ebb[ \left. \mathbf{1}_{\hsf |x|^{-\alpha} > \wsf_0} \right| \wsf_0 ] \drm x \right\} \nonumber \\
&=& \exp \left\{ - \lambda \int_{\Rbb^d} \Pbb( \left. \hsf |x|^{-\alpha} > \wsf_0 \right| \wsf_0 ) \drm x \right\}
\end{eqnarray}
Express the CCDF for $\hsf$ in terms of $|x|$ and apply Thm.\ \ref{thm:baker} to reduce the integral over $\Rbb^d$ to an integral over $\Rbb$:
\begin{eqnarray}
\Pbb\left( \left. \hat{\Phi}_{d,\lambda}(o) = \emptyset \right| \hsf_0 \right)
&=& \exp \left\{ - \lambda \int_{\Rbb^d} \Pbb \left( \left. \left(\frac{\hsf}{\wsf_0}\right)^{\frac{1}{\alpha}} > |x| \right| \wsf_0 \right) \drm x \right\} \label{eq:fadlbpf12} \\
&=& \exp \left\{ - \lambda d c_d \int_{\Rbb_+} \Pbb \left( \left. \left(\frac{\hsf}{\wsf_0}\right)^{\frac{1}{\alpha}} > r \right| \wsf_0 \right) r^{d-1} \drm r \right\} \nonumber
\end{eqnarray}
Exchange the order of integration as follows:
\begin{eqnarray}
\Pbb\left( \left. \hat{\Phi}_{d,\lambda}(o) = \emptyset \right| \hsf_0 \right)
&=& \exp \left\{ - \lambda d c_d \int_{\Rbb_+} \Pbb (\left. \hsf > \wsf_0 r^{\alpha} \right| \wsf_0) r^{d-1} \drm r \right\} \nonumber \\
&=& \exp \left\{ - \lambda d c_d \int_{\Rbb_+} \left( \int_{\wsf_0 r^{\alpha}}^{\infty} f_{\hsf}(h) \drm h \right) r^{d-1} \drm r \right\} \nonumber \\
&=& \exp \left\{ - \lambda d c_d \int_{\Rbb_+} \left( \int_0^{(h/\wsf_0)^{1/\alpha}} r^{d-1} \drm r \right) f_{\hsf}(h) \drm h \right\} \nonumber
\end{eqnarray}
Evaluate the inner integral and rearrange:
\begin{eqnarray}
\Pbb\left( \left. \hat{\Phi}_{d,\lambda}(o) = \emptyset \right| \hsf_0 \right)
&=& \exp \left\{ - \lambda d c_d \int_{\Rbb_+} \frac{1}{d} \left( \frac{h}{\wsf_0} \right)^{\delta}  f_{\hsf}(h) \drm h \right\} \nonumber \\
&=& \exp \left\{ - \lambda c_d \wsf_0^{-\delta} \int_{\Rbb_+} h^{\delta} f_{\hsf}(h) \drm h \right\} \nonumber \\
&=& \exp \left\{ - \lambda c_d \wsf_0^{-\delta} \Ebb[\hsf^{\delta}] \right\}
\end{eqnarray}
Substituting this last expression into \eqref{eq:mida} yields \eqref{eq:fadoplb}.
\end{proof}
Note the MGF $\Mmc[-\hsf_0^{-\delta}](\theta)$ is not in general expressible in closed form, and consequently the OP LB in Prop.\ \ref{pro:fadoplb} must be numerically computed, see Fig.\ \ref{fig:mgffadlb}.  Further, the MGF is not in general analytically invertible, and thus the corresponding TC UB must be numerically computed.  Fig.\ \ref{fig:fadexasylb} shows several OP and TC expressions from \S\ref{sec:fading} for the specific case of  Rayleigh fading (for both signal and interference channels), $d = 2$ and $\alpha = 4$ ($\delta = 1/2$), and no noise ($N=0$).  The other parameters include $\tau = 5$ and $u = 1$.  The plots include exact results from \S\ref{ssec:fadexact} (Cor.\ \ref{cor:optcrayfadlev}), asymptotic results from \S\ref{ssec:fadasymp} (Prop.\ \ref{pro:asymoptcfad}), and bound results from \S\ref{ssec:fadasymp} (Prop.\ \ref{pro:fadoplb}).
\begin{figure}[!htbp]
\centering
\includegraphics[width=0.75\textwidth]{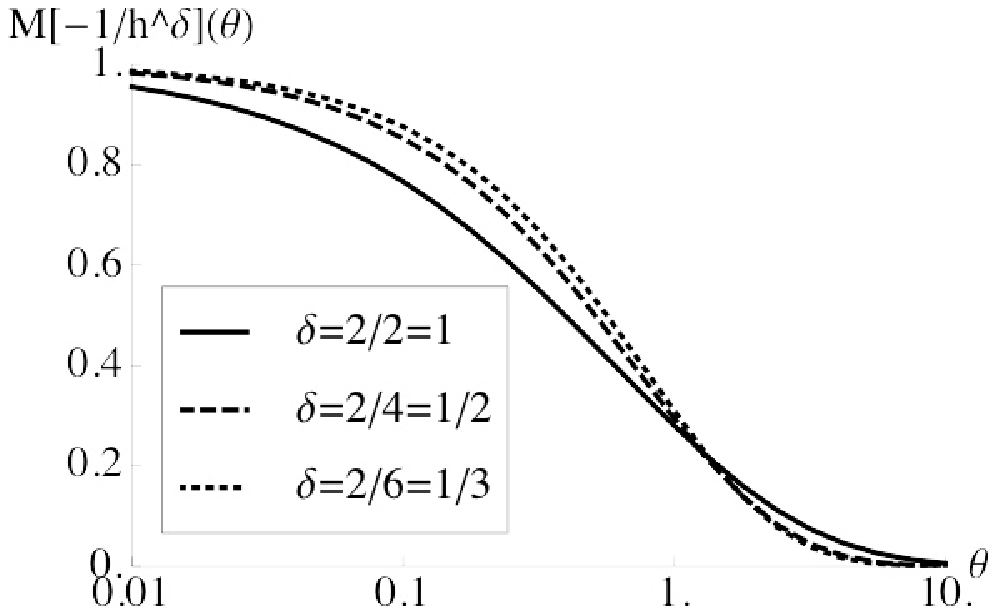}
\caption{The MGF for the RV $-\hsf^{-\delta}$ in Prop.\ \ref{pro:fadoplb} ($\Mmc[-\hsf^{-\delta}](\theta)$) vs.\ $\theta$, for $\hsf \sim \mathrm{Exp}(1)$ and $d = 2$ and $\alpha \in \{2,4,6\}$ ($\delta \in \{1,1/2,1/3\}$).}
\label{fig:mgffadlb}
\end{figure}

\begin{figure}[!htbp]
\centering
\includegraphics[width=0.49\textwidth]{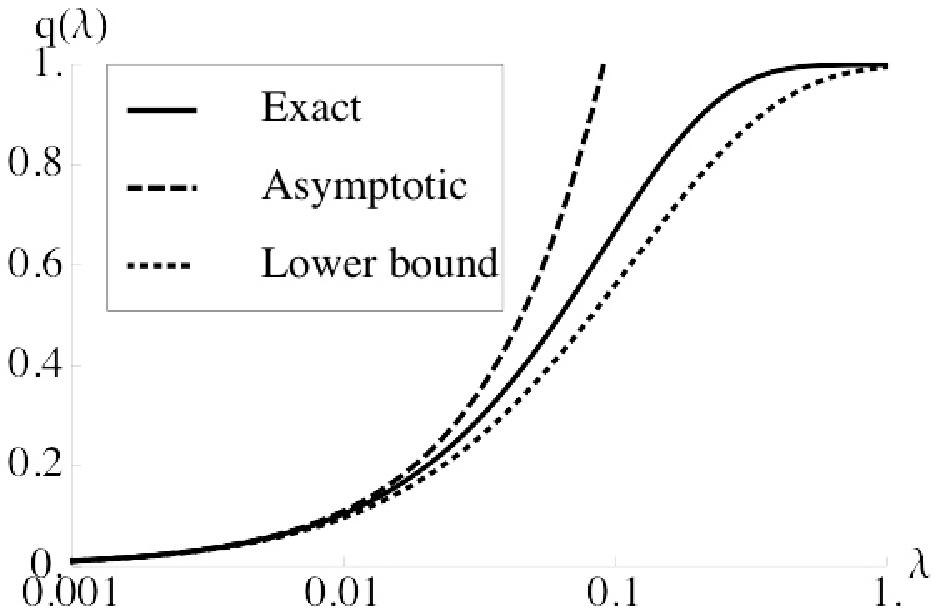}
\includegraphics[width=0.49\textwidth]{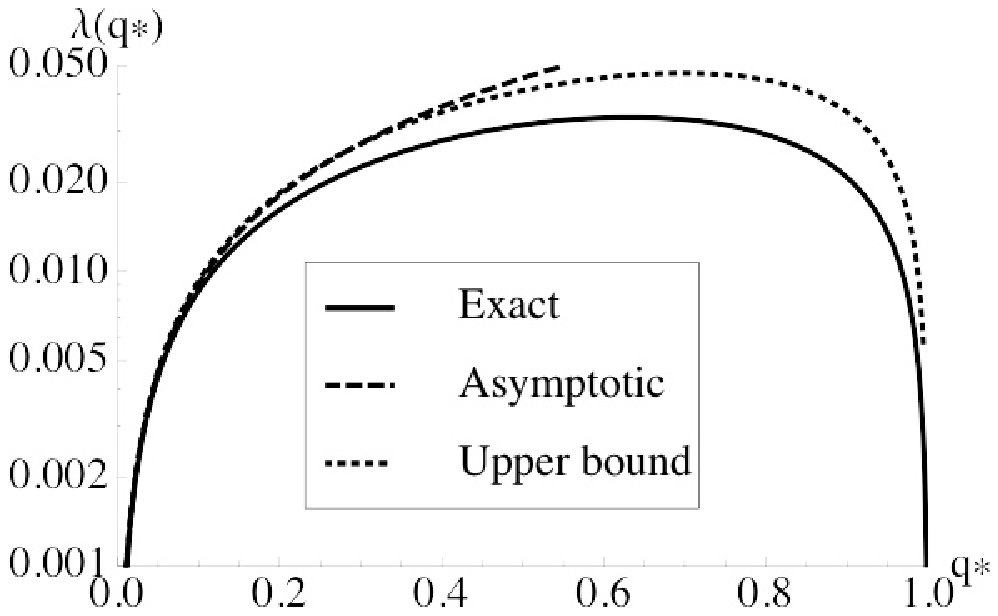}
\caption{The OP (left) and TC (right) for Rayleigh fading (signal and interference), $d = 2$, $\alpha = 4$ ($\delta = 1/2$) $\tau = 5$, $u=1$, and $N=0$.  Shown are the exact, asymptotic, and bound results from \S\ref{sec:fading}.  }
\label{fig:fadexasylb}
\end{figure}

\section{Variable link distances (VLD)}
\label{sec:vardist}

In this section we make perhaps the simplest enhancement to the basic model: the distance separating each Tx--Rx pair is allowed to be a RV $\usf \sim F_{\usf}$ iid across pairs, which we call variable link distances (VLD), as opposed to fixed link distances (FLD).  For simplicity of exposition we fix $\epsilon = 0$ throughout this section.  We retain the assumption of constant power $P$ for each node.  There are now two sources of randomness: the PPP $\Pi_{d,\lambda}$ yielding a random interference $\Sisf$, and the Tx-Rx separation distance $\usf$, yielding a random signal power $\Ssf = P \usf^{-\alpha}$, requiring an updated definition of SINR.
\begin{definition}
\label{def:vardistsinr}
{\bf SINR for VLD.}
Under VLD $\usf \sim F_{\usf}$ and $\epsilon = 0$, the SINR in Def.\ \ref{def:sinrnf} holds but with signal power a RV $\Ssf \equiv P \usf^{-\alpha}$.
\end{definition}
The total law of probability gives an updated definition for the OP.
\begin{definition}
\label{def:opvld}
The {\bf OP for VLD} with $\usf \sim F_{\usf}$ is
\begin{equation}
\label{eq:opvld}
q(\lambda) \equiv \Ebb[ \Pbb(\sinr(o) < \tau|\usf) ].
\end{equation}
\end{definition}
The impact of VLD on the OP and TC is most easily seen by applying the asymptotic OP and TC  from Prop.\ \ref{pro:asymoptc} for the case of no noise ($N=0$).  Denote the OP and TC with VLD $\tilde{q}(\lambda),\tilde{\lambda}(q^*)$ and denote the OP and TC with FLD as $q(\lambda),\lambda(q^*)$.
\begin{proposition}
\label{prop:vardistoptc}
{\bf Asymptotic OP ($\epsilon = 0$, $N=0$).}
For $\epsilon = 0$ and $N=0$ the asymptotic OP with FLD $u = \Ebb[\usf]$ and VLD $\usf \sim F_{\usf}$ are
\begin{eqnarray}
q(\lambda) &=& c_d \tau^{\delta} \Ebb[\usf]^d \lambda + \Omc(\lambda^2), ~ \lambda \to 0 \nonumber \\
\tilde{q}(\lambda) &=& c_d \tau^{\delta} \Ebb[\usf^d] \lambda + \Omc(\lambda^2), ~ \lambda \to 0
\end{eqnarray}
For $\epsilon = 0$ and $N=0$ the asymptotic TC with FLD $u = \Ebb[\usf]$ and VLD $\usf \sim F_{\usf}$ are
\begin{eqnarray}
\lambda(q^*) &=& \frac{1}{c_d \tau^{\delta} \Ebb[\usf]^d} q^* + \Omc(q^*)^2, ~ q^* \to q(0) \nonumber \\
\tilde{\lambda}(q^*) &=& \frac{1}{c_d \tau^{\delta} \Ebb[\usf^d]} q^* + \Omc(q^*)^2, ~ q^* \to q(0) 
\end{eqnarray}
\end{proposition}
Note for small $\lambda, q^*$ the ratio of OP and TC under VLD vs.\ FLD is:
\begin{equation}
\label{eq:vldratio}
\frac{\tilde{q}(\lambda)}{q(\lambda)} = \frac{\lambda(q^*)}{\tilde{\lambda}(q^*)} = \frac{\Ebb[\usf^d]}{\Ebb[\usf]^d} \geq 1.
\end{equation}
The inequality is shown by applying Jensen's inequality to the convex function $u^d$.  We conclude that VLD degrades performance for $d > 1$ in the regime of small $\lambda,q^*$ relative to FLD with the same mean.

We next apply the OP LB in Prop. \ref{pro:oplb} to Def.\ \ref{def:opvld} to express the OP LB for VLD in terms of the MGF of $-\usf^d$.
\begin{proposition}
\label{pro:oplbvld}
{\bf OP LB as an MGF.}
For $\epsilon = 0$ and $N=0$ the OP LB for VLD is the MGF of the RV $-\usf^d$ evaluated at $\theta = \lambda c_d \tau^{\delta}$:
\begin{equation}
\label{eq:oplbvld}
\tilde{q}^{\rm lb}(\lambda) = 1 - \Ebb \left[
\erm^{- \lambda c_d \tau^{\delta} \usf^d} \right] = 1 - \left. \Mmc[-\usf^d](\theta)\right|_{\theta = \lambda c_d \tau^{\delta}}.
\end{equation}
\end{proposition}
We specialize this example to a particular choice of $F_{\usf}$ for which we can explicitly compute the MGF $\Mmc[-\usf^d]$.  Fix $\usf$ to be the contact distribution for a PPP, \ie, the random distance from any point $x \in \Rbb^d$ to the nearest point in $\Pi_{d,\mu}$.  As motivation, suppose the PPP $\Pi_{d,\mu}$ represents locations of base stations, these locations induce a Voronoi tesselation of $\Rbb^d$, and a Tx in the PPP $\Pi_{d,\lambda}$ is assigned to the base station for the transmitter's cell.  The CCDF is immediate from the void probability in Prop.\ \ref{pro:void}.
\begin{proposition}
\label{pro:nnvld}
{\bf Nearest neighbor RV characteristics.}
The nearest neighbor contact RV $\usf$ for a PPP $\Pi_{d,\mu}$ has CCDF:
\begin{equation}
\label{eq:nnvld}
\bar{F}_{\usf}(u) = \Pbb(\usf > u) = \Pbb(\Pi_{d,\mu}(\brm_d(o,u)) = 0) = \erm^{-\mu c_d u^d}.
\end{equation}
The RVs $\usf$ and $\usf^d$ have means and ratio
\begin{equation}
\Ebb[\usf] = \frac{\frac{1}{d} \Gamma\left(\frac{1}{d}\right)}{(c_d \mu)^{\frac{1}{d}}}, ~~~ \Ebb[\usf^d] = \frac{1}{c_d \mu}, ~~~ \frac{\Ebb[\usf^d]}{\Ebb[\usf]^d} = \left( \frac{d}{\Gamma(1/d)} \right)^d,
\end{equation}
which evaluates to $1$ ($d=1$), $\approx 1.273$ ($d=2$), and $\approx 1.404$ ($d=3$).  The RV $-\usf^d$ has MGF:
\begin{equation}
\Mmc[-\usf^d](\theta) = \frac{c_d \mu}{\theta + c_d \mu}.
\end{equation}
The OP LB in Prop.\ \ref{pro:oplbvld} for this choice of $U$ is
\begin{equation}
\tilde{q}(\lambda) \geq \tilde{q}^{\rm lb}(\lambda) = 1 - \frac{c_d \mu}{c_d \lambda \tau^{\delta} + c_d \mu} = \frac{\lambda \tau^{\delta}}{\lambda \tau^{\delta}+ \mu}.
\end{equation}
\end{proposition}
We may further specialize this example to the case of $\delta = 1/2$ and employ the exact OP in Cor.\  \ref{cor:oplev} in Def.\ \ref{def:opvld}.
\begin{corollary}
\label{cor:vardistexactopnn}
{\bf Exact OP ($\epsilon = 0$, $N=0$, $\delta = \frac{1}{2}$).}
For $\usf$ in Prop.\ \ref{pro:nnvld}, the OP is
\begin{equation}
\tilde{q}(\lambda) = 2 \Ebb \left[ F_{\zsf} \left( \sqrt{\frac{\pi}{2}} \usf^d \sqrt{\tau} c_d \lambda \right) \right] - 1.
\end{equation}
\end{corollary}
Fig.\ \ref{fig:vld} shows six curves of OP vs.\ $\lambda$ for the case of $\tau = 1$, $\mu = 1/100$, $d = 2$, $\alpha = 4$ ($\delta=1/2$).  The six curves are the exact, LB, and asymptotic OP for both VLD using the nearest neighbor distance $\usf$ in Prop.\ \ref{pro:nnvld}, and the FLD OP (from Ch.\  \ref{cha:bm}) with $u = \Ebb[\usf] = 5$.  Observe the LBs lie below the exact OP as expected and are asymptotically tight as $\lambda \to 0$, and that the asymptotic approximations as $\lambda \to 0$ are valid.  Observe the VLD OP is slightly higher than the FLD OP for $\lambda \to 0$, as expected by \eqref{eq:vldratio}.  Note however, that the ordering for $\lambda \to 0$ does not hold for larger $\lambda$, as we in fact observe FLD OP exceeds VLD OP in this regime.  Although the example illustrates there may be some significant dependence of the OP (and hence the TC) on the link distance variability for $\lambda$ (and hence the OP) large, we also observe the impact of this variability to be minimal in the small $\lambda$ (and hence small OP) regime, which is typically of more practical interest.  In short, this section shows that a variable Tx-Rx link distance can be included in the OP and TC results without any major modifications; the link distance is simply conditioned on and averaged which, for a given link distribution, causes a fixed reduction in TC.

\begin{figure}[!htbp]
\centering
\includegraphics[width=0.49\textwidth]{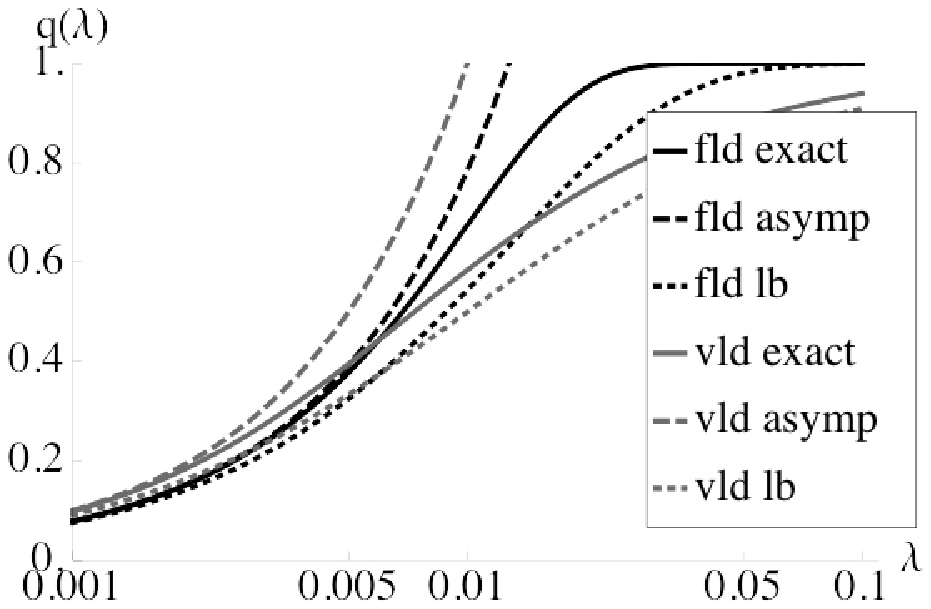}
\includegraphics[width=0.49\textwidth]{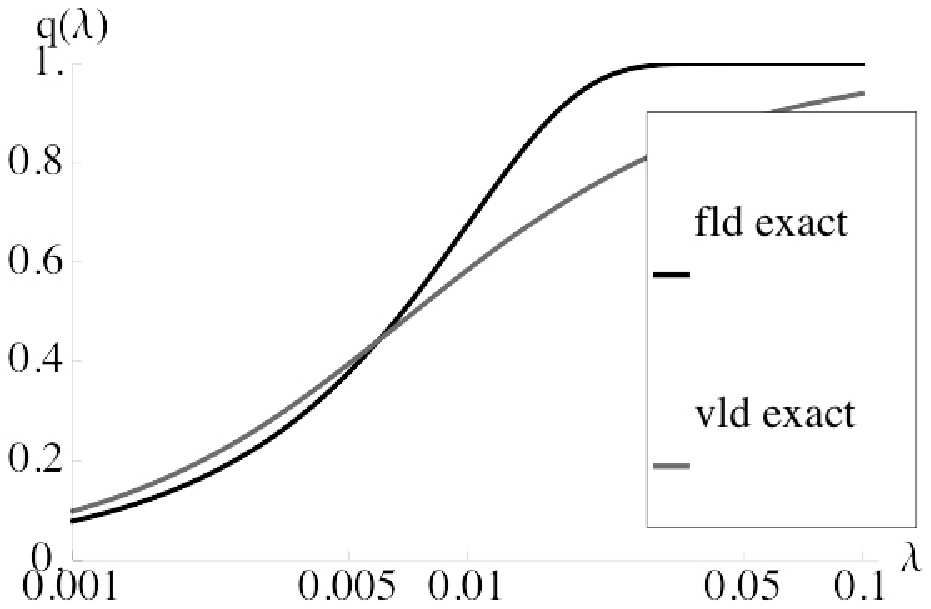}
\includegraphics[width=0.49\textwidth]{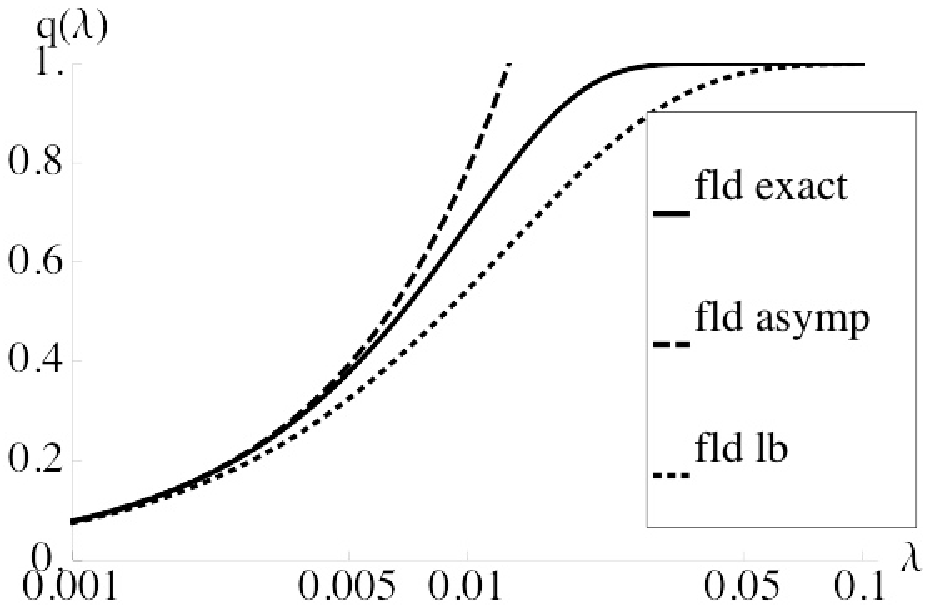}
\includegraphics[width=0.49\textwidth]{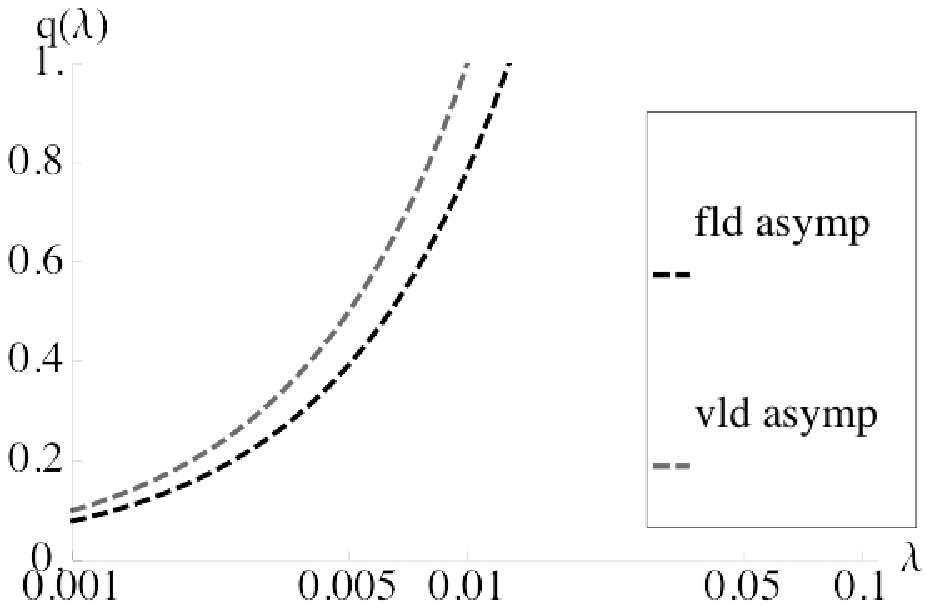}
\includegraphics[width=0.49\textwidth]{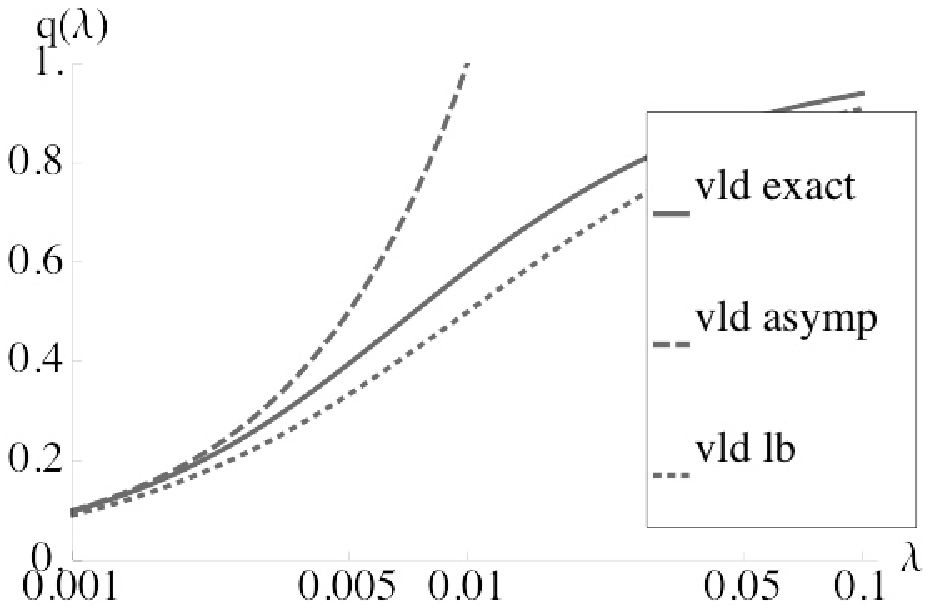}
\includegraphics[width=0.49\textwidth]{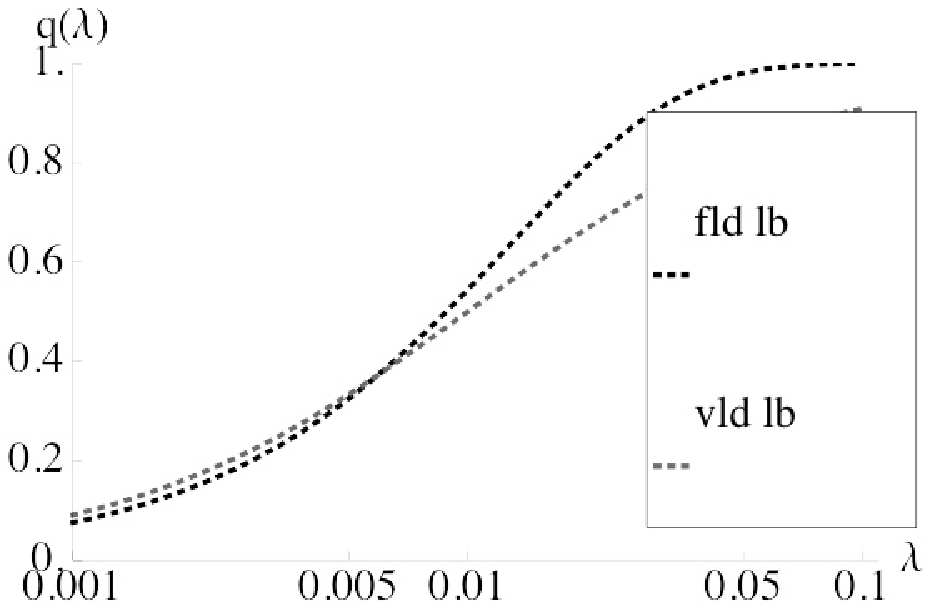}
\caption{The OP (exact, LB, and asymptotic approximation for $\lambda \to 0$) vs.\ $\lambda$ for both FLD and VLD.  The top left plot shows all six curves, various subsets of these are shown in the remaining five plots.}
\label{fig:vld}
\end{figure}

\section{Multihop TC}
\label{sec:multihop}

Since an \adhoc network -- or any wireless network -- may wish to use intermediate relays to transfer packets from source to destination, an extension of the TC metric to multiple hops is very desirable. This is particularly the case in ``power-limited'' networks; \ie, those where the source-destination distance $U$ (say) is sufficiently large that it cannot be bridged with a single transmission, but instead must be broken into $M$ shorter hops using other nodes in the network as relays.  This is a challenging extension, and in this section we overview one recent approach that retains fairly good tractability.  Building on the TC framework so far, in this section we assume:
\begin{enumerate}
\item All links experience unit mean Rayleigh fading.
\item All hops are equidistant of length $u=U/M$, and on a straight line between the source and destination.  This can easily be shown to be ``best case'' in terms of success probability.
\item A packet is repeatedly transmitted on each hop until it is successfully received by the next hop, or until a timeout corresponding to a maximum total number of end-to-end (e2e) attempts $A$ occurs.
\item Each transmission attempt experiences iid fading and interference, and hence has the same OP which is independent of all other attempts.
\item A packet must reach the final destination before a new one is injected by the source; \ie, there is no intra-route spatial reuse.  The spatial intensity of attemped transmissions at a point in time is $\lambda$ (nodes per square meter); this is also the spatial intensity of source nodes.
\end{enumerate}
This model can be conceptualized by Fig.\ \ref{fig:MH-model}.  Based on these assumptions, we introduce  multihop TC, originally defined in \cite{AndWeb2010}\footnote{This paper introduces a slightly different metric called {\em random access transport capacity}.}.  We acknowledge prior work \cite{StaRos2009} that proposed a similar model but focused more on conditions for queue stability over the multihop routes.
\begin{figure}[!htbp]
\centering
\includegraphics[width=1\textwidth]{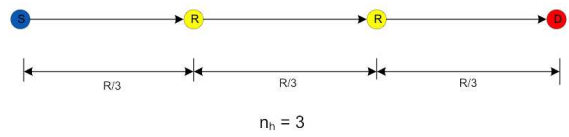}
\caption{The multihop TC model with $M=3$.}
\label{fig:MH-model}
\end{figure}

\begin{definition} \textbf{Multihop TC.}
\label{def:MHTC}
Define $\lambda_{\rm mh}(\lambda,U,M,A)$ as the spatial intensity of successfully delivered packets (packets per square meter) when the intensity of source nodes is $\lambda$, each source destination pair is separated by $U$ meters, each such pair employs $M-1$ relays positioned equidistant on the straight line connecting them, and each packet is permitted at most $A$ transmission attempts (end to end) before timeout.  Let $\Tsf_M \sim \mathrm{Pascal}(M,q(\lambda))$\footnote{A special case of the negative binomial distribution.  Sometimes these distributions are defined as the number of failures (instead of trials) until $M$ successes are achieved.} be the RV denoting the total number of independent transmission attempts required to achieve $M$ successes when each trial fails according to the OP $q(\lambda)$.  Then
\begin{equation}
\label{eq:mhtc1}
\lambda_{\rm mh}(\lambda,U,M,A) = \frac{\Pbb(\Tsf_M \leq A)}{\Ebb[\Tsf_M \land A]} \lambda,
\end{equation}
where we thin by the probability of successful end to end delivery $\Pbb(T_M \leq A)$ and divide by the average number of transmissions required for end to end delivery, $\Ebb[\Tsf_M \land A]$.  The multihop TC is defined as the maximum of \eqref{eq:mhtc1} over the number of hops $M$:
\begin{equation}
\label{eq:mhdefn}
\lambda_{\rm mh} = \lambda \max_{M \in [A]} \frac{\Pbb(\Tsf_M \leq A)}{\Ebb[\Tsf_M \land A]}.
\end{equation}
It has units of packets per unit area.
\end{definition}
Recall $[A] \equiv \{1,\ldots,A\}$.  Neither the numerator $\Pbb(\Tsf_M \leq A)$ nor the denominator $\Ebb[\Tsf_M \land A]$ in \eqref{eq:mhtc1} can be computed directly, however, a useful inequality is given in the following proposition.
\begin{proposition} 
\label{pro:MHTCineq}
{\bf Multihop TC inequality.} For all $A \in \Nbb$ and all $M \in \{1,\ldots,A\}$:
\begin{equation}
\frac{\Pbb(\Tsf_M \leq A)}{\Ebb[\Tsf_M \land A]} \leq \frac{1}{\Ebb[\Tsf_M]}.
\end{equation}
\end{proposition}
The proof is given in \cite{AndWeb2010}.  The main idea is that as $A$ increases the numerator $\Pbb(\Tsf_M \leq A) \to 1$ and the denominator $\Ebb[\Tsf_M \land A] \to \Ebb[\Tsf_M]$ both increase, but the numerator does so slightly more quickly, resulting in an upper bound.

With Ass.\ (2) above (equally spaced hops), the per-hop per-attempt OP is simply the standard Rayleigh fading OP (Cor.\ \ref{cor:optcrayfadall}) with per-hop distance $u=U/M$ and $\snr = P(U/M)^{-\alpha}/N$:
\begin{equation}
q(\lambda,M) = 1 - \exp \left\{ - \lambda \frac{\pi \delta c_d}{\sin (\pi \delta)} \tau^{\delta} \left( \frac{U}{M}\right)^d - \tau \frac{N}{P} \left(\frac{U}{M}\right)^{\alpha} \right\}.
\label{eq:MHoutage}
\end{equation}
Since each transmission attempt is iid, the RV $\Tsf_M$ is in fact the sum of $M$ independent geometric RVs, $\Tsf_M = \Tsf_M^{(1)} + \cdots \Tsf_M^{(M)}$, where $\Tsf_M^{(i)}$ is the number of independent trials required until success is achieved on hop $i$, and each trial fails with OP $q(\lambda,M)$, \ie,
\begin{equation}
\Pbb(\Tsf_M^{(i)} = t) = q(\lambda,M)^{t-1} (1-q(\lambda,M)), ~ t \in \Nbb.
\end{equation}
Therefore, the average number of attempts per hop is $\Ebb[\Tsf_M^{(i)}] = 1/(1-q(\lambda,M))$ and the expected number of total transmissions required to move a packet from source to destination is $\Ebb[\Tsf_M] = M \Ebb[\Tsf_M^{(i)}] = M/(1-q(\lambda,M))$.  Combining all these observations yields the following UB on the multihop TC.
\begin{proposition} 
\label{pro:MHTC-UB}
{\bf Multihop TC UB.}
An UB on the multihop TC from \eqref{eq:mhdefn} is
\begin{equation}
\lambda_{\rm mh} \leq \lambda_{\rm mh}^{\rm ub} \equiv \lambda \max_{M \in [A]} \frac{1-q(\lambda,M)}{M},
\label{eq:UB-A}
\end{equation}
for $q(\lambda,M)$ in \eqref{eq:MHoutage}.
\end{proposition}
\begin{figure}
\centering
\includegraphics[width=0.49\textwidth]{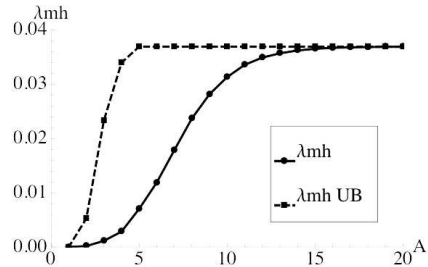}
\caption{The multihop TC $\lambda_{\rm mh}$ from \eqref{eq:mhdefn} and its upper bound $\lambda_{\rm mh}^{\rm ub}$ from \eqref{eq:UB-A} versus the allowed number of transmission attempts $A$.}
\label{fig:MHTC_Fig1}
\end{figure}
The upper bound is somewhat loose for small $A$ and is fairly tight for large $A$, as can be seen in Fig.\  \ref{fig:MHTC_Fig1}. The natural next step is to find the optimal hopcount $M^*(A)$, \ie, the value of $M$ that maximizes the multihop TC.
\begin{definition} 
\label{def:mstar}
{\bf Optimal number of hops.}
The (bound-) optimal number of hops using the multihop TC UB in Prop.\ \ref{pro:MHTC-UB} is
\begin{equation}
\label{eq:TCopt}
M^* \equiv \arg\max_{M \in [A]} \frac{1-q(\lambda,M)}{M}.
\end{equation}
\end{definition}
The following proposition characterizes $M^*$ for the special case $d=2$.
\begin{proposition}
\label{pro:mstar} 
{\bf Optimal number of hops.}  
The optimal number of hops $M^*$ in Def.\ \ref{def:mstar} for $d=2$ is the solution to the equation
\begin{equation}
\label{eq:mstar}
M^{\alpha} - 2 \lambda \tau^{2/\alpha} K_{\alpha} U^2 M^{\alpha -2} - \frac{\tau N U^{\alpha}}{P} \alpha = 0.
\end{equation}
This results in closed-form expressions for $M$ only when $\alpha \in \{3, 4, 6, 8\}$, in which case $M^*$ is the largest positive root of \eqref{eq:mstar}. The constant $K_{\alpha} = \pi^2 \delta \csc(\pi \delta)$ with $\delta = \frac{2}{\alpha}$.
\end{proposition}
\begin{proof}
The objective is
\begin{equation}
\frac{1-q(\lambda,M)}{M} = \frac{1}{M} \exp(-k_1 M^{-\alpha} - k_2 M^{-2})
\end{equation}
where $k_1 = \frac{\tau}{\snr}$ and $k_2 = \lambda \pi^2 \delta \csc(\pi\delta) \tau^{\delta} U^2$, which follows from \eqref{eq:MHoutage}. Setting the derivative equal to zero gives
\begin{equation}
\exp(k_1 M^{-\alpha} + k_2 M^{-2}) \left(1 - k_1\alpha M^{-\alpha} - 2 k_2 M^{-2}
\right) = 0
\end{equation}
which results in
\begin{equation}
M^{\alpha} - 2 k_2 M^{\alpha -2} - k_1 \alpha = 0. \label{eq:MHpoly}
\end{equation}
By the Abel-Ruffini theorem, a formula solution to a polynomial equation only exists when the degree of the polynomial is 4 or less, therefore $M$ can be found in closed-form only for $\alpha \in \{2, 3, 4, 6\}$, although it can also be found in principle for $\alpha \in \{1,8\}$. The solutions for $\alpha = 6$ and $\alpha = 8$ follow similarly to the $\alpha = 3$ and $\alpha = 4$ solutions due to the sparsity of the polynomial.\\
\end{proof}
As noted previously, the expected interference in a 2D PPP is infinite for $\alpha = 2$ (recall $d=2$ in Prop.\ \ref{pro:mstar}). Hence we give the results for $\alpha \in \{3,4\}$ in the following two corollaries.  Proofs follow standard algebraic arguments (although perhaps unfamiliar ones for $\alpha = 3$) and are given in \cite{AndWeb2010}.
\begin{corollary}
\label{cor:mstaralpha3}
{\bf Optimal number of hops for $\alpha = 3$.}
Solving \eqref{eq:MHpoly} with $\alpha = 3$ yields two possible solutions, depending on the polarity of the equation's discriminant $D = \frac{9 k_1^2}{4}-\frac{8k_2^3}{27}$. When $D \geq 0$,
\begin{equation}
M^*(3) = \tau^{\frac{1}{3}} U \left[\sqrt[3]{\frac{3N}{2P} + f} + \sqrt[3]{\frac{3N}{2P} - f}\right], \label{eq:PL3a}
\end{equation}
where
\begin{equation}
f = \sqrt{\left(\frac{3N}{2P}\right)^2- \frac{8 K^3_3}{27}\lambda^3},
\end{equation}
and $K_3 = \frac{4}{9} \sqrt{3} \pi^2 \approx 7.6$ is the $\pi^2 \delta \csc(\pi \delta)$ term evaluated at $\alpha = 3$. When $D < 0$
\begin{equation}
M^*(3) = 2\sqrt{\frac{2\lambda K_3}{3}}\tau^{\frac{1}{3}} U \cos\left( \frac{1}{3} \arc\cos\left[ \frac{3\sqrt{3}N}{4\sqrt{2}P(\lambda K_3)^{\frac{3}{2}}} \right]\right).
\label{eq:PL3b}
\end{equation}
\end{corollary}
Although the expressions \eqref{eq:PL3a} and \eqref{eq:PL3b} at first appear to be quite different, in fact they have a quite similar dependence in terms of the main parameters of interest.
\begin{corollary}
\label{cor:mstaralpha4}
{\bf Optimal number of hops for $\alpha = 4$.}
Solving \eqref{eq:MHpoly} with $\alpha = 4$ yields a unique maximum positive real solution for $M^*$:
\begin{equation}
M^*(4) = \tau^{\frac{1}{4}} U \sqrt{\lambda \frac{\pi^2}{2} + \sqrt{\lambda^2
\frac{\pi^4}{4} + \frac{4N}{P}}}.
\label{eq:alpha4}
\end{equation}
\end{corollary}
The tightness of the bound is shown in Fig.\ \ref{fig:MHTC_CvsM}, where it is seen that the quality of the bound improves as $A$ increases.

\begin{figure}
\centering
\includegraphics[width=0.49\textwidth]{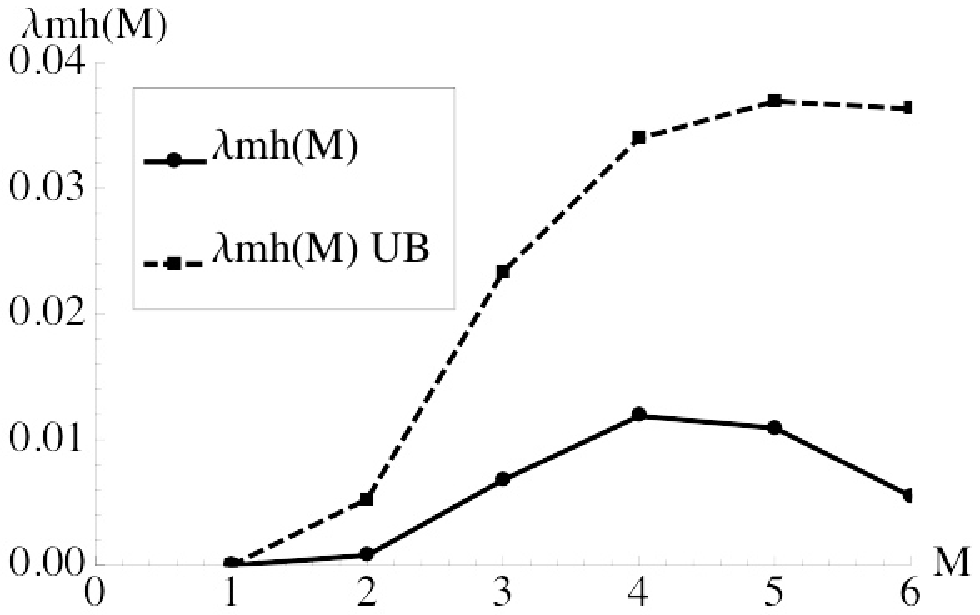}
\includegraphics[width=0.49\textwidth]{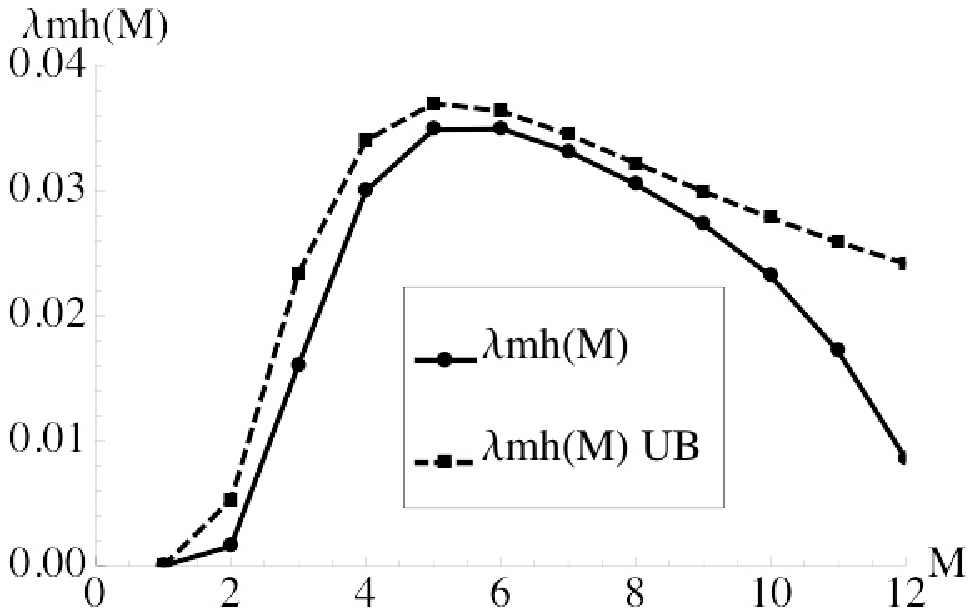}
\caption{The quantity $\lambda_{\rm mh}(\lambda,U,M,A)$ in \eqref{eq:mhtc1} and its upper bound versus the number of hops $M$ for each $M=1,\ldots,A$ when $A=6$ (left) and $A=12$ (right).  Note the maximizing $M^*$ for the exact solution and the UB are close to one another in both cases.}
\label{fig:MHTC_CvsM}
\end{figure}

We conclude by commenting on the five enumerated assumptions that were made at the beginning of this section in order to get to this result.  The first assumption, of Rayleigh fading, is fairly innocuous and used just because it allows a precise equality for per-hop outage probability.  Any of the other channel models discussed in this monograph could of course be used, resulting in appropriately modified bounds.  The second assumption of equi-distant relays on a straight line is perhaps the most physically questionable assumption, especially since the network interference is drawn from a PPP.  This model is best case, so it preserves the upper bound, but it can be legitimately asked how loose the bound might be compared to a network with relays that must be chosen from a random distribution of nodes.  This question has been explored in \cite{CheAnd2010}, but deserves further study.   The third assumption, of transmission until success or repeated failure (timout) is reasonable and does not require modification.  The fourth assumption, of each slot being iid, could be approximated in a frequency hopping or other diversity-harvesting system; in which each slot could truly see iid interference and fading, and thus have nearly independent success probability.  This can be relaxed, see \eg, \cite{Vaz2011} which builds upon \cite{AndWeb2010}.  And finally, the fifth assumption regarding no intra-route reuse is pessimistic for a given source-destination pair, especially for large $M$.  However, in a network-wide sense sparse reuse of a given route allows other source-destination pairs to transmit and so there is no net change to the multihop TC, since this is a spatially averaged metric.  Formally, the loss in multihop TC from $M$ in the denominator (which assumes no intra-route reuse) can be exactly balanced by a corresponding increase in $\lambda$ in the numerator, resulting in the same interference intensity.

%
%
\chapter{Design techniques for wireless networks}
\label{cha:destech}

In this chapter we consider four design techniques/issues for a decentralized wireless network based on the models of Ch.\  \ref{cha:bm} and \ref{cha:modenh}:
\begin{enumerate}
\item {\em Spectrum management:} the network bandwidth $W$ (Hz) is to be broken up into $B$ bands each of size $W/B$ (Hz).  This allows the density of interferers (per band) to be controlled, since the density will then be $\lambda/B$.
\item {\em Interference cancellation:} allow a Rx to cancel a fraction $\kappa$ of the interference generated by the $K$ strongest interferers.
\item {\em Threshold scheduling:} exploit fading by scheduling transmissions for Tx--Rx pairs with a strong channel.
\item {\em Power control:} select the transmission power to partially compensate for the channel to the Rx.
\end{enumerate}

Although this list of enhancements and design issues is by no means exhaustive, these four topics are major issues that any systems engineer designing a centralized network protocol would be faced with. We will see that the TC framework is able to provide insight and quantitative network-level performance analysis regarding these issues, whereas other approaches typically have been unable to achieve this, and must either use simulations or greatly simplified analytical models.  Chapter 6 continues this direction, focusing exclusively on multi-antenna transmission and reception.

\section{Spectrum management}
\label{sec:specman}

In this section we study the impact of multiple frequency bands on OP and TC \cite{JinAnd2008}.  We consider a bandwidth of $W$ (Hz) that is divided into $B$ uniform bands, each with bandwidth $W/B$ (Hz).  The objective is to select $B$ to maximize the TC.  The solution of this optimization is non-trivial on account of the following two dependencies: as $B$ increases there is simultaneously less interference on each band and a higher SINR required to achieve a fixed rate $R$.  The lower interference is on account of there being more bands from which to choose, while the higher SINR requirement is due to Shannon's formula $\csf = W \log_2(1+\sinr)$; the nature of this dependence will be made clear in what follows.
\begin{assumption}
\label{ass:specman}
{\bf Random band selection.}
Throughout this section fix $\epsilon = 0$ and set each $\hsf_i = 1$ (no fading) in Ass.\ \ref{ass:snp}.   Further, assume each Tx will independently and uniformly at random select a single band in $[B]\equiv\{1,\ldots,B\}$ on which to operate.  Let $\bsf_0$ be the random band selected by the reference Tx, and $\{\bsf_i\}$ be the random bands selected by each interferer $i \in \Pi_{d,\lambda}$.
\end{assumption}
\begin{remark}
\label{rem:spec1}
{\bf Thinned interference seen by reference Rx.}
Having $\bsf_0$ and $\{\bsf_i\}$ be uniform and independent implies that the SINRs on each band are iid, and hence, without loss of generality, we may assume the reference Tx selects $\bsf_0 = 1$ (say).  It follows that only interferers that select band $1$ are of relevance to the reference Rx at $o$.  More formally, the MPPP $\Phi_{d,\lambda} = \{(\xsf_i,\bsf_i)\}$ induces $B$ iid PPPs, each of intensity $\lambda/B$, and the PPP of interferers on band $1$ is
\begin{equation}
\Pi_{d,\lambda/B} = \{ \xsf_i : (\xsf_i,1) \in \Phi_{d,\lambda} \}.
\end{equation}
It furthermore follows that the (normalized) interference SN RV seen at $o$ on band $1$ is
\begin{equation}
\Sisf(o) = \Sisf^{\alpha,0}_{d,\lambda/B}(o) = \sum_{i \in \Pi_{d,\lambda/B}} |\xsf_i|^{-\alpha}.
\end{equation}
\end{remark}
The random channel capacity and channel spectral efficiency are defined below, along with equivalent definitions of outage event \eqref{eq:outage}.
\begin{definition}
\label{def:specdefs}
{\bf Noise, SINR, SNR, capacity, spectral efficiency.}
\begin{enumerate}
\item The noise power spectral density is uniform at $\eta$ (W/Hz).  The noise power over a band is $N(B) = \eta \frac{W}{B}$ (W) and the noise power over the full spectrum is $N = \eta W$ (W).
\item The SINR seen at $o$ on band 1 is
\begin{equation}
\sinr(o) \equiv \frac{P u^{-\alpha}}{P \Sisf(o) + \eta \frac{W}{B}} = \frac{1}{\Sisf(o) + \frac{1}{\snr B}},
\end{equation}
where the full-spectrum Rx SNR is
\begin{equation}
\label{eq:specsnr}
\snr \equiv \frac{P u^{-\alpha}}{\eta W}.
\end{equation}
\item The Shannon channel capacity at $o$ on a channel of $W/B$ (Hz), treating interference as noise (Ass.\ \ref{ass:keyass}), is the RV
\begin{equation}
\csf(o) \equiv \frac{W}{B}\log_2 (1 + \sinr(o)), ~ \mbox{(bps)}.
\end{equation}
\item The corresponding Shannon spectral efficiency is the RV
\begin{equation}
\frac{\csf(o)}{W/B} = \log_2 (1 + \sinr(o)), ~ \mbox{(b/s/Hz)}.
\end{equation}
\item The rate and spectral efficiency requirements are defined as
\begin{equation}
\label{eq:speceff}
R \equiv \frac{W}{B} \log_2 (1+\tau) ~~ \mbox{(bps)}, ~~ \nu \equiv \frac{R}{W} B = \log_2(1+\tau) ~~ \mbox{(b/s/Hz)}.
\end{equation}
\item The outage event \eqref{eq:outage} may be expressed in terms of both the rate and spectral efficiency:
\begin{equation}
\sinr(o) < \tau \Leftrightarrow \csf(o) < R \Leftrightarrow \frac{\csf(o)}{W/B} < \nu.
\end{equation}
\end{enumerate}
\end{definition}
The above definitions and discussion motivate the following modifications to the definitions of OP and TC.
\begin{definition}
\label{def:outageequiv}
The {\bf OP and TC under multiple bands} are defined as follows:
\begin{enumerate}
\item The OP at $o$ is a function of the per-band spatial intensity of interferers $\lambda/B$:
\begin{equation}
\label{eq:specop}
q(\lambda/B,B) \equiv \Pbb(\sinr(o)<\tau).
\end{equation}
\item The TC under outage constraint $q^* \in (0,1)$ is defined as
\begin{equation}
\label{eq:spectc}
\lambda(q^*,B^*) \equiv \max_{B \in \Nbb} B q^{-1}(q^*,B)(1-q^*),
\end{equation}
where the maximization is over all possible number of bands $B \in \Nbb$.  Here $q^{-1}(\cdot,B)$ is the inverse of the monotone increasing OP $q(\cdot,B)$ in \eqref{eq:specop}.  The multiplication by $B$ is on account of the fact that $q^{-1}(q^*,B)(1-q^*)$ is the average number of successful transmissions per unit area on band $1$.
\end{enumerate}
\end{definition}
The above definitions make clear that the TC in Def.\ \ref{def:specdefs} is the following modification of the TC (Prop.\ \ref{pro:tcexzereps}) in Ch.\  \ref{cha:bm}.
\begin{proposition}
\label{pro:tcspec}
Using Def.\ \ref{def:specdefs} the {\bf TC under multiple bands} is:
\begin{equation}
\label{eq:spectcM}
\lambda(q^*,B^*) = \kappa(q^*) \omega(B^*)
\end{equation}
where
\begin{equation}
\label{eq:kappaqst}
\kappa(q^*) = \frac{2 (1-q^*)}{c_d u^d (\bar{F}_{\Sisf}^{-1}(q^*))^{\delta}}
\end{equation}
and $\bar{F}_{\Sisf}^{-1}$ is the inverse CCDF of the RV ${\Sisf}_{1,1}^{1/\delta,0}(o)$.  Moreover,
\begin{equation}
\label{eq:omegaM}
\omega(B) = B\left(\frac{1}{2^{\frac{R}{W}B}-1} - \frac{1}{\snr B}\right)^{\delta},
\end{equation}
with
\begin{equation}
\label{eq:omegaMmax}
B^* \in \arg \max_{B \in [ \lfloor B_{\rm max} \rfloor ]} \omega(B),
\end{equation}
for $B_{\rm max}$ satisfying
\begin{equation}
\label{eq:mmax}
B_{\rm max} = \frac{W}{R} \log_2(1 + \snr B_{\rm max}).
\end{equation}
\end{proposition}
\begin{proof}
Multiplying the TC in Prop.\ \ref{pro:tcexzereps} by $B$ as in \eqref{eq:spectc} gives:
\begin{equation}
\lambda(q^*) =
\frac{2 B \left(\frac{u^{-\alpha}}{\tau(B)} - \frac{N(B)}{P}\right)^{\delta} (1-q^*)}{c_d (\bar{F}_{\Sisf}^{-1}(q^*))^{\delta}},
\end{equation}
where we write $\tau(B),N(B)$ to emphasize their dependence upon $B$.  Substitution of \eqref{eq:specsnr} and \eqref{eq:speceff} and algebra yields the proposition.
\end{proof}
\begin{remark}
\label{rem:optnumbands}
{\bf Optimal number of bands independent of target OP.}
The form \eqref{eq:spectcM} highlights the fact that the optimal number of bands $B^*$ is independent of the target OP $q^*$.  In fact $\kappa(q^*)$ captures the spatial components of the network through its dependence upon $u,d,\delta,q^*$, while $\omega(B)$ captures the spectral components of the network through its dependence upon $R,W,\snr,\delta$.
\end{remark}
It is natural to change the design variable from the number of bands $B$ to the spectral efficiency $\nu = \frac{R}{W} B$.  The following definition is central to what follows.
\begin{definition}
\label{def:ebno}
{\bf Energy per bit and receiver SNR.}
The energy per bit is
\begin{equation}
\ebno  \equiv \frac{P u^{-\alpha}}{\eta R} ~ \mbox{(J/bit)},
\end{equation}
and relates to the Rx SNR as
\begin{equation}
R \ebno = W \snr = \frac{P u^{-\alpha}}{\eta}.
\end{equation}
\end{definition}
\begin{remark}
\label{rem:Mdisccont}
{\bf Relaxation of integrality constraint.}
Although $B$ is naturally restricted to be integral, taking values in $[\lfloor B_{\rm max} \rfloor]$ for $B_{\rm max}$ in \eqref{eq:omegaMmax}, by continuity of the objective $\omega(B)$ we can safely relax the domain to the continuous interval $[0,B_{\rm max}]$ and then take the nearest integer.  It follows that the corresponding domain for $\nu$ is $[0,\nu_{\rm max}]$ defined below.
\end{remark}
Simple algebra gives the following corollary of Prop.\ \ref{pro:tcspec}.
\begin{corollary}
\label{cor:spectcebno}
Using Def.\ \ref{def:specdefs} the {\bf TC under multiple bands} is:
\begin{equation}
\label{eq:spectcnu}
\lambda(q^*,\nu^*) = \frac{W}{R} \kappa(q^*) \tilde{\omega}(\nu^*)
\end{equation}
where $\bar{F}_{\Sisf}^{-1}$ and $\kappa(q^*)$ are as in Prop.\ \ref{pro:tcspec},
\begin{equation}
\label{eq:omegaNu}
\tilde{\omega}(\nu) = \nu \left(\frac{1}{2^{\nu}-1} - \frac{1}{\ebno \nu}\right)^{\delta},
\end{equation}
with
\begin{equation}
\label{eq:omegaNumax}
\nu^* \in \arg \max_{\nu \in [0,\nu_{\rm max}]} \tilde{\omega}(\nu),
\end{equation}
and $\nu_{\rm max}$ determined by $\ebno$ and satisfying the equation
\begin{equation}
\label{eq:numax}
\nu_{\rm max} = \log_2 \left(1 + \ebno \nu_{\rm max} \right)
\end{equation}
\end{corollary}
Note \eqref{eq:numax} does not have a solution for all $\ebno$.
\begin{proposition}
\label{prop:ebnomin}
{\bf Minimum energy per bit required for solution.}
Equation \eqref{eq:numax} has a solution precisely for
\begin{equation}
\ebno > \log 2 \approx 0.693 \approx -1.59 ~ \mbox{(dB)}.
\end{equation}
\end{proposition}
\begin{proof}
Application of $x > \log(1+x)$ for $x \in \Rbb_+$ to \eqref{eq:numax} yields:
\begin{equation}
(\log 2) \nu_{\rm max}  = \log\left(1 + \ebno \nu_{\rm max} \right) > \ebno \nu_{\rm max}.
\end{equation}
\end{proof}
\begin{remark}
\label{rem:lowsnrregime}
{\bf Low SNR regime.}  The minimum energy per bit threshold of $-1.59$ (dB) is a fundamental quantity in the wideband regime.  The interested reader is referred to \cite{Ver2002} for a general discussion and to \cite{JinAnd2008} (footnote 4) for an elaboration in this context.
\end{remark}
Fig.\ \ref{fig:specomegam} shows $\omega(B)$ vs.\ $B$ \eqref{eq:omegaM}, $\tilde{\omega}(\nu)$ vs.\ $\nu$ \eqref{eq:omegaNu}, and two plots illustrating $\nu_{\rm max}(\ebno)$ \eqref{eq:numax} using parameters:
\begin{equation}
\label{eq:specplotvals}
\begin{array}{llll}
d=2 & u=1 \mbox{(m)} & \alpha \in \{3,4\} & P \in \{3,6\} ~ \mbox{(dBW)} \\
R = 1 ~ \mbox{(Mbps)} & W = 10 ~ \mbox{(MHz)} & \eta = 10^{-6} ~ \mbox{(W/Hz)} &
\end{array}
\end{equation}
These are the default parameters throughout this section except when indicated otherwise.  Observe that each of the four $(\alpha,P)$ pairs in Fig.\ \ref{fig:specomegam} has a distinct optimal $B^*$ (resp.\ $\nu^*$) that maximizes $\omega(B)$ (resp.\ $\tilde{\omega}(\nu)$).
\begin{figure}[!htbp]
\centering
\includegraphics[width=0.49\textwidth]{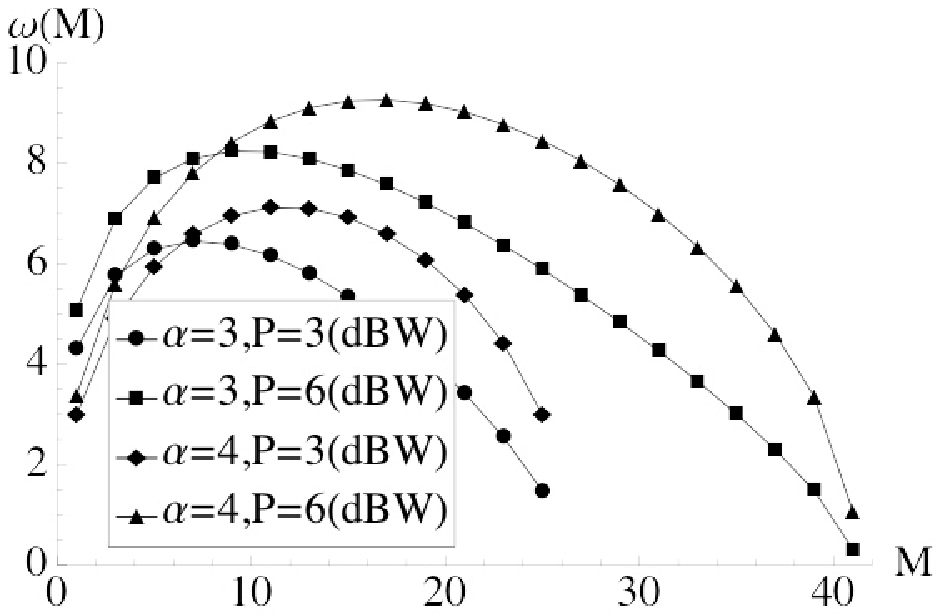}
\includegraphics[width=0.49\textwidth]{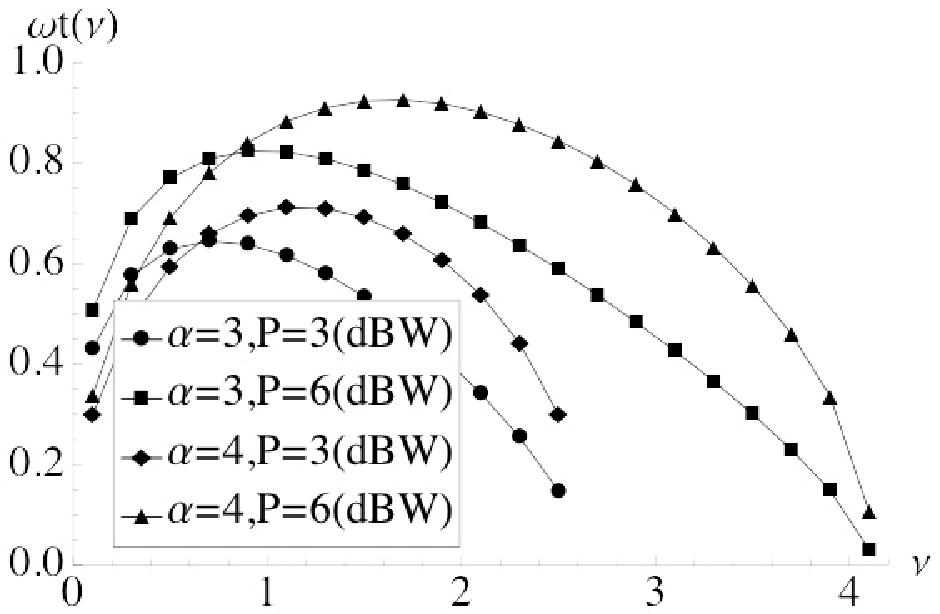}
\includegraphics[width=0.49\textwidth]{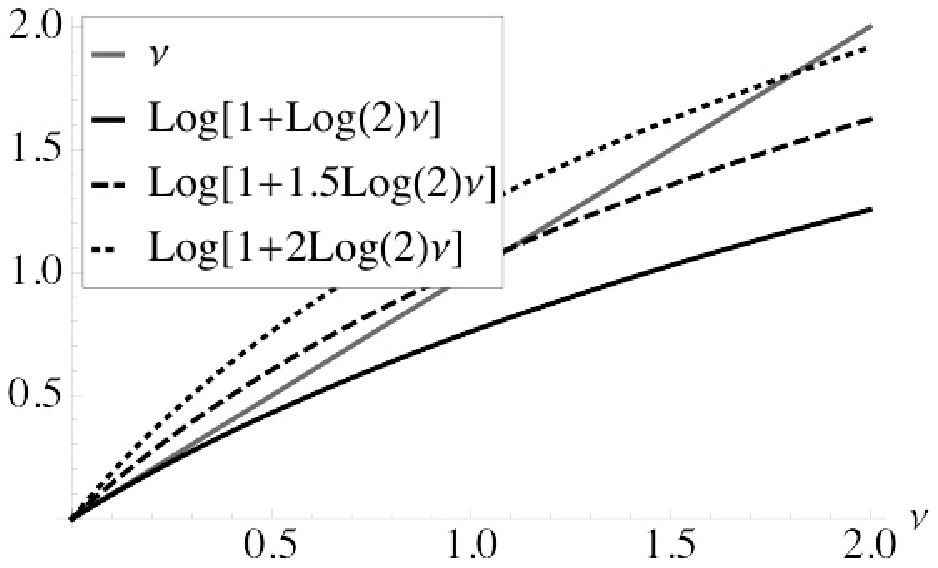}
\includegraphics[width=0.49\textwidth]{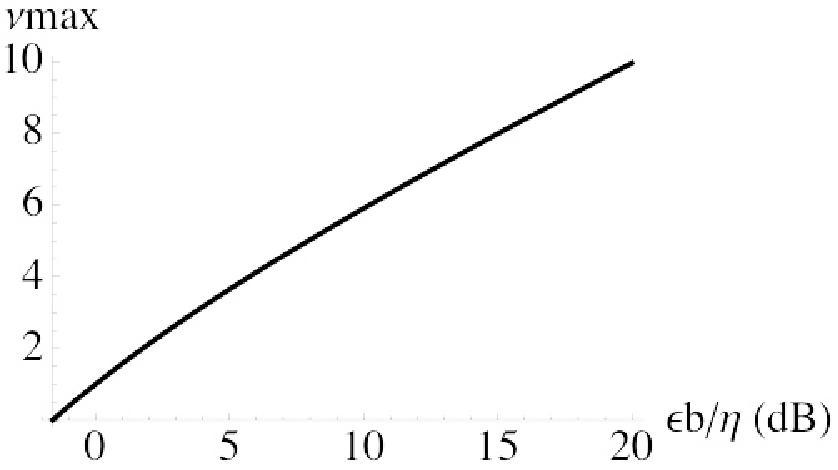}
\caption{{\bf Top:} the spectral component $\omega(B)$ (left) and $\tilde{\omega}(\nu)$ (right) vs.\  the number of bands $B$ (left) and the spectral efficiency $\nu$ (right) under Prop.\ \ref{pro:tcspec} (left) and Cor.\ \ref{cor:spectcebno} (right) using parameters \eqref{eq:specplotvals}.  {\bf Bottom:} the value $\nu_{\rm max}$ obeys $\nu_{\rm max} = \log_2(1+\ebno \nu_{\rm max})$ and thus is the intersection of the $\nu$ line with the $\log_2(1+ \ebno \nu)$ curve (left).  Equivalently, $\nu_{\rm max}(\ebno)$ vs.\ $\ebno$ (in dB) is shown on the right.  No solution $\nu_{\rm max}$ exists for $\ebno < \log 2$ ($\approx -1.59$ dB).  }
\label{fig:specomegam}
\end{figure}
The $B^*$ maximizing $\omega(B)$ in \eqref{eq:omegaMmax} is seen to depend upon three parameters ($R/W$, $\snr$, $\delta$) while the $\nu^*$ maximizing $\tilde{\omega}(\nu)$ in \eqref{eq:omegaNumax} is seen to only depend upon two parameters ($\ebno$ and $\delta$).  For this reason we restrict our attention in what follows to optimization over $\nu$ instead of over $B$, with the understanding that $B^* = \frac{W}{R} \nu^*$ (recall Rem.\  \ref{rem:Mdisccont}).  The following result characterizes $\nu^*$ (\cite{JinAnd2008} Thm.\  1).
\begin{proposition}
\label{pro:specoptnu}
The {\bf optimal spectral efficiency} $\nu^*$ in \eqref{eq:omegaNumax} is the unique positive solution of:
\begin{equation}
\label{eq:specoptnu}
\ebno \nu (2^{\nu}-1) - (1-\delta)(2^{\nu}-1)^2 - \delta (\log 2) \ebno \nu^2 2^{\nu} = 0.
\end{equation}
Furthermore, $\nu^*$ is increasing in $\ebno$ and decreasing in $\delta$.
\end{proposition}
The proof is in \cite{JinAnd2008} and requires only simple calculus.  Fig.\ \ref{fig:specnuoptebno} shows $\nu^*$ vs.\ $\ebno$ for various $\delta$ as well as the TC $\lambda(q^*,\nu^*)$ vs.\ $\ebno$ for various $q^*$.  Observe that $\nu^*$ is increasing in $\ebno$ and decreasing in $\delta$ as asserted in Prop.\ \ref{pro:specoptnu}.   Observe that the TC is increasing in both $\ebno$ and $q^*$.
\begin{figure}[!htbp]
\centering
\includegraphics[width=0.49\textwidth]{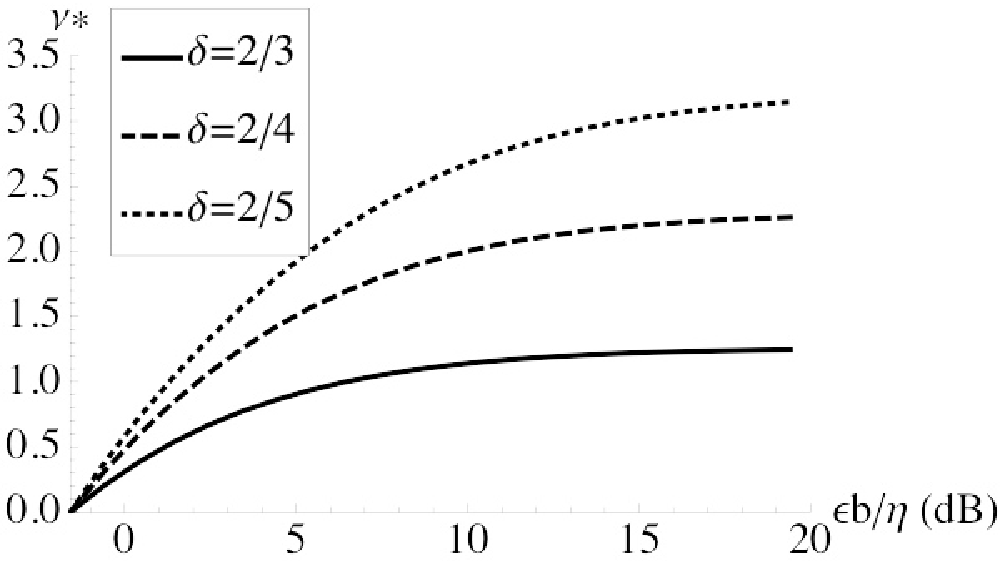}
\includegraphics[width=0.49\textwidth]{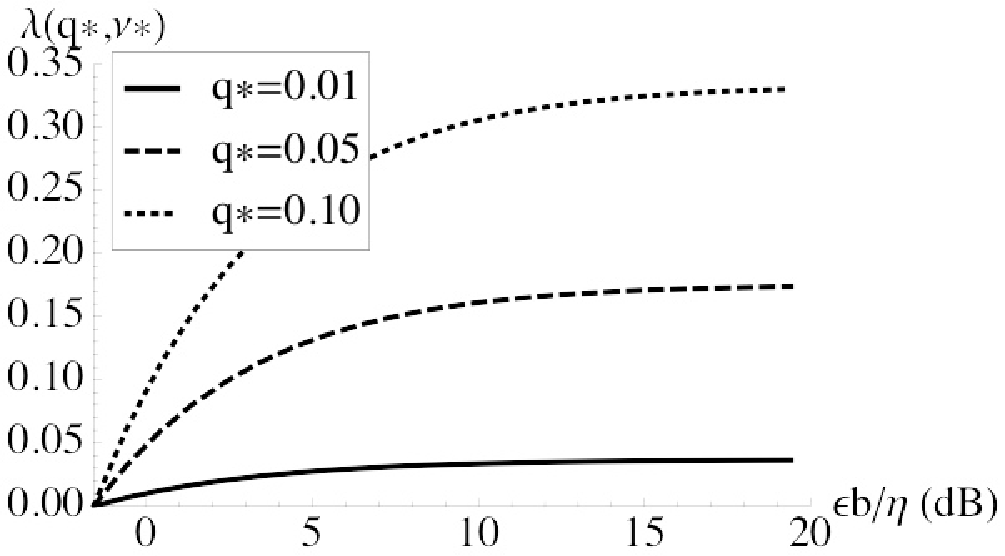}
\caption{{\bf Left:} the optimal $\nu^*$ satisfying \eqref{eq:specoptnu} vs.\ $\ebno$ for $\delta \in \{2/3,2/4,2/5\}$. {\bf Right:} the TC $\lambda(q^*,\nu^*)$ vs.\ $\ebno$ for $\delta = 1/2$ and $q^* \in \{0.01,0.05,0.10\}$.}
\label{fig:specnuoptebno}
\end{figure}

We next study the asymptotic behavior of $\nu^*$ in the high SNR ($\ebno \to \infty$) and low SNR ($\ebno \to \log 2$) regimes.
\begin{corollary}
\label{cor:specasymoptnu}
The {\bf asymptotic optimal spectral efficiency} $\nu^*$ in the high SNR regime ($\ebno \to \infty$) is the unique positive solution of:
\begin{equation}
1 - 2^{-\nu} = (\log 2) \delta \nu.
\end{equation}
\end{corollary}
The corollary is immediate upon observing the second term in \eqref{eq:specoptnu} is negligibly small as $\ebno \to \infty$.  This function $\nu^*(\delta)$ is shown in Fig.\ \ref{fig:FigSpecOptNuAsyHighSNR}.  Observe the convergence of $\nu^*$ to that predicted by Cor.\ \ref{cor:specasymoptnu} is slower for small $\delta$ (equivalently, for high $\alpha$).
\begin{figure}[!htbp]
\centering
\includegraphics[width=0.49\textwidth]{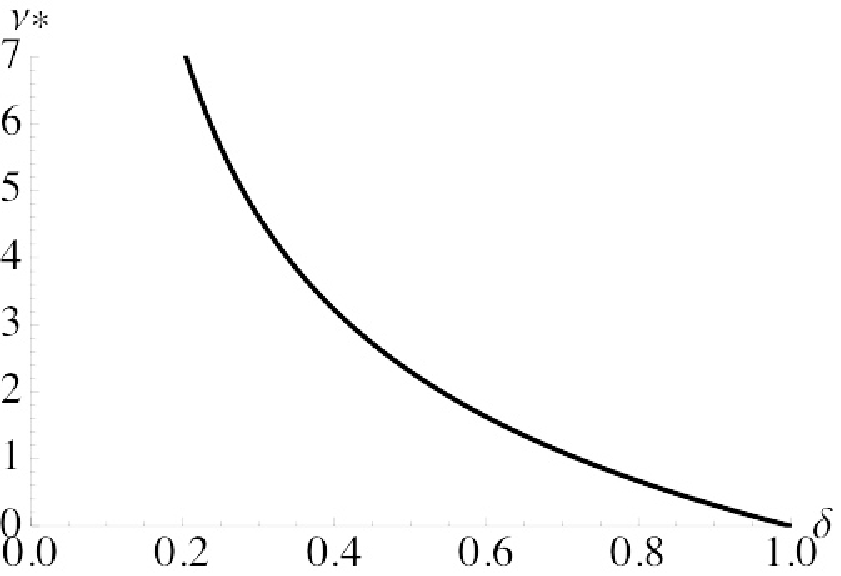}
\includegraphics[width=0.49\textwidth]{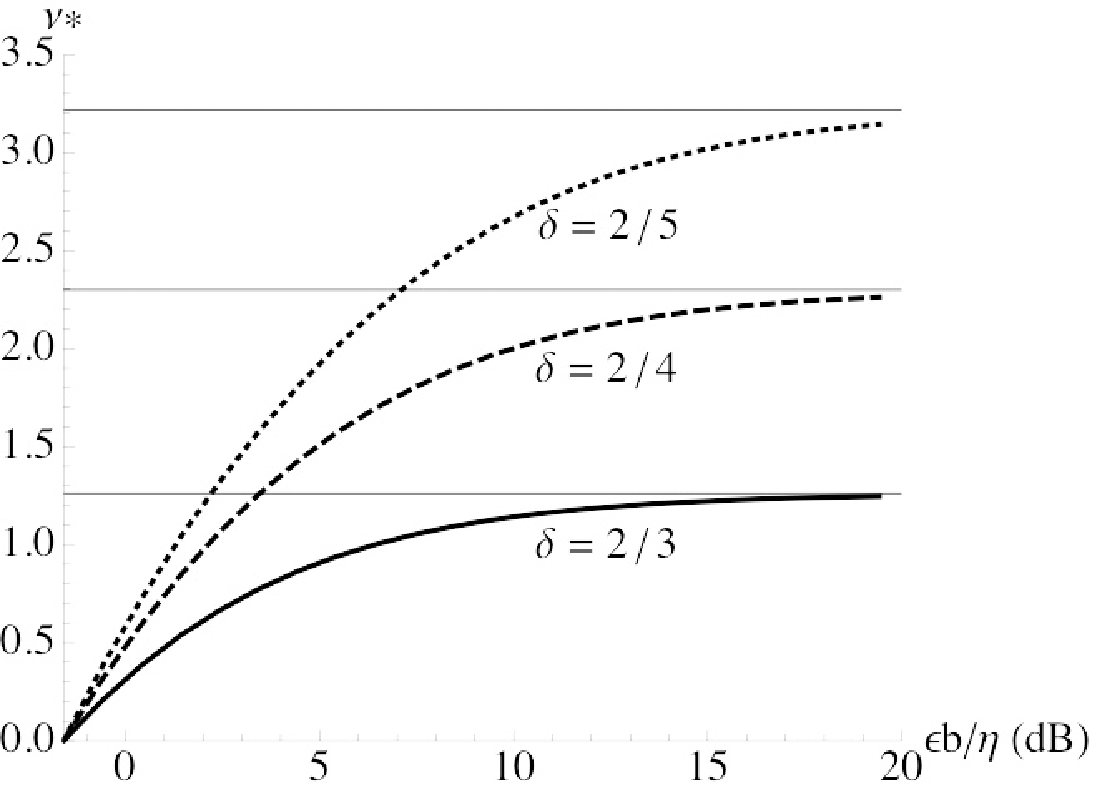}
\caption{{\bf Left:} the optimal high SNR spectral efficiency $\nu^*(\delta)$ in Cor.\ \ref{cor:specasymoptnu}.  {\bf Right:} the optimal spectral efficiency $\nu^*$ vs.\ $\ebno$ (in dB) from Fig.\ \ref{fig:specnuoptebno} (left) along with the high SNR asymptotes.}
\label{fig:FigSpecOptNuAsyHighSNR}
\end{figure}
\begin{proposition}
\label{pro:specasymnulowsnr}
The {\bf maximum possible spectral efficiency} $\nu_{\rm max}$ (obeying \eqref{eq:numax}) in the low SNR regime ($\ebno \to \log 2$) to first order in $(\ebno - \log 2)$ and second order in $\nu$ is
\begin{equation}
\label{eq:specnumaxhighsnr}
\nu_{\rm max} = \frac{2 (\ebno - \log 2)}{(\log 2)(2 \ebno - \log 2)} + \Omc(\nu^2) + \Omc \left(\ebno - \log 2 \right)^2.
\end{equation}
The optimal spectral efficiency (obeying \eqref{eq:specoptnu}) in the low SNR regime ($\ebno \to \log 2$) to first order in $(\ebno - \log 2)$ and second order in $\nu$ is
\begin{equation}
\label{eq:specnuoptighsnr}
\nu^* = \frac{2(1-\delta)(\ebno - \log 2)}{(\log 2)((\log 2) - (1 - 2\delta)(\ebno-\log 2))} + \Omc(\nu^2) + \Omc \left(\ebno - \log 2 \right)^2.
\end{equation}
\end{proposition}
\begin{proof}
The series expansion of $\log(1+ \ebno \nu) - \nu$ (viewed as a function of $(\ebno,\nu)$) around the point $(\log 2,0)$ to first order in $(\ebno - \log 2)$ and to second order in $\nu$ is:
\begin{eqnarray}
\log \left(1+ \ebno \nu \right) - \nu &=& \left(\frac{-\log 2}{2} \nu^2 + \Omc(\nu^3)\right) \nonumber \\
& & + \left( \frac{\nu}{\log 2} - \nu^2 + \Omc(\nu^3) \right)\left(\ebno - \log 2 \right) \nonumber \\
& & + \Omc \left(\ebno - \log 2 \right)^2.
\end{eqnarray}
Cancelling the common factor of $\nu$ from the first two terms, equating with zero, and solving for $\nu$ gives \eqref{eq:specnumaxhighsnr}.  Note the second order expansion for $\nu$ is required to be able to solve for $\nu$ --- a first order expansion in $\nu$ is inadequate in this regard.  The series expansion of \eqref{eq:specoptnu} to first order in $\ebno$ and third order in $\nu$ around $(\ebno,\nu) = (\log 2, 0)$ is
\begin{eqnarray}
\mbox{\eqref{eq:specoptnu}} &=& \left( -\frac{1}{2}(\log 2)^3 \nu^3 + \Omc(\nu^4) \right) \nonumber \\
& & + \left( (\log 2) (1-\delta) \nu^2 + \frac{1}{2}(\log 2)^2(1-2 \delta) \nu^3 + \Omc(\nu^4) \right) \times \nonumber \\
& & \left(\ebno - \log 2 \right)  + \Omc \left( \ebno - \log 2 \right)^2.
\end{eqnarray}
Cancelling the common factor of $\nu^2$ form the first two terms, equating with zero, and solving for $\nu$ gives \eqref{eq:specnuoptighsnr}.  Again, the third order expansion for $\nu$ is required to be able to solve for $\nu$.
\end{proof}
Fig.\ \ref{fig:FigSpecOptNuAsyLowSNR} (left) shows the exact value of $\nu_{\rm max}$ (obeying \eqref{eq:numax}) and the low SNR approximation from \eqref{eq:specnumaxhighsnr} vs.\ $\ebno$ (in dB).   Similarly, Fig.\ \ref{fig:FigSpecOptNuAsyLowSNR} (right) shows the exact value of $\nu^*$ (obeying \eqref{eq:specoptnu}) and the low SNR approximation from \eqref{eq:specnuoptighsnr} for $\delta \in \{2/3,2/4,2/5\}$.  Note the validity of both the max and optimal approximations quickly degrades as $\ebno$ increases, and that the optimal approximation is more accurate for high $\delta$ than for small $\delta$.
\begin{figure}[!htbp]
\centering
\includegraphics[width=0.49\textwidth]{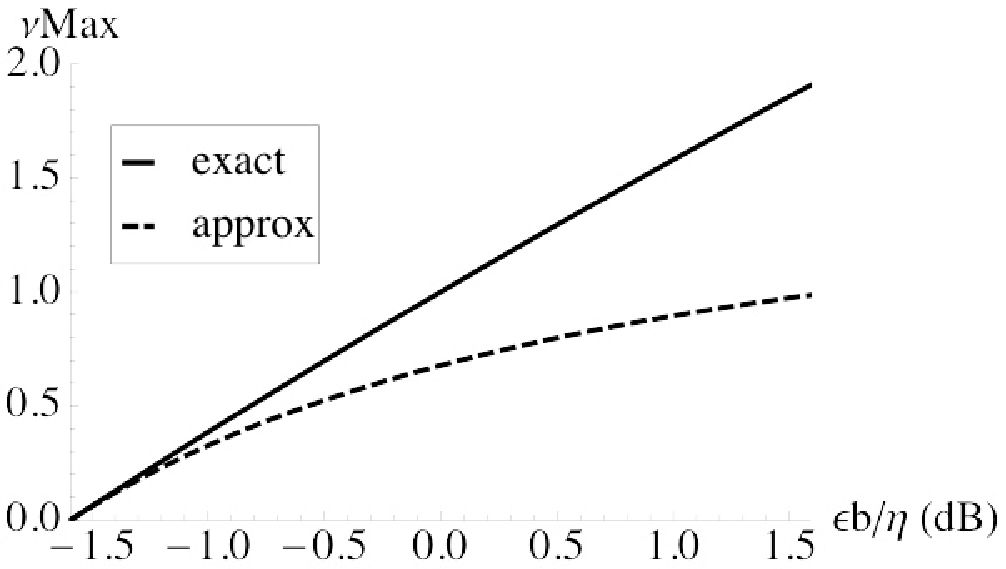}
\includegraphics[width=0.49\textwidth]{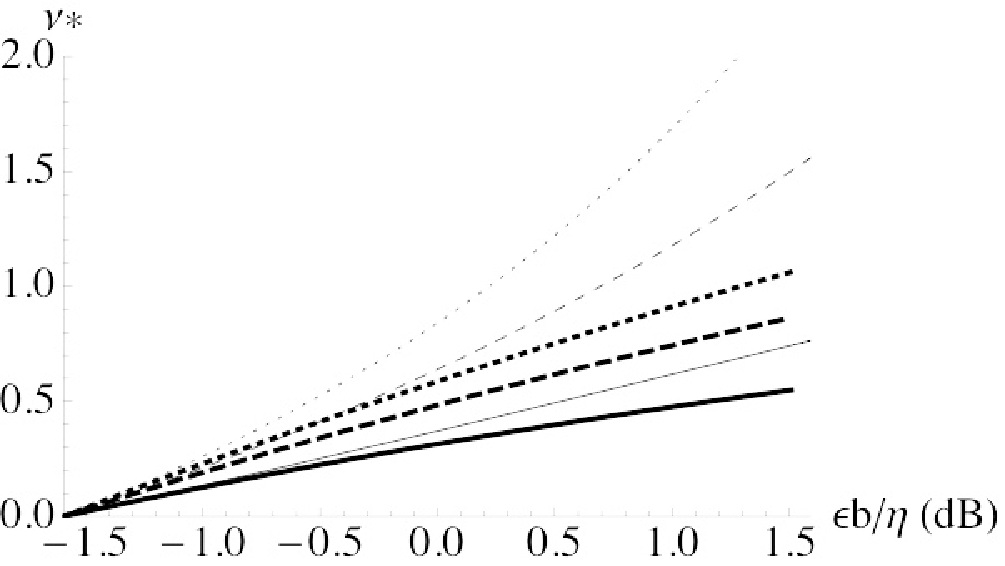}
\caption{Exact and approximate values for $\nu_{\rm max}$ (left) and $\nu^*$ (right) vs.\ $\ebno$ (in dB) using the low SNR approximations in Prop.\ \ref{pro:specasymnulowsnr}.  The $\nu^*$ plot shows exact (thick) and approximate (thin) values for $\delta$ of $2/3$ (solid), $2/4$ (dashed), $2/5$ (dotted).}
\label{fig:FigSpecOptNuAsyLowSNR}
\end{figure}

\section{Interference cancellation (IC)}
\label{sec:intcan}

In this section we study the impact of interference cancellation (IC) on the OP and TC \cite{WebAnd2007a}.  There are (at least) three major constraints in designing a Rx capable of cancellation.
\begin{definition}
\label{def:icparam}
{\bf The $(\kappa,K,P_{\rm min})$ IC model} is defined as follows.
\begin{enumerate}
\item Cancellation is inevitably imperfect, due to imperfect channel estimation and signal reconstruction.  To model this, we multiply interference from cancellable nodes by $\kappa \in [0,1]$.  Note $\kappa = 0$ is perfect cancellation and $\kappa = 1$ is no cancellation.
\item Processing delays and computational complexity constraints limit the number of cancellable interferers.  Therefore, we cancel at most the nearest $K \in \Nbb$ interferers.
\item Received power limits the decodability of an interferer, \ie, it is difficult to accurately estimate and subtract low SNR interferers, due for example to finite resolution receivers.  Thus, we only cancel interferers with a received power above $P_{\rm min} \in \Rbb_+$.
\end{enumerate}
\end{definition}
These constraints may be interdependent, \eg, there may be natural design tradeoff in $(\kappa,K)$ in that high quality cancellation (small $\kappa$) may require larger delays and thereby UB feasible $K$.  We make the following assumptions throughout this section.  Recall Ass.\ \ref{ass:pppdistorder} in \S\ref{sec:pppvd}.  We retain Ass.\ \ref{ass:snp} in \S\ref{sec:snp} and in particular assume $\epsilon = 0$.  We next define the SINR seen at $o$ under the $(\kappa,K,P_{\rm min})$ IC model.
\begin{definition}
\label{def:sicsinr}
The {\bf SINR at a $(\kappa,K,P_{\rm min})$ IC capable reference Rx} located at $o$ is as in Def.\ \ref{def:sinrnf}, except interference is split into partially cancelled and uncancelled components:
\begin{equation}
\sinr(o) \equiv \frac{P u^{-\alpha}}{\kappa P \Sisf^{\rm pc}(o) + P \Sisf^{\rm uc}(o) + N},
\end{equation}
where the partially cancelled / uncancelled interference components are
\begin{equation}
\Sisf^{\rm pc}(o) \equiv \sum_{i \in \Pi^{\rm pc}_{d,\lambda}(o)} |\xsf_i|^{-\alpha}, ~~~ \Sisf^{\rm uc}(o) \equiv \sum_{i \in \Pi^{\rm uc}_{d,\lambda}(o)} |\xsf_i|^{-\alpha}.
\end{equation}
Here, the set of interferers $\Pi_{d,\lambda} = \{\xsf_i\}$ is partitioned into the sets of partially cancelled and uncancelled interferers based on $(K,P_{\rm min})$:
\begin{eqnarray}
\Pi^{\rm pc}_{d,\lambda}(o) &\equiv& \left\{ i \in \Pi_{d,\lambda} : i \in [K] \mbox{ and } P |\xsf_i|^{-\alpha} > P_{\rm min} \right\} \nonumber \\
\Pi^{\rm uc}_{d,\lambda}(o) &\equiv& \Pi_{d,\lambda} \setminus \Pi^{\rm pc}_{d,\lambda}(o).
\end{eqnarray}
\end{definition}
Note the $i \in [K]$ requirement and recall the labeling convention in Ass.\ \ref{ass:pppdistorder}.  The following lemma translates the received power constraint ($P_{\rm min}$) and cancellation quantity constraint ($K$) for interferers in $\Pi_{d,\lambda}$ (from Def.\ \ref{def:icparam}) to a maximum distance from $o$ under $\Pi_{1,1}$, and the dominant vs.\ non-dominant criterion in $\Pi_{d,\lambda}$ (from Def.\ \ref{def:domint}) to a maximum distance from $o$ under $\Pi_{1,1}$.
\begin{lemma}
\label{lem:sictmax}
{\bf Constraint mapping.}
Satisfying both the received power constraint $P |\xsf_i|^{-\alpha} > P_{\rm min}$ and the cancellation quantity constraint $i \in [K]$ for $i \in \Pi_{d,\lambda}$ is equivalent to $|\tsf_i| < \tsf^{\rm pc}$ for $i \in \Pi_{1,1}$ where
\begin{equation}
\label{eq:sictmax1}
\tsf^{\rm pc} \equiv |\tsf_K| \land \frac{\lambda c_d}{2} \left( \frac{P}{P_{\rm min}} \right)^{\delta}.
\end{equation}
The requirement for a partially cancellable interferer $i \in \Pi^{\rm pc}_{d,\lambda}(o)$ to be dominant as in Def.\ \ref{def:domint} is equivalent to $|\tsf_i| < t^{\rm pc}_{\rm dom}$ for $i \in \Pi_{1,1}$ where
\begin{equation}
\label{eq:sictmax2}
t^{\rm pc}_{\rm dom} \equiv \frac{\lambda c_d}{2} \kappa^{\delta} \xi^d,
\end{equation}
for $\xi$ in Def.\ \ref{def:xisnr}.  The requirement for an uncancellable interferer $i \in \Pi^{\rm uc}_{d,\lambda}(o)$ to be dominant is equivalent to $|\tsf_i| < t^{\rm uc}_{\rm dom}$ for $i \in \Pi_{1,1}$ where
\begin{equation}
\label{eq:sictmax3}
t^{\rm uc}_{\rm dom} \equiv \frac{\lambda c_d}{2} \xi^d.
\end{equation}
\end{lemma}
\begin{proof}
All three quantities follow from Prop.\ \ref{pro:distmap}.  First, \eqref{eq:sictmax1}:
\begin{equation}
P |\xsf_i|^{-\alpha} > P_{\rm min} \Leftrightarrow |\xsf_i|^d < \left( \frac{P}{P_{\rm min}} \right)^{\delta} \Leftrightarrow |\tsf_i| < \frac{\lambda c_d}{2} \left( \frac{P}{P_{\rm min}} \right)^{\delta}.
\end{equation}
Take the minimum with $|\tsf_K|$ to also ensure $i \in [K]$.  Next, \eqref{eq:sictmax2}:
\begin{equation}
\frac{P u^{-\alpha}}{\kappa P |\xsf_i|^{-\alpha} + N} > \tau \Leftrightarrow |\xsf_i|^d < \kappa^{\delta} \xi^d \Leftrightarrow |\tsf_i| < \frac{\lambda c_d}{2} \kappa^{\delta} \xi^d.
\end{equation}
The constraint in \eqref{eq:sictmax3} is obtained by an identical development but without the $\kappa$ on the interference term.
\end{proof}
Note $\tsf^{\rm pc}$ is a function of the RV $|\tsf_K|$ and is therefore a RV.  The next lemma applies Lem.\  \ref{lem:sictmax} to define the sets of dominant and non-dominant nodes (again in the sense of Def.\ \ref{def:domint}) among the sets of cancellable and uncancellable nodes using thresholds $\tsf^{\rm pc},t^{\rm pc}_{\rm dom},t^{\rm uc}_{\rm dom}$.
\begin{lemma}
\label{lem:sicdomcan}
The set of {\bf partially cancellable / uncancellable nodes} in $\Pi_{1,1}$ is:
\begin{equation}
\label{eq:sicdomcan1}
\Pi^{\rm pc}_{1,1}(o) = \{ i \in \Pi_{1,1} : |\tsf_i| < \tsf^{\rm pc} \}, ~ \Pi^{\rm uc}_{1,1}(o) = \Pi_{1,1} \setminus \Pi^{\rm pc}_{1,1}(o),
\end{equation}
for $t_{\rm max}^{\rm pc}$ in Lem.\  \ref{lem:sictmax}.  The set of dominant and non-dominant partially cancellable and uncancellable nodes is:
\begin{eqnarray}
\label{eq:sicdomcan2}
\hat{\Pi}^{\rm pc}_{1,1}(o) &=& \left\{ i \in \Pi_{1,1} : |\tsf_i| < \min\left\{ \tsf^{\rm pc}, t^{\rm pc}_{\rm dom} \right\} \right\} \nonumber \\
\tilde{\Pi}^{\rm pc}_{1,1}(o) &=&  \left\{ i \in \Pi_{1,1} : t^{\rm pc}_{\rm dom} < |\tsf_i| < \tsf^{\rm pc} \right\} \nonumber \\
\hat{\Pi}^{\rm uc}_{1,1}(o) &=& \left\{ i \in \Pi_{1,1} : \tsf^{\rm pc}  < |\tsf_i| < t^{\rm uc}_{\rm dom}  \right\} \nonumber \\
\tilde{\Pi}^{\rm uc}_{1,1}(o) &=&  \left\{ i \in \Pi_{1,1} : \max \left\{ \tsf^{\rm pc}, t^{\rm uc}_{\rm dom} \right\} < |\tsf_i| \right\}
\end{eqnarray}
for $t_{\rm max}^{\rm pc,dom},t_{\rm max}^{\rm uc,dom}$ in Lem.\  \ref{lem:sictmax}.  The aggregate dominant and non-dominant partially cancellable and uncancellable interference is:
\begin{eqnarray}
\hat{\Sisf}^{{\rm pc},1/\delta,0}_{1,1}(o) &=& \sum_{i \in \hat{\Pi}^{\rm pc}_{1,1}(o)} |\tsf_i|^{-\frac{1}{\delta}} \nonumber \\
\tilde{\Sisf}^{{\rm pc},1/\delta,0}_{1,1}(o) &=& \sum_{i \in \tilde{\Pi}^{\rm pc}_{1,1}(o)} |\tsf_i|^{-\frac{1}{\delta}} \nonumber \\
\hat{\Sisf}^{{\rm uc},1/\delta,0}_{1,1}(o) &=& \sum_{i \in \hat{\Pi}^{\rm uc}_{1,1}(o)} |\tsf_i|^{-\frac{1}{\delta}} \nonumber \\
\tilde{\Sisf}^{{\rm uc},1/\delta,0}_{1,1}(o) &=& \sum_{i \in \tilde{\Pi}^{\rm uc}_{1,1}(o)} |\tsf_i|^{-\frac{1}{\delta}}
\end{eqnarray}
\end{lemma}
\begin{proof}
The sets in \eqref{eq:sicdomcan1} and \eqref{eq:sicdomcan2} are immediate from Lem.\  \ref{lem:sictmax}.  The aggregate interferences are simply the summed interference contributions under the corresponding interferer sets, using the SN RV mapping in Prop.\ \ref{pro:imap}.
\end{proof}
These sets are illustrated in Fig.\ \ref{fig:sicthresh}.
\begin{figure}[!htbp]
\centering
\includegraphics[width=0.75\textwidth]{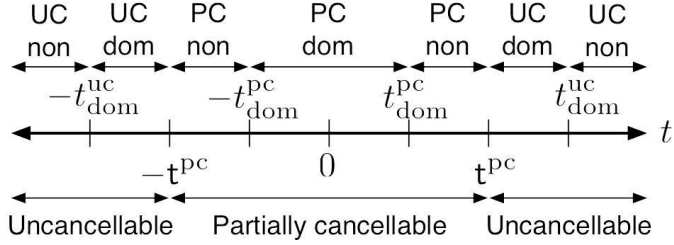}
\caption{The thresholds $\tsf^{\rm pc}$ partitions $\Rbb$ into the region of partially cancellable (PC) nodes and uncancellable (UC) nodes. The thresholds $\tsf^{\rm pc}_{\rm dom},\tsf^{\rm uc}_{\rm dom}$ further divide each of these regions into dominant and non-dominant regions.}
\label{fig:sicthresh}
\end{figure}
We see that the RV $\tsf^{\rm pc}$ is a function of $|\tsf_K|$, the distance of the $K$th nearest Rx from $o$ in $\Pi_{1,1}$.  The next two results characterize this RV.
The following theorem (from \cite{Hae2005}, Thm.\  1) extends Prop.\ \ref{pro:void}.
\begin{theorem}
\label{thm:eucdistnn}
{\bf Ordered distances marginal distributions} (\cite{Hae2005}).
The successive (ordered) random distances from $o$ of points in $\Pi_{d,\lambda} = \{\xsf_k\}$ with $|\xsf_1| < |\xsf_2| < \cdots$ have marginal (generalized Gamma) PDFs:
\begin{equation}
f_{|\xsf_k|}(t) = \frac{d (\lambda c_d t^d)^k}{t (k-1)!} \erm^{-\lambda c_d t^d}, ~ t \in \Rbb_+, ~ k \in \Nbb.
\end{equation}
\end{theorem}
By the mapping theorem we will have need only for $d=1$ and $\lambda = 1$.
\begin{corollary}
\label{cor:distkpi11nn}
{\bf Ordered distances in $\Pi_{1,1}$ marginal distributions.}
The successive (ordered) random distances from $o$ of points in $\Pi_{1,1} = \{\tsf_k\}$ with $|\tsf_1| < |\tsf_2| < \cdots$ have marginal (Gamma) PDFs:
\begin{equation}
f_{|\tsf_k|}(t) = \frac{(2 t)^k}{t (k-1)!} \erm^{-2t}, ~ t \in \Rbb_+, k \in \Nbb.
\end{equation}
Moreover,
\begin{equation}
\Ebb[|\tsf_k|^p] = \left\{ \begin{array}{ll}
\frac{\Gamma(k+p)}{2^p (k-1)!}, \; & p > -k \\
\infty, \; & \mbox{else}
\end{array} \right.
\end{equation}
In particular $\Ebb[|\tsf_k|] = \frac{k}{2}$.
\end{corollary}
The PDF is immediate from Thm.\ \ref{thm:eucdistnn}.  The moments are obtained by calculus.  The PDFs for $k \in \{1,\ldots,5\}$ are shown in Fig.\ \ref{fig:gammadist}.
\begin{figure}[!htbp]
\centering
\includegraphics[width=0.75\textwidth]{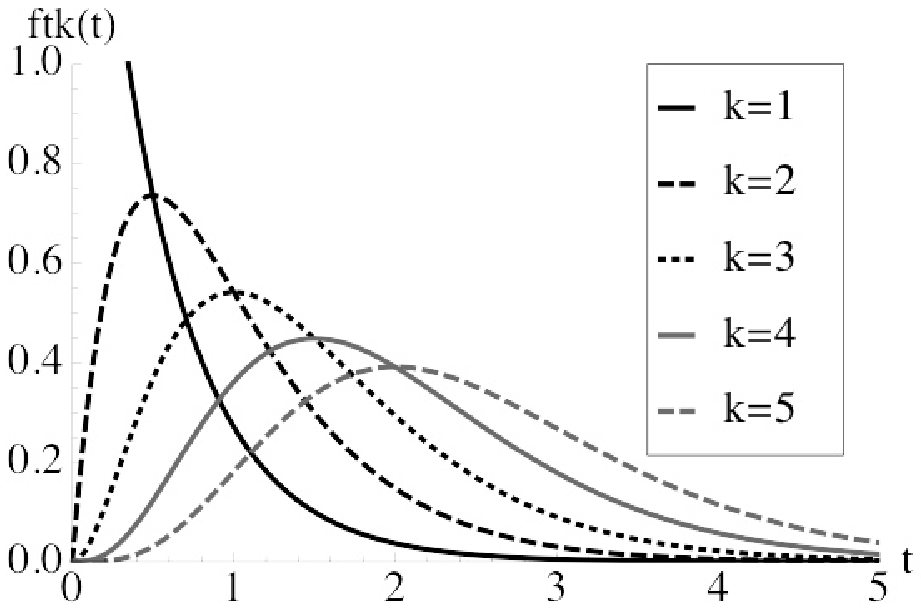}
\caption{The PDFs $f_{|\tsf_k|}(t)$ for $k \in \{1,\ldots,5\}$ from Cor.\ \ref{cor:distkpi11nn}.  Note $\Ebb[|\tsf_k|] = \frac{k}{2}$.}
\label{fig:gammadist}
\end{figure}
The main result of this section is the following LB on the OP.
\begin{proposition}
\label{pro:sicmain}
The {\bf OP LB under $(\kappa,K,P_{\rm min})$ IC Rx model} is:
\begin{equation}
\label{eq:sicmain}
q^{\rm lb}(\lambda) = 1 - \Ebb \left[ \exp \left\{ - 2 \min\left\{ \tsf^{\rm pc}, t^{\rm pc}_{\rm dom} \right\} \right\} \exp \left\{ - 2  (t^{\rm uc}_{\rm dom} - \tsf^{\rm pc})^+ \right\} \right]
\end{equation}
where $x^+ \equiv \max\{x,0\}$, and $\tsf^{\rm pc}, t^{\rm pc}_{\rm dom},t^{\rm uc}_{\rm dom}$ are given in Lem.\  \ref{lem:sictmax}, $\tsf^{\rm pc}$ is a RV that is a function of the RV $|\tsf_K|$ (from Prop.\ \ref{cor:distkpi11nn}), and the expectation is with respect to $|\tsf_K|$.
\end{proposition}
\begin{proof}
We first apply the SINR definition from Def.\ \ref{def:sicsinr} and solve for the interference:
\begin{eqnarray}
q(\lambda) &=& \Pbb(\sinr(o) < \tau) \nonumber \\
&=& \Pbb \left( \frac{P u^{-\alpha}}{P \Sisf(o) + N} < \tau \right) \nonumber \\
&=& \Pbb \left( \Sisf(o) > \frac{1}{\tau u^{\alpha}} - \frac{N}{P} \right) \nonumber \\
&=& \Pbb \left( \kappa \Sisf^{{\rm pc},\alpha,0}_{d,\lambda}(o) + \Sisf^{{\rm uc},\alpha,0}_{d,\lambda}(o) > \xi^{-\alpha} \right)
\end{eqnarray}
for $\xi$ in Def.\ \ref{def:xisnr}.  We apply the law of total probability by conditioning on $|\tsf_K|$ which effectively partitions $\Rbb$ into the regions in Fig.\ \ref{fig:sicthresh}.  With those regions fixed for each possible value of $|\tsf_K|$ we achieve a LB by only considering the dominant interferers (both partially cancellable and uncancellable).
\begin{eqnarray}
q(\lambda) &=& \Ebb \left[ \Pbb \left( \left. \kappa \Sisf^{{\rm pc},\alpha,0}_{d,\lambda}(o) + \Sisf^{{\rm uc},\alpha,0}_{d,\lambda}(o) > \xi^{-\alpha} \right| |\tsf_K| \right) \right] \nonumber \\
&>& \Ebb \left[ \Pbb \left( \left. \kappa \hat{\Sisf}^{{\rm pc},\alpha,0}_{d,\lambda}(o) + \hat{\Sisf}^{{\rm uc},\alpha,0}_{d,\lambda}(o) > \xi^{-\alpha} \right| |\tsf_K| \right) \right]
\end{eqnarray}
We then take the complement of the outage event, yielding the event that the sum dominant partially cancellable and dominant uncancellable interference is less than $\xi^{-\alpha}$.  By construction this event holds iff there are no dominant interferers of any type.
\begin{eqnarray}
q^{\rm lb}(\lambda)
&=& 1 - \Ebb \left[ \Pbb \left( \left. \kappa \hat{\Sisf}^{{\rm pc},\alpha,0}_{d,\lambda}(o) + \hat{\Sisf}^{{\rm uc},\alpha,0}_{d,\lambda}(o) \leq \xi^{-\alpha} \right| |\tsf_K| \right) \right] \nonumber \\
&=& 1 - \Ebb \left[ \Pbb \left( \left. \hat{\Pi}^{\rm pc}_{1,1}(o) = \emptyset \cap \hat{\Pi}^{\rm uc}_{1,1}(o) = \emptyset \right| |\tsf_K| \right) \right]
\end{eqnarray}
Because the dominant partially cancellable interferers and dominant uncancellable interferers are defined on disjoint regions of $\Rbb$, it follows by independence property of the PPP that their occupancies are independent RVs.
\begin{equation}
q^{\rm lb}(\lambda) = 1 - \Ebb \left[ \Pbb \left( \left. \hat{\Pi}^{\rm pc}_{1,1}(o) = \emptyset \right| |\tsf_K| \right) \Pbb \left( \left. \hat{\Pi}^{\rm uc}_{1,1}(o) = \emptyset \right| |\tsf_K| \right) \right]
\end{equation}
Prop.\ \ref{pro:void} with $d=\lambda=1$ applied to the two dominant regions gives \eqref{eq:sicmain}.
\end{proof}
\begin{figure}[!htbp]
\centering
\includegraphics[width=0.49\textwidth]{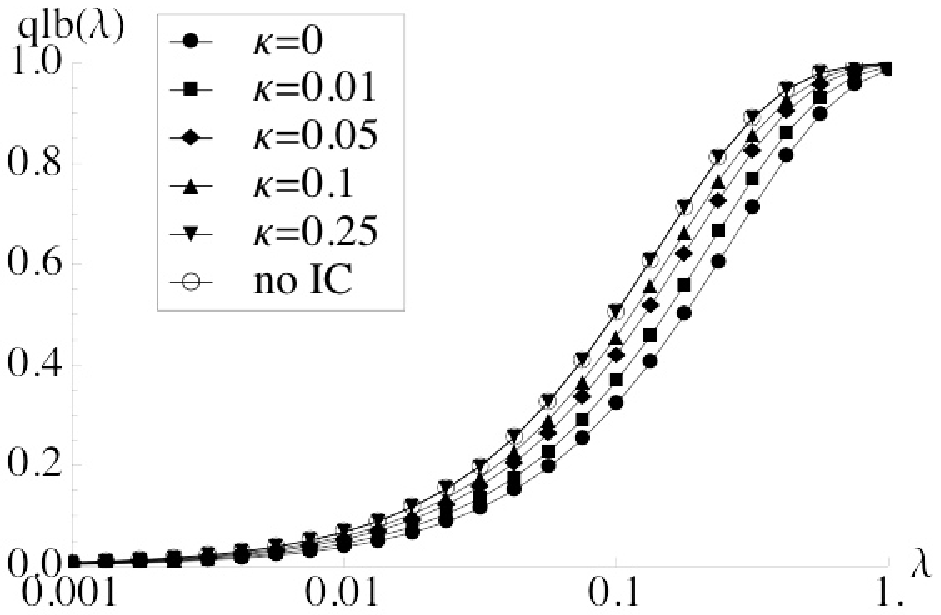}
\includegraphics[width=0.49\textwidth]{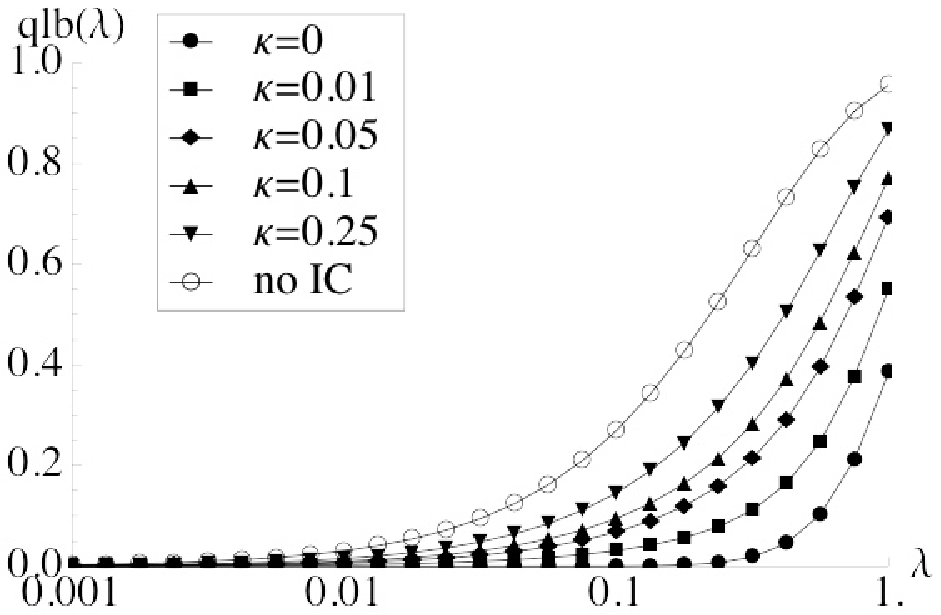}
\includegraphics[width=0.49\textwidth]{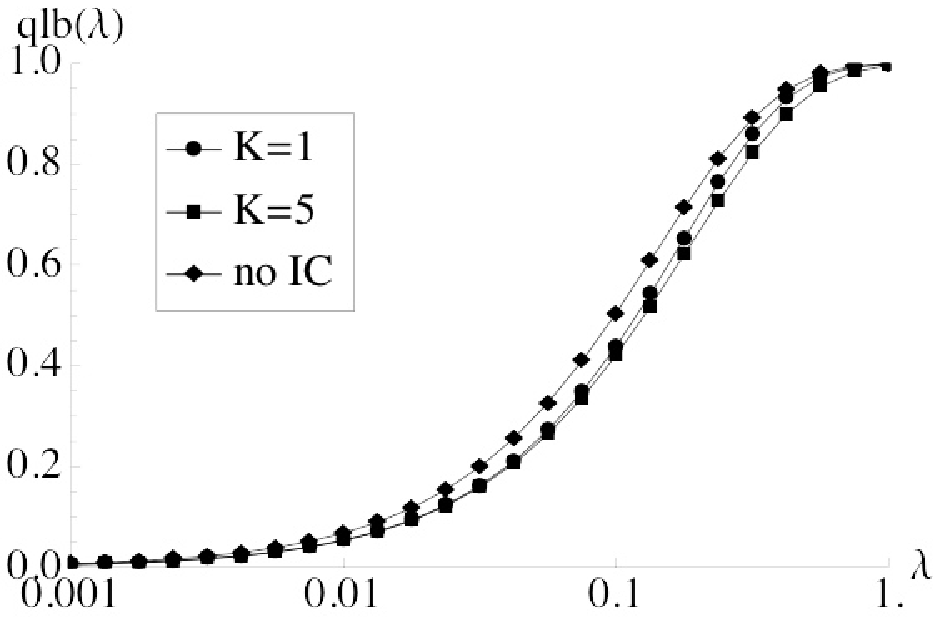}
\includegraphics[width=0.49\textwidth]{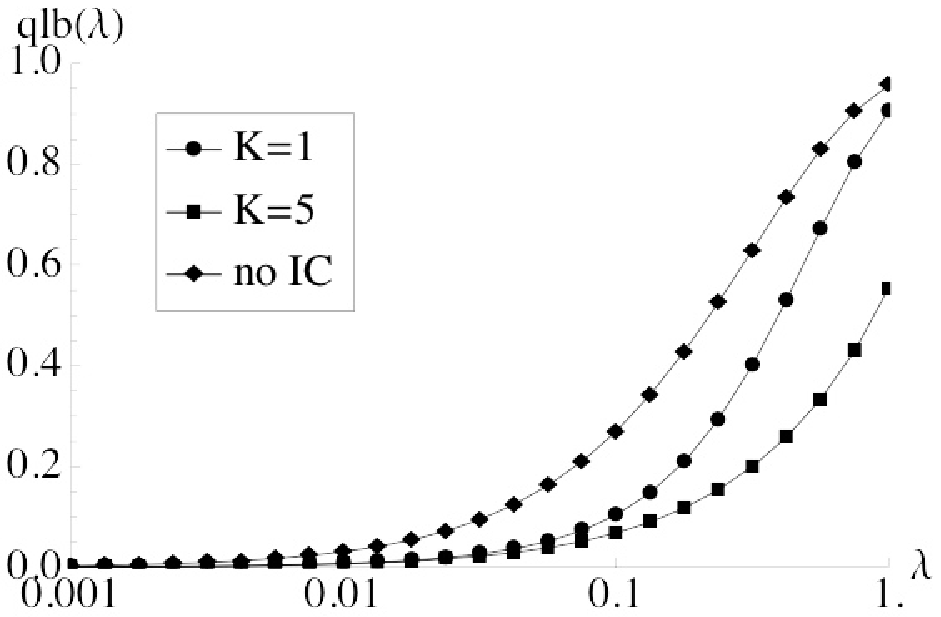}
\includegraphics[width=0.49\textwidth]{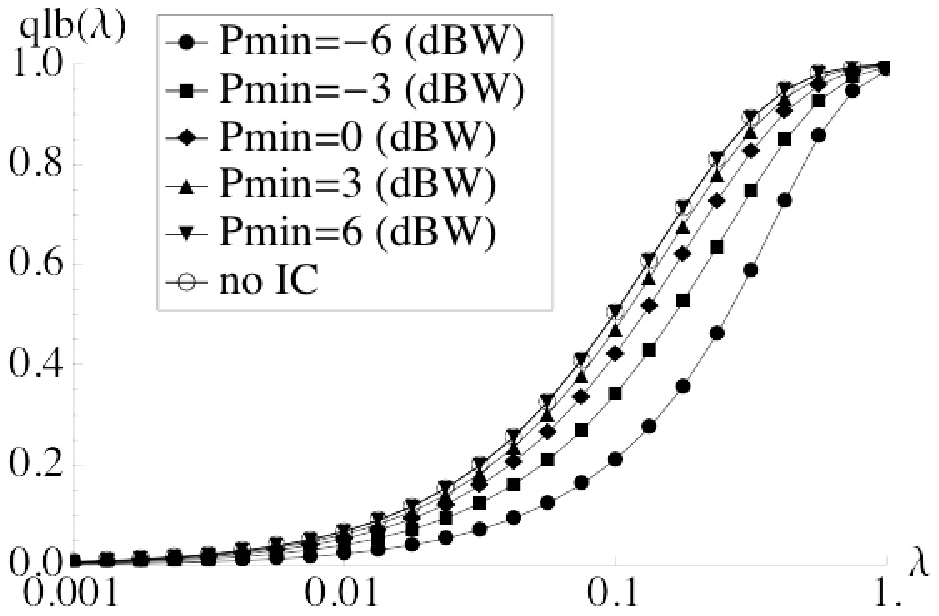}
\includegraphics[width=0.49\textwidth]{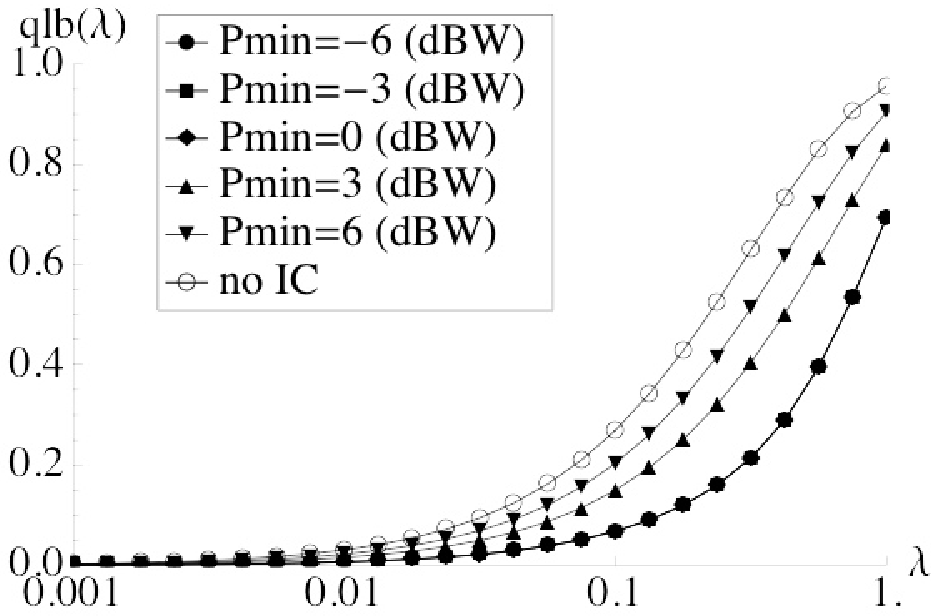}
\caption{The OP LB $q^{\rm lb}(\lambda)$ from Prop.\ \ref{pro:sicmain} for the $(\kappa,K,P_{\rm min})$ IC Rx model for $\tau = 5$ (left) and $\tau = 1$ (right).  The OP is shown vs.\ $\lambda$ for varying $\kappa$ (top), $K$ (middle), and $P_{\rm min}$ (bottom).  Default parameters are in \eqref{eq:sicdefparam}.  In all cases we also show the OP LB $q^{\rm lb}(\lambda)$ from Prop.\ \ref{pro:oplb} for no IC for the same parameters.}
\label{fig:sicoplb1}
\end{figure}
\begin{figure}[!htbp]
\centering
\includegraphics[width=0.49\textwidth]{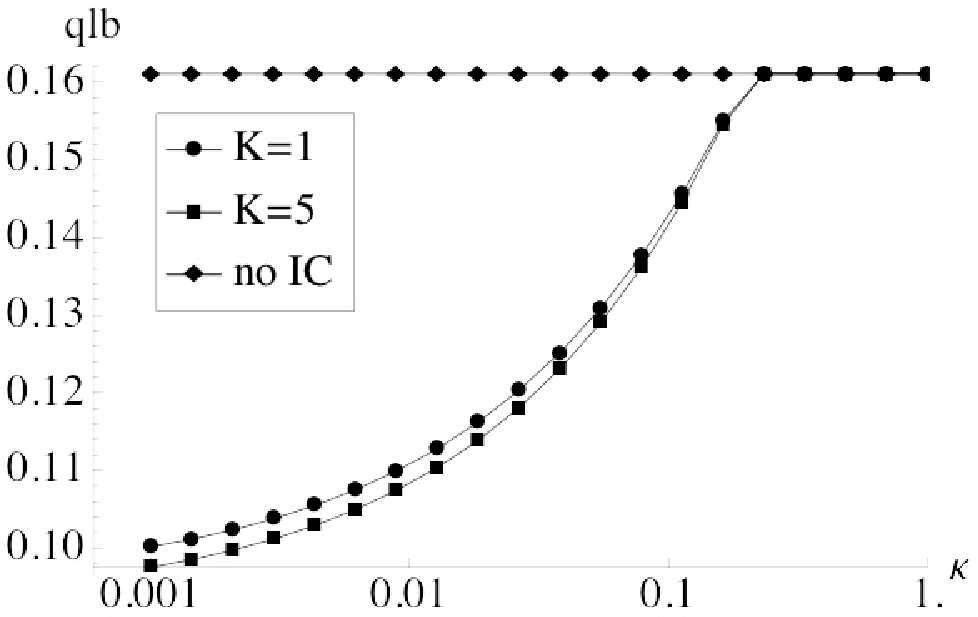}
\includegraphics[width=0.49\textwidth]{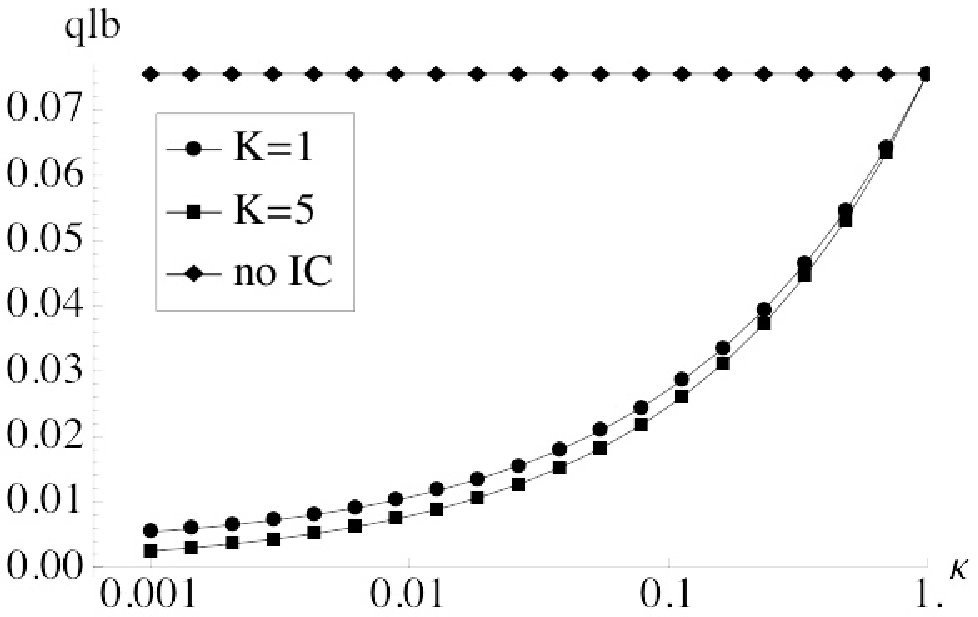}
\includegraphics[width=0.49\textwidth]{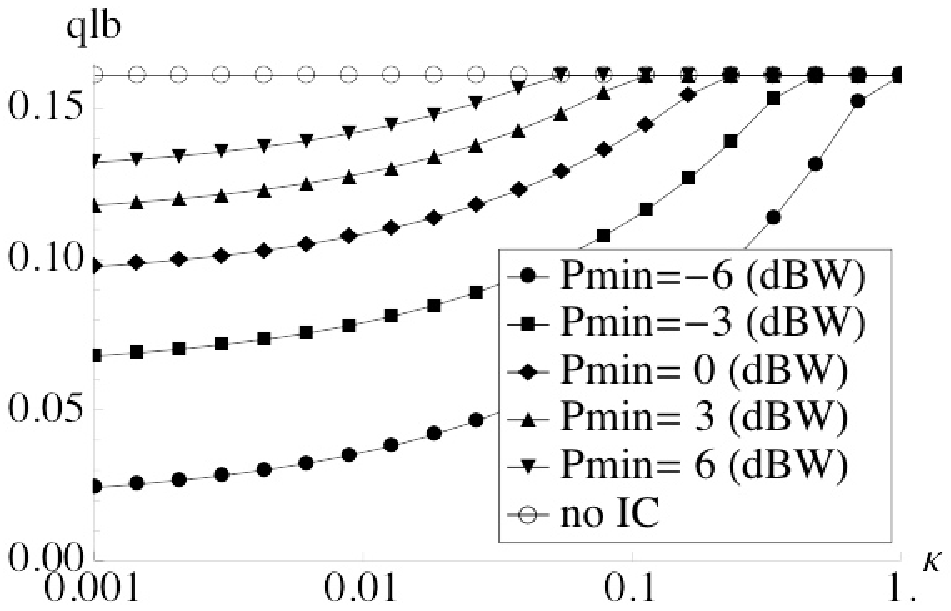}
\includegraphics[width=0.49\textwidth]{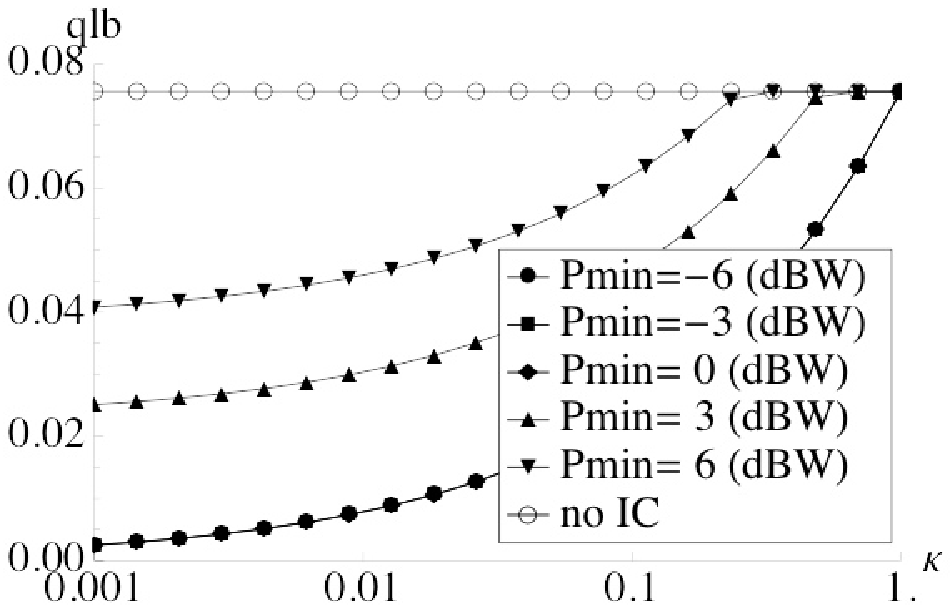}
\includegraphics[width=0.49\textwidth]{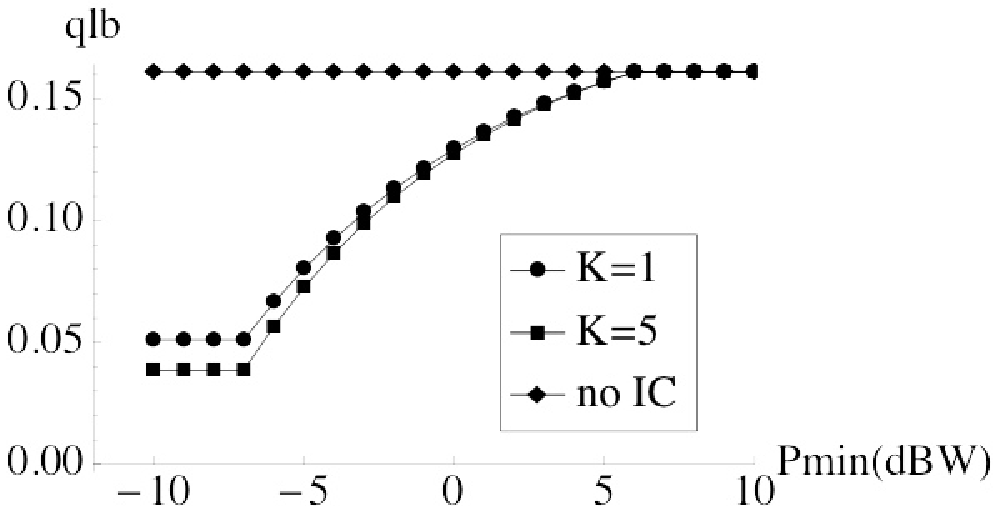}
\includegraphics[width=0.49\textwidth]{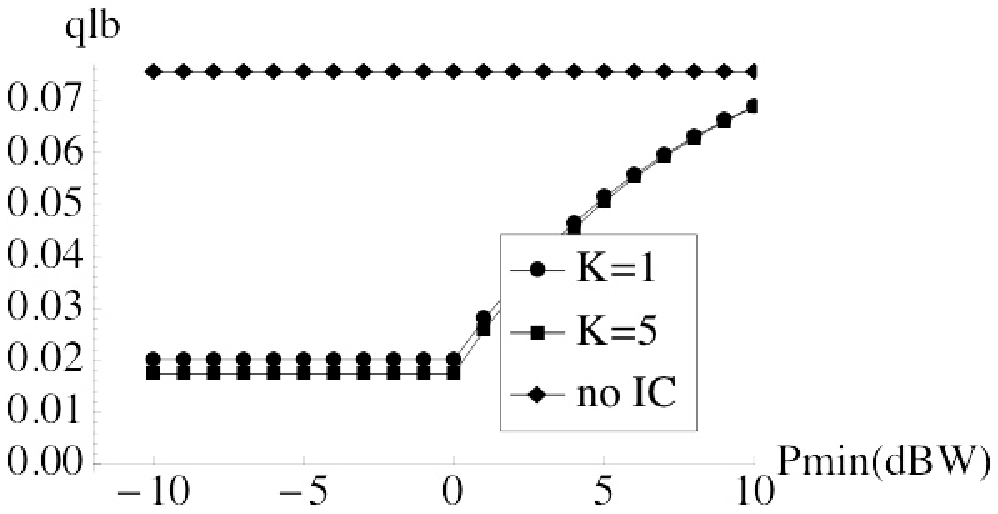}
\caption{The OP LB $q^{\rm lb}(\lambda)$ from Prop.\ \ref{pro:sicmain} for the $(\kappa,K,P_{\rm min})$ IC Rx model for $\tau = 5$ (left) and $\tau = 1$ (right).  The OP is shown vs.\ $\kappa$ for varying $K$ (top), vs.\ $\kappa$ for varying $P_{\rm min}$ (middle), and vs.\ $P_{\rm min}$ for varying $K$ (bottom).  Default parameters are in \eqref{eq:sicdefparam}.  In all cases we also show the OP LB $q^{\rm lb}(\lambda)$ from Prop.\ \ref{pro:oplb} for no IC for the same parameters.}
\label{fig:sicoplb2}
\end{figure}
Fig.\ \ref{fig:sicoplb1} and \ref{fig:sicoplb2} plot the OP LB $q^{\rm lb}(\lambda)$ from Prop.\ \ref{pro:sicmain} for various choices of the $(\kappa,K,P_{\rm min})$ SIC Rx model.  The default parameters are
\begin{equation}
\label{eq:sicdefparam}
\begin{array}{cccc}
\tau \in \{1,5\} & u = 1 & \alpha = 4 & d = 2 ~  P = 1 ~ N = 0 \\
\lambda = 0.025 & \kappa = 0.05 & K = 3 & P_{\rm min} = 1 = 0 \mbox{(dBW)}
\end{array}
\end{equation}
Several points merit comment.  First, note that in all cases Fig.\ \ref{fig:sicoplb1} shows the corresponding performance without SIC using the OP LB $q^{\rm lb}(\lambda)$ from Prop.\ \ref{pro:oplb} (with $\epsilon = 0$).  Recall that this bound was shown to be tight.  Second, in most (but not in all) cases the plots show greater sensitivity of the OP LB to the IC parameters for $\tau = 1$ (right) than for $\tau = 5$ (left).  Third, the plots illustrate the benefit in reducing the OP from improved IC is quite parameter dependent; we now give several examples of this.  In the top plots of Fig.\ \ref{fig:sicoplb1} we see that for $K=3$ and $P_{\rm min} = 1$ ($0$ dBW) the performance of IC with $\kappa = 0.25$ or more is indistinguishable from no IC for $\tau = 5$, but this is not true for $\tau = 5$.  In the middle plots of Fig.\ \ref{fig:sicoplb1} we see the OP is insensitive to $K$ for $K \geq 1$ when $\tau =5$, and this also holds for $\tau = 1$ and $\lambda$ small, but not for $\lambda$ larger.  In the bottom plots of Fig.\ \ref{fig:sicoplb1} we see the OP for $P_{\rm min} = 6$ or greater is the same as no SIC for $\tau = 5$, but that the OP under $P_{\rm min} = -6$ is the same as for $P_{\rm min} = 0$ for $\tau = 1$ (all powers in dBW).  In the top plots of Fig.\ \ref{fig:sicoplb2} we observe that improving the cancellation quality $\kappa$ by an order of magnitude reduces the OP by only a few percentage points in the given parameter regime.  In the bottom plots of Fig.\ \ref{fig:sicoplb2} we observe again limited dependence of OP upon $K \geq 1$, while the right plot ($\tau = 1$) shows insensitivity to $P_{\rm min}$ below $0$ dBW.
\begin{remark}
\label{rem:othericbounds}
{\bf Other bounds on OP and TC under IC} are found in the literature.  We specifically discuss \cite{MorLoy2009} centered upon the OP LB for perfect ($\kappa = 0$) IC:
\begin{equation}
\Pbb \left(\sum_{i \in \Pi_{d,\lambda} \setminus [K]} |\xsf_i|^{-\alpha} > y \right) \geq
\Pbb \left(|\xsf_{K+1}|^{-\alpha} > y \right) = \Pbb\left(|\xsf_{K+1}| \leq y^{-\frac{1}{\alpha}}\right),
\end{equation}
where Ass.\ \ref{ass:pppdistorder} dictates the labeling of the points.  The CDF from Thm.\ \ref{thm:eucdistnn} for $k=K+1$ may be used to give an explicit OP LB.  Note that this bound holds also for imperfect IC ($\kappa \in (0,1)$), but is not as insightful as the bound depends upon $K$ but not on $\kappa$.  Perfect cancellation naturally leads also to an UB on OP using the Markov inequality, as done in \cite{JinAnd2011} (Thm.\ 1 and Lem.\ 2), where the essential step is below:
\begin{equation}
\Pbb \left(\sum_{i \in \Pi_{d,\lambda} \setminus [K]} |\xsf_i|^{-\alpha} > y \right) \leq \frac{1}{y} \Ebb \left[ \sum_{i = K+1}^{\infty} |\xsf_i|^{-\alpha} \right] = \frac{1}{y} \sum_{i = K+1}^{\infty} \Ebb \left[ |\xsf_i|^{-\alpha} \right].
\end{equation}
The moment $\Ebb[\xsf_i|^{-\alpha}]$ may be computed using the PDF in Thm.\ \ref{thm:eucdistnn}, and this in turn be UBed using Kershaw's inequality on the Gamma function.
\end{remark}

\section{Fading threshold scheduling (FTS)}
\label{sec:sched}

In \S\ref{sec:fading} we extended the basic model of Ch.\  \ref{cha:bm} to incorporate fading, and Cor.\ \ref{cor:optcfadnoncomp} showed that the asymptotic OP and TC (as $\lambda \to 0$ and $q^* \to 0$, respectively) under fading were worse than without fading \cite{WebAnd2007b}.  The assumed slotted Aloha MAC protocol does not exploit fading as a source of diversity, and this (partially) explains why OP and TC degrade under fading.  In this section we seek to exploit fading through the use of scheduling; in \S\ref{sec:power} we will exploit fading by selecting the transmission power.  Only certain forms of scheduling will retain the critical features of the basic model that allow for analytical tractability, namely, the PPP independence property that the number of points in disjoint regions of $\Rbb^d$ are independent RVs.  Scheduling in its usual usage refers to inducing a negative spatial and or temporal correlation for activity among adjacent nodes so as to minimize the occurence of collisions, and in this sense the set of active nodes under scheduling will not form a PPP.  Instead, our usage is restricted to ``fading threshold scheduling'' (FTS): a MAC protocol where potential transmitting nodes decide to transmit when the fading coefficient of the channel to their intended Rx exceeds a threshold.  As we assume independent channels, it follows that the transmission decisions of two nodes are independent.  The following definition formalizes the model for this section.
\begin{definition}
\label{def:sched}
{\bf Fading coefficients, signal interference, and SINR.}
Extend Def.\ \ref{def:fad} throughout this section as follows.  Let $\{(\hsf_{i,i},\hsf_{i,0})\}$ for $i \in \{0,1,2,\ldots,\}$ be iid RVs indicating the fading coefficient for the channels between each Tx $i$ and its Rx and between each Tx $i$ and the reference Rx at $o$:
\begin{enumerate}
\itemsep=-2pt
\item $\hsf_{0,0}$: fading coefficient from reference Tx to the reference Rx.
\item $\hsf_{i,0}, i \in \Nbb$: fading coefficient from interferer $i$ to reference Rx.
\item $\hsf_{i,i}$: fading coefficient from interferer $i$ to its own Rx.
\end{enumerate}
Let $\Pi_{d,\lambda_{\rm pot}}$ denote the set of {\em potential} transmitters ({\em c.f.}\ Rem.\  \ref{rem:ptx}), and assume each Tx elects to transmit precisely when its fading coefficient to its intended Rx exceeds a threshold $\hat{h} \in \Rbb_+$.  The FTS threshold $\hat{h}$ must lie in the support of the RV $\hsf$, and moreover we assume the threshold to be such that the success is guaranteed in the absence of interference, \ie, ({\em c.f.}\ Rem.\  \ref{rem:fadlowout})
\begin{equation}
\label{eq:schedhhlb}
\frac{\hat{h}}{\tau u^{\alpha}} - \frac{N}{P} > 0 \Leftrightarrow \hat{h} > \frac{\tau}{\snr}.
\end{equation}
More formally, the MPPP $\Phi_{d,\lambda_{\rm pot}} = \{(\xsf_i,\hsf_{i,0},\hsf_{i,i})\}$ induces a MPPP of actual interferers
\begin{equation}
\label{eq:schedphi}
\hat{\Phi}_{d,\hat{\lambda}} \equiv \{ (\xsf_i,\hsf_{i,0}) \in \Phi_{d,\lambda_{\rm pot}} : \hsf_{i,i} > \hat{h} \}
\end{equation}
of homogeneous intensity $\hat{\lambda} \equiv \lambda_{\rm pot} \Pbb(\hsf > \hat{h})$.  Assume the reference Tx attempts transmission (since otherwise there is no chance of outage at the reference Rx), so that the channel coefficient between the reference Tx and Rx $\hat{\hsf}_{0,0}$ has a CCDF
\begin{equation}
\label{eq:schedcondfad}
\bar{F}_{\hat{\hsf}}(h) = \Pbb(\hsf > h | \hsf > \hat{h}) = \left\{ \begin{array}{ll}
\frac{\bar{F}_{\hsf}(h)}{\bar{F}_{\hsf}(\hat{h})}, \; & h \geq \hat{h} \\
1, \; & \mbox{else}
\end{array} \right.
\end{equation}
The signal, interference, and SINR seen at the reference Rx at $o$ are given by
\begin{equation}
\sinr(o) \equiv \frac{\Ssf(o)}{\Sisf(o) + N/P}, ~ \Ssf(o) \equiv \hat{\hsf}_{0,0} u^{-\alpha}, ~ \Sisf(o) \equiv \sum_{i \in \hat{\Phi}_{d,\hat{\lambda}}} \hsf_{i,0} |\xsf_i|^{-\alpha},
\end{equation}
where $\{\hsf_{i,0}\}$ are iid with distribution $F_{\hsf}$, $\hat{\hsf}_{0,0}$ is independent of $\{\hsf_{i,0}\}$ and has distribution \eqref{eq:schedcondfad}, and $\hat{\Phi}_{d,\hat{\lambda}}$ is as in \eqref{eq:schedphi}.
\end{definition}
In this section we modify the OP (Def.\ \ref{def:op}) and TP (Def.\ \ref{def:tp}) definitions as follows.
\begin{definition}
\label{def:schedoptptc}
{\bf OP and TP.}
The OP is the probability of success at $o$ conditioned on having a channel fade above the threshold.  Because the threshold parameter $\hat{h}$ directly controls the intensity of attempted transmissions ($\hat{\lambda} = \lambda_{\rm pot} \bar{F}_{\hsf}(\hat{h})$) as well as the distribution on the signal strength $\hat{\hsf}_{0,0}$, it is more natural to define the OP as a function of $\hat{h}$ instead of $\lambda_{\rm  pot}$ or $\hat{\lambda}$.
\begin{equation}
q(\hat{h}) \equiv \Pbb(\sinr(o) < \tau | \hsf_{0,0} > \hat{h}),
\end{equation}
where interference is from a PPP $\hat{\Phi}_{d,\hat{\lambda}}$ of intensity $\hat{\lambda}$.  Likewise, the  MAC TP is the spatial intensity of successful transmissions
\begin{equation}
\label{eq:tpsched}
\Lambda(\hat{h}) \equiv \hat{\lambda}(1-q(\hat{h})).
\end{equation}
\end{definition}
\begin{remark}
\label{rem:desobj}
{\bf Quantity vs.\ quality of transmissions through FTS.}
The design objective in this section is to maximize the TP $\Lambda(\hat{h})$ over $\hat{h}$.  Increasing $\hat{h}$ has two effects: the signal power increases on average, and the intensity of attempted transmissions decreases.  In other words, increasing $\hat{h}$ increases the quality but decreases the quantity of attempted transmissions.  Viewing $\hat{\lambda}$ as the quantity and $1-q(\hat{h})$ as the quality of transmissions, we see the TP is the product of quantity times quality, and thereby yields a natural optimization over $\hat{h}$.   There are also fairness and delay costs under $\hat{h}$: nodes that happen to have a fading coefficient above the threshold are effectively given priority over those with a fading coefficient below.  The fairness and delay issues is of particular concern for channels with long coherence time where waiting for a ``good fade'' may be prohibitively costly.
\end{remark}
Recall \S\ref{sec:fading} addressed OP and TC under fading from three perspectives: exact (\S\ref{ssec:fadexact}), asymptotic  (\S\ref{ssec:fadasymp}), and bound (\S\ref{ssec:fadlb}) results.  This section follows the same outline, but focuses on OP and TP.

We first explain why FTS spoils the analytical tractability of \S\ref{ssec:fadexact}.  Recall the derivation of the explicit expression for the OP and TC when the signal fading is Rayleigh ($\hsf_0 \sim \mathrm{Exp}(1)$) in Prop.\ \ref{pro:optcrayfadsig} was obtained by conditioning on the interference, applying the exponential CCDF, and recognizing the resulting expression as the LT of the interference, given in Prop.\ \ref{pro:lapintfad}.  If $\hsf_0 \sim \mathrm{Exp}(1)$ then
\begin{equation}
\bar{F}_{\hat{\hsf}_{0,0}}(h) = \left\{ \begin{array}{ll} \erm^{-(h-\hat{h})}, \; & h > \hat{h} \\
0, \; & \mbox{else} \end{array} \right. .
\end{equation}
The analogous development to Prop.\ \ref{pro:optcrayfadsig} gives (writing $\Sisf = \Sisf^{\alpha,\hsf}_{d,\hat{\lambda}}$):
\begin{eqnarray}
q(\hat{\lambda}) &=& \Pbb(\sinr(o) < \tau) \nonumber \\
&=& \Pbb(\hat{\hsf}_{0,0} > \tau u^{\alpha} (\Sisf + N/P)) \nonumber \\
&=& \Ebb[ \Pbb( \left. \hat{\hsf}_{0,0} > \tau u^{\alpha} (\Sisf + N/P) \right| \Sisf) ] \nonumber \\
&=& \Ebb \left[ \exp \left\{ - \left(\tau u^{\alpha} (\Sisf + N/P) - \hat{h} \right) \right\} \mathbf{1}_{\Sisf > \hat{y}} \right] \nonumber \\
&=& \erm^{-(\tau/\snr-\hat{h})} \Ebb \left[  \erm^{\tau u^{\alpha} \Sisf} \mathbf{1}_{\Sisf > \hat{y}} \right]
\end{eqnarray}
for $\hat{y} = \frac{\hat{h}}{\tau u^{\alpha}} - N/P$.  The key difficulty is the indicator function within the expectation that precludes using the LT of $\Sisf$ in Prop.\ \ref{pro:lapintfad}.

We next turn to establishing the asymptotic OP (as $\hat{\lambda} \to 0$) under FTS.  The following result is the FTS analogue of Prop.\ \ref{pro:asymoptcfad}.
\begin{proposition}
\label{pro:asymopsched}
{\bf Asymptotic OP under FTS} ($\hat{h} \to \infty$, $\hat{\lambda} \to 0$):
\begin{equation}
\label{eq:schedasyop}
q(\hat{h}) = c_d \Ebb[\hsf^{\delta}] \Ebb \left[ \left( \frac{\hat{\hsf}_{0,0}}{\tau u^{\alpha}} - \frac{N}{P}\right)^{-\delta} \right] \lambda_{\rm pot} \bar{F}_{\hsf}(\hat{h}) + \Omc(\hat{\lambda}^2).
\end{equation}
In the no noise case ($N=0$) this expression becomes
\begin{equation}
\label{eq:schedoptc2}
q(\hat{h}) = c_d \tau^{\delta} u^d \Ebb[\hsf^{\delta}] \Ebb[\hat{\hsf}_{0,0}^{-\delta}] \lambda_{\rm pot} \bar{F}_{\hsf}(\hat{h}) + \Omc(\hat{\lambda}^2).
\end{equation}
In the no noise and Rayleigh fading case ($\hsf_0$ and $\{\hsf_i\}$ exponential RVs):
\begin{equation}
\label{eq:schedoptc3}
q(\hat{h}) = c_d \tau^{\delta} u^d \Gamma(1+\delta)\Gamma(1-\delta,\hat{h},\infty) \lambda_{\rm pot} + \Omc(\hat{\lambda}^2),
\end{equation}
for $\Gamma(z,t_l,t_h)$ the incomplete Gamma function defined in Def.\ \ref{def:gamfun}.
\end{proposition}
\begin{proof}
The proof is analogous to that of Prop.\ \ref{pro:asymoptcfad} with the notable exception that the assumption on $\hat{h}$ in \eqref{eq:schedhhlb} ensures outage is impossible without interference; this leads to a simplified development below.
\begin{eqnarray}
q(\hat{h}) &=& \Pbb\left( \sinr(o) < \tau | \hsf_{0,0} > \hat{h} \right) \nonumber \\
&=& \Ebb \left[ \Pbb\left( \sinr(o) < \tau | \hat{\hsf}_{0,0} \right) \right] \nonumber \\
&=& \Ebb \left[ \Pbb\left( \left. \Sisf^{\alpha,\hsf}_{d,\hat{\lambda}} > \frac{\hat{\hsf}_{0,0}}{\tau u^{\alpha}} - \frac{N}{P} \right| \hat{\hsf}_{0,0} \right) \right]
\end{eqnarray}
Now apply Prop.\ \ref{pro:markdistmap} and \ref{pro:fadintserrep} noting that $\hat{\hsf}_{0,0}$ and $\Sisf$ are independent.
\begin{eqnarray}
q(\hat{h}) &=& \Ebb \left[ \Pbb \left( \left. \Sisf^{1/\delta,\hsf}_{1,1} > \left( \frac{\hat{\lambda} c_d}{2} \right)^{-\frac{1}{\delta}} \left(\frac{\hat{\hsf}_{0,0}}{\tau u^{\alpha}} - \frac{N}{P} \right) \right| \hat{\hsf}_{0,0} \right) \right] \nonumber \\
&=& \Ebb \left[ \Pbb \left( \left. \Sisf^{1/\delta,\hsf}_{1,1} > \hat{\ysf}  \right| \hat{\hsf}_{0,0} \right) \right], ~ \hat{\ysf} = \left( \frac{\hat{\lambda} c_d}{2} \right)^{-\frac{1}{\delta}} \left(\frac{\hat{\hsf}_{0,0}}{\tau u^{\alpha}} - \frac{N}{P} \right) \nonumber \\
&=& \Ebb \left[ \Ebb[\hsf^{\delta}] \hat{\ysf}^{-\delta} + \Omc(\ysf^{-2 \delta}) \right]
\end{eqnarray}
Simplification yields \eqref{eq:schedasyop}.  \eqref{eq:schedoptc2} is immediate upon substituting $N=0$.  Assuming Rayleigh fading we obtain \eqref{eq:schedoptc3} by observing
\begin{equation}
\Ebb[\hat{\hsf}_{0,0}^{-\delta}] \bar{F}_{\hsf}(\hat{h}) = \Gamma(1-\delta,\hat{h},\infty).
\end{equation}
\end{proof}
This result is applied below to obtain the asymptotic TP under FTS and its optimization over thresholds $\hat{h}$.  For simplicity we restrict our attention to the no noise case in what follows.
\begin{proposition}
\label{pro:asymtpsched}
The {\bf asymptotic TP under FTS} and no noise ($N=0$) as $\hat{h} \to \infty$ ($\hat{\lambda} \to 0$) is:
\begin{equation}
\label{eq:asymtpsched1}
\Lambda(\hat{h}) = \lambda_{\rm pot} \bar{F}_{\hsf}(\hat{h}) \left(1 - b \lambda_{\rm pot} \Ebb[\hat{\hsf}_{0,0}^{-\delta}] \bar{F}_{\hsf}(\hat{h}) \right) + \Omc(\hat{\lambda}^3)
\end{equation}
for $b = c_d \tau^{\delta} u^d \Ebb[\hsf^{\delta}]$.  The asymptotic TP optimal $\hat{h}$ satisfies
\begin{equation}
\label{eq:asymhhoptsched}
\Ebb[\hsf^{-\delta}\mathbf{1}_{\hsf > \hat{h}}] + \Ebb[ \hat{h}^{-\delta} \mathbf{1}_{\hsf > \hat{h}}] = \frac{1}{b \lambda_{\rm pot}}.
\end{equation}
\end{proposition}
\begin{proof}
Expression \eqref{eq:asymtpsched1} follows upon substitution of \eqref{eq:schedoptc2} into \eqref{eq:tpsched}.  The derivative of the TP with respect to $\hat{h}$ is
\begin{eqnarray}
\frac{1}{\lambda_{\rm pot}} \Lambda'(\hat{h}) &=& - f_{\hsf}(\hat{h}) \left(1 - b \lambda_{\rm pot} \Ebb[\hat{\hsf}_{0,0}^{-\delta}] \bar{F}_{\hsf}(\hat{h}) \right) \nonumber \\
& & - b \lambda_{\rm pot} \bar{F}_{\hsf}(\hat{h}) \frac{\drm}{\drm \hat{h}} \left\{ \Ebb[\hat{\hsf}_{0,0}^{-\delta}] \bar{F}_{\hsf}(\hat{h}) \right\}.
\end{eqnarray}
Observe that \eqref{eq:schedcondfad} gives:
\begin{eqnarray}
\Ebb[\hat{\hsf}_{0,0}^{-\delta}] \bar{F}_{\hsf}(\hat{h}) &=& \int_{\hat{h}}^{\infty} h^{-\delta} f_{\hsf}(h) \drm h \nonumber  = \Ebb[\hsf^{-\delta} \mathbf{1}_{\hsf > \hat{h}}] \\
\frac{\drm}{\drm \hat{h}} \left\{ \Ebb[\hat{\hsf}_{0,0}^{-\delta}] \bar{F}_{\hsf}(\hat{h}) \right\} &=& - \hat{h}^{-\delta} f_{\hsf}(\hat{h})
\end{eqnarray}
Substitution, equating with zero, and rearranging gives \eqref{eq:asymhhoptsched}.
\end{proof}
Fig.\ \ref{fig:schedasymp} presents some numerical results for Prop.\ \ref{pro:asymtpsched}: the left plot shows the asymptotic TP \eqref{eq:asymtpsched1} vs.\ the fading threshold $\hat{h}$, and the right plot shows the optimal threshold $\hat{h}_{\rm opt}$ solving \eqref{eq:asymhhoptsched}.  The optimal $\hat{h}$ trades off quantity with quality of attempted transmissions --- as expected both the optimal threshold and the associated optimum TP increase with $\lambda_{\rm pot}$ as the MAC can afford to be more selective on link quality without sacrificing the quantity of transmission attempts.
\begin{figure}[!htbp]
\centering
\includegraphics[width=0.49\textwidth]{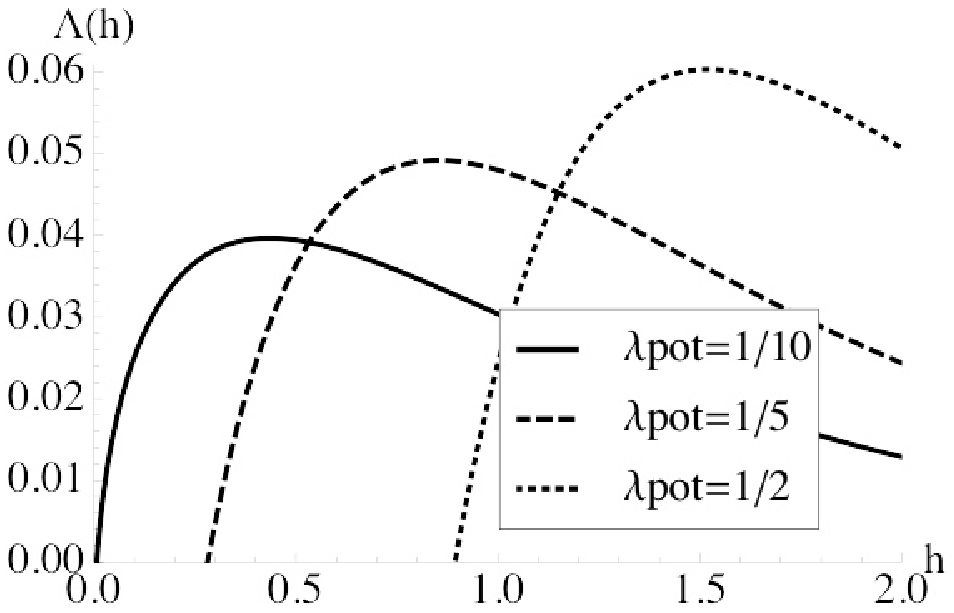}
\includegraphics[width=0.49\textwidth]{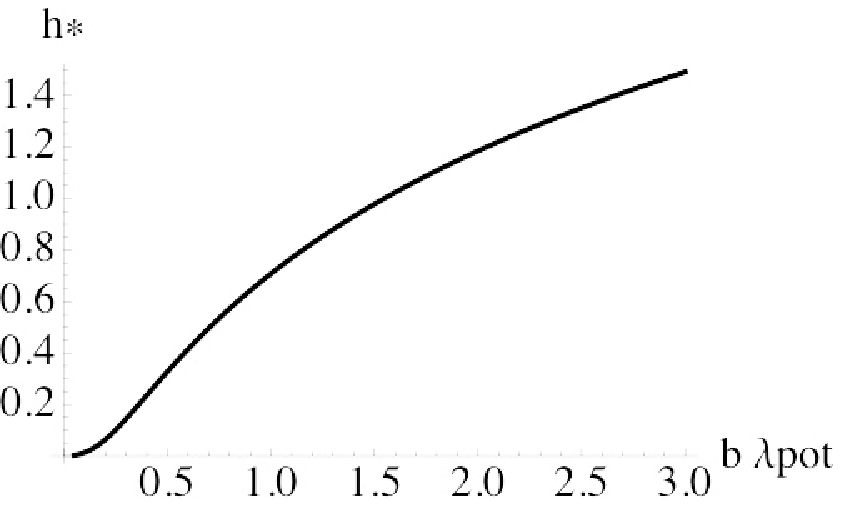}
\caption{{\bf Left:} the asymptotic TP $\Lambda(\hat{h})$ \eqref{eq:asymtpsched1} vs.\ the fading  threshold $\hat{h}$ for $\lambda_{\rm opt} \in \{1/10,1/5,1/2\}$.  {\bf Right:} the TP asymptotic optimal threshold $\hat{h}_{\rm opt}$ that maximizes \eqref{eq:asymtpsched1} as a function of $b \lambda_{\rm opt}$.  All fading is Rayleigh; other parameter values are $d=2$, $\alpha=4$ ($\delta=1/2$), $\tau=5$, $u=1$ ($b \approx 3.11$).}
\label{fig:schedasymp}
\end{figure}
We next show that the asymptotic expressions for TP without fading, with fading but without scheduling, and with fading under FTS are ordered when the signal and interference fading distribution are equal.  This result extends Cor.\ \ref{cor:optcfadnoncomp}.
\begin{proposition}
\label{pro:schedasympcomp}
{\bf FTS exploits fading to improve performance.}
Assume the signal and interference fading coefficients are equal in distribution: $\hsf_{0,0} \stackrel{\drm}{=} \hsf$.  Assume no noise ($N=0$) and define $a = \frac{1}{2} c_d \tau^{\delta} u^d$.  Then the asymptotic TPs (as $\hat{\lambda} \to 0$ and $\hat{h} \to \infty$) without fading (NF, Prop.\ \ref{pro:asymoptc}), with fading but without scheduling (FNS, Prop.\ \ref{pro:asymoptcfad}), and with FTS (Prop.\ \ref{pro:asymtpsched}) are:
\begin{eqnarray}
\Lambda^{\rm NF}(\hat{\lambda}) &=& \hat{\lambda}(1 - a \hat{\lambda} + \Omc(\lambda^2)) \nonumber \\
\Lambda^{\rm FNS}(\hat{\lambda}) &=& \hat{\lambda}(1 - a \Ebb[\hsf^{\delta}]\Ebb[\hsf_{0,0}^{-\delta}] \hat{\lambda} + \Omc(\lambda^2)) \nonumber \\
\Lambda^{\rm FTS}(\hat{\lambda}) &=& \hat{\lambda}(1 - a \Ebb[\hsf^{\delta}]\Ebb[\hat{\hsf}_{0,0}^{-\delta}] \hat{\lambda} + \Omc(\lambda^2)),
\end{eqnarray}
where $\hat{\lambda}$ in $\Lambda^{\rm FTS}$ implies selecting $\hat{h}$ so that $\lambda_{\rm pot} \bar{F}_{\hsf}(\hat{h}) = \hat{\lambda}$, \ie, $\hat{h} = \bar{F}_{\hsf}^{-1}\left(\hat{\lambda}/\lambda_{\rm pot}\right)$.  Large $\hat{h}$ ($\hat{h} > \Ebb[\hsf^{\delta}]^{\frac{1}{\delta}}$ is sufficient) ensures
\begin{equation}
\label{eq:schedasympcompineq1}
\Ebb[\hsf^{\delta}]\Ebb[\hat{\hsf}_{0,0}^{-\delta}] \leq 1 \leq \Ebb[\hsf^{\delta}] \Ebb[\hsf_{0,0}^{-\delta}],
\end{equation}
which in turn ensures
\begin{equation}
\label{eq:schedasympcompineq2}
\Lambda^{\rm FNS}(\hat{\lambda}) \leq \Lambda^{\rm NF}(\hat{\lambda}) \leq \Lambda^{\rm FTS}(\hat{\lambda}).
\end{equation}
\end{proposition}
\begin{proof}
The asymptotic TP expressions for $\Lambda^{\rm NF}(\hat{\lambda})$ and $\Lambda^{\rm FNS}(\hat{\lambda})$ are immediate from the TP definition $\Lambda(\hat{\lambda}) = \hat{\lambda}(1-q(\hat{\lambda}))$, and the expression for $\Lambda^{\rm FTS}(\hat{\lambda})$ is immediate from Prop.\ \ref{pro:asymtpsched} using the assumed $\lambda_{\rm pot} \bar{F}_{\hsf}(\hat{h}) = \hat{\lambda}$.  The second inequalities in \eqref{eq:schedasympcompineq1} and \eqref{eq:schedasympcompineq2} are not asymptotic, they hold for all $\hat{\lambda}$.  To prove the first inequalities, note the assumption $\hat{h}^{\delta} > \Ebb[\hsf^{\delta}]$ ensures the second inequality below:
\begin{equation}
\Ebb[\hsf^{\delta}] \Ebb[\hat{\hsf}_{0,0}^{-\delta}] = \Ebb[\hsf^{\delta}] \frac{\Ebb[ \hsf_{0,0}^{-\delta} \mathbf{1}_{\hsf_{0,0} > \hat{h}}]}{\bar{F}_{\hsf_{0,0}}(\hat{h})}  \leq \Ebb[\hsf^{\delta}] \frac{\hat{h}^{-\delta} \Ebb[\mathbf{1}_{\hsf_{0,0} > \hat{h}}]}{\bar{F}_{\hsf_{0,0}}(\hat{h})} \leq 1.
\end{equation}
\end{proof}
Fig.\ \ref{fig:schedasympcomp} illustrates the asymptotic TP under no fading, fading without scheduling, and fading with threshold scheduling.  The figure makes clear that threshold scheduling transforms fading from an overall performance penalty (fading without scheduling being inferior to no fading) to a performance enhancement (fading with threshold scheduling being superior to no fading).  Threshold scheduling successfully exploits the available fading diversity.  Note as $\alpha \to 2$ ($\delta \to 1$) that fading with threshold scheduling is equivalent (in terms of asymptotic TP) to no fading; one can also show that as $\alpha \to \infty$ ($\delta \to 0$) fading without scheduling is equivalent to no fading.
\begin{figure}[!htbp]
\centering
\includegraphics[width=0.49\textwidth]{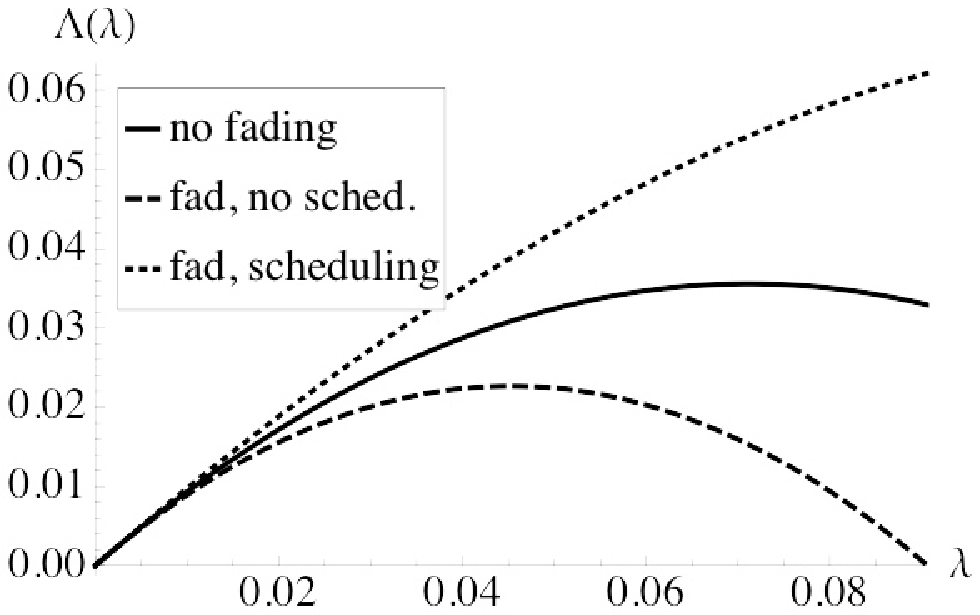}
\includegraphics[width=0.49\textwidth]{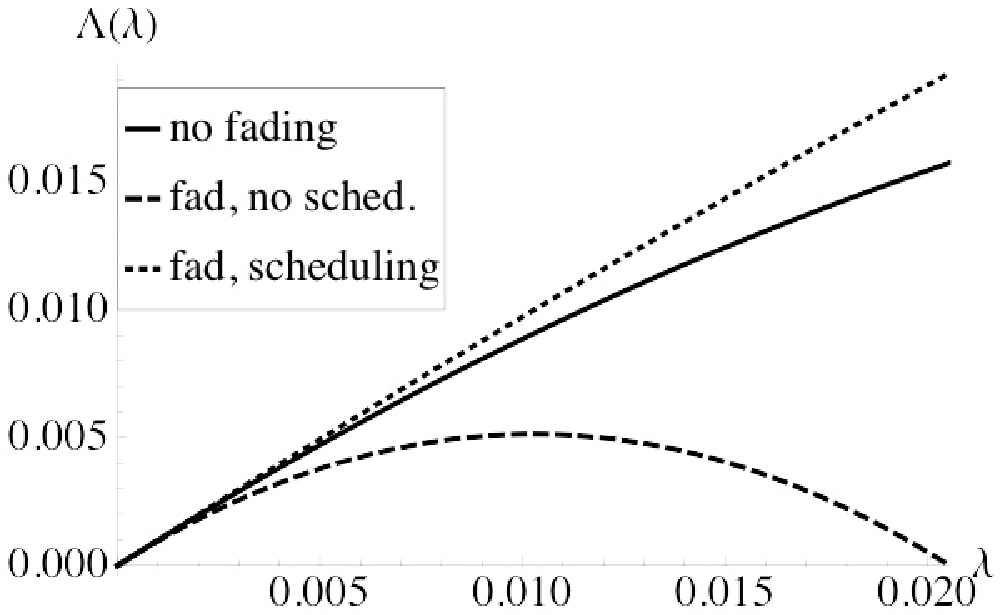}
\caption{The asymptotic $(\hat{\lambda} \to 0$) spatial TP $\Lambda(\hat{\lambda})$ vs.\ the spatial intensity of attempted transmissions $\hat{\lambda}$ under $i)$ no fading, $ii)$ fading without scheduling, and $iii)$ fading with threshold scheduling for $\alpha = 4$ (left) and $\alpha = 2.5$ (right).  Other parameters are $d=2$, $u=1$, $\tau=5$, $\lambda_{\rm pot}=1$.}
\label{fig:schedasympcomp}
\end{figure}

Finally, we establish a LB on the OP and an UB on TP under fading threshold scheduling.  These bounds follow immediately by the natural extension of Def.\ \ref{def:domintfad} and its subsequent Prop.\ \ref{pro:fadoplb} in \S\ref{ssec:fadlb}, replacing $\hsf_0$ with $\hat{\hsf}_{0,0}$, yielding the following result.
\begin{proposition}
\label{pro:schedlb}
{\bf OP LB and TP UB.}
Extend the concepts defined in Def.\ \ref{def:domintfad} to apply to threshold scheduling by replacing $\hsf_0$ with $\hat{\hsf}_{0,0}$.  Then Prop.\ \ref{pro:fadoplb} yields the OP LB
\begin{equation}
q^{\rm lb}(\hat{h}) = 1 - \Ebb \left[ \exp \left\{ - \lambda_{\rm pot} \bar{F}_{\hsf}(\hat{h}) c_d \Ebb[\hsf^{\delta}]  \left(\frac{\hat{\hsf}_{0,0}}{\tau u^{\alpha}} - \frac{N}{P} \right)^{-\delta} \right\}  \right]
\label{eq:schedoplb}
\end{equation}
where the outer expectation is w.r.t. the random signal fade $\hat{\hsf}_{0,0}$.  In the case of no noise ($N=0$) the OP LB is:
\begin{equation}
\label{eq:schedoplbnn}
q^{\rm lb}(\hat{h}) = 1 - \Ebb \left[ \exp \left\{ - \lambda_{\rm pot} \bar{F}_{\hsf}(\hat{h}) c_d \tau^{\delta} u^d \Ebb[\hsf^{\delta}] \hat{\hsf}_{0,0}  \right\}  \right],
\end{equation}
In particular, with no noise the OP LB is expressible in terms of the MGF of the RV $-\hat{\hsf}_{0,0}^{-\delta}$ at a certain $\theta$:
\begin{equation}
q^{\rm lb}(\hat{h}) = 1 - \left. \Mmc[-\hat{\hsf}_{0,0}^{-\delta}](\theta)\right|_{\theta = \lambda_{\rm pot} \bar{F}_{\hsf}(\hat{h}) c_d \tau^{\delta} u^d \Ebb[\hsf^{\delta}]}.
\end{equation}
In all cases the TP UB is
\begin{equation}
\label{eq:schedtpub}
\Lambda^{\rm ub}(\hat{h}) = \lambda_{\rm pot} \bar{F}_{\hsf}(\hat{h}) \left(1 - q^{\rm lb}(\hat{h}) \right).
\end{equation}
\end{proposition}
Fig.\ \ref{fig:schedlb} compares the TP UB from Prop.\ \ref{pro:schedlb} with the asymptotic TP from Prop.\ \ref{pro:asymtpsched}.
\begin{figure}[!htbp]
\centering
\includegraphics[width=0.49\textwidth]{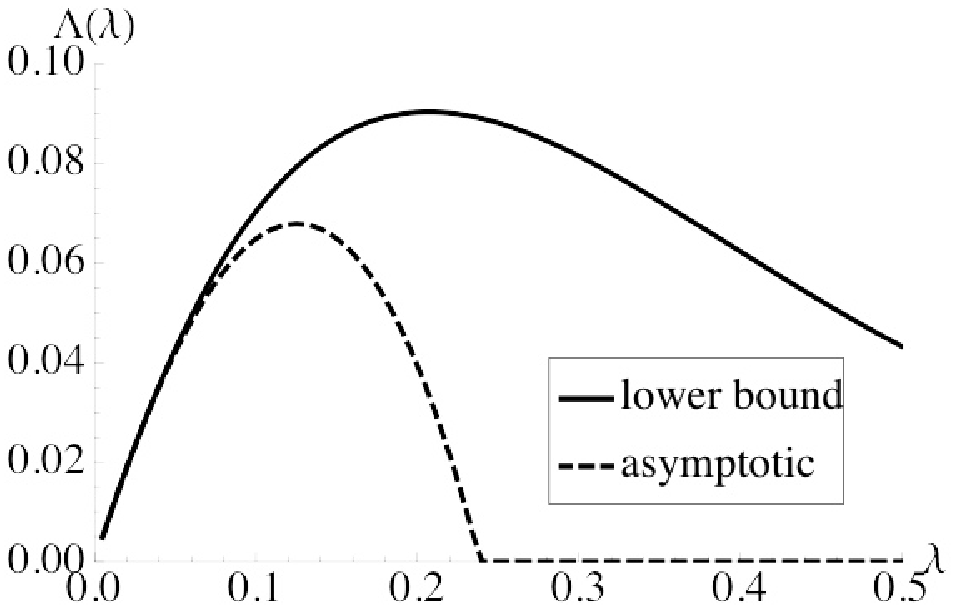}
\includegraphics[width=0.49\textwidth]{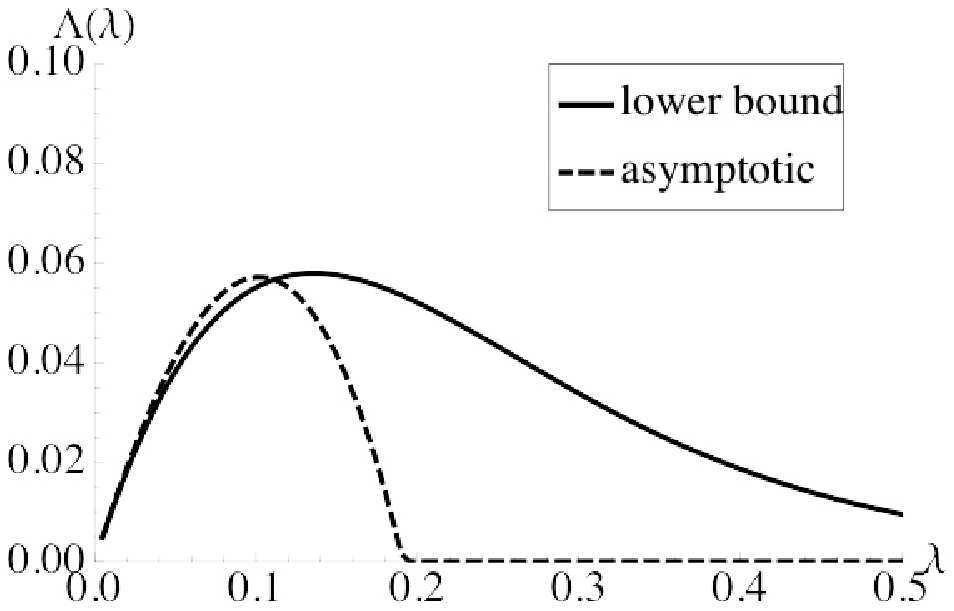}
\caption{The TP UB and the asymptotic $(\hat{\lambda} \to 0$) spatial TP $\Lambda(\hat{\lambda})$ vs.\  the spatial intensity of attempted transmissions $\hat{\lambda}$ for $\alpha = 4$ (left) and $\alpha = 2.5$ (right).  Other parameters are $d=2$, $u=1$, $\tau=5$, $\lambda_{\rm pot}=1$.}
\label{fig:schedlb}
\end{figure}
In closing, we refer the interested reader to recent work on threshold scheduling by Kim, Baccelli, and de Veciana \cite{KimBac2011sub}, and Cho and Andrews \cite{ChoAnd2009}.

\section{Fractional power control (FPC)}
\label{sec:power}

In \S\ref{sec:sched} we found that exploiting fading through threshold scheduling improved performance (measured as spatial TP) compared with no fading; in this section we exploit fading through a specific form of power control (PC).  In general, PC in decentralized wireless networks refers to an algorithm where each node $i$ at each slot $t$ updates its power $P_{i,t}$ based on feedback so that the nodes jointly converge in a distributed manner to a globally (or at least locally) optimal power allocation vector.  As with scheduling in \S\ref{sec:sched}, analytical tractability requires we adopt a rather restricted definition of PC.  Again, the key property we seek to maintain is that the set of interferers at some point in time form a PPP, and this requires each interferer act independently of the other interferers.  Consequently, power control in this section refers to each interferer $i \in \Pi_{d,\lambda}$ selecting a (random) power $\Psf_i$ that is a function of the fading coefficient to its intended Rx $\hsf_{i,i}$ \cite{JinWeb2008}.  The following definition of fractional power control (FPC) makes these concepts precise.
\begin{definition}
\label{def:power}
{\bf Transmission powers and SINR.}
Consider the fading model and MPPP used in Def.\ \ref{def:sched} in \S\ref{sec:sched}:
\begin{equation}
\Phi_{d,\lambda} \equiv \left\{ (\xsf_i,\hsf_{i,i},\hsf_{i,0})\right\},
\end{equation}
and the same reference pair fading coefficient $\hsf_{0,0}$.  As in Def.\ \ref{def:sched}, all fading coefficients are independent.  Assume $\{\hsf_{i,i}\}$ are iid with CDF $F_{\hsf_{1,1}}$, $\{\hsf_{i,0}\}$ are iid with CDF $F_{\hsf_{1,0}}$, and $\hsf_{0,0}$ has CDF $F_{\hsf_{0,0}}$.  Let $f \in \Rbb$ be the FPC exponent and let $P$ be the average transmission power.  The reference Tx ($0$) and each interferer $i \in \Pi_{d,\lambda}$ select a random transmission power according to
\begin{equation}
\label{eq:fpcpdef}
\Psf_0 \equiv \frac{P}{\Ebb[\hsf_{0,0}^{-f}]} \hsf_{0,0}^{-f}, ~~~
\Psf_i \equiv \frac{P}{\Ebb[\hsf_{1,1}^{-f}]} \hsf_{i,i}^{-f}.
\end{equation}
The SINR at $o$ is
\begin{equation}
\sinr(o) \equiv \frac{\Ssf(o)}{\Sisf(o)+N},
\end{equation}
where the received signal and interference powers at $(o)$ are
\begin{equation}
\Ssf(o) \equiv \frac{P}{\Ebb[\hsf_{0,0}^{-f}]} \hsf_{0,0}^{1-f} u^{-\alpha}, ~~
\Sisf(o) \equiv \frac{P}{\Ebb[\hsf_{1,1}^{-f}]} \sum_{i \in \Pi_{d,\lambda}} \hsf_{i,i}^{-f} \hsf_{i,0} |\xsf_i|^{-\alpha}.
\end{equation}
\end{definition}
To clarify, $\Ssf(o)$ is proportional to $\hsf_{0,0}^{1-f}$ due to FPC ($\hsf_{0,0}^{-f}$) and fading itself ($\hsf_{0,0}$).  Interference from $i$ seen at $(o)$ is proportional to $\hsf_{i,i}^{-f}$ due to FPC and to $\hsf_{i,0}$ and due to fading.  Observe the transmission powers are selected so that $\Ebb[\Psf_0] = \Ebb[\Psf_i] = P$ for all $i \in \Pi_{d,\lambda}$.  Two special cases for $f$ merit mention: (i) $f=0$ is no PC, \ie, $\Psf_0 = \Psf_i = P$, and we revert to the results up until now; (ii) $f=1$ is channel inversion, \ie, the transmitted signal power inverts the fading coefficient so that the received signal power is not dependent upon the fading coefficient.  Therefore, FPC can be viewed as a fairly general and flexible form of power control. The following lemma gives the moments and variance of $\Psf$ under Rayleigh fading.
\begin{lemma}
\label{lem:fpcvar}
{\bf Transmitted power RV moments under FPC.}
For Rayleigh fading ($\hsf_{0,0},\hsf_{i,i} \sim \mathrm{Exp}(1)$) the moments and variance of $\Psf_0,\Psf_i$ under FPC in \eqref{eq:fpcpdef} are
\begin{eqnarray}
\Ebb[\Psf^p] &=& \left\{ \begin{array}{cl}
\left( \frac{P}{\Gamma(1-f)} \right)^p \Gamma(1-pf), \; & p f < 1 \\
\infty,\; & \mbox{else}
\end{array} \right., \\
\mathrm{Var}(\Psf) &=& \left\{ \begin{array}{cl}
P^2 \left( \frac{\Gamma(1-2f)}{\Gamma(1-f)^2} - 1\right), \; & f < 1/2 \\
\infty, \; & \mbox{else}
\end{array} \right.
\end{eqnarray}
\end{lemma}
The variance $\mathrm{Var}(\Psf)$ is plotted in Fig.\ \ref{fig:fpcvar} vs.\ the FPC exponent $f$.
\begin{figure}[!htbp]
\centering
\includegraphics[width=0.75\textwidth]{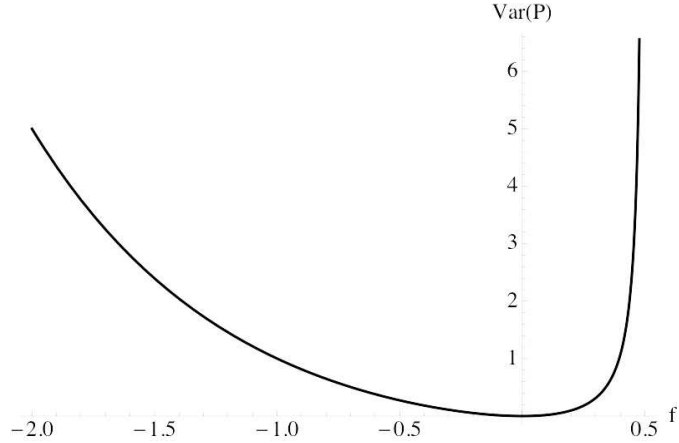}
\caption{The variance of the transmitted power $\mathrm{Var}(\Psf)$ vs.\ the FPC exponent $f$ for $P = 1$.  Note $\mathrm{Var}(\Psf) = \infty$ for $f \geq 1/2$.}
\label{fig:fpcvar}
\end{figure}

The design objective in this section is the selection of the FPC exponent $f \in \Rbb$.  The outline of this section mirrors that of \S\ref{sec:sched}: we first discuss asymptotic OP (as $\lambda \to 0$) and TC (as $q^* \to 0$), and then discuss a LB on OP (UB on TC).  Although in \S\ref{sec:sched} we assumed $\hat{h}$ to be sufficiently large to preclude the possibility of a bad fade guaranteeing outage even in the absence of interference ({\em c.f.}\ Rem.\  \ref{rem:fadlowout}), under FPC there is again this possibility.  We update Rem.\  \ref{rem:fadlowout} for FPC below.
\begin{remark}
\label{rem:fpclowout}
{\bf Fading and outage with no interference.}
With random signal fading under FPC there is the possibility of a bad fade causing outage even in the absence of interference.  This outage event $\Emc_{0,f}$ has the three equivalent forms:
\begin{equation}
\frac{ \Psf_0 u^{-\alpha} \hsf_{0,0}}{N} < \tau \Leftrightarrow \hsf_{0,0} < \left( \frac{\tau}{\snr} \Ebb[\hsf_{0,0}^{-f}] \right)^{\frac{1}{1-f}} \Leftrightarrow \frac{\hsf_{0,0}^{1-f}}{\Ebb[\hsf_{0,0}^{-f}] \tau u^{\alpha}} - \frac{N}{P} < 0,
\end{equation}
for $\snr$ defined in Ass.\ \ref{ass:epsnoi}, and has probability
\begin{equation}
q_f(0) \equiv \Pbb\left( \Emc_{0,f} \right) = F_{\hsf_{0,0}}\left(\left( \frac{\tau}{\snr} \Ebb[\hsf_{0,0}^{-f}] \right)^{\frac{1}{1-f}}\right).
\end{equation}
Note $q_f(0)$ is the OP evaluated at $\lambda = 0$.  Denote the complement of $\Emc_{0,f}$ by $\bar{\Emc}_{0,f}$, and its probability by
\begin{equation}
\bar{q}_f(0) \equiv 1 -q_f(0) = \Pbb(\bar{\Emc}_{0,f}) = \bar{F}_{\hsf_{0,0}}\left(\left( \frac{\tau}{\snr} \Ebb[\hsf_{0,0}^{-f}] \right)^{\frac{1}{1-f}}\right).
\end{equation}
All analysis of OP (and hence TC) must therefore condition on $\hsf_{0,0}$ being above or below $\left( \frac{\tau}{\snr} \Ebb[\hsf_{0,0}^{-f}] \right)^{\frac{1}{1-f}}$ to distinguish between the case of outage being possible vs.\ outage being guaranteed.  The range of the OP $q(\lambda)$ is $[q_f(0),1]$, and the domain of $q^*$ in the TC $\lambda(q^*)$ is $[q_f(0),1]$.
\end{remark}
The following result gives the asymptotic OP and TC under FPC.
\begin{proposition}
\label{pro:fpcasymp}
The {\bf asymptotic OP under FPC} as $\lambda \to 0$ is:
\begin{equation}
\label{eq:opfpcasymp}
q(\lambda) = 1 - (1 - a \lambda) \bar{q}_f(0) + \Omc(\lambda^2) \nonumber \\
\end{equation}
for
\begin{equation}
a = c_d \Ebb\left[ \hsf_{1,1}^{-f \delta}\right] \Ebb\left[ \hsf_{1,0}^{\delta} \right] \Ebb \left[\hsf_{1,1}^{-f} \right]^{-\delta} \Ebb \left[ \left. \left( \frac{\hsf_{0,0}^{1-f}}{\Ebb[\hsf_{0,0}^{-f}] \tau u^{\alpha}} - \frac{N}{P} \right)^{-\delta} \right| \bar{\Emc}_{0,f} \right].
\end{equation}
The asymptotic TC under fading as $q^* \to q_f(0)$ is
\begin{equation}
\label{eq:tcfpcasymp}
\lambda(q^*) = \frac{1}{a} \left( \frac{q^* - q_f(0)}{1-q_f(0)} \right) + \Omc(q^*-q_f(0))^2.
\end{equation}
In the no noise case ($N=0$, $q_f(0) = 0$) these expressions become
\begin{eqnarray}
q(\lambda) &=& \lambda c_d \tau^{\delta} u^d \Ebb\left[ \hsf_{1,1}^{-f \delta}\right] \Ebb\left[ \hsf_{1,0}^{\delta} \right] \Ebb \left[\hsf_{1,1}^{-f} \right]^{-\delta} \Ebb \left[\hsf_{0,0}^{-f} \right]^{\delta} \Ebb \left[ \hsf_{0,0}^{-(1-f)\delta}\right] \nonumber \\
& & + \Omc(\lambda^2) \nonumber \\
\lambda(q^*) &=& \frac{q^*}{ c_d \tau^{\delta} u^d \Ebb\left[ \hsf_{1,1}^{-f \delta}\right] \Ebb\left[ \hsf_{1,0}^{\delta} \right] \Ebb \left[\hsf_{1,1}^{-f} \right]^{-\delta} \Ebb \left[\hsf_{0,0}^{-f} \right]^{\delta} \Ebb \left[ \hsf_{0,0}^{-(1-f)\delta}\right]} \nonumber \\
& & + \Omc(q^*)^2 \label{eq:optcfpcasympnn}
\end{eqnarray}
For no noise and Rayleigh fading ($\hsf_{0,0},\hsf_{i,0},\hsf_{i,i}$ exponential RVs):
\begin{eqnarray}
q(\lambda) &=& \lambda c_d \tau^{\delta} u^d \Gamma(1+\delta) \Gamma(1-f\delta) \Gamma(1-(1-f)\delta) + \Omc(\lambda^2) \nonumber \\
\lambda(q^*) &=& \frac{q^*}{ c_d \tau^{\delta} u^d \Gamma(1+\delta) \Gamma(1-f\delta) \Gamma(1-(1-f)\delta)} + \Omc(q^*)^2 \label{eq:optcfpcasympnnrf}
\end{eqnarray}
provided $\delta, f \delta, (1-f)\delta$ are all in $(0,1)$.
\end{proposition}
\begin{proof}
The proof is analogous to that of Prop.\ \ref{pro:asymoptcfad}.  We condition on $\hsf_{0,0}$ and isolate the random normalized interference.
\begin{eqnarray}
q(\lambda) &=& \Pbb\left( \sinr(o) < \tau \right) \nonumber \\
&=&  \Pbb\left( \sinr(o) < \tau| \bar{\Emc}_{0,f} \right) \bar{q}_f(0) + 1 (1-\bar{q}_f(0)) \nonumber \\
&=& 1 - \left(1 - \Pbb\left( \sinr(o) < \tau| \bar{\Emc}_{0,f} \right) \right) \bar{q}_f(0) \nonumber \\
&=& 1 - \left(1 - \underbrace{\Ebb \left[ \left. \Pbb\left( \left. \sinr(o) < \tau \right|\hsf_{0,0} \right) \right| \bar{\Emc}_{0,f} \right]}_{\Ebb[\cdot]} \right) \bar{q}_f(0)
\end{eqnarray}
\begin{eqnarray}
\Ebb[\cdot] &=& \Ebb \left[ \left. \Pbb \left( \left. \frac{\frac{P \hsf_{0,0}^{1-f}}{\Ebb[\hsf_{0,0}^{-f}] } u^{-\alpha}}{\frac{P}{\Ebb[\hsf_{1,1}^{-f}]} \sum_{i \in \Pi_{d,\lambda}} \hsf_{i,i}^{-f} \hsf_{i,0} |\xsf_i|^{-\alpha}+N} < \tau \right| \hsf_{0,0} \right) \right| \bar{\Emc}_{0,f} \right] \nonumber \\
&=& \Ebb \left[ \left. \Pbb \left( \left. \Sisf^{\alpha,\hsf}_{d,\lambda}  > \Ebb[\hsf_{1,1}^{-f}] \left( \frac{\hsf_{0,0}^{1-f}}{\Ebb[\hsf_{0,0}^{-f}] \tau u^{\alpha}} - \frac{N}{P} \right) \right| \hsf_{0,0} \right) \right| \bar{\Emc}_{0,f} \right]
\end{eqnarray}
for
\begin{equation}
\label{eq:fpcasympi}
\Sisf^{\alpha,\hsf}_{d,\lambda} = \sum_{i \in \Pi_{d,\lambda}} \hsf_{i,i}^{-f} \hsf_{i,0} |\xsf_i|^{-\alpha}.
\end{equation}
Define the RV
\begin{equation}
\label{eq:fpcasymppfy}
\ysf \equiv \left( \frac{\lambda c_d}{2} \right)^{-\frac{1}{\delta}} \Ebb[\hsf_{1,1}^{-f}] \left( \frac{\hsf_{0,0}^{1-f}}{\Ebb[\hsf_{0,0}^{-f}] \tau u^{\alpha}} - \frac{N}{P} \right),
\end{equation}
as a function of the RV $\hsf_{0,0}$ and apply Prop.\ \ref{pro:markdistmap} to obtain
\begin{equation}
\Ebb[\cdot] = \Ebb \left[ \left. \Pbb \left( \left. \sum_{i \in \Pi_{1,1}} \hsf_{i,i}^{-f} \hsf_{i,0} |\tsf_i|^{-\frac{1}{\delta}}  > \ysf \right| \hsf_{0,0} \right) \right| \bar{\Emc}_{0,f} \right].
\end{equation}
Now apply Prop.\ \ref{pro:fadintserrep} to obtain the asymptotic CCDF
\begin{equation}
\label{eq:fpcasympq}
\Ebb[\cdot] = \Ebb \left[ \left. \Ebb\left[ \left( \hsf_{1,1}^{-f} \hsf_{1,0} \right)^{\delta} \right] \ysf^{-\delta} + \Omc(\ysf^{-2 \delta}) \right| \ysf > 0 \right].
\end{equation}
Note the equivalence of the events $\bar{\Emc}_{0,f}$ and $\ysf > 0$.  Substitution of \eqref{eq:fpcasymppfy} into \eqref{eq:fpcasympq} yields \eqref{eq:opfpcasymp}.  Solving for $\lambda$ gives \eqref{eq:tcfpcasymp}.  Expressions \eqref{eq:optcfpcasympnn} are immediate after substitution of $N=0$.  Expressions \eqref{eq:optcfpcasympnnrf} follow after application of Lem.\  \ref{lem:expmom}.
\end{proof}
The next result gives the asymptotic optimal choice for $f$ in the special case of no noise and when the signal fading coefficient $\hsf_{0,0}$ is equal in distribution to the interference fading coefficients $\hsf_{1,1}$.
\begin{proposition}
\label{pro:fpcasympoptf}
{\bf Asymptotic optimality of $f=1/2$.}
Setting $f = 1/2$ minimizes the asymptotic OP and maximizes the asymptotic TC under no noise ($N=0$) and $\hsf_{0,0} \stackrel{\drm}{=} \hsf_{1,1}$.
\end{proposition}
\begin{proof}
It is clear that the optimal $f$ for minimizing the asymptotic OP equals the optimal $f$ for maximizing the asymptotic TC, hence we restrict our attention to minimizing the asymptotic OP.
The asymptotic OP under the assumption of no noise and $\hsf_{0,0} \stackrel{\drm}{=} \hsf_{1,1}$ (which we denote in this proof by $\hsf$) gives the optimization problem
\begin{equation}
\label{eq:fpcasympoptfpf}
\min_{f \in \Rbb} \Ebb\left[ \hsf^{-f \delta}\right] \Ebb \left[ \hsf^{-(1-f)\delta}\right].
\end{equation}
We will show the following more general result: for any RV $\xsf$ the function $g(f) = \Ebb[\xsf^{-f}] \Ebb[\xsf^{-\bar{f}}]$ for $f \in \Rbb$ (and $\bar{f} = 1 -f$) is convex in $f$ and has a unique global minimizer at $f = 1/2$.  The proposition then follows by using this result for the RV $\xsf = \hsf^{\delta}$.  Recall that a function is said to be log-convex if it satisfies
\begin{equation}
g(c f_1 + \bar{c} f_2) \leq g(f_1)^c g(f_2)^{\bar{c}}, ~ \forall c \in [0,1], f_1,f_2 \in \Rbb,
\end{equation}
where $\bar{c} = 1-c$.  Establishing log-convexity of $g$ is sufficient to establish convexity of $g$ as log-convex functions are convex.  Recall H\"{o}lder's inequality asserts
\begin{equation}
\Ebb[\xsf_1 \xsf_2] \leq \Ebb \left[ \xsf_1^p \right]^{\frac{1}{p}} \Ebb \left[ \xsf_1^q \right]^{\frac{1}{q}}, ~ \forall (p,q) : \frac{1}{p} + \frac{1}{q} = 1.
\end{equation}
We apply H\"{o}lder's inequality twice with $p_1 = p_2 = \frac{1}{c}$ and $q_1 = q_2 = \frac{1}{\bar{c}}$:
\begin{eqnarray}
g(c f_1 + \bar{c} f_2) & = & \Ebb[ \xsf^{-(c f_1 + \bar{c} f_2)}] \Ebb[ \xsf^{-(1-(c f_1 + \bar{c} f_2))} ] \nonumber \\
&=& \Ebb[ \xsf^{-c f_1} \xsf^{- \bar{c} f_2}] \Ebb[ \xsf^{-c \bar{f_1}} \xsf^{-\bar{c}\bar{f_2}} ] \nonumber \\
& \leq & \Ebb \left[ \xsf^{-c f_1 p_1} \right]^{\frac{1}{p_1}} \Ebb \left[ \xsf^{-\bar{c} f_2 q_1} \right]^{\frac{1}{q_1}} \Ebb \left[ \xsf^{-c \bar{f_1} p_2} \right]^{\frac{1}{p_2}} \Ebb \left[ \xsf^{-\bar{c} \bar{f_2} q_2} \right]^{\frac{1}{q_2}} \nonumber \\
&=& \Ebb \left[ \xsf^{-f_1} \right]^c \Ebb \left[ \xsf^{-f_2} \right]^{\bar{c}} \Ebb \left[ \xsf^{- \bar{f_1}} \right]^c \Ebb \left[ \xsf^{-\bar{f_2}} \right]^{\bar{c}} \nonumber \\
& = & g(f_1)^c g(f_2)^{\bar{c}}.
\end{eqnarray}
As $g$ is convex, its global minimum (if it exists) is found by equating the derivative of $g(f)$ with zero and solving for $f$.  The derivative is:
\begin{eqnarray}
g'(f) &=& \frac{\drm}{\drm f} \left\{  \Ebb\left[ \xsf^{-f}\right] \Ebb \left[ \xsf^{-\bar{f}}\right] \right\} \\
&=& \Ebb\left[ \xsf^{-f}\right] \Ebb \left[ \frac{\drm}{\drm f}  \xsf^{-\bar{f}} \right] + \Ebb \left[ \xsf^{-\bar{f}}\right] \Ebb \left[ \frac{\drm}{\drm f} \xsf^{-f} \right] \nonumber \\
&=& \Ebb\left[ \xsf^{-f}\right] \Ebb \left[ (\log \xsf) \xsf^{-\bar{f}} \right] + \Ebb \left[ \xsf^{-\bar{f}}\right] \Ebb \left[ - (\log \xsf) \xsf^{-f} \right] \nonumber
\end{eqnarray}
It is clear that $g'(f) = 0$ for $f=1/2$, establishing the result.
\end{proof}
\begin{remark}
\label{rem:optfpcvspowvar}
{\bf Optimal FPC may incur large power variance.}
Prop.\ \ref{pro:fpcasympoptf} establishes $f=1/2$ is the FPC exponent that optimizes the asymptotic OP and TC (under the assumptions $N=0$ and $\hsf_{0,0} \stackrel{\drm}{=} \hsf_{1,1}$), while Lem.\  \ref{lem:fpcvar} establishes $\mathrm{Var}(\Psf)$ is only finite under Rayleigh fading provided $f < 1/2$.  Thus at $f=1/2$ we optimize the asymptotic OP and TC but the variability of power required is large, and in practice the radio power constraint would limit the ability of nodes with poor fades to achieve \eqref{eq:fpcpdef}.  Note that the asymptotic OP and TC under the above assumptions are symmetric in $f$ around $1/2$, \ie, the OP at $f$ equals the OP at $1-f$ ({\em c.f.}\ \eqref{eq:fpcasympoptfpf}), and in particular the asymptotic OP and TC under $f=0$ equal that under $f=1$.  Although no PC ($f=0$) has the same asymptotic performance to channel inversion ($f=1$), the former is greatly preferred to the latter on the basis of power variance (Lem.\  \ref{lem:fpcvar}).
\end{remark}

Interestingly, recent work modeling and analyzing Qualcomm's FlashLINQ peer-to-peer protocol \cite{BacLi2010}, which has a CSMA-type MAC, found a similar result. Namely, with a different interference and transmission model (namely, SIR-based scheduling), but also using stochastic geometric tools, they found that inverse square root power control is optimal for the FlashLINQ protocol, which they verified with extensive simulations.  This shows that this earlier result on FPC in a simpler Aloha-type MAC was fairly robust to the proposed model, and provides some degree of confidence that the other design-type results given in this chapter will be as well.

Our next result gives an OP LB and TC UB under FPC.  We modify the definition of dominant interferer and interference (Def.\ \ref{def:domintfad}) for FPC.
\begin{definition}
\label{def:domintfpc}
{\bf Dominant and maximum interferers.}
An interferer $i$ in the MPPP $\Phi_{d,\lambda}$ in Def.\ \ref{def:power} is dominant at $o$ under threshold $\tau$, signal fade $\hsf_{0,0}$, and FPC exponent $f$ if its interference contribution is sufficiently strong to cause an outage for the reference Rx at $o$:
\begin{equation}
\frac{\Psf_0 \hsf_{0,0} u^{-\alpha}}{\Psf_i \hsf_{i,0} |\xsf_i|^{-\alpha} + N} < \tau \Leftrightarrow \hsf_{i,i}^{-f} \hsf_{i,0} |\xsf_i|^{-\alpha} > \Ebb[\hsf_{1,1}^{-f}] \left( \frac{\hsf_{0,0}^{1-f}}{\Ebb[\hsf_{0,0}^{-f}]\tau u^{\alpha}} - \frac{N}{P} \right).
\end{equation}
Else $i$ is non-dominant.  The set of dominant and non-dominant interferers at $o$ under $(\tau,\hsf_0)$ is
\begin{equation}
\label{eq:fpclbdomint3}
\hat{\Phi}_{d,\lambda} \equiv \left\{ i \in \Phi_{d,\lambda} : \hsf_{i,i}^{-f} \hsf_{i,0} |\xsf_i|^{-\alpha} > \Ebb[\hsf_{1,1}^{-f}] \left( \frac{\hsf_{0,0}^{1-f}}{\Ebb[\hsf_{0,0}^{-f}]\tau u^{\alpha}} - \frac{N}{P} \right) \right\},
\end{equation}
and $\tilde{\Phi}_{d,\lambda} \equiv \Phi_{d,\lambda} \setminus \hat{\Phi}_{d,\lambda}$.  The dominant and non-dominant interference at $o$ under $(\tau,\hsf_0)$
\begin{equation}
\label{eq:fpclbdomint2}
\hat{\Sisf}^{\alpha,\hsf}_{d,\lambda}(o) \equiv \sum_{i \in \hat{\Phi}_{d,\lambda}} \hsf_{i,i}^{-f} \hsf_{i,0} |\xsf_i|^{-\alpha}, ~
\tilde{\Sisf}^{\alpha,\hsf}_{d,\lambda}(o) \equiv \sum_{i \in \tilde{\Phi}_{d,\lambda}} \hsf_{i,i}^{-f} \hsf_{i,0} |\xsf_i|^{-\alpha}
\end{equation}
are the interference generated by the dominant and non-dominant nodes.  Note $\Sisf^{\alpha,\hsf}_{d,\lambda}(o) = \hat{\Sisf}^{\alpha,\hsf}_{d,\lambda}(o) + \tilde{\Sisf}^{\alpha,\hsf}_{d,\lambda}(o)$.
\end{definition}
\begin{proposition}
\label{pro:fpcoplb}
{\bf The OP LB under FPC} is
\begin{equation}
\label{eq:fpcoplb}
q^{\rm lb}(\lambda) = 1 - \Ebb \left[ \left. \exp \left\{ - a \left( \frac{\hsf_{0,0}^{1-f}}{\Ebb[\hsf_{0,0}^{-f}] \tau u^{\alpha}} - \frac{N}{P} \right)^{-\delta} \lambda \right\} \right| \bar{\Emc}_{0,f} \right] \bar{q}_f(0),
\end{equation}
where $\bar{\Emc}_{0,f}$ and $\bar{q}_f(0)$ are defined in Rem.\  \ref{rem:fpclowout} and
\begin{equation}
a = c_d \Ebb [\hsf_{1,0}^{\delta}] \Ebb [\hsf_{1,1}^{-f \delta}] \Ebb[\hsf_{1,1}^{-f}]^{-\delta}.
\end{equation}
In the case of no noise ($N=0$, $q_f(0)=0$) the LB is:
\begin{equation}
\label{eq:fpcoplbnn}
q^{\rm lb}(\lambda) = 1 - \Ebb \left[ \exp \left\{ - b \lambda \hsf_{0,0}^{-(1-f)\delta}  \right\} \right] ,
\end{equation}
where
\begin{equation}
b = c_d \tau^{\delta} u^d \Ebb[\hsf_{1,0}^{\delta}] \Ebb [\hsf_{1,1}^{-f \delta}] \Ebb[\hsf_{1,1}^{-f}]^{-\delta} \Ebb[\hsf_{0,0}^{-f}]^{\delta}.
\end{equation}
In particular, with no noise the LB is expressible in terms of the MGF of the RV $-\hsf_{0,0}^{-(1-f)\delta}$ at a certain $\theta$:
\begin{equation}
q^{\rm lb}(\lambda) = 1 - \left. \Mmc[  -\hsf_{0,0}^{-(1-f)\delta} ](\theta)\right|_{\theta = b \lambda}.
\end{equation}
\end{proposition}
\begin{proof}
The proof requires only minor adaptation of that of Prop.\ \ref{pro:fadoplb}.
\begin{eqnarray}
q(\lambda) &=& \Pbb(\sinr(o) < \tau) \nonumber \\
&=&  \Pbb\left( \sinr(o) < \tau| \bar{\Emc}_{0,f} \right) \bar{q}_f(0) + 1 (1-\bar{q}_f(0)) \nonumber \\
&=& 1 - \left(1 - \underbrace{\Pbb\left( \sinr(o) < \tau| \bar{\Emc}_{0,f} \right)}_{\Pbb(\cdot)} \right) \bar{q}_f(0)
\end{eqnarray}
where
\begin{equation}
\Pbb(\cdot) = \Ebb \left[ \left. \Pbb\left( \left. \Sisf^{\alpha,\hsf}_{d,\lambda}(o) > \Ebb[\hsf_{1,1}^{-f}] \left( \frac{\hsf_{0,0}^{1-f}}{\Ebb[\hsf_{0,0}^{-f}] \tau u^{\alpha}} - \frac{N}{P} \right) \right|\hsf_{0,0} \right) \right| \bar{\Emc}_{0,f} \right]
\end{equation}
and $\Sisf^{\alpha,\hsf}_{d,\lambda}(o)$ in \eqref{eq:fpcasympi}.  Lower bound in terms of $\hat{\Sisf}^{\alpha,\hsf}_{d,\lambda}(o)$ in \eqref{eq:fpclbdomint2}, and then express the LB outage event in terms of $\hat{\Phi}_{d,\lambda}$ in \eqref{eq:fpclbdomint3}:
\begin{eqnarray}
\Pbb(\cdot) & > & \Ebb \left[ \left. \Pbb\left( \left. \hat{\Sisf}^{\alpha,\hsf}_{d,\lambda}(o) > \Ebb[\hsf_{1,1}^{-f}] \left( \frac{\hsf_{0,0}^{1-f}}{\Ebb[\hsf_{0,0}^{-f}] \tau u^{\alpha}} - \frac{N}{P} \right)  \right|\hsf_{0,0} \right) \right| \bar{\Emc}_{0,f} \right] \nonumber \\
&=& \Ebb \left[ \left. \Pbb\left( \left. \hat{\Phi}_{d,\lambda} \neq \emptyset \right|\hsf_{0,0} \right) \right| \bar{\Emc}_{0,f} \right] \nonumber \\
&=& 1  - \Ebb \left[ \left.  \Pbb\left( \left. \hat{\Phi}_{d,\lambda} = \emptyset \right| \hsf_{0,0} \right) \right| \bar{\Emc}_{0,f} \right] \bar{q}_f(0) \label{eq:midab}
\end{eqnarray}
In what follows we denote $\hat{\hsf} = \hsf_{1,1}$ and $\tilde{\hsf} = \hsf_{1,0}$.
The PPP $\Phi_{d,\lambda}$ is a homogeneous PPP with intensity measure given by Thm.\ \ref{thm:mark} with $\lambda(x) = \lambda$ and a joint PDF on mark pairs $(\hat{\hsf},\tilde{\hsf})$ independent of $\xsf$:
\begin{equation}
f_{\hat{\hsf},\tilde{\hsf}|\xsf}(\hat{h},\tilde{h}|x) = f_{\hat{\hsf}}(\hat{h}) f_{\tilde{\hsf}}(\tilde{h}).
\end{equation}
The key observation is this: given $h_{0,0}$, the probability that $\hat{\Phi}_{d,\lambda}$ is empty equals the void probability for $\Phi_{d,\lambda}$ on the set
\begin{eqnarray}
C_0 &=& \{ (x,\hat{h},\tilde{h}) : \hat{h}^{-f} \tilde{h} |x|^{-\alpha} > w_0\} \nonumber \\
w_0 &=& \Ebb[\hat{h}^{-f}] \left( \frac{h_{0,0}^{1-f}}{\Ebb[\hsf_{0,0}^{-f}] \tau u^{\alpha}} - \frac{N}{P} \right)
\end{eqnarray}
Note in what follows that $\hsf_{0,0}$ is random and hence so is $\Csf_0$ and $\wsf_0$.  Using Prop.\ \ref{pro:voidnonhomo} and simplifying gives:
\begin{eqnarray}
\Pbb\left( \left. \hat{\Phi}_{d,\lambda} = \emptyset \right| \hsf_{0,0} \right) &=&
\Pbb\left( \left. \Phi_{d,\lambda}(\Csf_0) = 0 \right| \hsf_{0,0} \right) \\
&=& \exp \left\{ - \lambda \int_{(x,\hat{h},\tilde{h}) \in \Csf_0} f_{\hat{\hsf}}(\hat{h}) f_{\tilde{\hsf}}(\tilde{h}) \drm x \drm \hat{h} \drm \tilde{h} \right\} \nonumber \\
&=& \exp \left\{ - \lambda \int_{\Rbb^d} \Pbb \left( \left. \bar{\hsf} |x|^{-\alpha} > \wsf_0 \right| \wsf_0 \right) \drm x \right\} \nonumber
\end{eqnarray}
where we have defined the RV $\bar{\hsf} = \tilde{\hsf}/\hat{\hsf}^{f}$.  Now echo the development in the proof of Prop.\ \ref{pro:fadoplb} starting at \eqref{eq:fadlbpf12}, replacing $\hsf$ with $\bar{\hsf}$, yielding:
\begin{eqnarray}
\Pbb\left( \left. \hat{\Phi}_{d,\lambda} = \emptyset \right| \hsf_{0,0} \right) &=& \exp \left\{ - \lambda c_d \wsf_0^{-\delta} \Ebb[\bar{\hsf}^{\delta}] \right\} \nonumber \\
&=& \exp \left\{ - \lambda c_d \wsf_0^{-\delta} \Ebb [\hsf_{1,0}^{\delta}] \Ebb [\hsf_{1,1}^{-f \delta}] \right\}
\end{eqnarray}
where in the last step we have exploited the assumed independence of $\tilde{\hsf}$ and $\hat{\hsf}$.  Substituting this last expression into \eqref{eq:midab} yields \eqref{eq:fpcoplb}.
\end{proof}
Fig.\ \ref{fig:FPC1} shows the asymptotic and LB OP vs.\ the PCE $f$ (all for no noise ($N=0$) and Rayleigh fading on all coefficients).  We also show the no fading asymptotic and LB OP, and observe the OP under FPC is always larger than the OP without fading.  Although threshold scheduling was able to exploit fading and improve the OP/TC to lie above that of no fading, this is seen to not be possible under FPC.  All plots use $d=2,\alpha=4,(\delta=1/2),u=1,\tau=5$.  Further, the asymptotic and LB OP are seen to predict the same optimal PCE $f^* = 1/2$ for small $\lambda$, while for larger $\lambda$ the LB optimal FPC is seen to be $f^* = 0$.  Fig.\ \ref{fig:FPC2} shows the OP LB vs.\ $\lambda$, the TC UB vs.\ $q^*$, as well as the asymptotic OP and TC using the same parameters as Fig.\ \ref{fig:FPC1}.    Similar comments apply: FPC cannot exploit fading to improve performance relative to that of no fading, and the (bound) optimal PCE is $f=1/2$ for small $\lambda$ and $f=0$ for large $\lambda$.

\begin{figure}[!htbp]
\centering
\includegraphics[width=0.45\textwidth]{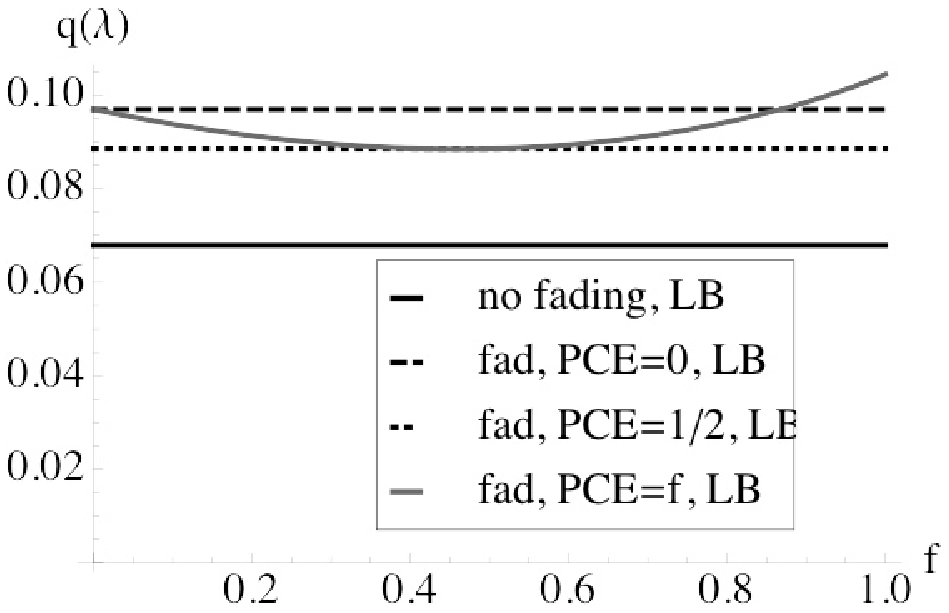}
\includegraphics[width=0.45\textwidth]{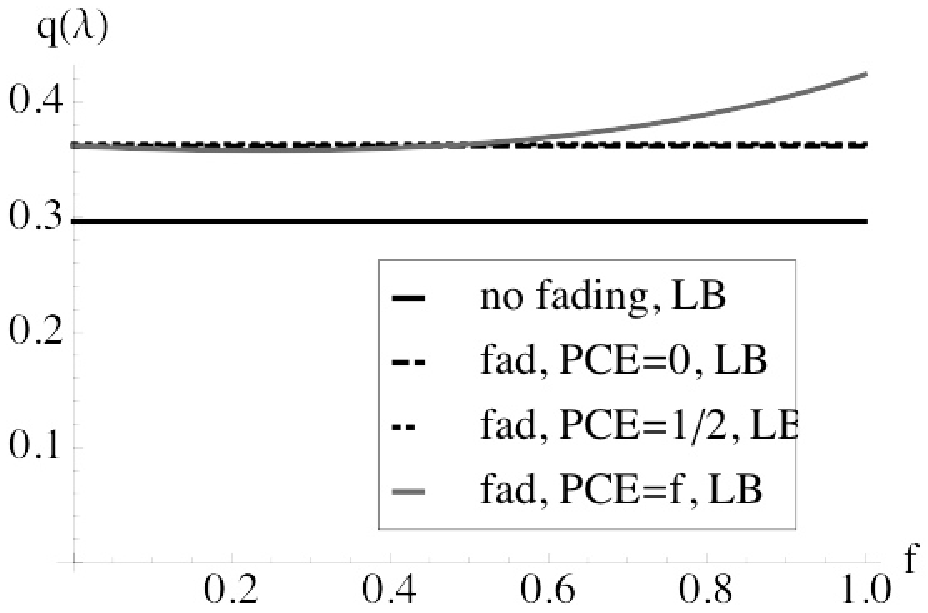}
\includegraphics[width=0.45\textwidth]{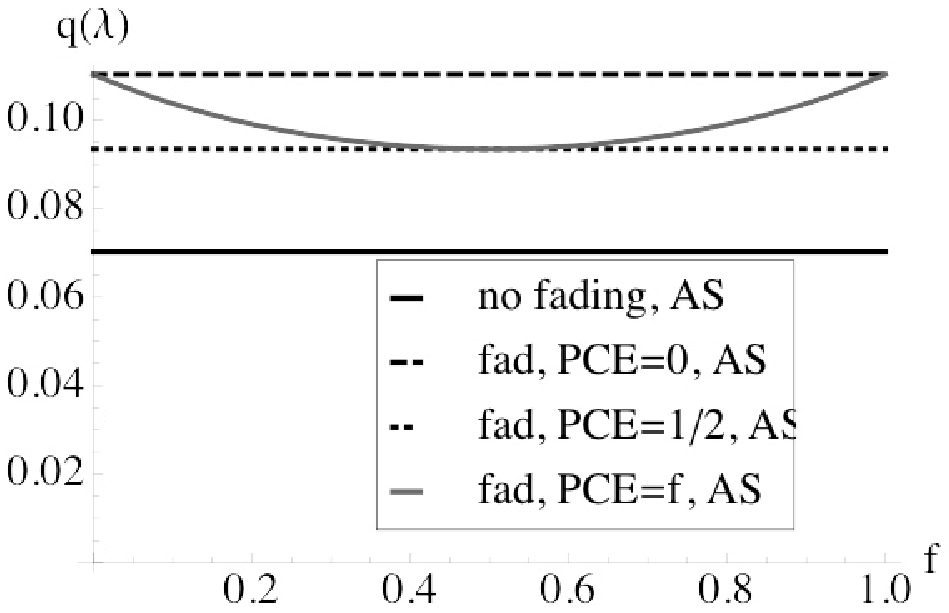}
\includegraphics[width=0.45\textwidth]{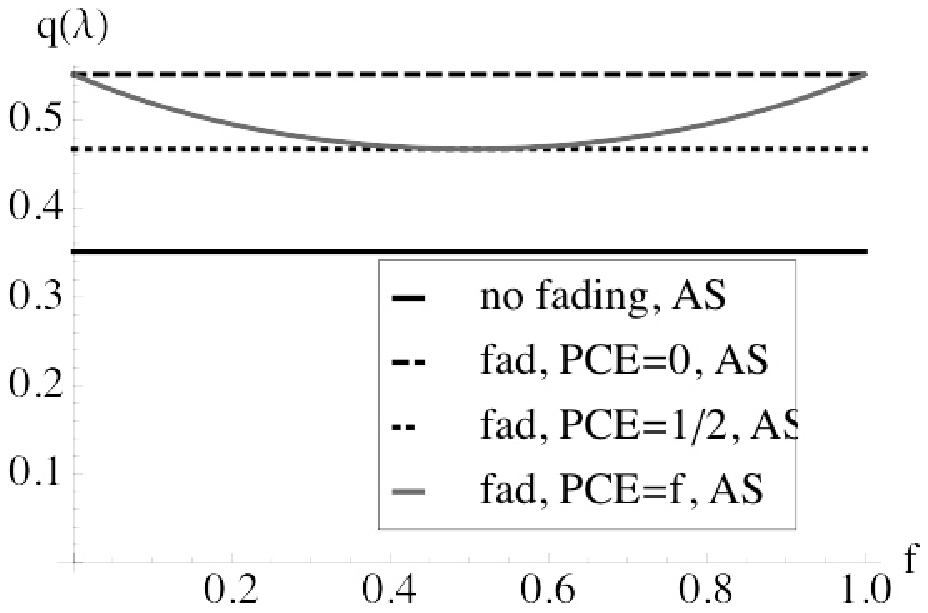}
\caption{The OP $q(\lambda)$ vs.\ the FPC exponent $f$ for $\lambda = 0.01$ (left) and $\lambda = 0.05$ (right).  The top plots are the OP LB, the bottom plots are the asymptotic OP (as $q^* \to 0$).  The OP is shown for no fading, fading without PC ($f=0$), fading with PC at $f=1/2$, and fading vs.\ PCE $f$.  For small $\lambda$ the asymptotic and LB OP agree that the optimal PCE is $f=1/2$, while for larger $\lambda$ the LB optimal PCE is $f = 0$.}
\label{fig:FPC1}
\end{figure}
\vspace{-0.5in}
\begin{figure}[!ht]
\centering
\includegraphics[width=0.45\textwidth]{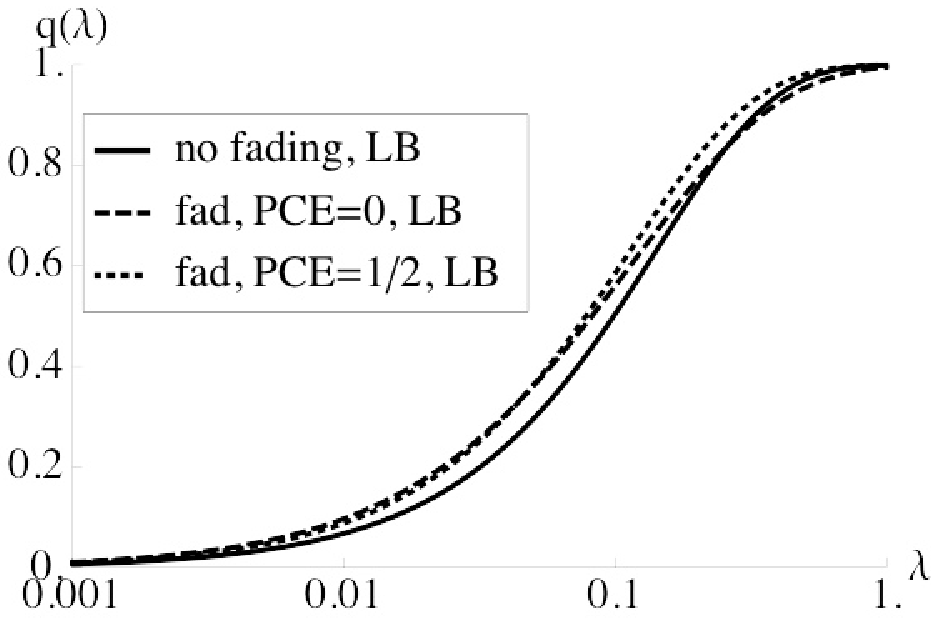}
\includegraphics[width=0.45\textwidth]{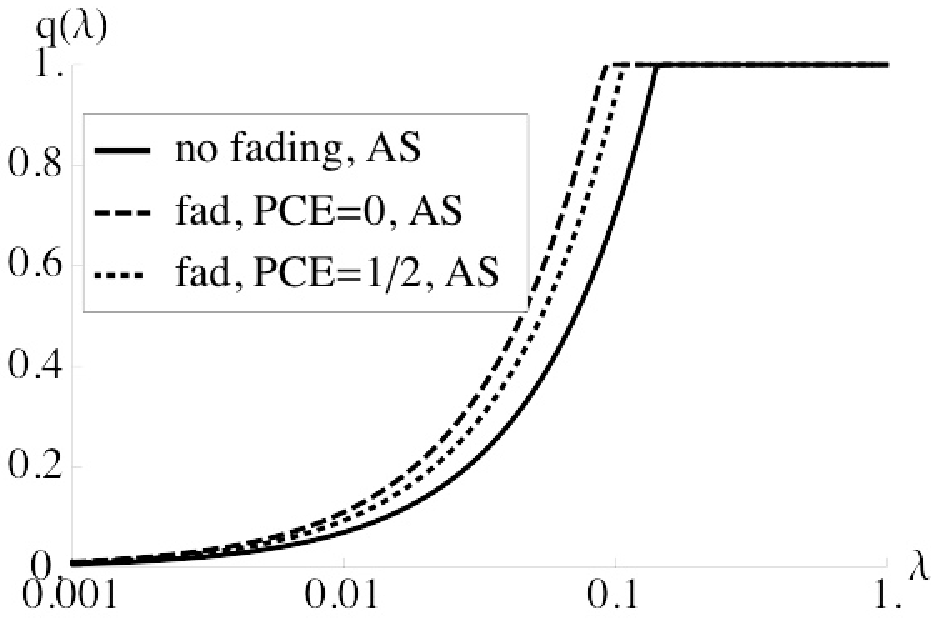}
\includegraphics[width=0.45\textwidth]{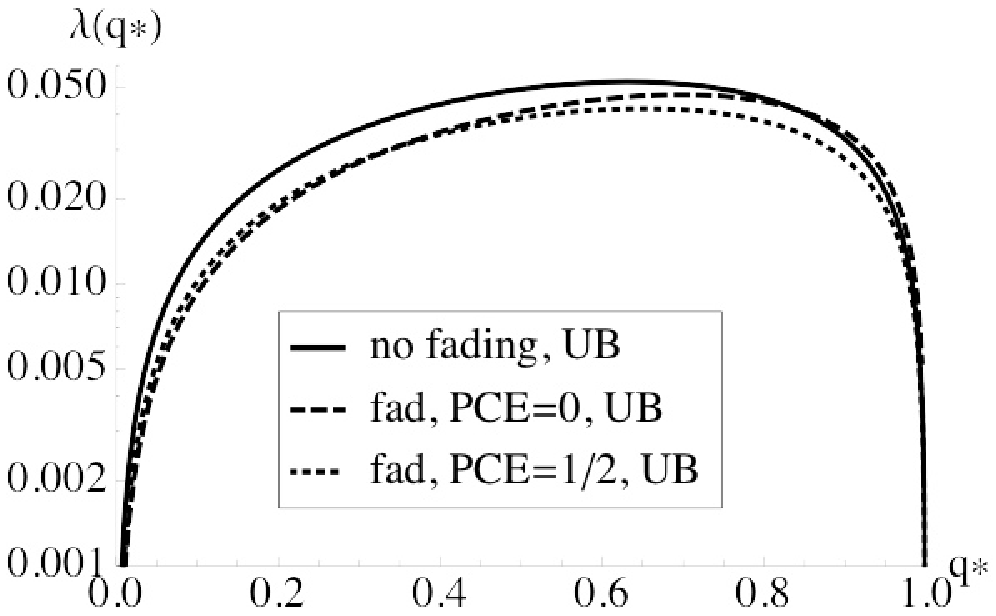}
\includegraphics[width=0.45\textwidth]{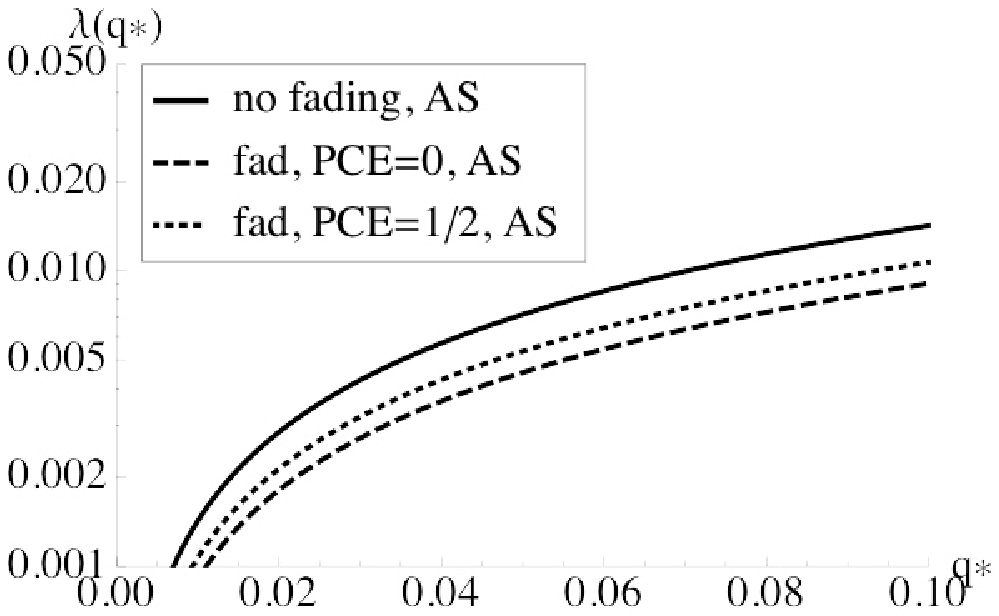}
\caption{The OP $q(\lambda)$ vs.\ $\lambda$, and the TC $\lambda(q^*)$ vs.\ $q^*$.  The LB OP and UB TC are on the left, while the asymptotic OP and TC are on the right.  Curves are shown for no fading, and for fading with PCE $f=0$ and $f=1/2$.  The performance without fading is superior to performance with fading, for both $f=0$ and $f=1/2$.  The bound on OP and TC is seen to be better under $f=1/2$ ($f=0$) for small (large) $\lambda$.}
\label{fig:FPC2}
\end{figure}

%
%
\chapter{Multiple antennas}
\label{cha:MIMO}

\section{MIMO with interference}
\label{sec:MIMOintro}

In this chapter we consider the use of multiple antennas at the transmitter (Tx) and/or receiver (Rx).  Multiple antenna transmission and reception, often broadly referred to as ``MIMO'' (multiple-input multiple-output), has been one of the most extensively studied topics in physical layer communications over the past fifteen years.  It is now a successful commercial technology, having strong support in modern cellular standards, \eg, LTE and WiMAX, and more recent instantiations of WiFi, namely 802.11n and its successors.  Nevertheless, when considering the realities of both WiFi and cellular, not to mention \adhoc networks, one notices that nearly all theoretical work on MIMO ignores the role of interference.

Understandably, the heralded early results on point-to-point MIMO neglected the role of interference or multiple concurrent users.  From those results we learned that MIMO offered the possibility of not only increased reliability through diversity and combining gains, but also held out the promise of linear capacity scaling (at high SNR) with the number of antennas, assuming they were equally balanced at the Tx and Rx \cite{Fos1996,Tel1999}.  This led naturally to considering multiple simultaneous users, for example the downlink of a cellular system where the base station has $\nt$ (Tx) antennas, and each user has $\nr < \nt$ (Rx) antennas. This situation is typically referred to as space division multiple access (SDMA) or multiuser MIMO. Here as well, the linear capacity (now summed over the users) scaling can still be maintained by using a pre-cancellation technique at the Tx (in the downlink), or Rx SIC (uplink) \cite{CaiSha2003,VisJin2003,YuCio2004,WeiSte2006}.

Notably, however, these well-known results all ignore the spurious interference that occurs from ``uncontrolled'' interferers in the network.  Only a very small body of work has considered the effect of interference on MIMO transmission, despite fairly clear warnings that the classic results may not hold for moderate to high levels of interference, and in fact contradictory results may hold instead \cite{CatDri2001,BluWin2002,Blu2003,AndCho2007}.  This is probably because it appears to be quite difficult to analyze MIMO systems (which already have significant randomness to deal with from the random matrix channel) when considering non-Gaussian interference as well.  The aforementioned cautionary works primarily used simulations, as did contemporaneous work on \adhoc networks \cite{RamRed2005,Ram2001,CheGan2006,YeBlu2004}.  Even today, the most promising theoretical MIMO techniques, namely spatial multiplexing (SM) and SDMA, have proved largely disappointing in the field, in large part due to the effect of interference.

One of the most popular uses of the TC framework has been to attempt to better understand MIMO systems that are subject to interference.  The results we present in this chapter are again for the decentralized one-hop \adhoc networks with an Aloha-type MAC, but extensions to other examples of interference-limited networks along these lines seem well within reach, and are discussed more in \S\ref{sec:takeaways}.

\section{Categorizing MIMO in decentralized networks}
\label{sec:model}

There are two main differences that emerge when considering MIMO in a large wireless network, versus an isolated point-to-point or downlink/uplink scenario.  As we just emphasized, there is the random (non-Gaussian) interference aspect, which causes the tradeoffs that emerged from traditional MIMO analysis to no longer necessarily hold.  Second, there is the spatial reuse aspect to consider, since in a decentralized network the key metric is throughput per unit area.  A scheme that gets high throughput for a typical Tx-Rx pair at the expense of poor spatial reuse may result in a poor TC, or any related network-wide throughput metric.  Therefore, in this section we distinguish between ``single stream'' and ``multi-stream'' TC, dividing multiple antenna techniques accordingly.  Of course, multi-stream techniques reduce to single stream as a special case, but for ease of exposition we consider them separately.

\begin{figure}[!htbp]
\centering
\includegraphics[width=\textwidth]{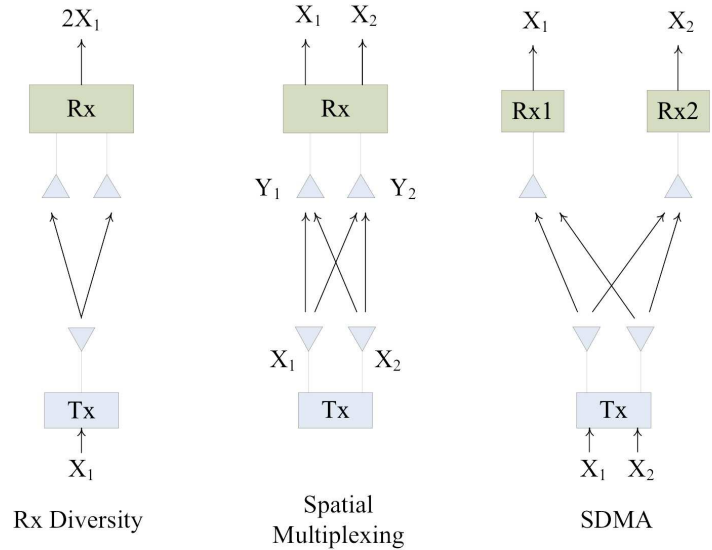}
\caption{Simple configurations of $1 \times 2$ receive diversity, $1 \times 2$ spatial multiplexing (SM), and $K=2$ user space division multiple access (SDMA).}
\label{fig:MIMO}
\end{figure}

\subsection{Single stream techniques}
\label{ssec:modelsinglestream}

We now provide a general model for a MIMO \adhoc network when only a single data stream is \emph{transmitted} by any given node.  Our focus is on linear transmitters and receivers but the approach here could be generalized to nonlinear systems without any major conceptual differences. Among single stream techniques, there are two major competing approaches for how to use the antenna arrays.  The first we refer to broadly as \emph{diversity}, where the spatial degrees of freedom (DoF) are used to strengthen the desired portion of the received signal.  The other is \emph{interference cancellation} (IC), whereby the DoF are used to suppress interference.  To the best of our knowledge, the first paper on this topic was \cite{HunAnd2007} (extended in \cite{HunAnd2008}), with \cite{GovBli2007} shortly thereafter. It has since been extended by numerous other works to many scenarios.

The model used in this section starts with the SINR given in \S\ref{sec:fading}, since these multi-antenna techniques are only relevant when the channels have temporal and/or spatial selectivity.  \emph{In this chapter, we assume that all channels experience iid Rayleigh fading, \ie, the channels are iid complex Gaussian with zero mean and unit variance.}  Further, we assume the dimension $d=2$ throughout this chapter, and that there is no guard zone, \ie, $\epsilon = 0$.  The relevant SINR is still
\begin{equation}
\sinr(o) = \frac{\Ssf(o)}{\Sisf(o)+N/P},
\end{equation}
and the diversity techniques attempt to increase $\Ssf(o)$ (and reduce its variance) while IC reduces $\Sisf(o)$ by removing a portion of one or more interfering nodes.
\begin{definition}
{\bf MIMO single stream SINR.}
\label{def:MIMO-SINR}
For both diversity and IC approaches, when linear Tx/Rx filters are used and the number of Tx and Rx antennas used per node are $\nt$ and $\nr$, respectively; a common model can be adopted for SINR, which is:
\begin{equation}
\sinr(o) = \frac{u^{-\alpha} |\wsf_0^* \Hsf_0 \vsf_0|^2}{\sum_{i \in \Pi_{2,\lambda}} |\xsf_i|^{-\alpha}|\wsf_i^* \Hsf_i \vsf_i|^2  + N/P},
\label{eq:MIMO-SINR}
\end{equation}
where channel fading is captured by $\nr \times \nt$ random matrices $\Hsf_0$ for the reference Tx-Rx pair, $\Hsf_i$ for the channel from interferer $i$ to the reference Rx, and filtering is modeled by $\nt \times 1$ Tx filters $\{\vsf_0,\vsf_1,\ldots\}$ and $\nr \times 1$ Rx filters $\{\wsf_0,\wsf_1,\ldots\}$.  These filters are random in that they are in general chosen as functions of the random channels.
\end{definition}
The relevant design question is how to pick the various $\vsf$ and $\wsf$ to maximize SINR. For single stream MIMO approaches, the definitions of OP and TC developed in the preceding chapters still apply without modification.

\subsubsection*{Diversity}

Diversity techniques pick $\vsf$ and $\wsf$ to maximize the numerator of \eqref{eq:MIMO-SINR}, while ignoring the interference. Although some feedback may be necessary for the design of the Tx filter $\vsf$, this class of techniques has the merit that interference does not need to be estimated, learned, or otherwise acquired; such acquisition can be expensive in terms of overhead and computation, and in many cases (\eg, mobility) not very robust.

An important point is that in this model, the ``diversity'' affects the SINR in two ways.  First, by sending a single stream over multiple effectively uncorrelated channels, variations in the SINR are significantly reduced and the formerly random SINR \emph{hardens} to a constant value as the antenna arrays grow large.  However, the effect of hardening on TC is small, since the TC already considers a spatial throughput average.  Second, the average SINR increases because of \emph{array gain}, which aligns the transmitted energy towards the dominant eigenvalue of the matrix channel.  This latter effect has a more significant effect on the TC.

\subsubsection*{Interference cancellation (IC)}

IC is considered in this chapter for single stream techniques, \ie, we design $\wsf$ to be roughly orthogonal to a subset of nearby interferers, \eg, the $\nr - 1$ strongest interferers. Additional single stream IC gain can in theory be achieved if each Tx is able to learn the channels to its $\nt-1$ closest active ``victim'' receivers, in which case $\vsf$ can be designed to be approximately orthogonal to those channels.  However, this introduces significant overhead and is unlikely to be very robust or practical --- it also complicates the exposition considerably.  For the purposes of this chapter, we limit ourselves to Rx IC when considering single stream transmission.

Three different linear Rx IC approaches are considered in the next section.  Two of them also achieve a diversity gain.  The three techniques are:
\begin{enumerate}
\item Zero-forcing (ZF), where $\wsf$ is chosen to be precisely orthogonal to the channels of the strongest $\nr-1$ interferers. There is no diversity gain.
\item Partial zero-forcing (PZF), where $\wsf$ is chosen to be orthogonal to the $z \leq \nr-1$  strongest interferers, and the remaining $\nr-z$ dimensions are used for Maximal Ratio Combining (MRC).
\item Minimum Mean Square Error (MMSE), achieves both IC and diversity and maximizes the SINR (and hence TC), but the tradeoff between them is not as easy to observe as in PZF.
\end{enumerate}
Although it is impossible to eliminate all interference with a finite number of antennas, for $\nr$ sufficiently large and for low to moderate SINR thresholds $\tau$, all appreciable interference can in theory be removed (since sufficiently weak interferers will not cause outages even in aggregate).

\subsection{Multi-stream models: spatial multiplexing and SDMA}
\label{ssec:modelmultistream}

We would also like to allow the communication of  multiple simultaneous information-bearing streams that originate and/or terminate at a single user in the network. Although with single stream transmissions, multiple Tx-Rx pairs communicate one stream simultaneously, we refer here to the case where a single node transmits (or receives) multiple streams.  Three cases are of interest.
\begin{enumerate}
  \item Spatial Multiplexing (SM): Tx 0 sends $K$ streams to Rx 0.  All other active transmitters follow suit to their respective receivers.
  \item Tx SDMA: Tx 0 sends $K$ streams total, one each to Rx's $0, 1, ..., K-1$. For simplicity and maximum contrast to SM, we limit our attention to this case, but it is straightforward conceptually to extend this to $K_i$ streams sent to Rx $i$, where $\sum K_i \leq \nt$.
  \item Rx SDMA: Similarly, the simplest scenario is where Tx's $0, 1, ..., K-1$ send one stream each to Rx 0.
\end{enumerate}
For such a setup, it is necessary to modify the definitions of OP and TC, given in Ch.\ \ref{cha:int}. Since multiple streams are sent over a matrix channel, different SINR statistics will be seen on different streams, motivating a per-stream outage constraint $q^*_k \in (0,1)$.
\begin{definition}
{\bf Multi-stream OP and OCD.}
\label{def:MIMO-K-SINR}
Fix a common maximum permissible per-stream OP for each stream $q_k^* = q^* \in (0,1)$.  The maximum spatial intensity subject to an OP of $q^*$ is the same as in the single stream case, but must hold for all $K$ streams: $q^{-1}(q^*)$.  For the multi-stream case, we call $q^{-1}(q^*)$ the \emph{optimal contention density} (OCD).
\end{definition}
For such a definition, one stream will become the limiting factor in the OCD, and hence the TC. Because there are now potentially multiple streams per transmission, the OCD and TC are no longer interchangeable, leading to the following generalized definition for TC.
\begin{definition}
{\bf Multi-stream TC.}
\label{def:multistreamtc}
The TC with $K \leq \nt$ data streams per transmission is the optimal spatial density of concurrent multi-stream transmissions $q^{-1}(q^*)$ per unit area allowed subject to an OP constraint $q^*$, \ie,
\begin{equation}
\Cmc(q^*) = K q^{-1}(q^*) (1-q^*).
\label{eq:multicapa}
\end{equation}
\end{definition}
This definition generalizes our previous TC definition in Def.\ \ref{def:tc}, which was simply for $K=1$.

\section{Single stream MIMO TC results}
\label{sec:singlestream}

We now summarize a subset of instructive results on single stream TC with multiple antennas.  For the balance of the chapter we neglect noise when it makes the exposition more complicated, and include it when the exposition is not made significantly more difficult.  It is generally reasonable to neglect noise since in attempting to maximize TC, we are by definition packing transmissions as tightly as possible, rendering thermal noise quite inconsequential compared to the background interference level, even after strong interferers are cancelled.

\subsection{Diversity}
\label{ssec:diversity}

We first consider diversity techniques, which attempt only to improve the desired signal, while ignoring interference.
\begin{definition}
\label{def:optdiversityfilter}
{\bf Single stream MIMO optimal linear diversity filters.} 
In a single stream $\nt \times \nr$ MIMO channel $\Hsf_0$, the optimum linear diversity Tx filter $\vsf_0$ and Rx filter $\wsf_0$ are the left and right singular vectors corresponding to the maximum eigenvalue of $\Hsf_0$, which is denoted as $\phi_0$.
\end{definition}
Because these filters ignore interference, they do not maximize SINR or TC.  They do however maximize the strength of the desired received signal since the optimum eigenmode $\phi_0$ is used for transmission. Any other Tx filter(s) would not concentrate the entire energy on $\phi_0$, and hence have a strictly inferior value of $\Ssf$, the SIR numerator.  The main challenge to deriving OP and TC for the diversity case, relative to the single antenna results of Ch.\ \ref{cha:modenh}, is that now multiple iid channels (both desired and interference) are combined at the Rx by the Rx filter. We begin with the simpler special case of $1 \times \nr$ MRC, where only the filter $\wsf_0$ is required, before preceding to general $\nt \times \nr$ eigenbeamforming.  We adapt our channel notation to the assumed $1 \times \nr$ case in the natural way --- the channel matrices $\Hsf_0,\Hsf_i$ are denoted by the $\nr$-vectors $\hsf_0,\hsf_i$.
\begin{theorem}
\label{thm:MRC}
The {\bf Maximal Ratio Combiner (MRC)} for a vector channel $\hsf_0$ which maximizes the desired received energy is $\wsf_0 =
\hsf_0^*/||\hsf_0||$.  For the $\nr$-branch MRC with only Gaussian noise, the post-combining SNR is the sum of the per-branch SNRs.
\end{theorem}
This well-known result (\eg, \cite{Gol2005}) means that one simply weights each branch in proportion to the gain on that branch.  Since the interference is not Gaussian, the branch SIRs cannot simply be summed.  The post-combining SIR in our network model is
\begin{equation}
\sir(o) = \frac{u^{-\alpha} |\hsf_0^* \hsf_0|^2}{\sum_{i \in \Pi_{2,\lambda}} |\xsf_i|^{-\alpha}|\hsf_0^* \hsf_i|^2} = \frac{u^{-\alpha} \|\hsf_0\|^2}{\sum_{i \in \Pi_{2,\lambda}} |\xsf_i|^{-\alpha} \left|\frac{\hsf_0^*}{\|\hsf_0\|} \hsf_i \right|^2}.
\label{eq:MRC-SIR}
\end{equation}
\begin{remark}
\label{rem:intdistunchanged}
{\bf Interference distribution is unchanged.}
Since $\hsf_0$ and $\hsf_i$ are iid and $\frac{\hsf_0^*}{\|\hsf_0\|}$ is unit norm, $\frac{\hsf_0^*}{\|\hsf_0\|} \hsf_i$ is simply a linear combination of complex Gaussian RVs.  As proven in \cite{ShaHai2000}, a linear combination of Gaussian RVs is again Gaussian, and due to the normalization of $\hsf_0$ the mean and variance are maintained as 0 and 1, respectively.  Therefore, $|\frac{\hsf_0^*}{\|\hsf_0\|} \hsf_i|^2$ is exponentially distributed, and the interference is unchanged from the single antenna case in Cor.\ \ref{cor:optcrayfadall}.
\end{remark}
As a consequence, the SIR can simply be expressed as
\begin{equation}
\sir(o) = \frac{u^{-\alpha} \|\hsf_0\|^2}{\Sisf^{\alpha,0}_{2,\lambda}},
\end{equation}
where $\Sisf^{\alpha,0}_{2,\lambda}$ is the standard aggregate single antenna interference.
\begin{remark}{\bf Signal distribution.}
\label{rem:DivSigDis}
Since $u$ is deterministic and $\Ssf_0 = \|\hsf_0\|^2$ is a $\chi^2$ RV with $2\nr$ DoF, the numerator is $\chi^2$ distributed as well.  Clearly for $\nr=1$ antenna this reduces to an exponential distribution and the result in Cor.\ \ref{cor:optcrayfadall} holds.  For $\nr \geq 1$ the CCDF of $\Ssf_0$ is $\bar{F}_{\Ssf_0}(x)=\erm^{-x}\sum_{k=0}^{\nr-1}\frac{x^k}{k!}$ and the associated PDF is $f_{\Ssf_0}(x)=\erm^{-x}\frac{x^{\nr-1}}{(\nr-1)!}$.
\end{remark}
We now generalize \S\ref{sec:fading} to $1 \times \nr$ with MRC to give the asymptotic ($\lambda \to 0$) OP.
\begin{proposition}
\label{pro:OP-MRC}
{\bf OP with MRC} \cite{HunAnd2008}.
The OP for $1\times \nr$ MRC is
\begin{equation}
q(\lambda) = \lambda \tau^{\delta} u^2 C_{\alpha} \left( 1 + \sum_{k=1}^{\nr-1} \frac{1}{k!} \prod_{l=0}^{k-1}(l- \delta)\right) + \Theta(\kappa^2), ~ \lambda \to 0,
\end{equation}
where $\kappa = \lambda \tau^{\delta} u^2 C_{\alpha}$ and $C_{\alpha} = \pi^2 \delta \csc(\pi \delta)$.
\end{proposition}
\begin{proof}
The proof requires significant algebra and is given in \cite{HunAnd2008}.  A brief version with key steps is as follows.  We begin as in the proof of Prop.\ \ref{pro:optcrayfadsig}.  Denote $\Sisf^{\alpha,0}_{2,\lambda}$ by $\Sisf$ for this proof with PDF $f_{\Sisf}(t)$.  Then
\begin{eqnarray}
1 - q(\lambda) &=& \Pbb (\Ssf_0 \geq \tau u^{\alpha} \Sisf) = \int_0^{\infty}\bar{F}_{\Ssf_0}(st)f_{\Sisf}(t)\drm t \nonumber \\
& = &\int_0^{\infty}\left(\erm^{-st}\sum_k \frac{1}{k!}(st)^k\right)f_{\Sisf}(t)\drm t = \sum_k \frac{1}{k!} (-s)^k \frac{\drm^k}{\drm s^k} \mathcal{L}[\Sisf](s) \nonumber \\
\label{eq:MRCkeystep}
\end{eqnarray}
where the second equality in \eqref{eq:MRCkeystep} uses the LT property $t^n f(t)\longleftrightarrow (-1)^n\frac{\mathrm{d}^n}{\mathrm{d}s^n}\mathcal{L}[f(t)](s)$ and $s = \tau u^{\alpha}$.  For Rayleigh fading, as in \S\ref{ssec:fadexact}, $\mathcal{L}[\Sisf](s) = \exp(-\lambda C_{\alpha} s^{\delta})$. Forming the first order Taylor expansion for the $j$th derivative around $\kappa=\lambda s^{\delta} C_\alpha=0$, note that any term with $\kappa^k$ for $k>1$ is $o(\kappa)$ and can be discarded.  Eventually it can be shown that
\begin{eqnarray}
\frac{1}{k!} \left(-s\right)^k \frac{\mathrm{d}^k}{\mathrm{d}s^k} \mathcal{L}[\Sisf](s)
&=& -\frac{1}{k!} \lambda s^{\delta}C_\alpha \prod_{l=0}^{k-1}(l-\delta)+\Theta(\kappa^2)
\end{eqnarray}
as $\lambda \to 0$, which after straightforward manipulation, yields the result.
\end{proof}
Note that Prop.\ \ref{pro:OP-MRC} is equivalent to Cor.\ \ref{cor:optcrayfadall} for $\nr=1$, no noise, and given a Taylor series expansion around $\kappa$.
\begin{proposition}
\label{pro:TC-MRC}
{\bf TC with MRC} \cite{HunAnd2008}.
When each Tx transmits on a single antenna and each Rx performs MRC with $\nr$ antennas; or equivalently each Tx performs MRT with $\nt = \nr$ antennas and each Rx uses a single antenna; the asymptotic (as $q^* \to 0$) TC under iid Rayleigh fading for all links, with no noise, is
\begin{equation}
\lambda(q^*)=\frac{q^*}{C_{\alpha}\tau^{\delta}u^2 \left( 1 + \sum_{k=1}^{\nr-1} \frac{1}{k!} \prod_{l=0}^{k-1}(l- \delta)\right)} + \Theta(q^*)^2, ~ q^* \to 0.
\label{eq:TC-MRC}
\end{equation}
Further,
$\lambda(q^*)$ is $\Theta(\nr^{\delta})$ (as $\nr \to \infty$) and ignoring $\Theta(q^*)^2$ terms can be bounded as
\begin{equation}
\label{eq:boundsmrc}
1 \leq
\frac{C_\alpha \tau^{\delta} u^2}{\nr^{\delta}q^*} \lambda(q^*) \leq \Gamma(1-\delta).
\end{equation}
\end{proposition}
The ``exact'' result \eqref{eq:TC-MRC} follows immediately by solving the result of Prop.\ \ref{pro:OP-MRC} for $\lambda$.  The bounds are more involved to show, but it can be shown that $\left( 1 + \sum_{k=1}^{\nr-1} \frac{1}{k!} \prod_{l=0}^{k-1}(l- \delta) \right)$ can be bounded by $\nr{^\delta}\Gamma(1-\delta)$ (lower) and $\nr{^\delta}$ (upper) \cite{HunAnd2008}. The analysis for $\nr$ Tx antennas and 1 Rx antenna follows precisely the same development. However, we do not account for the cost of acquiring (non-causally) the necessary channel state information $\hsf_0$ at the Tx.
\begin{remark}
\label{rem:MRCsublineargain}
{\bf Sublinear gain of MRC in $\nr$.}
The $\nr$ antennas utilized in MRC provides a gain of $\nr^{\delta} = \nr^{\frac{2}{\alpha}}$ which is sublinear since $\alpha > 2$.  The same scaling holds if noise is included \cite{HunAnd2008}.  Interestingly, the gain is larger for small path loss exponents which may seem counter-intuitive since it would seem that diversity is increasingly desirable when the signal propagation is poor.
\end{remark}
The results on MRC can be extended to the $\nt \times \nr$ eigenbeamforming, where the Tx uses what is sometimes called Maximal Ratio Transmission (MRT), and the Rx uses MRC.  The Tx and Rx filters simply put all energy onto the dominant eigenvalue of the channel matrix $\hsf_0$, \ie, $\vsf_0$ and $\wsf_0$ are set equal to the input and output singular vectors of $\hsf_0$ corresponding to the maximum singular value of $\hsf_0$.
\begin{proposition}
\label{pro:OCDeigenbeamforming}
{\bf OCD of $\nt \times \nr$ eigenbeamforming} \cite{HunAnd2008}.
The asymptotic (as $q^* \to 0$) OCD can be bounded (neglecting $\Theta(q^*)^2$ terms) as
\begin{equation}
\frac{ (\max\{\nt,\nr\})^\delta q^*}{C_\alpha u^2 \tau^\delta} \leq q^{-1}(q^*) \leq \frac{ \Gamma(1-\delta) (\nt \nr)^\delta q^*}{C_\alpha u^2 \tau^\delta}.
\end{equation}
\end{proposition}
\begin{remark}
\label{rem:eigen}
{\bf TC scaling with $\nt \times \nr$ eigenbeamforming.} 
The scaling of $\max\{\nt,\nr\}$ observed in the LB follows immediately from Prop.\ \ref{pro:TC-MRC}, since simply using MRT at the Tx with $\nr = 1$ gives $\nt^\delta$ scaling, and similarly using MRC at the Rx with $\nt = 1$ gives $\nr^\delta$.  Thus the scaling with $\nt \times \nr$ must be at least the larger of those two, since the smaller array must be at least as good as using a single antenna. The UB indicates the possibility of superlinear TC scaling for $\nr = \nt$ and $\alpha < 4$.  Although a tighter UB is not available, we conjecture based on other results in random matrix theory and simulations that the LB is more accurate and the scaling is $\max\{\nt,\nr\}$.
\end{remark}

\subsection{Interference cancellation (IC)}
\label{ssec:MIMO-IC}

One can instead design the Rx beamforming filter to suppress interference from a subset of the nearby interfering nodes, rather than for enhancing the desired signal strength. As in the last section, we assume $\nt =1$.  To illuminate the tradeoff between these two opposing beamforming design philosophies, we adopt a general but suboptimal structure for the Rx beamformer $\wsf_0$ which we term \emph{partial zero-forcing} (PZF).  We briefly visit the optimum but less illustrative MMSE Rx towards the end of this section.
\begin{definition}
\label{def:PZF}
{\bf Partial zero forcing (PZF) Rx.}
The PZF-$z$ Rx uses a beamforming vector $\wsf_0$ that is orthogonal to the channel vectors of the $z$ strongest interferers, where $z\leq \nr-1$.  That is, $\hsf_i \perp \wsf_0 ~ i = 1, 2, \ldots z.$  Furthermore, $\|\wsf_0\|^2 =1$ and the $\nr-z$ remaining DoF are used to maximize the desired received power.  Formally, if the columns of the $\nr \times (\nr - z)$ matrix $\Qsf$ form an orthonormal basis for the nullspace of $(\hsf_1,\ldots, \hsf_z)$, which can be found by performing a full QR decomposition of matrix $[\hsf_1 \cdots \hsf_z]$, then the Rx filter is chosen as:
\begin{eqnarray}
\wsf_0 = \frac{ \Qsf^* \hsf_0}{ \| \Qsf^* \hsf_0\| }.
\end{eqnarray}
\end{definition}
Note that if $z=0$, $\Qsf= \bI$ and $\wsf_0$ is the MRC beamformer of Thm.\ \ref{thm:MRC}.  If $z=\nr-1$ then we have conventional (``full'') zero forcing (ZF); note that ``full'' ZF is only on the closest $z$ interferers and the rest of the interferers are treated as background noise.
\begin{proposition}
\label{pro:pzfsinr}
{\bf PZF SINR.}
The SINR for PZF-$z$ is
\begin{eqnarray}
\sinr^{\textrm{pzf}-z}(o) = \frac{ \Ssf_0 } {u^{\alpha} \sum_{i =z+1}^{\infty} |\xsf_i|^{-\alpha} \Hsf_i + \snr^{-1}} 
\label{eq:PZF-SINR}
\end{eqnarray}
where $\Ssf_0$ is $\chi^2_{2(\nr-z)}$, the $\Hsf_i \equiv |\wsf_0^* \hsf_i|^2$ terms are iid unit-mean exponential RVs and also independent of $\Ssf_0$, the quantities $|\xsf_{z+1}|^2, |\xsf_{z+2}|^2,
\ldots$ are the $z+1, z+2, \ldots$ ordered points of a $1$-dim.\ PPP (\ie, the closest) with
intensity $\pi \lambda$, and the ordered points are independent of the signal and
interference terms. Finally, $\snr \equiv Pu^{-\alpha}/N$ as in Def.\ \ref{def:xisnr}.
\end{proposition}
We note that $\Hsf_1,\ldots,\Hsf_z$ terms are not in the expression because those interfering nodes have been (perfectly) cancelled.  Unlike in the MRC case, we will need to concern ourselves with the interference distribution, and so define the aggregate interference power for PZF-$z$ as:
\begin{eqnarray}
\label{eq:byeeto}
\Sisf_z \equiv u^{\alpha} \sum_{i = z+1}^{\infty} |\xsf_i|^{-\alpha} \Hsf_i .
\end{eqnarray}
To understand how PZF-$z$ (and the choice of $z$) affects the TC, it is sufficient to consider upper and lower bounds, since we shall see that the scaling in $\nr$ is the same for both bounds, and also for the MMSE Rx.
\begin{proposition}
\label{pro:oppzfz}
{\bf OP for PZF-$z$} \cite{JinAnd2011}.
The OP with PZF-$z$ has UB:
\begin{equation}
q^{\textrm{pzf}-z} (\lambda) \leq \frac{ \tau \left(
\left( \pi u^2 \lambda \right)^{\frac{\alpha}{2}} \left(\frac{\alpha}{2} - 1\right)^{-1} \left( z
- \left\lceil \frac{\alpha}{2} \right\rceil
\right)^{1-\frac{\alpha}{2}}  + \frac{1}{\snr} \right) } {\nr - z -1}
\end{equation}
for $\left\lceil \frac{\alpha}{2} \right\rceil < z < \nr - 1$.
\end{proposition}
\begin{proof}
First, rewrite the OP as the tail probability of the RV $1/\sinr$ and then apply Markov's inequality as follows:
\begin{eqnarray}
q^{\textrm{pzf}-z} (\lambda)
&=& \Pbb \left( \frac{1}{\sinr^{\textrm{pzf}-z}} \geq \frac{1}{\tau} \right)
 \stackrel{(a)}{\leq}  \tau  \cdot \Ebb \left[ \frac{1}{\sinr^{\textrm{pzf}-z}} \right] \nonumber \\
& \stackrel{(b)}{=} & \tau \cdot \Ebb \left[ \Sisf_z + \frac{1}{\snr}\right] \Ebb \left[ \frac{1}{\Ssf_0} \right] \label{eq:PZF-OP1},
\end{eqnarray}
where (a) is due to Markov's inequality, (b) is from \eqref{eq:PZF-SINR} and the independence of $\Ssf_0$ and $\Hsf_i$.  We note that the first expectation term $\Ebb \left[ \Sisf_z + \frac{1}{\snr} \right] = \Ebb[\Sisf_z] + \frac{1}{\snr}$ and corresponds to the effect of IC, whereas $\Ebb \left[ \frac{1}{\Ssf_0} \right] = \frac{1}{\nr- z -1}$ since $\Ssf_0$ is $\chi^2_{2(\nr-z)}$, and this term corresponds to the signal power boost from the remaining DoF.  The remaining task is to find $\Ebb[\Sisf_z]$, and an upper bound on it (suppressing the leading $u^{\alpha}$ term) can be found as
\begin{equation}
\Ebb \left[ \sum_{i=z+1}^{\infty} |\xsf_i|^{-\alpha} \Hsf_i \right] =  \sum_{i=z+1}^{\infty} \Ebb \left[ |\xsf_i|^{-\alpha} \Hsf_i \right] =  \sum_{i=z+1}^{\infty} \Ebb \left[ |\xsf_i|^{-\alpha} \right],
\end{equation}
from the independence of $|\xsf_i|$ and $\Hsf_i$ and unit mean fading. Because $|\xsf_1|^2, |\xsf_2|^2, \ldots$ are a $1$-dim.\ PPP with intensity $\pi \lambda$, RV $\pi \lambda |\xsf_i|^2$ is $\chi^2_{2i}$ and thus has PDF $f_{\xsf}(x) = \frac{\erm^{-x} x^{i-1} }{(i-1)!}$ (see Rem.\ \ref{rem:DivSigDis}).  Therefore,
\begin{equation}
\Ebb \left[ \left(|\xsf_i|^2 \right)^{-\alpha/2} \right]
= \left( \pi \lambda \right)^{\frac{\alpha}{2}} \frac{ \Gamma
\left(i- \frac{\alpha}{2} \right)  } {\Gamma(i)}.
\end{equation}
This quantity is finite only for $i > \frac{\alpha}{2}$, and thus
the expected power from the nearest uncancelled interferer is finite
only if $z +1 > \frac{\alpha}{2}$.  As in \cite{HuaAndSub}, using Kershaw's inequality,
\begin{equation}
\frac{ \Gamma \left(i- \frac{\alpha}{2} \right)  } {\Gamma(i)} <
\left(i - \left\lceil \frac{\alpha}{2} \right\rceil
\right)^{-\frac{\alpha}{2}}
\end{equation}
where $\lceil {\cdot} \rceil$ is the ceiling function and $i > \left\lceil \frac{\alpha}{2} \right\rceil$. Therefore
\begin{eqnarray}
\sum_{i = z+1}^{\infty} \frac{ \Gamma \left(i- \frac{\alpha}{2}\right)} {\Gamma(i)} &<& \sum_{i = z+1}^{\infty} \left(i -\left\lceil \frac{\alpha}{2} \right\rceil\right)^{-\frac{\alpha}{2}} \nonumber \\
& \leq & \int_z^{\infty} \left(x - \left\lceil \frac{\alpha}{2} \right\rceil \right)^{-\frac{\alpha}{2}} \drm x \nonumber \\
&=& \left(\frac{\alpha}{2} - 1 \right)^{-1} \left( z - \left\lceil
\frac{\alpha}{2} \right\rceil \right)^{1-\frac{\alpha}{2}},
\end{eqnarray}
where the inequality in the second line holds because $x^{-\frac{\alpha}{2}}$ is a decreasing function.  Inserting this expression for $\Ebb[\Sisf_z]$ along with $\Ebb \left[ \frac{1}{\Ssf_0} \right]$ into \eqref{eq:PZF-OP1} gives the desired result.
\end{proof}
This result for OP is easy to invert for $\lambda$ which yields our next proposition, and a useful interpretation of it.
\begin{proposition}
\label{pro:tcpzflb}
{\bf TC for PZF LB.} 
The TC of PZF is lower bounded by:
\begin{equation}
\lambda(q^*)^{\textrm{pzf}-z} \geq   \left( \frac{q^*}{\tau}\right)^{\frac{2}{\alpha}}  \frac{\left(\frac{\alpha}{2} - 1\right)^{\frac{2}{\alpha}} }{\pi u^2 } \left(\nr - z - 1 - \frac{\tau}{q^* ~ \snr} \right)^{\frac{2}{\alpha}} \left( z - \left\lceil \frac{\alpha}{2} \right\rceil \right)^{1-\frac{2}{\alpha}}
\label{eq:PZF-LB}
\end{equation}
for any $z$ satisfying $\left\lceil \frac{\alpha}{2} \right\rceil < z < \nr - 1 - \frac{\tau}{q^* ~ \snr}$. Furthermore, if $z=\theta \nr$ for $0 < \theta < 1$, then $\lambda(q^*)^{\textrm{pzf}-z} = \Omc(\nr)$ (as $\nr \to \infty$), and the optimum value of $\theta$ is $\theta^* = 1 - \frac{2}{\alpha}$.
\end{proposition}
The linear scaling of PZF can be observed by simply plugging in $z=\theta \nr$ and allowing $\nr$ to grow large. Similarly, the optimum value of $\theta$ can be found taking the derivative of \eqref{eq:PZF-LB} w.r.t. $\theta$ after plugging in $z = \theta \nr$, setting equal to zero, and solving for $\theta$.
\begin{remark}
\label{rem:pzflinearscaling}
{\bf PZF linear scaling in $\nr$.} 
The linear scaling can be interpreted as adaptively combining the MRC portion of the array (to get $\nr^{\frac{2}{\alpha}}$) with the IC portion (to get $\nr^{1-\frac{2}{\alpha}}$), for a total scaling of $\nr$.  In particular, the amount used for each must grow with $\nr$, any fixed value of $z$ does not achieve linear scaling in $\nr$.
\end{remark}
To be certain that PZF cannot achieve superlinear TC scaling in $\nr$, we turn our attention to appropriate upper bounds.  One can consider two different approaches.  The first and most direct is simply to lower bound the PZF OP and invert it to observe a TC UB.  The second is to consider the MMSE Rx, which is by definition strictly better than PZF, and hence upper bounds its TC.  For completeness, and because the proof techniques are very similar for each, we provide both upper bounds.  First, however, we must define the MMSE Rx beamformer.
\begin{definition}
\label{def:mmserxfilter}
The {\bf MMSE Rx filter} is given as
\begin{eqnarray}
\label{eq-mmse_filter}
\wsf_0 = \frac{ \Lisf^{-1} \hsf_0 }{\| \Lisf^{-1} \hsf_0 \|}.
\end{eqnarray}
where $\Lisf$ is the random spatial covariance of the interference plus noise and can be expressed as
\begin{equation}
\Lisf \equiv \frac{1}{\snr} \Ibf + u^{\alpha} \sum_{i \in \Pi_{2,\lambda}} |\xsf_i|^{-\alpha} \hsf_i \hsf_i^*,
\label{eq:SINR-MMSE}
\end{equation}
and $\Ibf$ is the $\nr \times \nr$ identity matrix.
\end{definition}
\begin{remark}
\label{rem:mmserxfilter}
{\bf MMSE Rx filter.} 
The MMSE filter maximizes SINR, which is $\sinr = \hsf_0^* \Lisf^{-1} \hsf_0$, but is not easy to express in a way amenable to analysis. Furthermore, it should be noted that although it appears superficially from \eqref{eq:SINR-MMSE} that a great deal of information (namely the distances and fading values of every single interferer) is needed to compute $\Lisf$, in fact this covariance matrix can be estimated over time and is generally more robust to imperfections than the ZF matrix computation, which requires the BF vector $\wsf_0$ to be precisely orthogonal to each of the $z$ interfering channels.
\end{remark}
With this background on the MMSE filter (see \cite{GovBli2007,JinAnd2011,LouMcK2011,AliCar2010} for more detailed discussion), we can now state the two TC upper bounds in the following proposition.
\begin{proposition}
\label{pro:MMSE-UB}
{\bf TC UB for PZF and MMSE} \cite{JinAnd2011}.
The TC for an MMSE Rx with $\nr$ antennas at high SNR is upper bounded by
\begin{eqnarray}
\lambda^{\textrm{mmse}}(q^*) \leq \frac{2 \nr + 1 + \frac{\alpha}{2}} {\pi u^2 \tau^{\frac{2}{\alpha}} (1-q^*)^{\frac{2}{\alpha}}},
\end{eqnarray}
while for PZF the corresponding bound with $l$ uncancelled interferers is
\begin{equation} \label{pzf_density_upper}
\lambda^{\textrm{pzf}-z}(q^*) \leq \frac{z + l + \frac{\alpha}{2}}{\pi u^2 \tau^{\frac{2}{\alpha}} (1-q^*)^{\frac{2}{\alpha}}} \left( \frac{\nr -z}{l-1} \right)^{\frac{2}{\alpha}}.
\end{equation}
The PZF upper bound holds for any $0 \leq z \leq \nr -1$ and any integer $l \geq 2$.  These bounds both scale as $\Omc(\nr)$ (as $\nr \to \infty$).
\end{proposition}
\begin{proof}
The complete proof of the MMSE TC result is given in \cite{JinAnd2011}. The first step is showing that an UB on \emph{success} probability (neglecting noise) can be written as
\begin{equation}
1 - q^{\textrm{mmse}} (\lambda) \leq  \Pbb \left( \frac{ u^{-\alpha}
\| \hsf_0 \|^2 }{ \sum_{i=\nr}^{\infty}  |\xsf_i|^{-\alpha} \Hsf_i} \geq \tau \right),
\end{equation}
which follows from an OP lower bound for MMSE receivers developed in \cite{GaoSmi1998}.  Using Markov's inequality and a similar approach to the proof of Prop.\ \ref{pro:tcpzflb} -- namely using Kershaw's inequality, dropping various terms (while preserving the bound), and exploiting the independence of various terms -- the desired result can be achieved. The proof of the PZF UB follows the same method, starting with the PZF SINR expression.  Both results scale linearly with $\nr$ which can be observed by inspection.
\end{proof}
This proposition, paired with Prop.\ \ref{pro:tcpzflb}, establishes that by appropriately combining diversity (array) gains with IC, the TC of a wireless network can in fact be increased linearly with just the Rx antenna array size. In contrast, this is not possible without the combination.  This can also be observed from the above upper bound on PZF TC, which reduces to MRC for $z=0$ and to full ZF for $z=\nr-1$.  Plugging those values in yields the following Corollary.
\begin{corollary}
\label{cor:tcmrcub}
{\bf TC UB for MRC.} 
An upper bound on the TC with MRC is given by $z=0$ and is
\begin{eqnarray}
\lambda(q^*)^{\textrm{mrc}} \leq \frac{2 + \frac{\alpha}{2}}{ \pi u^2 \tau^{\frac{2}{\alpha}} (1-q^*)^{\frac{2}{\alpha}}} \nr^{\frac{2}{\alpha}},
\end{eqnarray}
while for full ZF the TC is giving with $z=\nr-1$ and
\begin{eqnarray}
\lambda(q^*)^{\textrm{zf}}(\nr-1) \leq \frac{2 + \alpha/(2 \nr)}{ \pi u^2 \tau^{\frac{2}{\alpha}} (1-q^*)^{\frac{2}{\alpha}}} \nr^{1-\frac{2}{\alpha}}.
\end{eqnarray}
Both of these expressions use $l=2$ in Prop.\ \ref{pro:MMSE-UB}.
\end{corollary}
These expressions make plain the $\nr^{\frac{2}{\alpha}}$ scaling for MRC and $\nr^{1-\frac{2}{\alpha}}$ scaling for full ZF, respectively.  These same scalings were observed above for MRC in Prop.\ \ref{pro:TC-MRC}  and for full ZF in \cite{HuaAndSub}, respectively.  Notably, in \cite{HuaAndSub}, it was found that using multiple Tx antennas did not change the scaling if the Tx beamforming vector is not adapted to the channel or interference.

\begin{figure}[!htbp]
\centering
\includegraphics[width=\textwidth]{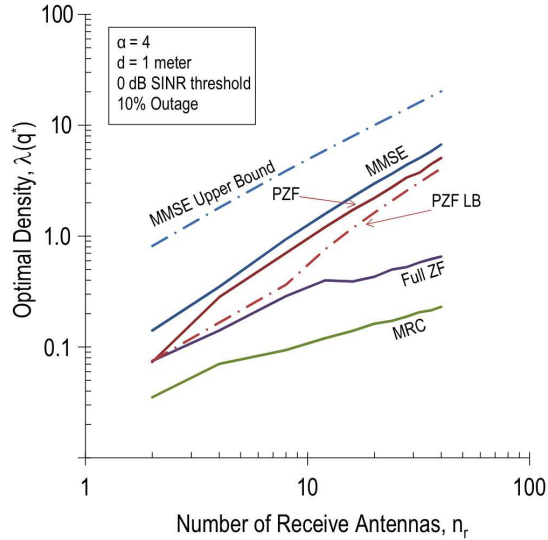}
\caption{The OCD vs.\ $\nr$ for PZF, MMSE, ZF, and MRC, for $\alpha = 4$.  The superior scaling of MMSE and PZF is readily observed.}
\label{fig:PZF}
\end{figure}

\section{Main results on multiple stream TC}
\label{sec:multistream}

We now consider the case where nodes communicate $K \geq 1$ streams simultaneously.  In the case of spatial multiplexing, these $K$ streams all originate and terminate between a single Tx-Rx pair.  For (Tx) SDMA, the streams all originate at a single Tx but terminate at $K$ different receivers.

\subsection{Spatial multiplexing}
\label{ssec:SM}

The fundamental tradeoff in a spatial multiplexing (SM) system is one of rate versus reliability.  Famously formalized in \cite{ZheTse2003}, the basic idea is that transmitting multiple streams, \eg, $K = \nt$, tends to maximize the data rate while $K = 1$ tends to maximize reliability. We will see that a related tradeoff also occurs in decentralized wireless networks following the TC model.  In our case the tradeoff is typically more complex to quantify because the TC of each Tx-Rx pair is affected by all the other Tx-Rx pairs in the network, who also are adjusting their number of Tx streams.  However, as a metric the TC already inherently includes both throughput and reliability, so maximizing the TC inherently trades off between rate (higher SINR or $\lambda$) and reliability (outage).

As noted previously, the TC framework operates on a specific system model, which renders a SINR expression specific to that model.  In the case of SM there is a fairly large library of possible Tx and Rx structures, each of which yield a different SINR, and hence possibly a different tradeoff. For example, a SM Tx could use multimode beamforming, a singular value decomposition (SVD) pre-filter, or it could simply transmit each of the $K$ streams on $K \leq \nt$ different antennas --- either in a fixed pattern or rotating over the antennas (as is done in BLAST).  At the Rx, the number of options is even larger, including the Maximum Likelihood (ML) detector or approximations to it like the sphere decoder; various forms of linear interference cancellation (including ZF, PZF, and MMSE); nonlinear interference cancellation (BLAST or SIC); or other approaches like MRC or SVD post-filtering.  Explaining the details of all these different well-known Rx structures is outside the scope of this chapter, and readers lacking this background are referred to any of \cite{PauGor2003,TseVis2005,Gol2005,GhoZha2010} for details. We will focus on an instructive subset of these approaches (mainly limited to Tx and Rx very similar to those considered in \S\ref{sec:singlestream}), and see that several broad trends, as well as a few interesting differences, hold across the different approaches.

\subsubsection*{Simple $K$-stream transmission}

First, we consider a class of results where the Tx has the simplest possible structure: it simply transmits a single unique data stream from each of one or more of the $\nt$ Tx antennas.  So, for $K=1$ these results reduce to the results of \S\ref{sec:singlestream}.  For $2 \leq K \leq \min(\nt,\nr)$, the results will generalize our prior results.  For this simple transmit structure, we will consider four possible different receivers, all of which are linear:
\begin{enumerate}
\item The MRC Rx \cite{LouMcK2011}.
\item The ZF Rx, where only self-interference from the interfering $K-1$ streams is cancelled \cite{StaPro2010b,LouMcK2011}.
\item The PZF Rx, where in general both self-interference and interference from other users is cancelled \cite{VazHea2009}.
\item The MMSE Rx \cite{LouMcK2011b}.
\end{enumerate}
We now catalogue and discuss the main results of interest for these cases. The proofs are not given but can be found in the referenced papers.  To express the results simply, we assume in all cases that $\nt \leq \nr$ and that the number of transmitted streams $K \leq \nt$, and so of course $K^* \leq \nt$ also.  Due to the transmit structure assumed in this section, the other $\nt - K$ Tx antennas are simply not used.
\begin{proposition}
\label{pro:SM-MRC}
{\bf SM with MRC} \cite{LouMcK2011}. 
With a MRC Rx and large $\nt, \nr$, the optimal number of streams to transmit is
\begin{equation}
K^* = \nr \frac{1-\frac{2}{\alpha}}{\tau(1 + \snr^{-1})},
\end{equation}
and when $K = K^*$ the resulting TC is $\Theta(\nr)$ (as $\nr \to \infty$).
\end{proposition}
\begin{proposition}
\label{pro:SM-ZF}
{\bf SM with ZF Rx} \cite{StaPro2010b,LouMcK2011}.
With a ZF Rx applied to the $K-1$ interfering streams, other-user interference treated as noise, and large $\nt, \nr$, the optimal number of streams to transmit is
\begin{equation}
K^* = \nr \frac{1-\frac{2}{\alpha}}{1 + \tau \snr^{-1}},
\end{equation}
and when $K = K^*$ the resulting TC is $\Theta(\nr)$ (as $\nr \to \infty$).  Furthermore, if $K = \nt = \nr$ the scaling falls to $\Theta(\nr^{1-\frac{2}{\alpha}})$ (as $\nr \to \infty$).
\end{proposition}
Several interesting observations can be made already from these two results.  At high SNR, the optimum number of streams is quite different for MRC and ZF, namely ZF uses $\tau$ more streams than MRC.  They both prefer more streams at high path loss exponents: in this case the other-user interference is attenuated more rapidly.  And we see that with ZF, there is a significant penalty attached to sending streams from the full antenna array.
\begin{proposition}
\label{pro:SM-PZF-HighSNR}
{\bf SM with PZF at high SNR} \cite{VazHea2009}. 
When a PZF Rx is applied to cancel the $K-1$ interfering streams from the same user, and the $K$ streams from each of $z$ closest interfering users, the optimal number of streams is $K^* =1$, the optimal number of users to cancel is $z^* = \nr (1-\frac{2}{\alpha})$, and the TC is $\Theta(\nr)$ (as $\nr \to \infty$).
\end{proposition}
This proposition indicates that the optimal PZF setup is the single stream Tx of \cite{JinAnd2011} that was already extensively discussed in \S\ref{ssec:MIMO-IC}.  In short, there is in principle no gain to doing SM if one is able to cancel interference instead.  An important caveat here is that one must be able to increase the contention density arbitrarily to get the linear scaling gains (\eg, by allowing a higher contention probability in the Aloha-like random access protocol).  SM may therefore still be desirable since it does not require an increased network density to achieve its TC gains.
\begin{corollary}
\label{cor:pzfrxnoic}
{\bf PZF Rx without cancellation} \cite{VazHea2009}.
If the PZF Rx operates only on the streams of the desired user, \ie, $z = 0$, then $K^* = \nr (1-\frac{2}{\alpha})$ and the TC is still $\Theta(\nr)$ (as $\nr \to \infty$).
\end{corollary}
This corollary follows from Prop.\ \ref{pro:SM-ZF} and from \cite{VazHea2009}, because the PZF Rx is in this case exactly the ZF Rx.  The TC scaling is the same whether the Rx cancels interference from the other streams of the same Tx; or cancels interference from a number of nearby single stream transmitters \cite{LouMcK2011}. Furthermore, the optimum number of streams/users to cancel turns out to be the same, which is not obvious since the statistical properties of the interference are quite different in each case.
\begin{corollary}
\label{cor:pzfrxstrongestic}
{\bf PZF Rx cancelling strongest interferer} \cite{VazHea2009}.
If the PZF Rx instead cancels the \emph{strongest} interferers, then $K^*=1$, $z^* = \nr-1$, and the TC is still $\Theta(\nr)$ (as $\nr \to \infty$) but can be further characterized as $\Theta(\nr (q^*)^{\frac{1}{\nr}})$ (as $\nr \to \infty$).
\end{corollary}
This tells us that there is an additional gain in terms of the outage constraint from measuring and cancelling the strongest users, which would typically also be more practical since the Rx would not know where the interferers are located, but would be able to measure the receive signal strengths.
\begin{remark}
\label{rem:smmmserx}
{\bf SM with the MMSE Rx.}
Exact results for the TC in this case are given in \cite{LouMcK2011b} but are quite complex, and not amenable to a simple conclusion on the orderwise scaling or optimum number of streams $K$. In view of the previous results it is safe to conjecture that the scaling will be at least linear in $\nr$ if the optimum number of streams are used.  In fact from Prop.\ \ref{pro:MMSE-UB} we know that the scaling is linear even if $K=1$.
\end{remark}

\subsubsection*{Spatial multiplexing enhancements: multimode BF and DBLAST}

The prior section used the simplest possible Tx structure, and four different linear receivers.  In this section we consider two other representative setups that introduce logical enhancements at both the Tx and Rx.  These are:
\begin{enumerate}
\item Multimode (eigen)-beamforming transmission \cite{VazHea2009}.
\item DBLAST, which includes a diversity Tx and non-linear interference cancelling Rx \cite{StaPro2010b}.
\end{enumerate}
The multimode beamformer transmits $K$ streams on all $\nt$ antennas, where the $K$ streams are placed on the $K$ dominant eigenmodes of the Tx-Rx channel. For simplicity we now let $\nt = \nr$ and so for the balance of the section $\nt$ and $\nr$ can be used interchangeably.  Obviously, the Tx must acquire the Tx-Rx channel to implement such a precoder, which typically requires the singular value decomposition (SVD) of the matrix channel.  Such an approach achieves the capacity (with appropriate power allocation across the eigenmodes) for a point-to-point MIMO channel with an optimal Rx \cite{Tel1999}.  This approach was the first considered approach in the TC framework for \adhoc networks as well \cite{HunAnd2008b} and is difficult to analyze; Vaze and Heath made progress by instead using a PZF Rx.  Their main result can be summarized by the following proposition.
\begin{proposition}
\label{pro:SM-BF-PZF}
{\bf SM with multimode beamforming and PZF receivers} \cite{VazHea2009}.
For a multimode beamformer with a PZF Rx, the optimum number of Tx streams is $K^* = 1$, the optimal number of cancelled interferers is $z = \nr-1$, and the TC scales as $\Theta(\nr)$ (as $\nr \to \infty$).
\end{proposition}
In short, applying an optimum Tx pre-filter does not change the PZF results we have previously seen; however it does increase the optimum number of cancelled interferers since the additional Rx DoF are no longer needed for diversity and array gain, since that is accomplished at the Tx.  There is a fixed (independent of $\nt,\nr$) gain of about a factor of $4$ with multimode beamforming.  With an optimum single-user Rx (computed on the SVD of the Tx-Rx channel), but no IC, it was observed numerically in \cite{HunAnd2008b} that the optimal number of streams reverts to be compatible with the recurring $\nt (1 - \frac{2}{\alpha})$ expression.

Turning to nonlinear Rx structures, Stamatiou {\em et al.}\ considered both Vertical and Diagonal BLAST architectures in \cite{StaPro2010b}. BLAST receivers successively decode and cancel the $K$ transmitted streams.  V-BLAST is more bandwidth efficient but does not achieve as much diversity as D-BLAST, which rotates the symbols across the antennas.
\begin{proposition}
\label{pro:SM-BLAST}
{\bf SM with the BLAST architecture} \cite{StaPro2010b}. 
For a D-BLAST architecture the optimum number of streams is $K^* = 2 (\nt +1) (1-\frac{2}{\alpha})$ and for both D-BLAST and V-BLAST the TC is still $\Theta(\nt)$ (as $\nt \to \infty$).  Furthermore, the TC of V-BLAST is higher than D-BLAST by a factor of $2^{1-\delta}$ for $\alpha \leq 4$ and by a factor $2^{-\delta} \delta^{-\delta} (1-\delta)^{\delta-1}$ for $\alpha > 4$, recalling $\delta = \frac{2}{\alpha}$.
\end{proposition}
It is notable that BLAST appears to make SM more robust, which results in more streams being used especially for high path loss (lower interference).  For example, if $\alpha \geq 4$, $K^* = \nt$.

\subsection{Space division multiple access (SDMA)}
\label{ssec:SDMA}

The main distinction between SDMA, also often called multiuser MIMO, and spatial multiplexing is that in SDMA, the multiple simultaneous streams originate or terminate at different users. We will focus on the Tx SDMA case, and to maximally distinguish between SDMA and SM, we will assume that each SDMA stream is sent to a different Rx. We will now formalize key aspects of the model, before providing key TC results on SDMA, and discussing their implications, and how they differ and/or agree with cellular SDMA results.

\subsubsection*{Details of the SDMA model}
\label{sec:SDMA-details}

Each Tx has $\nt$ antennas and communicates simultaneously with $\left|\mathcal{K}\right|=K \leq \nt$ receivers, which are each equipped with $\nr$ Rx antennas. Each of the $K$ streams contains a separate message destined to a different Rx, and each Tx and its set of intended receivers $\mathcal{K}_i$ form a {\em broadcast cluster}.  The closest Rx in the cluster is a distance $u_{\rm min}$ from the Tx, and the farthest is $u_{\rm max}$, with the rest of the receivers proportionally located in a random direction in between $u_{\rm min}$ and $u_{\rm max}$.  Although $u_{\rm min}$ and $u_{\rm max}$ show up in the TC bounds, they can take on any arbitrary value and do not affect the scaling results we will present.

\subsubsection*{Dirty paper coding (DPC) transmission}

It is well-known that the optimum Tx SDMA strategy in a Gaussian broadcast channel (\ie, a downlink channel with only Gaussian noise/interference) is \emph{dirty paper coding} (DPC) \cite{Cos1983}, which can be viewed as a form of successive interference pre-cancellation, and provides an effectively interference free channel for each stream \cite{CaiSha2003,VisJin2003,WeiSte2006}.  Therefore, the SDMA sum capacity (achieved using DPC) with single antenna receivers at high SNR is $O( \nt \log \snr)$, and adding $\nr$ Rx antennas to each node does not increase the scaling by more than a fairly inconsequential $\log\log \nr$ term, which can be achieved with antenna selection (which increases the $\snr$ by a $\log \nr$ factor asymptotically, and hence the capacity by $\log \log \nr$).  However, until recently, almost nothing was known about SDMA's performance in the presence of uncontrolled (spurious) interference.

The next several results, due to Kountouris \cite{KouAndSub,KouAnd2009a}, provide mathematical expressions that show SDMA's effect on throughput in a decentralized network.  The proofs can be found in the sources referenced when not provided here.

\begin{proposition}
\label{pro:DPC-TC}
{\bf Bounds on multistream DPC TC with MRC} \cite{KouAndSub}.
The multi-stream TC of DPC transmission with MRC receivers at high SNR subject to small OP constraint $q^*$ is bounded as
\begin{equation}
\frac{K q^* \left(1-q^*\right) \mathcal{F}_{d} }{\mathcal{J}_K\tau^{\frac{2}{\alpha}}u_{\rm max}^2}
\leq \mathcal{C}_{\rm DPC}^{\rm mrc} \leq \frac{K[4\nt(\nr-K+1)]^{\frac{2}{\alpha}}(1-q^*)\log\left(\frac{1}{1-q^*}\right)}{\mathcal{J}_K\tau^{2/\alpha}u_{\rm min}^2}
\end{equation}
where for diversity order $d = \nt(\nr-K+1)$ we have
\begin{equation}
\mathcal{F}_d = \left[\displaystyle \sum_{j=0}^{d}\binom{d}{j}(-1)^{j+1}j^{\frac{2}{\alpha}}\right]^{-1},
\end{equation}
and the constant $\mathcal{J}_K$ is
\begin{equation}
\mathcal{J}_K = \frac{2\pi\Gamma\left(\frac{2}{\alpha}\right)}{\alpha\Gamma(K)}\displaystyle \sum_{m=0}^{K-1}\binom{K}{m}\Gamma(m+1)\Gamma\left(K-m-\frac{2}{\alpha}\right),
\end{equation}
which depends only on  the number of streams $K$ and $\alpha$.
\end{proposition}
\begin{proof}
A sketch of the proof is provided, which follows the standard TC framework. First, derive an upper and lower bound on the OP, which results in bounds that are given by
\begin{equation}
1 - \mathcal{L}[\Sisf](\frac{\tau u_{\rm min}^{\alpha}}{4d})
\leq \Pbb(\sinr \leq \tau) \leq \sum_{j=0}^{d}\binom{d}{j}(-1)^j\mathcal{L}[\Sisf](j \tau u_{\rm max}^{\alpha}),
\end{equation}
where $\mathcal{L}[\Sisf](s)$ is the LT of the aggregate interference term with multi-stream transmission (given interference power marks that are distributed as chi-squared with $2K$ DoF), and can be given by \cite{HunAnd2008b}
\begin{equation}
\mathcal{L}[\Sisf](\zeta) = \exp\left(-\lambda \zeta^{\frac{2}{\alpha}} \mathcal{J}_K\right).
\end{equation}
The closest Rx in the cluster gives the lowest OP (almost surely in probability), and the farthest similarly gives the highest. These expressions are then set equal to $q^*$ and inverted for $\lambda$ to find corresponding lower and upper bounds on the contention density $\lambda(q^*)$, which multiplied by the the number of transmitted streams $K$ and the success probability $1-q^*$ gives the result.
\end{proof}
These somewhat involved expressions are amenable to the following scaling laws, which clearly show the dependence on the number of streams and antennas.
\begin{proposition}
\label{pro:DPC-scaling}
{\bf SDMA scaling laws with a DPC Tx and MRC Rx} \cite{KouAndSub}.
\begin{eqnarray}
\mathcal{C}_{\rm DPC}^{\rm mrc} &=& \Omega(K^{1-\frac{2}{\alpha}}[(\nt-K+1)(\nr-K+1)]^{\frac{2}{\alpha}}) \\
\mathcal{C}_{\rm DPC}^{\rm mrc} &=& O(K^{1-\frac{2}{\alpha}}[\nt(\nr-K+1)]^{\frac{2}{\alpha}})
\label{capa_dpc_scale}
\end{eqnarray}
where both asymptotic order expressions hold as $\nt,\nr \to \infty$.
\end{proposition}
\begin{proof}
Given Prop.\ \ref{pro:DPC-TC}, for asymptotically large number of antennas $\nt, \nr$, and streams $K$, as in the proof Prop.\ \ref{pro:TC-MRC}, that
\begin{equation}
\displaystyle \lim_{K\to\infty}\frac{\mathcal{J}_K}{K^{\frac{2}{\alpha}}} = \pi\Gamma(1-2/\alpha)
\end{equation}
and for asymptotically large $d$ we have $\mathcal{F}_{d} = \Theta(d^{\frac{2}{\alpha}})$ \cite{WebYan2005}. Thus, the lower bound divided by $K^{1-\frac{2}{\alpha}}[(\nt-K+1)(\nr-K+1)]^{\frac{2}{\alpha}}$converges to a constant as $\nt, \nr \to \infty$, and similarly for the upper bound divided by $K^{1-\frac{2}{\alpha}}(\nt(\nr-K+1))^{\frac{2}{\alpha}}$.
\end{proof}
\begin{corollary}
\label{cor:sdmadpcmrc}
{\bf TC scaling and $K^*$ for SDMA with DPC Tx and MRC Rx} \cite{KouAndSub}.
Letting $\nt = \nr$ and the number of streams $K = \theta \nt$, for $0 < \theta < 1$ in Prop.\ \ref{pro:DPC-scaling}, then
\begin{equation}
\mathcal{C}_{\rm DPC}^{\rm mrc} = \Omega(\nt), ~~~ \mathcal{C}_{\rm DPC}^{\rm mrc} = O(\nt^{1 + \frac{2}{\alpha}}), ~ \nt \to \infty 
\end{equation}
and it can be shown that
\begin{equation}
K^*_{\rm lb} = \frac{\alpha - 2}{\alpha + 2} \nt = 1- \frac{2}{\alpha(\frac{1}{2}+ \frac{1}{\alpha})} \nt, ~~~ K^*_{\rm ub} = 1 - \frac{2}{\alpha}(\nt + 1).
\end{equation}
\end{corollary}
Considering just the UB, we observe that the optimum number of streams $K^*_{\rm ub}$ again follows the same theme as the PZF and SM results, with a similar but slightly different trend for $K^*_{\rm lb}$.  But unique thus far to SDMA, we see the possibility of \emph{superlinear} scaling in the number of antennas.

If we consider a trivial Rx that has a single Rx antenna, the following scaling law can be determined for such a multi-input single output (MISO) scenario.
\begin{proposition}
\label{pro:DPC-MISO}
{\bf TC with DPC Tx and single antenna Rx} \cite{KouAndSub}.
At low $q^*$ and high SNR, The TC with a DPC precoder and single antenna receivers scales as
\begin{equation}
\mathcal{C}_{\rm DPC}^{\rm miso} = \Omega(K^{1-\frac{2}{\alpha}}(\nt-K+1)^{\frac{2}{\alpha}}) \hspace{5mm} \text{and} \hspace{5mm} \mathcal{C}_{\rm DPC}^{\rm miso} = O(K^{1-\frac{2}{\alpha}}\nt^{\frac{2}{\alpha}}),
\end{equation}
(as $\nt \to \infty$), and if $K = \theta \nt$ for any $\theta \in (0,1)$, then $\mathcal{C}_{\rm DPC}^{\rm miso} = \Theta(\nt)$.  Finally, the asymptotically optimum number of streams is $K^* = (1 - \frac{2}{\alpha})\nt$.
\end{proposition}
Prop.\ \ref{pro:DPC-MISO} --- where $K$ streams are transmitted to $K$ different single antenna receivers in space --- forms an agreeable symmetry with Prop.\ \ref{pro:tcpzflb} where $z$ streams are transmitted by $z$ different single antenna transmitters to $z$ different $\nr$-antenna receivers in space.  In both cases the scaling laws are linear in $\nt$ and $\nr$, respectively, and optimum number of streams is the fraction $1 - \frac{2}{\alpha}$.  Furthermore in both cases, fixing the number of streams to some constant that does not depend on $\nt$ of $\nr$ loses the linear scaling.

\section{Practical issues and further research}
\label{sec:takeaways}

In this chapter, we have considered numerous ways to use multiple antennas at the Tx, Rx, or both. The results provided here follow a few broad themes, which we now briefly summarize.

\subsection{Summary of main design insights}
\label{ssec:designinsights}

\textbf{Diversity} techniques typically provide linear SINR gains coming from the array gain of the antenna array. This then results in a TC gain of $L^{\frac{2}{\alpha}}$ in most cases, where $L$ is the amount of diversity, \eg, $L = \nt \nr$.  For example, for a pathloss exponent of $4$, the TC scales as $\sqrt{L}$.  This is considerably better than in a hypothetical centralized system where interference does not exist, in which case the throughput scales about as $\log(L)$ for diversity.

\noindent \textbf{Linear capacity increases.}   We observe that there are in effect three ways to increase the number of streams per unit area.  The first is to increase the density of independent transmissions and tolerate higher interference, \eg, using Rx IC.  The second is through SM between a given Tx and Rx.  And the third is via SDMA.  Interestingly, we see that all three of these can linearly increase the TC with the number of antennas.  Rx IC and SDMA require only the Rx or Tx, respectively, to have multiple antennas, whereas SM requires both the Tx and Rx to have multiple antennas.  However, it may not be possible to simply increase the density (Rx interference cancellation) and the SDMA result requires perfect CSIT and the feedback for this has not been accounted for.

\noindent \textbf{Balancing diversity and multiplexing.}   An interesting recurring theme in many of the results is that the number of streams sent should not be sent to the maximum, but rather a fraction $1 - \frac{2}{\alpha}$, with the remaining DoF going to diversity.  For example, with $\alpha = 4$, equal antenna resources should be devoted to diversity and spatial packing/multistream transmission.  No such corresponding result exists, to the best of our knowledge, in classical MIMO theory (without spurious interference).

\subsection{Caveats and practical issues}
\label{ssec:caveats}

There were numerous simplifications made in this chapter that will be easily recognized by those familiar with MIMO systems and/or \adhoc network design.  We comment on a few of the significant ones.

\noindent \textbf{Random access MAC.}  Throughout this chapter, indeed for any results using the PPP as a spatial interference model, the implicit MAC is an uncoordinated slotted Aloha-type MAC.  Such a MAC is pessimistic versus CSMA (which is quite practical) or a centralized MAC (which is likely not practical for an \adhoc or self-forming network).  Moving away from an Aloha MAC (which gives worst-case interference) may change some of the relative benefits and make SDMA/SM more desirable versus interference suppression \cite{ZorZei2006}.  Some initial progress on extending the MIMO results to a CSMA setup are given in \cite{HunGan2010}.

\noindent \textbf{Scaling results} are focused on in this chapter, to see how the TC slope scales for large $\nt, \nr$. In most practical scenarios, regime of small $\nt$ and $\nr$ is the most important one, and in such a regime, the use of the antennas to decrease OP may be more important from a TC point of view.  For example, \cite{HuaAndSub} indicates that for small $\nr$ all DoF should be used for IC, in which case the OP falls as $\lambda^{\nr+1}$, or equivalently the TC can tolerate a stricter outage constraint at the same $\lambda$.  Then, as the array size grows, some DoF can be converted to use for MRC or sending multiple streams.

\noindent \textbf{Cellular systems} do not directly follow from the results given here, which are for decentralized networks. The main modeling difference in a (downlink) cellular system is that interferers should not be closer/stronger than the desired signal, or else one would simply handoff to this stronger signal.  A tractable framework for cellular analysis which uses tools similar to those in this monograph was recently proposed in \cite{AndBac2011}, where the base stations are drawn from a PPP, and hence the interference has a shot noise format and many of the results in this monograph can be modified and applied.   We conjecture that many of the key results and scalings shown here may hold in a MIMO-enhanced cellular network, but this is an interesting topic for future work.

\noindent \textbf{Feedback and overhead.}  We have neglected the cost of obtaining the channel state information at the Rx (\eg, using pilot symbols) or more importantly for SDMA and possibly SM, at the Tx.  SDMA in particularly requires frequent feedback.  However, \cite{GovBli2010sub,KouAndSub2} have considered exactly this problem, and \cite{GovBli2010sub,KouAndSub2} has shown that the feedback requirement for preserving the SDMA linear capacity scaling is $\nt \log (\tau)$, which follows the well-known broadcast channel result in \cite{Jin2006}, but interestingly includes the effect of interference and thus uses quite different analytical techniques.

%
%
\appendix

\begin{notationsandacronyms}

\begin{table}
\begin{center}
\caption{Notation}
\footnotesize\begin{tabular*}{\textwidth}{lll }\hline
{\bf Chapter 1} & & \\ \hline
$\csf(o)$ & (random) channel capacity for Rx at $o$ & Ass.\ \ref{ass:keyass} \\
$R$ & codebook rate & Def.\ \ref{def:op} \\
$\tau$ & SINR threshold required to support rate $R$ & Def.\ \ref{def:op} \\
$q(\lambda)$ & outage probability (OP) & Def.\ \ref{def:op} \\
$d $ & dimension of the wireless network ($\{1,2,3\}$) & Def.\ \ref{def:ppp} \\
$\lambda$ & intensity of attempted transmissions per unit area & Def.\ \ref{def:ppp} \\
$\Pi_{d,\lambda}$ & (random, homogeneous) PPP on $\Rbb^d$ of intensity $\lambda$ & Def.\ \ref{def:ppp} \\
$q^*$ & target OP & Def.\ \ref{def:tc} \\
$\lambda(q^*)$ & transmission capacity with target OP $q^*$ & Def.\ \ref{def:tc} \\
$\lambda_{\rm pot}$ & spatial intensity of potential interferers & Rem.\ \ref{rem:ptx} \\
$p_{\rm tx}$ & Aloha transmission probability & Rem.\ \ref{rem:ptx} \\ \hline
{\bf Chapter 2} & & \\ \hline
$\brm_d(c,r)$ & ball in $\Rbb^d$ centered at $c$ with radius $r$ & Def.\ \ref{def:ball} \\
$\arm_d(c,r_1,r_2)$ & annulus in $\Rbb^d$ centered at $c$ with radius $r_1,r_2$ & Def.\ \ref{def:ball} \\
$c_d$ & coeffiicent on volume of a ball: $|\brm_d(o,r)| = c_d r^d$ & Prop.\ \ref{pro:ballvol} \\
$\Gamma(z),\Gamma(z,t_1,t_2)$ & Gamma and incomplete Gamma functions & Def.\ \ref{def:gamfun} \\
$F_{\xsf},\bar{F}_{\xsf},f_{\xsf}$ & CDF, CCDF, PDF for RV $\xsf$ & Def.\ \ref{def:standardprobdefs} \\
$\xsf \sim N(\mu,\sigma)$ & a normal RV with mean $\mu$ and standard dev.\ $\sigma$ & Def.\ \ref{def:standardprobdefs} \\
$\zsf \sim N(0,1)=F_{\zsf}$ & standard normal RV & Def.\ \ref{def:standardprobdefs} \\
$F_{\xsf}^{-1},\bar{F}_{\xsf}^{-1}$ & inverse CDF and CCDF for RV $\xsf$ & Def.\ \ref{def:standardprobdefs} \\
$\Lmc[\xsf](s)$ & Laplace transform (LT) for $\xsf$ with $s \in \Cbb$ & Def.\ \ref{def:standardprobdefs} \\
$\phi[\xsf](t)$ & characteristic function (CF) for $\xsf$ with $t \in \Rbb$ & Def.\ \ref{def:standardprobdefs} \\
$\Mmc[\xsf](\theta)$ & moment generating function (MGF) for $\xsf$ with $\theta \in \Rbb_+$ & Def.\ \ref{def:standardprobdefs} \\
$\Hmc[\xsf](x)$ & hazard rate function (HRF) for $\xsf$ with $x \in \Rbb$ & Def.\ \ref{def:standardprobdefs} \\
$l(r)$ & a generic impulse response function & Def.\ \ref{def:snp} \\
$\Sisf_{\Pi}^l(x)$ & (random, cumulative) shot noise (SN) RV & Def.\ \ref{def:snp} \\
$\Msf_{\Pi}^l(x)$ & (random) max shot noise (SN) RV &  Def.\ \ref{def:snp} \\
$\alpha$ & pathloss exponent & Rem.\ \ref{rem:path}\\
$\epsilon$ & nulling radius around Rx & Rem.\ \ref{rem:path}\\
$l_{\alpha,\epsilon}(r)$ & pathloss attenuation function over distance $r \in \Rbb_+$ & Rem.\ \ref{rem:path} \\
$\Sisf^{\alpha,\epsilon}_{d,\lambda}(o)$ & (random) SN at $o$ under PPP $\Pi_{d,\lambda}$ and $l_{\alpha,\epsilon}$ & Def.\ \ref{def:intsn} \\
$\delta = \frac{d}{\alpha}$ & characteristic exponent & Def.\ \ref{def:intsn} \\
$\gamma$ & dispersion coefficient of a stable RV/CDF & Def.\ \ref{def:staparam} \\
$\Gmc[\Pi,\nu]$ & point processs probability generating functional & Def.\ \ref{def:pgfl} \\ \hline
{\bf Chapter 3} & & \\ \hline
$\sinr(o)$ & (random) SINR at $o$ & Def.\ \ref{def:sinrnf} \\
$P$ & transmission power & Def.\ \ref{def:sinrnf} \\
$N$ & noise power & Def.\ \ref{def:sinrnf} \\
$S$ & Rx signal power & Def.\ \ref{def:sinrnf} \\
$u$ & Tx-Rx pair separation distance & Def.\ \ref{def:sinrnf} \\
$\xi$ & a constant & Def.\ \ref{def:xisnr} \\
$\snr$ & Rx SNR & Def.\ \ref{def:xisnr} \\
$\hat{\Pi}_{d,\lambda},\hat{\Sisf}^{\alpha,\epsilon}_{d,\lambda}(o)$ & (random) dominant interferers and interference & Def.\ \ref{def:domint} \\
$\tilde{\Pi}_{d,\lambda},\tilde{\Sisf}^{\alpha,\epsilon}_{d,\lambda}(o)$ & (random) non-dominant interferers and interference & Def.\ \ref{def:domint} \\
$\Lambda(\lambda)$ & spatial throughput (TP) & Def.\ \ref{def:tp} \\ \hline
\end{tabular*}
\end{center}
\end{table}

\begin{table}
\begin{center}
\caption{Notation (continued)}
\footnotesize\begin{tabular*}{\textwidth}{lll }\hline
{\bf Chapter 4} & & \\ \hline
$\Sisf^{\alpha,\hsf}_{d,\lambda}(o)$ & (random, cumulative) interference at $o$ under fading ($\hsf$) & Def.\ \ref{def:fad} \\
$\hsf_i$ & (random) channel fading coef.\ from $i \in \Pi_{d,\lambda}$ to Rx at $o$ & Rem.\ \ref{rem:fadsigint} \\
$\hsf_0$ & (random) channel fading from ref.\ Tx to ref.\ Rx at $o$ & Rem.\ \ref{rem:fadsigint} \\
$\Phi_{d,\lambda}$ & (random) marked PPP & Rem.\ \ref{rem:MPPP} \\
$\Emc_0,q(0),\bar{q}(0)$ & event/probability of outage w/o interference & Rem.\ \ref{rem:fadlowout} \\
$\hat{\Phi}_{d,\lambda},\hat{\Sisf}^{\alpha,\hsf}_{d,\lambda}(o)$ & (random) dominant interferers and interference & Def.\ \ref{def:domintfad} \\
$\tilde{\Phi}_{d,\lambda},\tilde{\Sisf}^{\alpha,\hsf}_{d,\lambda}(o)$ & (random) non-dominant interferers and interference & Def.\ \ref{def:domintfad} \\
$\usf \sim F_{\usf}$ & (random) link distance under VLD & Def.\ \ref{def:vardistsinr} \\
$q(\lambda),\lambda(q^*)$ & OP and TC under FLD & Prop.\ \ref{prop:vardistoptc} \\
$\tilde{q}(\lambda),\tilde{\lambda}(q^*)$ & OP and TC under VLD & Prop.\ \ref{prop:vardistoptc} \\
$U$ & end-to-end (source-destination) distance & \S\ref{sec:multihop} \\
$M$ & number of hops & \S\ref{sec:multihop} \\
$A$ & maximum allowable end-to-end transmission attempts & Def.\ \ref{def:MHTC} \\
$\Tsf$ & (random) total number of transmission attempts & Def.\ \ref{def:MHTC} \\
$\lambda_{\rm mh}$ & multihop transmission capacity &  Def.\ \ref{def:MHTC} \\
$\lambda_{\rm mh}^{\rm ub}$ & upper bound on $\lambda_{\rm mh}$  &  Prop.\ \ref{pro:MHTC-UB} \\
$K_{\alpha}$ & $\pi^2 \delta \csc(\pi \delta)$ & Prop.\ \ref{pro:mstar}\\ \hline

{\bf Chapter 5} & & \\ \hline
$W$ & total bandwidth (Hz) available to network & Ass.\ \ref{ass:specman} \\
$B$ & number of bands employed, each of BW $W/B$ (Hz) & Ass.\ \ref{ass:specman} \\
$\eta$ & noise power spectral density (W/Hz) & Def.\ \ref{def:specdefs} \\
$N(B),N$ & noise power over a band, full spectrum & Def.\ \ref{def:specdefs} \\
$\nu$ & spectral efficiency requirement & Def.\ \ref{def:specdefs} \\
$q(\lambda/B,B)$ & OP for intensity $\lambda$ and $B$ bands & Def.\ \ref{def:outageequiv} \\
$\lambda(q^*,B^*)$ & TC for target OP $q^*$ using optimal bands $B$ & Def.\ \ref{def:outageequiv} \\
$\kappa(q^*),\omega(B^*)$ & spatial component and spectral component of the TC & Rem.\ \ref{rem:optnumbands} \\
$\ebno$ & energy per bit & Def.\ \ref{def:ebno} \\
$\kappa$ & interference cancellation effectiveness & Def.\ \ref{def:icparam} \\
$K$ & maximum number of cancellable interferers & Def.\ \ref{def:icparam} \\
$P_{\rm min}$ & minimum receive power required for cancellation & Def.\ \ref{def:icparam} \\
$\Sisf^{\rm pc}(o),\Pi^{\rm pc}_{d,\lambda}(o)$ & (random) partially cancellable (pc) int./interferers at $o$ &  Def.\ \ref{def:sicsinr} \\
$\Sisf^{\rm uc}(o),\Pi^{\rm uc}_{d,\lambda}(o)$ & (random) uncancellable (uc) int./interferers at $o$ &   Def.\ \ref{def:sicsinr} \\
$\tsf^{\rm pc},t^{\rm pc}_{\rm dom},t^{\rm uc}_{\rm dom}$ & mapped constraint thresholds & Lem.\ \ref{lem:sictmax} \\
$\hsf_{0,0}$ & (random) fading channel coef.\ from reference Tx to $o$ & Def.\ \ref{def:sched} \\
$\hsf_{i,i}$ & (random) fading channel coef.\ from interferer $i$ to its Rx & Def.\ \ref{def:sched} \\
$\hsf_{i,0}$ & (random) fading channel coef.\ from interferer $i$ to $o$ & Def.\ \ref{def:sched} \\
$\hat{h}$ & fading threshold for transmission & Def.\ \ref{def:sched} \\
$\hat{\lambda}$ & intensity of attempted transmissions under threshold $\hat{h}$ & Def.\ \ref{def:sched} \\
$\hat{\Phi}_{d,\hat{\lambda}}$ & (random) MPPP of attempted Tx's under threshold $\hat{h}$ & Def.\ \ref{def:sched} \\
$q(\hat{h}),\Lambda(\hat{h})$ & OP and TP under threshold $\hat{h}$ & Def.\ \ref{def:sched} \\
$f$ & FPC exponent & Def.\ \ref{def:power} \\
$\Psf_0,\Psf_i$ & (random) Tx power by reference Tx and interferer $i$ & Def.\ \ref{def:power} \\
$\Emc_{0,f},q_f(0),\bar{q}_f(0)$ & event/probability of outage w/o interference & Rem.\ \ref{rem:fpclowout} \\
$\hat{\Phi}_{d,\lambda},\hat{\Sisf}^{\alpha,\hsf}_{d,\lambda}(o)$ & (random) dominant interferers and interference & Def.\ \ref{def:domintfpc} \\
$\tilde{\Phi}_{d,\lambda},\tilde{\Sisf}^{\alpha,\hsf}_{d,\lambda}(o)$ & (random) non-dominant interferers and interference & Def.\ \ref{def:domintfpc} \\ \hline
\end{tabular*}
\end{center}
\end{table}

\begin{table}
\begin{center}
\caption{Notation (continued)}
\footnotesize\begin{tabular*}{\textwidth}{lll }\hline
{\bf Chapter 6} & & \\ \hline
$\vsf_i$ & (random) Tx beamforming vector for Tx $i$ & Def. \ref{def:MIMO-SINR} \\
$\wsf_i$ & (random) Rx beamforming vector for Rx $i$ & Def. \ref{def:MIMO-SINR} \\
$\nt$ & number of Tx antennas  & Def. \ref{def:MIMO-SINR} \\
$\nr$ & number of receive antennas  & Def. \ref{def:MIMO-SINR} \\
$\Hsf,\hsf$ & (random) matrix, vector channel coefficients  & \S\ref{ssec:diversity} \\
$\Ssf_0$ & (random) signal power, $\Ssf_0 = ||\hsf_0||^2$  & Rem. \ref{rem:DivSigDis} \\
$K$ & number of independent streams of data  & Def. \ref{def:MIMO-K-SINR} \\
$z$ & number of interferers cancelled in PZF Rx, $z < \nr$. & Def. \ref{def:PZF} \\
$\theta \in [0,1]$ & fraction of antenna array used for a given purpose & Prop. \ref{pro:tcpzflb} \\
$u_{\rm min},u_{\rm max}$ & min.\ and max.\ distances to a desired Rx in a SDMA cluster & \S\ref{sec:SDMA-details}\\ \hline
\end{tabular*}
\end{center}
\end{table}

\begin{table}
\begin{center}
\caption{Acronyms}
\footnotesize\begin{tabular*}{\textwidth}{ll}\hline
ASE & Area spectral efficiency \\
BF & Beamforming \\
BLAST & Bell Labs space time (a MIMO Rx) \\
BPP & Binomial point process \\
iid & Independent and identically distributed \\
CCDF & Complementary cumulative distribution function \\
CDF & Cumulative distribution function \\
CF & Characteristic function \\
CSMA & Carrier sense multiple access \\
DoF & Degrees of freedom \\
DPC & Dirty paper coding \\
FLD & Fixed link distances \\
FPC & Fractional power control \\
FTS & Fading threshold scheduling \\
HRF & Hazard rate function \\
IC & Interference cancellation \\
LB & Lower bound \\
LT & Laplace transfrom \\
MAC & Medium access control \\
MGF & Moment generating function \\
MIMO & Multiple-input multiple-output (multiple antennas) \\
MISO & Multiple-input single-output \\
MMSE & Minimum mean square error \\
MPPP & Marked Poisson point process \\
MRC & Maximal ratio combining \\
MRT & Maximal ratio transmission \\
OCD & Optimal contention density \\
OP & Outage probability \\
PC & Power control \\
PDF & Probability density function \\
PGFL & Probability generating functional \\
PZF & Partial zero forcing \\
PPP & Poisson point process \\
RV & Random variable \\
Rx & Receiver \\
SDMA & Space division multiple access \\
SINR & Signal to interference plus noise ratio \\
SIR & Signal to interference ratio \\
SM & Spatial multiplexing \\
SN & Shot noise \\
SNR & Signal to noise ratio \\
TC & Transmission capacity \\
TP & Throughput \\
Tx & Transmitter \\
UB & Upper bound \\
VLD & Variable link distances \\
ZF & Zero forcing \\ \hline
\end{tabular*}
\end{center}
\end{table}

\end{notationsandacronyms}

\chapter{List of results by chapter}

\begin{table}[h] \footnotesize
\caption{Results in Ch.\ \ref{cha:int}: Introduction and preliminaries.}
\begin{tabular*}{\textwidth}{@{\extracolsep{\fill}}|c|r|l|} \hline
\S & \ref{sec:motass} &  Motivation and assumptions \\
Ass. & \ref{ass:keyass} & Key assumptions \\
Fig. & \ref{fig:overview} & Reference Rx and Tx and PPP of interferers \\ \hline
\S & \ref{sec:def} & Key definitions: PPP, OP, and TC \\
Def. & \ref{def:op} & Outage probability \\
Def. & \ref{def:ppp} & Homogeneous Poisson point process (PPP) \\
Fig. & \ref{fig:PPP} & Instance of a PPP \\
Fact & \ref{fac:opinv} & OP $q(\lambda)$ is continuous, strictly increasing, onto $(0,1)$ \\
Def. & \ref{def:tc} & Transmission capacity (TC) \\
Rem. & \ref{rem:ptx} & Potential vs. actual transmitters and Aloha \\ \hline
\S & \ref{sec:sum} & Overview of the results \\ \hline
\end{tabular*}
\label{tab:cha1}
\end{table}

\begin{table}[ht] \footnotesize
\caption{Results in Ch.\ \ref{cha:matpre}: Mathematical preliminaries.}
\begin{tabular*}{\textwidth}{|c|r|@{\extracolsep{\fill}}l|} \hline
Def. & \ref{def:ball} &	Ball and annulus \\
Prop. & \ref{pro:ballvol} & Ball and annulus volume \\
Def. & \ref{def:gamfun} & Gamma function \\  \hline
\S & \ref{sec:mccine} & Probability: notations, definitions, key inequalities \\
Rem. & \ref{rem:notation} & RV notation \\
Def. & \ref{def:standardprobdefs} & Standard probability definitions \\
Prop. & \ref{pro:jensen} & Jensen's inequality \\
Prop. & \ref{pro:mar} & Markov's inequality \\
Prop. & \ref{pro:cheb} & Chebychev's inequality \\
Prop. & \ref{pro:cher} & Chernoff's inequality \\ \hline
\S & \ref{sec:pppvd} & PPP void probabilities and distance mappings \\
Ass. & \ref{ass:pppdistorder} & Labeling convention for PPP \\
Prop. & \ref{pro:void} & Void probability \\
Thm. & \ref{thm:map} & Mapping theorem \\
Prop. & \ref{pro:distmap} & Distance mapping \\ \hline
\S & \ref{sec:snp} & Shot noise (SN) processes \\
Fig. & \ref{fig:SN} & SN process \\
Def. & \ref{def:snp} & SN process \\
Rem. & \ref{rem:indexconvention} & SN index convention \\
Rem. & \ref{rem:radsym} & Radial symmetry \\
Ass. & \ref{ass:snp} & Power law impulse response \\
Def. & \ref{def:intsn} & Power law SN and characteristic exponent \\
Rem. & \ref{rem:path} & Pathloss attenuation and the singularity at the origin \\
Def. & \ref{def:frechet} & Frech\'{e}t distribution \\
Cor. & \ref{cor:maxsnrvdis} & Max SN RV CDF \\
Prop. & \ref{pro:imap} & Interference mapping \\
Thm. & \ref{thm:baker} & Integration of radially symmetric functions \\
Thm. & \ref{thm:cam} & Campbell-Mecke \\
Prop. & \ref{pro:camz} & Shot noise mean and variance \\
Prop. & \ref{pro:snserexp} & Shot noise series expansion \\
Cor. & \ref{cor:snasypdf} & Asymptotic PDF and CCDF of the SN RV \\ \hline
\S & \ref{sec:stadis} & Stable distributions, Laplace transforms, and PGFL \\
Def. & \ref{def:stadis} & Stable RV and distribution \\
Def. & \ref{def:staparam} & Stable CF \\
Def. & \ref{def:levy} & L\'{e}vy distribution \\
Prop. & \ref{pro:stamom} & Stable moments \\
Fig. & \ref{fig:levpdf} & L\'{e}vy PDFs and CDFs \\
Def. & \ref{def:pgfl} & Point process PGFL \\
Prop. & \ref{pro:pgfl} & PPP PGFL \\
Cor. & \ref{cor:pgflsn} & Point process SN LT \\
Cor. & \ref{cor:pgflsnppp} & PPP SN LT \\
Cor. & \ref{cor:mgfint} & Pathloss SN MGF \\
Cor. & \ref{cor:lapint} & Pathloss SN CF \\
Cor. & \ref{cor:levint} & Pathloss SN for $\delta=\frac{1}{2}$ \\ \hline
\S & \ref{sec:maxsum} & Maximums and sums of RVs \\
Prop. & \ref{pro:intsummaxtight} & Sum and max SN CCDF ratio \\
Fig. & \ref{fig:intsummaxtight} & CCDFs for sum and max SN \\
Def. & \ref{def:subexp} & Subexponential distribution \\
Prop. & \ref{pro:subexp} & Sufficient subexponential condition) \\
Def. & \ref{def:BPP} & binomial point process (BPP) \\
Lem. & \ref{lem:binint} & BPP distances and interference \\
Cor. & \ref{cor:subexp} & Subexponential BPP interferences \\ \hline
\end{tabular*}
\label{tab:cha2}
\end{table}

\begin{table}[ht] \footnotesize
\caption{Results in Ch.\  \ref{cha:bm}: Basic model.}
\begin{tabular*}{\textwidth}{@{\extracolsep{\fill}}|c|r|l|} \hline
Def. & \ref{def:sinrnf} & Basic model SINR \\
Def. & \ref{def:xisnr} & Rx SNR	 \\				
Ass. & \ref{ass:epsnoi} & SNR LB \\ \hline
\S & \ref{sec:exactTC} & Exact OP and TC \\
Prop. & \ref{pro:opnf} & OP is SN CCDF \\
Cor. & \ref{cor:oplev} & Explicit OP for $\delta=\frac{1}{2}$ \\
Prop. & \ref{pro:tcexzereps} & TC ($\epsilon = 0$) \\
Cor. & \ref{cor:tclev} & TC ($\epsilon = 0$ and $\delta = \frac{1}{2}$) \\
Fig. & \ref{fig:optclev} & Exact OP and TC \\ \hline
\S & \ref{sec:asympTC} & Asymptotic OP and TC \\
Prop. & \ref{pro:asymoptc} & Asymptotic OP and TC \\
Rem. & \ref{rem:spherepack} & TC as sphere packing \\
Fig. & \ref{fig:spherepack} & Taylor series expansions used in asymptotic TC \\ \hline
\S & \ref{sec:ubTC} & Upper bound on TC and lower bound on OP \\
Def. & \ref{def:domint} & Dominant interferers and interference \\
Prop. & \ref{pro:oplb} & OP LB and TC UB \\
Rem. & \ref{rem:dommaxequiv} & Dominant and maximum interferers \\
Cor. & \ref{cor:bndexcomp} & OP and TC bounds ($\epsilon = 0$, $N=0$, $\delta = \frac{1}{2}$)   \\
Fig. & \ref{fig:exactboundcomp} & OP and TC exact results vs.\ bounds \\
Fig. & \ref{fig:FZLB} & Lower bound on normal CDF \\ \hline
\S & \ref{sec:tpandtc} & Throughput (TP) and TC \\
Def. & \ref{def:tp} & MAC layer TP \\
Prop. & \ref{pro:slotaloha} & Slotted Aloha TP and OP \\
Fig. & \ref{fig:slotaloha} & Slotted Aloha TP and OP \\
Prop. & \ref{pro:tp} & MAC layer TP UB \\
Prop. & \ref{pro:tcopt} & TC is constrained TP maximization \\
Prop. & \ref{pro:optTPandTC} & Maximum TP equals maximum TC \\
Cor. & \ref{cor:optTPandTC} & Maximizing TP and TC UBs \\
Fig. & \ref{fig:tp12} & Throughput and TC UBs \\ \hline
\S & \ref{sec:lbTC} & Lower bounds on TC and upper bounds on OP \\
Prop. & \ref{pro:domnonop} & Exact OP in terms of OP LB \\
Prop. & \ref{pro:oplbmar} & Markov inequality OP UB \\
Prop. & \ref{pro:oplbcheb} & Chebychev inequality OP UB \\
Prop. & \ref{pro:oplbcher} & Chernoff inequality OP UB \\
Fig. & \ref{fig:optcmarkchebcher} & The three OP upper and TC LBs \\ \hline
\end{tabular*}
\label{tab:cha3}
\end{table}

\begin{table}[ht] \footnotesize
\caption{Results in Ch.\  \ref{cha:modenh}: Extensions of the basic model}
\begin{tabular*}{\textwidth}{@{\extracolsep{\fill}}|c|r|l|} \hline
\S & \ref{sec:fading} & Channel fading \\
Def. & \ref{def:fad} & SINR under fading \\
Rem. & \ref{rem:fadsigint} & Signal and interference fading coefficients \\ \hline
\S\S & \ref{ssec:fadexact} & Exact OP and TC with fading \\
Prop. & \ref{pro:lapintfad} & LT of the interference \\
Prop. & \ref{pro:optcrayfadsig} & OP and TC under Rayleigh signal fading \\
Lem. & \ref{lem:expmom} & Moments of exponential RV \\
Cor. & \ref{cor:optcrayfadall} & OP and TC under Rayleigh fading \\
Cor. & \ref{cor:optcrayfadlev} & OP and TC under Rayleigh fading ($\delta=\frac{1}{2}$, $N=0$) \\
Fig. & \ref{fig:fadnoncomp} & OP and TC under Rayleigh fading \\ \hline
\S\S & \ref{ssec:fadasymp} & Asymptotic OP and TC with fading \\
Rem. & \ref{rem:MPPP} & Marked PPP (MPPP) \\
Thm. & \ref{thm:mark} & PPP marking theorem \\
Prop. & \ref{pro:voidnonhomo} & Void probability of the non-homogeneous MPPP \\
Prop. & \ref{pro:markdistmap} & MPPP distance and interference mapping \\
Prop. & \ref{pro:fadintserrep} & Shot noise series expansion \\
Thm. & \ref{thm:fadstab} & Interference with fading is stable \\
Rem. & \ref{rem:fadint} & Interference and fading moments \\
Rem. & \ref{rem:fadlowout} & Fading and outage with no interference \\
Prop. & \ref{pro:asymoptcfad} & Asymptotic OP and TC under fading \\
Rem. & \ref{rem:fadoptc} & OP and fading moments \\
Cor. & \ref{cor:optcfadnoncomp} & Fading degrades performance \\
Fig. & \ref{fig:fadnonasymp} & Fading degrades performance \\ \hline
\S\S & \ref{ssec:fadlb} & Lower (upper) bound on OP (TC) with fading \\
Def. & \ref{def:domintfad} & Dominant interferers and interference \\
Prop. & \ref{pro:fadoplb} & OP LB \\
Fig. & \ref{fig:mgffadlb} & MGF for the RV $-\hsf^{-\delta}$ \\
Fig. & \ref{fig:fadexasylb} & Exact, asymptotic, bound OP and TC under fading \\ \hline
\S & \ref{sec:vardist} & Variable link distances (VLD) \\
Def. & \ref{def:vardistsinr} & SINR for VLD \\
Def. & \ref{def:opvld} & OP for VLD \\
Prop. & \ref{prop:vardistoptc} & Asymptotic OP and TC ($\epsilon = 0$, $N=0$) \\
Prop. & \ref{pro:oplbvld} & OP LB as an MGF \\
Prop. & \ref{pro:nnvld} & Nearest neighbor RV characteristics \\
Cor. & \ref{cor:vardistexactopnn} & Exact OP ($\epsilon = 0$, $N=0$, $\delta = \frac{1}{2}$) \\
Fig. & \ref{fig:vld} & Exact, LB, and asymptotic OP \\ \hline
\S & \ref{sec:multihop} & Multihop TC \\ \hline
Fig. & \ref{fig:MH-model} & The multihop TC model with $M=3$ \\
Def. & \ref{def:MHTC} & Multihop TC \\
Prop. & \ref{pro:MHTCineq} & Multihop TC inequality \\
Prop. & \ref{pro:MHTC-UB} & Multihop TC UB \\
Fig. & \ref{fig:MHTC_Fig1} & Multihop TC and its UB vs.\ allowed number of Tx attempts $A$ \\
Def. & \ref{def:mstar} & Optimal number of hops \\
Prop. & \ref{pro:mstar} & Optimal number of hops \\
Cor. & \ref{cor:mstaralpha3} & Optimal number of hops for $\alpha = 3$ \\
Cor. & \ref{cor:mstaralpha4} & Optimal number of hops for $\alpha = 4$ \\ 
Fig. & \ref{fig:MHTC_CvsM} & $\lambda_{\rm mh}(\lambda,U,M,A)$ vs.\ $M$ for $A \in \{6,12\}$ \\ \hline
\end{tabular*}
\label{tab:cha4}
\end{table}

\begin{table}[ht] \footnotesize
\caption{Results in Ch.\  \ref{cha:destech}: Design techniques for wireless networks}
\begin{tabular*}{\textwidth}{@{\extracolsep{\fill}}|c|r|l|} \hline
\S & \ref{sec:specman}& Spectrum management \\
Ass. & \ref{ass:specman}& Random band selection \\
Rem. & \ref{rem:spec1} & Thinned interference seen by reference Rx \\
Def. & \ref{def:specdefs} & Noise, SINR, SNR, capacity, spectral efficiency \\
Def. & \ref{def:outageequiv} & OP and TC under multiple bands \\
Prop. & \ref{pro:tcspec} & TC under multiple bands \\
Rem. & \ref{rem:optnumbands} & Optimal number of bands independent of target OP \\
Def. & \ref{def:ebno} & Energy per bit anad receive SNR \\
Rem. & \ref{rem:Mdisccont} & Relaxation of integrality constraint \\
Cor. & \ref{cor:spectcebno} & TC under multiple bands \\
Prop. & \ref{prop:ebnomin} & Minimum energy per bit required for solution \\
Rem. & \ref{rem:lowsnrregime} & Low SNR regime \\
Fig. & \ref{fig:specomegam} & Spectral component of performance \\
Prop. & \ref{pro:specoptnu} & Optimal spectral efficiency \\
Fig. & \ref{fig:specnuoptebno} & Optimal spectral efficiency \\
Cor. & \ref{cor:specasymoptnu} & Asymptotic optimal spectral efficiency (high SNR) \\
Fig. & \ref{fig:FigSpecOptNuAsyHighSNR} & Optimal high SNR spectral efficiency \\
Prop. & \ref{pro:specasymnulowsnr} & Maximum possible spectral efficiency (low SNR) \\
Fig. & \ref{fig:FigSpecOptNuAsyLowSNR} & Optimal spectral efficiency \\  \hline
\S & \ref{sec:intcan} & Interference cancellation (IC) \\
Def. & \ref{def:icparam} & The $(\kappa,K,P_{\rm min})$ IC model \\
Def. & \ref{def:sicsinr} & SINR at a $(\kappa,K,P_{\rm min})$ IC capable reference Rx \\
Lem. & \ref{lem:sictmax} & Constraint mapping \\
Lem. & \ref{lem:sicdomcan} & Partially cancellable / uncancellable nodes \\
Fig. & \ref{fig:sicthresh}& Illustration of the thresholds \\
Thm. & \ref{thm:eucdistnn} & Ordered distances marginal distributions \\
Cor. & \ref{cor:distkpi11nn} & Ordered distances in $\Pi_{1,1}$ marginal distributions \\
Fig. & \ref{fig:gammadist} & PDFs for the ordered distances in $\Pi_{1,1}$ \\
Prop. & \ref{pro:sicmain} & OP LB under $(\kappa,K,P_{\rm min})$ IC Rx model \\
Fig. & \ref{fig:sicoplb1} & OP vs.\ $\lambda$  \\
Fig. & \ref{fig:sicoplb2} & OP vs.\ $\kappa$ and $P_{\rm min}$ \\
Rem. & \ref{rem:othericbounds} & Other bounds on OP and TC under IC \\  \hline
\S & \ref{sec:sched} & Fading threshold scheduling (FTS) \\
Def. & \ref{def:sched}& Fading coefficients, signal interference, and SINR \\
Def. & \ref{def:schedoptptc} & OP and TP \\
Rem. & \ref{rem:desobj} & Quantity vs.\ quality of transmissions through FTS \\
Prop. & \ref{pro:asymopsched} & Asymptotic OP under FTS \\
Prop. & \ref{pro:asymtpsched} & Asymptotic TP under FTS \\
Fig. & \ref{fig:schedasymp} & Asymptotic TP and optimal FTS threshold  \\
Prop. & \ref{pro:schedasympcomp} & FTS exploits fading to improve performance \\
Fig. & \ref{fig:schedasympcomp} & Throughput under FTS \\
Prop. & \ref{pro:schedlb} & OP LB and TP UB \\
Fig. & \ref{fig:schedlb} & Asymptotic TP and UB \\  \hline
\S & \ref{sec:power} & Fractional power control (FPC) \\
Def. & \ref{def:power}& Transmission powers and SINR \\
Lem. & \ref{lem:fpcvar} & Transmitted power RV moments under FPC \\
Fig. & \ref{fig:fpcvar} & Variance of the transmitted power under FPC \\
Rem. & \ref{rem:fpclowout}& Fading and outage with no interference \\
Prop. & \ref{pro:fpcasymp} & Asymptotic OP and TC under FPC \\
Prop. & \ref{pro:fpcasympoptf} & Asymptotic optimality of $f=1/2$ \\
Rem. & \ref{rem:optfpcvspowvar} & Optimal FPC may incur large power variance \\
Def. & \ref{def:domintfpc} & Dominant and maximum interferers \\
Prop. & \ref{pro:fpcoplb} & The OP LB under FPC \\
Fig. & \ref{fig:FPC1} & OP vs.\ FPC exponent $f$ \\
Fig. & \ref{fig:FPC2} & OP vs.\ intensity of attempted transmissions $\lambda$ \\  \hline
\end{tabular*}
\label{tab:cha5}
\end{table}

\begin{table}[ht] \footnotesize
\caption{Results in Ch.\ \ref{cha:MIMO}: Multiple antennas}
\begin{tabular*}{\textwidth}{@{\extracolsep{\fill}}|c|r|l|} \hline
\S & \ref{sec:MIMOintro} & MIMO with interference \\ \hline
\S & \ref{sec:model} & Categorizing MIMO in decentralized networks\\ 
Fig. & \ref{fig:MIMO} & Receive diversity, SM, and SDMA \\ 
\S\S & \ref{ssec:modelsinglestream} & Single stream techniques \\
Def. & \ref{def:MIMO-SINR} & MIMO single stream SINR\\ 
\S\S & \ref{ssec:modelmultistream} & Multi-stream models: spatial multiplexing and SDMA \\
Def. & \ref{def:MIMO-K-SINR} & Multi-stream OP and OCD \\
Def. & \ref{def:multistreamtc} & Multi-stream TC \\ \hline 
\S & \ref{sec:singlestream} & Single stream MIMO TC results\\ 
\S\S & \ref{ssec:diversity} & Diversity \\
Def. & \ref{def:optdiversityfilter} & Single stream MIMO optimal linear diversity filters \\
Thm. & \ref{thm:MRC} & Maximal Ratio Combiner (MRC) \\
Rem. & \ref{rem:intdistunchanged} & Interference distribution is unchanged \\
Rem. & \ref{rem:DivSigDis} & Signal distribution \\
Prop. & \ref{pro:OP-MRC} & OP with MRC \\
Prop. & \ref{pro:TC-MRC} & TC with MRC \\
Rem. & \ref{rem:MRCsublineargain} & Sublinear gain of MRC in $\nr$ \\
Prop. & \ref{pro:OCDeigenbeamforming} & OCD of $\nt \times \nr$ eigenbeamforming \\
Rem. & \ref{rem:eigen} & TC scaling with $\nt \times \nr$ eigenbeamforming \\ 
\S\S & \ref{ssec:MIMO-IC} & Interference cancellation (IC) \\
Def. & \ref{def:PZF} & Partial zero forcing (PZF) Rx \\
Prop. & \ref{pro:pzfsinr} & PZF SINR \\
Prop. & \ref{pro:oppzfz} & OP for PZF-$z$ \\
Prop. & \ref{pro:tcpzflb} & TC for PZF LB \\
Rem. & \ref{rem:pzflinearscaling} & PZF linear scaling in $\nr$ \\
Def. & \ref{def:mmserxfilter} & MMSE Rx filter \\
Rem. & \ref{rem:mmserxfilter} & MMSE Rx filter \\
Prop. & \ref{pro:MMSE-UB} & TC UB for PZF and MMSE \\
Cor. & \ref{cor:tcmrcub} & TC UB for MRC \\ 
Fig. & \ref{fig:PZF} & OCD vs.\ $\nr$ for PZF, MMSE, ZF, and MRC, for $\alpha = 4$ \\ \hline
\S & \ref{sec:multistream} & Main results on multiple stream TC \\ 
\S\S & \ref{ssec:SM} & Spatial multiplexing \\
Prop. & \ref{pro:SM-MRC} & SM with MRC \\
Prop. & \ref{pro:SM-ZF} & SM with ZF Rx \\
Prop. & \ref{pro:SM-PZF-HighSNR} & SM with PZF at high SNR \\
Cor. & \ref{cor:pzfrxnoic} & PZF Rx without cancellation \\
Cor. & \ref{cor:pzfrxstrongestic} & PZF Rx cancelling strongest interferer \\
Rem. & \ref{rem:smmmserx} & SM with the MMSE Rx \\
Prop. & \ref{pro:SM-BF-PZF} & SM with multimode beamforming and PZF receivers \\
Prop. & \ref{pro:SM-BLAST} & SM with the BLAST architecture \\ 
\S\S & \ref{ssec:SDMA} & Space division multiple access (SDMA) \\
Prop. & \ref{pro:DPC-TC} & Bounds on multistream DPC TC with MRC \\
Prop. & \ref{pro:DPC-scaling} & SDMA scaling laws with a DPC Tx and MRC Rx \\
Cor. & \ref{cor:sdmadpcmrc} & TC scaling and $K^*$ for SDMA with DPC Tx and MRC Rx \\
Prop. & \ref{pro:DPC-MISO} & TC with DPC Tx and single antenna Rx \\ \hline
\S & \ref{sec:takeaways} & Practical issues and further research \\
\S\S & \ref{ssec:designinsights} & Summary of main design insight \\
\S\S & \ref{ssec:caveats} & Caveats and practical issue \\ \hline
\end{tabular*}
\label{tab:cha6}
\end{table}

\begin{acknowledgements}
The authors wish to acknowledge several individuals with whom we have collaborated in the development of the transmission capacity framework over the last eight years.  First, many of the ideas and results in this volume would not exist without our active collaboration with Dr. Nihar Jindal (now at Broadcom).  In particular, he was lead author on several journal papers (\cite{JinAnd2008,JinWeb2008,JinAnd2011}) on the topics discussed in Chapters 5 and 6.

The authors wish to also acknowledge the seminal contributions made by Prof.\ de Veciana (UT Austin) and his former Ph.D.\ student Dr.\ Xiangying Yang (Apple) in two early papers on transmission capacity.  The authors learned stochastic geometry largely via a course (and the accompanying course notes) taught by Dr.\ de Veciana in Spring 2003 --- this course led directly to the research collaboration on transmission capacity \cite{WebYan2005,WebAnd2007a}.  Other collaborators and colleagues who provided valuable input and/or feedback to this monograph are as follows, in alphabetical order:
\begin{itemize}
\item Prof.\ Fran\c{c}ois Baccelli (ENS and INRIA)
\item Prof.\ Radha Krishna Ganti (IIT Madras)
\item Prof.\ Martin Haenggi (University of Notre Dame)
\item Prof.\ Robert W.\ Heath Jr.\ (The University of Texas at Austin)
\item Prof.\ Kaibin Huang (Yonsei University)
\item Mr.\ Andrew Hunter (The University of Texas at Austin)
\item Prof.\ Marios Kountouris (Sup\'{e}lec)
\item Dr.\ Chun-Hung Liu (The University of Texas at Austin)
\item Dr.\ Raymond Louie (Univ. of Sydney)
\item Prof.\ Rahul Vaze (TIFR Mumbai)
\end{itemize}
\end{acknowledgements}
The support of a 2006--2009 NSF grant to Weber, Andrews, and Jindal is gratefully acknowledged, as is NSF CAREER CCF-0643508 (Andrews).  This volume would not have come about without the intellectual environment and financial support provided by the DARPA IT-MANET project (\#W911NF-07-1-0028).

\bibliographystyle{ieeesort}
\bibliography{WeberAndrews-Final}

\end{document}